\documentclass[runningheads,envcountsect,envcountsame]{llncs}

\pdfoutput=1

\usepackage{standalone}

\usepackage{ifpdf}

\ifpdf
\else
\usepackage{gastex}
\fi

\usepackage{amsmath}
\usepackage{amsfonts}
\usepackage{amssymb}

\newtheorem{maintheorem}{Theorem}

\usepackage[utf8]{inputenc}

\usepackage{microtype}
\usepackage{ellipsis}

\usepackage{url}

\usepackage[rflt]{floatflt}

\usepackage{numprint}
\npdecimalsign{.}
\npthousandsep{,}

\usepackage{hyperref}

\usepackage{amssymb}
\usepackage{tikz}

\newcommand{\accdiaplus}{%
\ensuremath{\diamondplus}}

\newcommand{\rawdiaplus}{%
  \begin{tikzpicture}
    \useasboundingbox (-0.7ex, -0.9ex) rectangle (0.7ex, 0.9ex);
    \node (w) at (-0.7ex,0) {};
    \node (e) at (+0.7ex,0) {};
    \node (s) at (0,-0.9ex) {};
    \node (n) at (0,+0.9ex) {};
    \draw (n.center) -- (e.center) -- (s.center) -- (w.center) -- (n.center);
    \draw (n.center) -- (s.center);
    \draw (e.center) -- (w.center);
  \end{tikzpicture}}

\newsavebox{\diamondplusbox}
\savebox{\diamondplusbox}{\rawdiaplus}
\DeclareRobustCommand{\diamondplus}{\texorpdfstring{\mathop{\raisebox{-0.25ex}{\resizebox{1.45ex}{!}{\usebox{\diamondplusbox}}}}\nolimits}{diamondplus}}

\newcommand{\rawdiaminus}{%
  \begin{tikzpicture}
    \useasboundingbox (-0.7ex, -0.9ex) rectangle (0.7ex, 0.9ex);
    \node (w) at (-0.7ex,0) {};
    \node (e) at (+0.7ex,0) {};
    \node (s) at (0,-0.9ex) {};
    \node (n) at (0,+0.9ex) {};
    \draw (n.center) -- (e.center) -- (s.center) -- (w.center) -- (n.center);
    \draw (e.center) -- (w.center);
  \end{tikzpicture}}

\newsavebox{\diamondminusbox}
\savebox{\diamondminusbox}{\rawdiaminus}

\newcommand{\tudparagraph}[2]{%
\vspace*{#1}

\noindent
{\bf #2}
}

\newcommand{\cE}{\mathcal{E}}

\newcommand{\cG}{\mathcal{G}}

\newcommand{\cI}{\mathcal{I}}

\newcommand{\cK}{\mathcal{K}}

\newcommand{\cM}{\mathcal{M}}
\newcommand{\cN}{\mathcal{N}}
\newcommand{\cO}{\mathcal{O}}

\newcommand{\cR}{\mathcal{R}}

\newcommand{\<}{\langle}
\renewcommand{\>}{\rangle}

\newcommand{\eqdef}{\ensuremath{\stackrel{\text{\tiny def}}{=}}}

\newcommand{\Ende}{\hfill ${\scriptscriptstyle \blacksquare}$}

\renewcommand{\Pr}{\mathrm{Pr}}

\newcommand{\threshold}{\vartheta}

\newcommand{\SchedThreshold}[1]{\mathit{ThresAlgo}[#1]}

\newcommand{\Exp}[2]{\mathrm{E}^{#1}_{#2}}

\newcommand{\CExp}[1]{\mathbb{CE}^{#1}}

\newcommand{\CExpState}[2]{\mathbb{CE}^{#1}_{#2}}

\newcommand{\ExpRew}[2]{\mathbb{E}^{#1}_{#2}}

\newcommand{\probability}{\mathrm{prob}}

\newcommand{\sinit}{s_{\mathit{\scriptscriptstyle init}}}

\newcommand{\Act}{\mathit{Act}}

\newcommand{\act}{\alpha}

\newcommand{\action}{\mathit{action}}

\newcommand{\Size}{\mathit{size}}

\newcommand{\fpath}{\pi}
\newcommand{\finpath}{\varrho}
\newcommand{\cycle}{\xi}

\newcommand{\infpath}{\varsigma}

\newcommand{\first}{\mathit{first}}
\newcommand{\last}{\mathit{last}}

\newcommand{\prefix}[2]{\mathit{pref}(#1,#2)}

\newcommand{\extstate}[4]{t_{#1,#4}}

\newcommand{\Cyl}{\mathit{Cyl}}

\newcommand{\turning}{\Re}
\newcommand{\saturation}{\wp}
\newcommand{\saturationNaive}{\saturation_0}

\newcommand{\md}{\mathit{md}}

\newcommand{\Sched}{\mathit{Sched}}

\newcommand{\SchedOpt}{\mathit{SchedOpt}}

\newcommand{\sched}{\mathfrak{S}}
\newcommand{\tsched}{\mathfrak{T}}
\newcommand{\usched}{\mathfrak{U}}
\newcommand{\vsched}{\mathfrak{V}}

\newcommand{\psched}{\mathfrak{P}}

\newcommand{\maxsched}{\mathfrak{M}}

\newcommand{\ub}{\mathrm{ub}}

\newcommand{\residual}[2]{#1 {\uparrow} {#2}}

\newcommand{\redefresidual}[3]{#1 \lhd_{#2} #3}

\newcommand{\REC}{\mathit{REC}}

\newcommand{\MEC}{\mathit{MEC}}

\newcommand{\simMEC}{\sim_{\text{MEC}}}

\newcommand{\rew}{\mathit{rew}}

\newcommand{\rewparam}{\mathfrak{r}}

\newcommand{\ActEC}{\mathfrak{A}}

\newcommand{\lift}[1]{\mathit{lift}(#1)}

\newcommand{\acc}[1]{#1_{\mathit{acc}}}

\newcommand{\Hut}[1]{\hat{#1}}
\newcommand{\sHutinit}{\Hut{s}_{\mathit{\scriptscriptstyle init}}}

\newcommand{\afterF}[1]{#1^{\mathit{F}}}
\newcommand{\afterG}[1]{#1^{\mathit{G}}}

\newcommand{\goal}{\mathit{goal}}
\newcommand{\fail}{\mathit{fail}}
\newcommand{\final}{\mathit{final}}
\newcommand{\Fail}{\mathit{Fail}}

\newcommand{\accept}{\mathit{accept}}
\newcommand{\reject}{\mathit{reject}}

\newcommand{\neXt}{\bigcirc}

\DeclareMathOperator{\Until}{\ensuremath{\mathrm{U}}}

\newcommand{\Nat}{\mathbb{N}}
\newcommand{\Rational}{\mathbb{Q}}
\newcommand{\Real}{\mathbb{R}}
\newcommand{\Integer}{\mathbb{Z}}

\newcommand{\overto}[1]{\stackrel{#1}{\longrightarrow}}

\newcommand{\CiteAppendix}[1]{}

\newcommand{\includeGastex}[1]{%
\includegraphics{Gastex/#1.pdf}%
}

\newcommand{\numStates}[1]{\npfourdigitnosep\numprint{#1}}
\newcommand{\numSec}[1]{\npfourdigitnosep\numprint{#1}}
\newcommand{\numValue}[1]{\npfourdigitnosep\numprint{#1}}

\newcommand{\tableC}[1]{\multicolumn{1}{c}{#1}}
\newcommand{\tableCL}[1]{\multicolumn{1}{c|}{#1}}

\newcommand{\consensus}{\textsc{Consensus}}
\newcommand{\wlan}{\textsc{Wlan}}

\renewcommand{\CiteAppendix}[1]{#1}

\usepackage{hyperref}

\title{Maximizing the Conditional Expected Reward for Reaching the Goal\\
{\normalsize (extended version)}%
\thanks{%
The authors are supported by the DFG through
        the collaborative research centre HAEC (SFB 912),
        the Excellence Initiative by the German Federal and State Governments (cluster of excellence cfAED),
        the Research Training Group QuantLA (GRK 1763), and
        the DFG-project BA-1679/11-1.}}

\author{Christel Baier,
  Joachim Klein,
  Sascha Kl\"uppelholz,
  Sascha Wunderlich}
\institute{%
  Institute for Theoretical Computer Science\\
  Technische Universit\"at Dresden, Germany}

\titlerunning{Maximizing the Conditional Expected Reward for Reaching the Goal}

\begin{document}

\maketitle
\vspace{-.25cm}

\begin{abstract}
The paper addresses the problem of computing maximal conditional
expected accumulated rewards until reaching a target state
(briefly called \emph{maximal conditional expectations})
in finite-state Markov decision processes where the
condition is given as a reachability constraint.
Conditional expectations of this type
can, e.g., stand for the maximal expected termination time of
probabilistic programs with non-determinism,
under the condition that the program eventually terminates,
or for the worst-case expected penalty to be paid, assuming that
at least three deadlines are missed.
The main results of the paper are (i) a polynomial-time algorithm to
check the finiteness of maximal conditional expectations,
(ii) PSPACE-completeness  for the threshold problem
in acyclic Markov decision processes
where the task is to check whether the maximal conditional expectation
exceeds a given threshold,
(iii) an exponential-time algorithm for the threshold problem
in the general (cyclic) case,
and (iv) an exponential-time algorithm for computing the maximal
conditional expectation and an optimal scheduler.
\end{abstract}

\section{Introduction}
\label{sec:introduction}

Stochastic shortest (or longest) path problems
are a prominent class of optimization problems
where the task is to find a policy for traversing a probabilistic
graph structure such that the expected value of the generated
paths satisfying a certain objective is minimal (or maximal).
In the classical setting
(see e.g.~\cite{BerTsi91,Puterman,deAlf99,Kallenberg}),
the underlying graph structure
is given by a finite-state Markov decision process (MDP),
i.e., a state-transition graph with nondeterministic choices
between several actions for each of its non-terminal states,
probability distributions specifying the probabilities
for the successor states for each state-action pair
and a reward function that assigns rational values to the
state-action pairs.
The stochastic shortest (longest) path problem  asks
to find a scheduler, i.e., a function that resolves
the nondeterministic choices, possibly in a history-dependent way,
which minimizes (maximizes) the expected accumulated reward until
reaching a goal state.
To ensure the existence of the expectation for given schedulers,
one often assumes that the given MDP is contracting,
i.e., the goal is reached almost surely under all schedulers, in which
case the optimal expected accumulated reward is achieved
by a memoryless deterministic scheduler that optimizes
the expectation from each state and is computable using a linear program
with one variable per state (see e.g.~\cite{Kallenberg}).
The contraction assumption can be relaxed by requiring
the existence of at least one
scheduler that reaches the goal almost surely and taking the extremum
over all those schedulers \cite{BerTsi91,deAlf99,BerYu16}.
These algorithms and corresponding value or policy iteration approaches
have been implemented in various tools and used in many application areas.

The restriction to schedulers that reach the goal almost surely, however,
limits the applicability and significance of the results.
First, the known algorithms for computing extremal expected accumulated
rewards are not applicable for models where
the probability for never visiting a goal state is positive
under each scheduler.
Second, statements about the expected rewards for schedulers that reach
the goal with probability 1 are not sufficient to draw any conclusion
for the best- or worst-case behavior, if there exist schedulers that
miss the goal with positive probability.
This motivates the consideration of \emph{conditional stochastic path problems}
where the task is to compute the optimal expected accumulated
reward until reaching a goal state, under the condition that a goal state
will indeed be reached and where the extrema are taken over all schedulers
that reach the goal with positive probability.
More precisely, we address here a slightly more general
problem where we are given two sets
$F$ and $G$ of states in an MDP $\cM$ with non-negative integer rewards
and ask for the maximal expected accumulated
reward until reaching $F$, under the condition that $G$ will be
visited (denoted $\ExpRew{\max}{\cM,\sinit}(\accdiaplus F|\Diamond G)$
where $\sinit$ is the initial state of $\cM$).
Computation schemes for conditional expectations of this type can, e.g.,
be used to answer the following questions (assuming the underlying model
is a finite-state MDP):
\begin{enumerate}
\item [(Q1)]
   What is the maximal termination time of a probabilistic
   and nondeterministic program,
   under the condition that the program indeed terminates?
\item [(Q2)]
   What are the maximal expected costs of the repair mechanisms
   that are triggered
   in cases where a specific failure scenario occurs,
   under the condition that the failure scenario indeed occurs?
\item [(Q3)]
   What is the maximal energy consumption, under the condition that
   all jobs of a given list will be successfully executed within one hour?
\end{enumerate}
The relevance of question (Q1) and related problems
becomes clear from the work
  \cite{GretzKatMcIver14,JKKGMcI15,KatoenGJKO15,BartheEFH16,ChatterjeeFG16}
on the semantics of probabilistic programs
where no guarantees for almost-sure termination can be given.
Question (Q2) is natural for a worst-case analysis of resilient systems
or other types of systems
where conditional probabilities serve to provide performance guarantees on the
protocols triggered in exceptional cases that appear with positive, but
low probability.
Question (Q3) is typical when the task is to study the trade-off
between cost and utility functions (see e.g.~\cite{BDKKW14-Festschrift}).
Given the work on anonymity and related notions
for information leakage using conditional probabilities
in MDP-like models
\cite{AndPalRosSok11,ChaPalBraun16} or the formalization of
posterior vulnerability as an expectation \cite{AlvChaMcIMorPalSm16},
the concept of conditional accumulated
excepted rewards might also be useful to specify the degree of
protection of secret data or to study the trade-off
between privacy and utility, e.g., using gain functions %
\cite{AlChaPalSmith12,AlAndChaDegPal15}.
Other areas where conditional expectations play
a crucial role are risk management where the conditional value-at-risk
is used to formalize the
expected loss under the assumption that very large losses occur
\cite{Uryasev00,AcerbiTasche02}
or regression analysis where conditional expectations serve to
predict the relation between random variables~\cite{SeberLee}.

\begin{figure}[t]
\vspace*{-2.5ex}
   \includeGastex{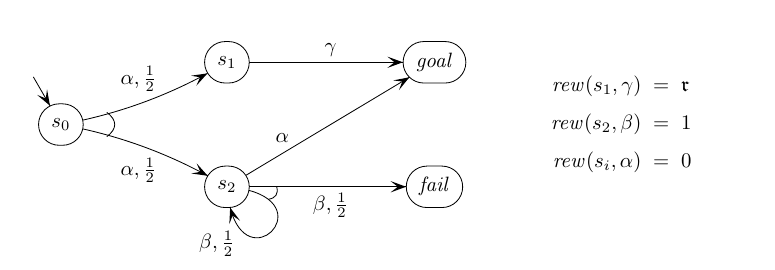}
\vspace*{-5ex}
\caption{MDP $\cM[\rewparam]$ for Example \ref{example:running-example}}
\label{fig:running-example}
\vspace*{-2.5ex}
\end{figure}

\begin{example}
\label{example:running-example}
To illustrate the challenges for designing algorithms to compute
maximal conditional expectations we regard the
MDP $\cM[\rewparam]$ shown in Figure \ref{fig:running-example}.
The reward of the state-action
pair $(s_1,\gamma)$ is given by a reward parameter $\rewparam \in \Nat$.
Let $\sinit=s_0$ be the initial state and $F=G=\{\goal\}$.
The only nondeterministic choice is in state $s_2$, while
states $s_0$ and $s_1$ behave purely probabilistic and $\goal$ and $\fail$
are trap states.
Given a scheduler $\sched$, we write
$\CExp{\sched}$ for the conditional expectation
$\ExpRew{\sched}{\cM[\rewparam],s_0}(\accdiaplus \goal |\Diamond \goal)$.
  (See also Section \ref{sec:preliminaries} for our notations.)
For the two memoryless schedulers
that choose $\alpha$~resp.~$\beta$ in
state $s_2$ we have:
$$
  \CExp{\alpha} \ \ = \ \
  \frac{\frac{1}{2} \cdot \rewparam 
  \ + \ \frac{1}{2}\cdot 0}{\frac{1}{2}+\frac{1}{2}}
  \ \ = \ \ \frac{\rewparam}{2}
  \ \ \ \ \ \text{and} \ \ \ \ \
  \CExp{\beta} \ \ = \ \
  \frac{\frac{1}{2} \cdot \rewparam \ + \ 0}{\frac{1}{2}+0}
  \ \ = \ \ \rewparam
$$
We now regard the schedulers $\sched_n$ for $n =1,2,\ldots$
that choose $\beta$ for the first
$n$ visits of $s_2$ and action $\alpha$ for the $(n{+}1)$-st visit of $s_2$.
Then:
$$
  \CExp{\sched_n} \ \ = \ \
  \frac{\frac{1}{2} \cdot \rewparam 
  \ + \ \frac{1}{2}\cdot \frac{1}{2^n} \cdot n}
       {\frac{1}{2} \ + \ \frac{1}{2}\cdot \frac{1}{2^n}}
  \ \ = \ \
  \rewparam \ + \ \frac{n-\rewparam}{2^n{+}1}
$$
Thus, $\CExp{\sched_n} > \CExp{\beta}$ iff $n > \rewparam$, and the
maximum is achieved for $n=\rewparam{+}2$.

This example illustrates three phenomena
that distinguish conditional and unconditional expected accumulated
rewards and make reasoning about maximal conditional expectations
harder than about unconditional ones.
First, optimal schedulers for $\cM[\rewparam]$ need a counter for the number
of visits in state $s_2$.
Hence, memoryless schedulers are not powerful enough to maximize
the conditional expectation.
Second, while the maximal conditional expectation for $\cM[\rewparam]$ 
with initial state $\sinit = s_0$ is finite, 
the maximal conditional expectation
for $\cM[\rewparam]$ with starting state $s_2$ is infinite as:
\begin{center}
  $\sup\limits_{n\in \Nat} \
  \ExpRew{\sched_n}{\cM[\rewparam],s_2}(\accdiaplus \goal | \Diamond \goal)
  \ \ = \ \
  \sup\limits_{n\in \Nat}
   \frac{\ \frac{n}{2^{n}} \ }{\frac{1}{2^n}} \ \ = \ \ \infty$
\end{center}
Third, as $\sched_2$ maximizes the conditional expected accumulated
reward for $\rewparam=0$, while $\sched_3$ is optimal
for $\rewparam=1$, optimal decisions for paths ending in state $s_2$
depend on the reward value $r$ of the $\gamma$-transition from state $s_1$,
although state $s_1$ is not reachable from $s_2$.
Thus,
optimal decisions for a path $\fpath$
do not only depend on the past (given by $\fpath$)
and possible future
(given by the sub-MDP that is reachable from
$\fpath$'s last state),
but require global reasoning.
\Ende
\end{example}

\noindent
The main results of this paper are the following theorems.
We write $\CExp{\max}$ for the maximal conditional expectation,
i.e., the supremum of the conditional expectations
$\ExpRew{\sched}{\cM,\sinit}(\accdiaplus F|\Diamond G)$,
when ranging over all schedulers $\sched$ where
$\Pr^{\sched}_{\cM,\sinit}(\Diamond G)$ is positive
and $\Pr^{\sched}_{\cM,\sinit}(\Diamond F |\Diamond G)=1$.
  (See also Section \ref{sec:preliminaries} for our notations.)
\begin{maintheorem}[Checking finiteness and upper bound]
 \label{thm:finiteness}
   There is a polynomial-time algorithm that checks if
   $\CExp{\max}$ is finite. If so, an upper bound
   $\CExp{\ub}$ for $\CExp{\max}$ is computable in pseudo-polynomial time
   for the general case and in polynomial time if
   $F=G$ and
   $\Pr^{\min}_{\cM,s}(\Diamond G) >0$ for all states
   $s$ with $s \models \exists \Diamond G$.
\end{maintheorem}
The threshold problem asks whether the maximal conditional expectation
exceeds or misses a given rational threshold $\threshold$.

\begin{maintheorem}[Threshold problem]
 \label{thm:threshold-problem}
  The problem ``does $\CExp{\max} \bowtie \threshold$ hold?''
  (where $\bowtie \ \in  \{>,\geqslant,<,\leqslant\}$)
  is PSPACE-hard and
  solvable in exponential time.
  \footnote{In an earlier version of this paper, the problem was wrongly claimed to be solvable in pseudo-polynomial time (see erratum \url{https://wwwtcs.inf.tu-dresden.de/ALGI/PUB/erratum_pseudo_polynomial.pdf}.}
  It is PSPACE-complete for acyclic MDPs.
\end{maintheorem}
For the computation of an optimal scheduler, we suggest an iterative
scheduler-improvement algorithm that interleaves calls of the
threshold algorithm with linear programming techniques to handle
zero-reward actions. This yields:
\begin{maintheorem}[Computing optimal schedulers]
  \label{thm:computing-cexpmax}
  The value $\CExp{\max}$ and an optimal scheduler $\sched$
  are computable in exponential time.
\end{maintheorem}
  Algorithms for checking finiteness and computing an upper bound
  (Theorem~\ref{thm:finiteness}) will be sketched in
  Sections \ref{sec:finiteness}.
  Section \ref{sec:threshold} presents
  an exponential threshold algorithm and a
  polynomially space-bounded algorithm for acyclic MDPs
  (Theorem~\ref{thm:threshold-problem}) as well as
  an exponential-time computation scheme for the
  construction of an optimal scheduler (Theorem \ref{thm:computing-cexpmax}).
  Further details,
  soundness proofs and a proof for the PSPACE-hardness
  as stated in Theorem \ref{thm:threshold-problem}
  can be found
\CiteAppendix{in the appendix. }%
   The general feasibility of the algorithms will be shown
   by experimental studies with a prototypical implementation
   (for details, see 
\CiteAppendix{Appendix~\ref{appendix:implementation}). }%

\tudparagraph{1ex}{Related work.}
Although conditional expectations appear rather naturally in many applications
and despite the large amount of publications on variants of stochastic path
problems and other forms of expectations in MDPs
(see e.g.~\cite{BBCFK-TwoViews14,RanRasSan15}),
we are not aware that they have been addressed in the context of MDPs.
Computation schemes for extremal conditional probabilities
$\Pr^{\max}(\varphi | \psi )$ or
$\Pr^{\min}(\varphi | \psi )$
where both
the objective $\varphi$ and the assumption $\psi$ are
path properties specified in some temporal logic
have been studied in
\cite{AndRoss08,AndresPhD2011,BKKM14}.
For reachability properties $\varphi$ and $\psi$, the
algorithm of \cite{AndRoss08,AndresPhD2011} has exponential time complexity,
while the algorithm of \cite{BKKM14} runs in polynomial time.
Although the approach of \cite{BKKM14} is not applicable for
 calculating maximal conditional expectations 
\CiteAppendix{(see Appendix~\ref{appendix:reset-Methode}), }%
it can be used to compute an upper bound for $\CExp{\max}$
(see Section \ref{sec:finiteness}).
Conditional expected rewards in
Markov chains can be computed using the rescaling
technique of \cite{BKKM14} for finite Markov chains
or the approximation techniques of \cite{BraKuc05,AbdHenMayr07}
for certain classes of infinite-state Markov chains.
The conditional weakest precondition operator
of \cite{KatoenGJKO15}
yields a technique to compute conditional expected rewards for
purely probabilistic programs (without non-determinism).

 \section{Preliminaries}

\label{sec:preliminaries}
\label{sec:prelim}

We briefly summarize our notations used for
Markov decision processes.
Further details can be found in
textbooks,
see e.g.~\cite{Puterman,Kallenberg} or Chapter 10 in \cite{BaierKatoen08}.

A \emph{Markov decision process} (MDP) is a tuple $\cM = (S,\Act,P,\sinit,\rew)$
where $S$ is a finite set of states,
$\Act$ a finite set of actions,
$\sinit \in S$ the initial state,
$P : S \times \Act \times S \to [0,1] \cap \Rational$ is the
transition probability function and
$\rew : S \times \Act \to \Nat$ the reward function.
We require that
$\sum_{s'\in S} P(s,\act,s') \in \{0,1\}$
for all $(s,\alpha)\in S\times \Act$.
We write $\Act(s)$ for the set of actions that are enabled in $s$,
i.e., $\act \in \Act(s)$ iff $P(s,\act,\cdot)$ is not the null function.
State $s$ is called a \emph{trap} if $\Act(s)=\varnothing$.
The paths of $\cM$ are finite or
infinite sequences $s_0 \, \act_0 \, s_1 \, \act_1 \, s_2 \, \act_2 \ldots$
where states and actions alternate such that
$P(s_i,\act_i,s_{i+1}) >0$ for all $i\geqslant 0$.
A path $\fpath$ is called \emph{maximal} if it is either infinite or
finite and its last state is a trap.
If $\fpath =
    s_0 \, \act_0 \, s_1 \, \act_1 \, s_2 \, \act_2 \ldots \act_{k-1} \, s_k$
is finite then
$\rew(\fpath)=
 \rew(s_0,\act_0) + \rew(s_1,\act_1) + \ldots + \rew(s_{k-1},\act_{k-1})$
denotes the accumulated reward
and $\first(\fpath)=s_0$, $\last(\fpath)=s_k$ its first resp.~last state.
The \emph{size} of $\cM$, denoted $\Size(\cM)$,
is the sum of the number of states
plus the total sum of the logarithmic lengths of the non-zero
probability values
$P(s,\alpha,s')$ and the reward values $\rew(s,\alpha)$.%
\footnote{%
  The logarithmic length of an integer $n$
  is the number of bits required for a representation of
  $n$ as a binary number.
  The logarithmic length of a rational number $a/b$
  is defined as the sum of the logarithmic lengths of its numerator $a$
  and its denominator $b$, assuming that $a$ and $b$ are coprime integers
  and
  $b$ is positive.}

An \emph{end component} of $\cM$ is a strongly connected sub-MDP. End components
can be formalized as pairs $\cE = (E,\ActEC)$ where $E$ is a nonempty subset
of $S$ and $\ActEC$ a function that assigns to each state $s\in E$ a nonempty
subset of $\Act(s)$ such that the graph induced by $\cE$ is strongly connected.

A \emph{(randomized) scheduler} for $\cM$,
often also called policy or adversary,
is a function $\sched$ that assigns to each finite path $\fpath$ where
$\last(\fpath)$ is not a trap
a probability distribution over $\Act(\last(\fpath))$.
$\sched$ is called memoryless if $\sched(\fpath)=\sched(\fpath')$ for
all finite paths $\fpath$, $\fpath'$ with $\last(\fpath)=\last(\fpath')$,
in which case $\sched$ can be viewed as a function
that assigns to each non-trap state $s$ a distribution over $\Act(s)$.
$\sched$ is called deterministic if $\sched(\fpath)$ is a Dirac distribution
for each path $\fpath$,
in which case $\sched$ can be viewed as a function that assigns an action
to each finite path $\fpath$ where $\last(\fpath)$ is not a trap.
We write $\Pr^{\sched}_{\cM,s}$ or briefly $\Pr^{\sched}_{s}$
to denote the probability measure induced by $\sched$ and $s$.
Given a measurable set $\psi$ of maximal paths, then
$\Pr^{\min}_{\cM,s}(\psi) = \inf_{\sched} \Pr^{\sched}_{\cM,s}(\psi)$
and
$\Pr^{\max}_{\cM,s}(\psi) = \sup_{\sched} \Pr^{\sched}_{\cM,s}(\psi)$.
We will use LTL-like notations to specify
measurable sets of maximal paths.
For these it is well-known that optimal deterministic schedulers
exists.
If $\psi$ is a reachability condition then even optimal deterministic
memoryless schedulers exist.

Let $\varnothing \not= F \subseteq S$.
For a comparison operator $\bowtie\ \in \{=,>,\geqslant,<,\leqslant\}$
and $r\in \Nat$,
$\Diamond^{\bowtie r} F$ denotes the event ``reaching $F$
along some finite path $\fpath$ with $\rew(\fpath)\bowtie r$''.
The notation $\accdiaplus F$ will be used for
the random variable that assigns to each maximal
path $\infpath$ in $\cM$ the reward $\rew(\fpath)$ of the shortest prefix
$\fpath$ of $\infpath$ where $\last(\fpath)\in F$.
If $\infpath \not\models \Diamond F$ then $(\accdiaplus F)(\infpath)=\infty$.
If $s\in S$ then $\ExpRew{\sched}{\cM,s}(\accdiaplus F)$ denotes
the expectation of $\accdiaplus F$ in $\cM$ with starting state $s$
under $\sched$, which is infinite if
$\Pr^{\sched}_{\cM,s}(\Diamond F) <1$.
$\ExpRew{\max}{\cM,s}(\accdiaplus F) \in \Real \cup \{\pm\infty\}$ stands for
$\sup_{\sched} \ExpRew{\sched}{\cM,s}(\accdiaplus F)$ where the supremum
is taken over all schedulers $\sched$ with
$\Pr^{\sched}_{\cM,s}(\Diamond F)=1$.
Let $\psi$ be a measurable set of maximal paths.
$\ExpRew{\sched}{\cM,s}(\accdiaplus F|\psi)$ stands for the expectation
of $\accdiaplus F$ w.r.t.~the conditional probability
measure $\Pr^{\sched}_{\cM,s}(\ \cdot \ | \psi)$ given by
$\Pr^{\sched}_{\cM,s}(\varphi | \psi) =
 \Pr^{\sched}_{\cM,s}(\varphi \wedge \psi)/\Pr^{\sched}_{\cM,s}(\psi)$.
$\ExpRew{\max}{\cM,s}(\accdiaplus F|\psi)$ is the supremum
of $\ExpRew{\sched}{\cM,s}(\accdiaplus F|\psi)$ 
where $\Pr^{\sched}_{\cM,s}(\psi)>0$ and
$\Pr^{\sched}_{\cM,s}(\Diamond F| \psi)=1$,
and
$\Pr^{\max}_{\cM,s}(\varphi | \psi) = 
 \sup_{\sched} \Pr^{\sched}_{\cM,s}(\varphi | \psi)$
where $\sched$ ranges over all schedulers with
$\Pr^{\sched}_{\cM,s}(\psi) >0$ and $\sup \varnothing = -\infty$.

For the remainder of this paper,
we suppose that two nonempty subsets $F$ and $G$ of $S$ are given
such that
$\Pr^{\max}_{\cM,s}(\Diamond F |\Diamond G)=1$.
The task addressed in this paper is to compute
the maximal conditional expectation given by:
\begin{center}
 $\CExpState{\max}{\cM,s} \ \eqdef  \
  \sup\limits_{\sched} \
    \CExpState{\sched}{\cM,s} \in \Real \cup\{\infty\}$
 \quad where \quad
  $ \CExpState{\sched}{\cM,s} \ = \
    \ExpRew{\sched}{\cM,s}(\accdiaplus F|\Diamond G)$
 \vspace*{-1ex}
\end{center}
Here, $\sched$ ranges over all schedulers $\sched$ with
$\Pr^{\sched}_{\cM,s}(\Diamond G)>0$ and
$\Pr^{\sched}_{\cM,s}(\Diamond F |\Diamond G)=1$.
If $\cM$ and its initial state
are clear from the context, we often 
simply write $\CExp{\max}$ resp.~$\CExp{\sched}$.
We assume that 
all states in $\cM$ are reachable from $\sinit$ and
$\sinit \notin F \cup G$
(as $\CExp{\max}=0$ if $s\in F$ and 
$\CExp{\max}=\ExpRew{\max}{\cM,\sinit}(\accdiaplus F)$ 
if $s\in G \setminus F$).

\section{Finiteness and upper bound}

\label{sec:finiteness}

\tudparagraph{0ex}{Checking finiteness.}
We sketch a polynomially time-bounded algorithm that takes as input an
MDP $\cM = (S,\Act,P,\sinit,\rew)$ with two distinguished subsets $F$ and $G$
of $S$ such that $\Pr^{\max}_{\cM,\sinit}(\Diamond F|\Diamond G)=1$.
If $\CExp{\max}=\ExpRew{\max}{\cM,\sinit}(\accdiaplus F |\Diamond G)=\infty$ 
then the output is ``no''.
Otherwise, the output is
an MDP $\Hut{\cM} = (\Hut{S},\Hut{\Act},\Hut{P},\sHutinit,\Hut{\rew})$
with two trap states $\goal$ and $\fail$
such that:
\begin{enumerate}
\item [(1)]
  $\ExpRew{\max}{\cM,\sinit}(\accdiaplus F|\Diamond G) \ \ = \ \
   \ExpRew{\max}{\Hut{\cM},\sHutinit}(\accdiaplus \goal | \Diamond \goal)$,
\item [(2)]
  $\Hut{s} \models \exists \Diamond \goal$ and
  $\Pr^{\min}_{\Hut{\cM},\Hut{s}}\bigl(\Diamond (\goal \vee \fail)\bigr)=1$
  for all states $\Hut{s}\in \Hut{S} \setminus \{\fail\}$, and
\item [(3)]
  $\Hut{\cM}$ does not have critical schedulers where a scheduler
  $\usched$ for $\Hut{\cM}$ is said to be critical iff
  $\Pr^{\usched}_{\Hut{\cM},\sHutinit}(\Diamond \fail)=1$
  and there is a reachable positive $\usched$-cycle.%
\footnote{The latter means a $\usched$-path
  $\fpath = s_0 \, \act_0 \, s_1 \, \act_1 \ldots \act_{k-1} \, s_k$
  where $s_0=\sHutinit$ and $s_i=s_k$ for some $i\in \{0,1,\ldots ,k{-}1\}$
  such that $\Hut{\rew}(s_j,\act_j)>0$ for some $j\in \{i,\ldots,k{-}1\}$.}
\end{enumerate}
We provide here the main ideas of the algorithms and
refer to
\CiteAppendix{Appendix~\ref{appendix:finiteness} }%
for the details.
The algorithm first transforms $\cM$ into an MDP $\tilde{\cM}$
that permits to assume $F=G = \{\goal \}$.
Intuitively, $\tilde{\cM}$ simulates $\cM$, while operating in four modes:
``normal mode'', ``after $G$'', ``after $F$'' and ``goal''.
$\tilde{\cM}$ starts in normal mode where it
behaves as $\cM$ as long as neither $F$ nor $G$ have been visited.
If a $G\setminus F$-state %
has been reached in normal mode
then $\tilde{\cM}$ switches to the mode ``after $G$''.
Likewise, as soon as an $F\setminus G$-state %
has been reached in normal mode
then $\tilde{\cM}$ switches to the mode ``after $F$''.
$\tilde{\cM}$ enters the goal mode (consisting of a single trap state $\goal$)
as soon as a path fragment
containing a state in $F$ and a state in $G$ has been generated.
This is the case if $\cM$ visits an $F$-state in mode ``after $G$''
or a $G$-state in mode ``after $F$'',
or a state in $F \cap G$ in the normal mode.
The rewards in the normal mode and in mode ``after $G$''
are precisely as in $\cM$,
while the rewards are 0 in all other cases.
We then remove all states $\tilde{s}$ in the ``after $G$'' mode
with $\Pr^{\max}_{\tilde{\cM},\tilde{s}}(\Diamond \goal) <1$,
collapse all states $\tilde{s}$ in $\tilde{\cM}$
with $\tilde{s}\not\models \exists \Diamond \goal$
into a single trap state called $\fail$
and add zero-reward transitions to $\fail$ from all states
$\tilde{s}$ that are not in the ``after $G$'' mode and
$\Pr^{\max}_{\tilde{\cM},\tilde{s}}(\Diamond \goal) =0$.
Using techniques as in the unconditional case \cite{deAlf99} we can check
whether $\tilde{\cM}$ has positive end components,
i.e., end components with at least one state-action pair
$(s,\alpha)$ with $\rew(s,\alpha)>0$. If so, then
$\ExpRew{\max}{\cM,\sinit}(\accdiaplus F|\Diamond G) = \infty$.
Otherwise, we collapse each maximal end component of $\tilde{\cM}$
into a single state.

Let $\Hut{\cM}$ denote the resulting MDP. It satisfies (1) and (2).
Property (3) holds iff
$\ExpRew{\max}{\Hut{\cM},\sHutinit}
   (\accdiaplus \goal|\Diamond \goal) < \infty$.
This condition can be checked in polynomial time
using a graph analysis in the sub-MDP
of $\Hut{\cM}$ consisting of the states
$\Hut{s}$ with $\Pr^{\min}_{\Hut{\cM},\Hut{s}}(\Diamond \goal)=0$
\CiteAppendix{(see Prop.~\ref{proposition:exprew-infinite-critical-scheduler} and Appendix \ref{summary:check-finiteness}). }%

\tudparagraph{1ex}{Computing an upper bound.}
Due to the transformation used for checking finiteness of the
maximal conditional expectation, we can now suppose that
$\cM=\Hut{\cM}$, $F=G=\{\goal\}$ and that (2) and (3) hold.
We now present a technique to compute an upper
bound $\CExp{\ub}$ for $\CExp{\max}$. The upper bound will be used
later to determine a saturation point from which on
optimal schedulers behave memoryless (see Section \ref{sec:threshold}).

We consider the MDP $\cM'$ simulating $\cM$, while operating
in two modes. In its first mode, $\cM'$ attaches the reward accumulated
so far to the states. More precisely, the states of $\cM'$ in its
first mode have the form $\<s,r\> \in S \times \Nat$ where
$0 \leqslant r \leqslant R$ and
$R = \sum_{s\in S'}
 \max \{ \rew_{\cM'}(s,\alpha):\alpha \in \Act_{\cM'}(s) \}$.
The initial state of $\cM'$ is $\sinit'=\<\sinit,0\>$.
The reward for the state-action pairs $(\<s,r\>,\alpha)$
where $r{+}\rew(s,\alpha) \leqslant R$ is 0.
If $\cM'$ fires an action $\alpha$ in state $\<s,r\>$ where
$r'\eqdef r{+}\rew(s,\alpha) > R$
then it switches to the second mode,
while earning reward $r'$.
In its second mode $\cM'$ behaves as $\cM$
without additional annotations of the states and earning the same rewards
as $\cM$.
From the states $\<\goal,r\>$, $\cM'$ moves to $\goal$ with probability 1
and reward $r$.
There is a one-to-one correspondence between the schedulers for
$\cM$ and $\cM'$ and the
switch from $\cM$ to $\cM'$ does not affect the probabilities
and the accumulated rewards until reaching $\goal$.

Let $\cN$ denote the MDP resulting from $\cM'$
by adding reset-transitions from $\fail$ (as a state of the
second mode) and the copies $\<\fail,r\>$ in the first mode
to the initial state $\sinit'$.
The reward of all reset transitions is 0.
The reset-mechanism has been taken from \cite{BKKM14} where it has been
introduced as a technique to
compute maximal conditional probabilities for reachability properties.
Intuitively, $\cN$ ``discards'' all paths of $\cM'$ that
eventually enter $\fail$ and ``redistributes'' their probabilities
to the paths that eventually enter the goal state.
In this way, $\cN$ mimics the conditional probability measures
$\Pr^{\sched}_{\cM',\sinit'}(\ \cdot \ |\Diamond \goal) =
 \Pr^{\sched}_{\cM,\sinit}(\ \cdot \ |\Diamond \goal)$
for prefix-independent path properties.
Paths $\fpath$ from $\sinit$ to $\goal$ in $\cM$ are
simulated in $\cN$ by paths of the form
$\finpath= \cycle_1; \ldots \cycle_k ; \fpath$ where
$\cycle_i$ is a cycle in $\cN$ with $\first(\cycle_i)=\sinit'$
and $\cycle_i$'s last transition is a reset-transition from some fail-state
to $\sinit'$. Thus, $\rew(\fpath) \leqslant \rew_{\cN}(\finpath)$.
The distinction between the first and second mode together with property (3)
ensure that the
new reset-transitions do not generate positive end components in $\cN$.
By the results of \cite{deAlf99}, the maximal unconditional expected
accumulated reward in $\cN$ is finite and we have:
\begin{center}
  $\ExpRew{\max}{\cM,\sinit}(\accdiaplus \goal | \Diamond \goal)
   \ \ = \ \
   \ExpRew{\max}{\cM',\sinit'}(\accdiaplus \goal | \Diamond \goal)
   \ \ \leqslant \ \
   \ExpRew{\max}{\cN,\sinit'}(\accdiaplus \goal)$
\end{center}
Hence, we can deal with
$\CExp{\ub}=\Exp{\max}{\cN,\sinit'}(\accdiaplus \goal)$, which is computable
in time polynomial in the size of $\cN$
by the algorithm proposed in \cite{deAlf99}.
As $\mathit{size}(\cN)=\Theta(R \cdot \mathit{size}(\cM))$
we obtain a pseudo-polynomial time bound for the general case.
If, however, $\Pr^{\min}_{\cM,s}(\Diamond \goal)>0$
for all states $s\in S \setminus \{\fail\}$  then there is no need for the
detour via $\cM'$ and we can apply the reset-transformation
$\cM \leadsto \cN$ by adding a reset-transition from $\fail$ to $\sinit$
with reward 0, in which case the upper bound
$\CExp{\ub}=\ExpRew{\max}{\cN,\sinit}(\accdiaplus \goal)$ is obtained
in time polynomial in the size of $\cM$.
For details we refer to
\CiteAppendix{the proof of Prop.~\ref{proposition:exprew-infinite-critical-scheduler} and Section \ref{sec:upper-bound} in the appendix. }%

\section{Threshold algorithm and computing optimal schedulers}

\label{sec:threshold}

In what follows, we suppose that $\cM=(S,\Act,P,\sinit,\rew)$ is an MDP
with two trap states $\goal$ and $\fail$ such that
$s\models \exists \Diamond \goal$ for all states $s\in S \setminus \{\fail\}$
and $\min_{s\in S} \Pr^{\min}_{\cM,s}(\Diamond (\goal \vee \fail))=1$
and
$\CExp{\max} =
 \ExpRew{\max}{\cM,\sinit}(\accdiaplus \goal | \Diamond \goal)< \infty$.

A scheduler $\sched$ is said to be \emph{reward-based} if
$\sched(\fpath)=\sched(\fpath')$ for all finite paths $\fpath$, $\fpath'$ with
$(\last(\fpath),\rew(\fpath)) = (\last(\fpath'),\rew(\fpath'))$.
Thus, deterministic reward-based schedulers can be seen as functions
$\sched : S \times \Nat \to \Act$.
\CiteAppendix{Prop.~\ref{prop:det-reward-based} in the appendix shows that }%
$\CExp{\max}$ equals the supremum of the values
$\CExp{\sched}$,
when ranging over all deterministic reward-based schedulers $\sched$
with $\Pr^{\sched}_{\cM,\sinit}(\Diamond \goal)>0$.

The basis of our algorithms are the following two observations.
First, there exists
a saturation point $\saturation \in \Nat$ such that the optimal decision
for all paths $\fpath$ with $\rew(\fpath)\geqslant \saturation$ is to maximize
the probability for reaching the goal state
(see Prop.~\ref{prop:saturation-maxsched} below).
The second observation is a technical statement that
will be used at several places.
Let $\rho,\theta,\zeta,r,x,y,z,p\in \Real$ with
$0\leqslant p,x,y,z \leqslant 1$, $p>0$, $y > z$ and $x+z>0$
and let
\begin{center}
  $\mathsf{A} 
   \  =  \ \frac{\displaystyle \rho + p(ry + \theta)}{\displaystyle x+py}$,
  \quad 
  $\mathsf{B} 
   \ =  \ \frac{\displaystyle \rho + p(rz + \zeta)}{\displaystyle x+pz}$
  \quad \text{and} \quad
  $\mathsf{C} = \max \{\mathsf{A},\mathsf{B}\}$
\end{center}
Then:
\begin{equation}
    \label{magic-constraint}
    \mathsf{A} \geqslant \mathsf{B}
    \ \ \ \ \text{iff} \ \ \ \
    r + \frac{\theta {-} \zeta}{y {-} z}
    \ \geqslant \ \mathsf{C}
    \ \ \ \ \text{iff}  \ \ \ \
    \theta - (\mathsf{C}{-}r)y \ \geqslant \  \zeta - (\mathsf{C}{-}r)z
    \tag{$\dagger$}
\end{equation}
and the analogous statement for $>$ rather than $\geqslant$.
\CiteAppendix{This statement is a consequence of Lemma \ref{lemma:decision-for-s-r} in the appendix. }%
We will apply this observation in different nuances.
To give an idea how to apply statement \eqref{magic-constraint},
suppose $\mathsf{A} = \CExp{\tsched}$ and
$\mathsf{B}=\CExp{\usched}$ where $\tsched$ and $\usched$ are
reward-based schedulers that agree for all paths $\finpath$
that do not have a prefix $\fpath$ with $\rew(\fpath)=r$ where
$\last(\fpath)$ is a non-trap state, in which case
$x$ denotes the probability for reaching $\goal$ from $\sinit$ along
such a path $\finpath$ 
and $\rho$ stands for the corresponding partial expectation, while
$p$ denotes the probability of the paths 
$\fpath$ from $\sinit$ to some non-trap state with $\rew(\fpath)=r$.
The crucial observation is that $r+(\theta{-}\zeta)/(y{-}z)$ does not
depend on $x,\rho,p$. Thus,
if $r+(\theta{-}\zeta)/(y{-}z) \geqslant \CExp{\ub}$ 
for some upper bound  $\CExp{\ub}$ of $\CExp{\max}$
then \eqref{magic-constraint} allows to conclude that
$\tsched$'s decisions for the state-reward pairs $(s,r)$ are better 
than $\usched$, independent of $x,\rho$ and $p$.

Let $R\in \Nat$ and $\sched$, $\tsched$ be reward-based schedulers.
The \emph{residual} scheduler
$\residual{\sched}{R}$ is given by 
$(\residual{\sched}{R})(s,r) = \sched(s,R{+}r)$.
$\redefresidual{\sched}{R}{\tsched}$
denotes the unique scheduler 
that agrees with $\sched$ for all state-reward pairs
$(s,r)$ where $r < R$ and 
$\residual{(\redefresidual{\sched}{R}{\tsched})}{R} = \tsched$.
We write $\Exp{\sched}{\cM,s}$ for the \emph{partial expectation}
\begin{center}
  $\Exp{\sched}{\cM,s}
   \ \ = \ \
   \sum\limits_{r=0}^{\infty} \
      \Pr^{\sched}_{\cM,s}(\Diamond^{=r}\goal) \cdot r$
\end{center}
Thus,
$\Exp{\tsched}{\cM,s}=\ExpRew{\tsched}{\cM,s}(\accdiaplus \goal)$
if $\Pr^{\tsched}_{\cM,s}(\Diamond \goal)=1$, while
$\Exp{\tsched}{\cM,s} < \infty = \ExpRew{\tsched}{\cM,s}(\accdiaplus \goal)$
if $\Pr^{\tsched}_{\cM,s}(\Diamond \goal)<1$.

\begin{proposition}
 \label{prop:saturation-maxsched}
  There exists
  a natural number $\saturation$ (called \emph{saturation point} of $\cM$)
  and a deterministic memoryless scheduler $\maxsched$
  such that:
  \begin{enumerate}
  \item [(a)]
     $\CExp{\tsched} \ \leqslant \
      \CExp{\redefresidual{\tsched}{\saturation}{\maxsched}}$
     for each scheduler $\tsched$ with
     $\Pr^{\tsched}_{\cM,\sinit}(\Diamond \goal)>0$, and
  \item [(b)]
     $\CExp{\sched} \ =\ \CExp{\max}$ for some
     deterministic reward-based scheduler $\sched$ such that
     $\Pr^{\sched}_{\cM,\sinit}(\Diamond \goal)>0$ and
     $\residual{\sched}{\saturation}=\maxsched$.
  \end{enumerate}
\end{proposition}
The proof of Prop.~\ref{prop:saturation-maxsched}
\CiteAppendix{(see Appendices \ref{appendix:saturation} and \ref{appendix:compute-saturation}) }%
is constructive and yields a polynomial-time algorithm
for generating a scheduler $\maxsched$
as in Prop.~\ref{prop:saturation-maxsched} and
a pseudo-polynomial algorithm for the computation of
a saturation point $\saturation$.

Scheduler $\maxsched$ maximizes the probability to reach $\goal$ 
from each state.
If there are two or more such schedulers, then
$\maxsched$ is one where the conditional expected accumulated reward
until reaching goal
is maximal under all  schedulers $\usched$ with
$\Pr^{\usched}_{\cM,s}(\Diamond \goal) = \Pr^{\max}_{\cM,s}(\Diamond \goal)$
for all states $s$.
Such a scheduler $\maxsched$ is computable in polynomial time using
linear programming techniques.
\CiteAppendix{(See Lemma \ref{lemma:Sched-max-exp} in the appendix.) }%

The idea for the computation of the saturation point 
is to compute the threshold $\saturation$
above which the scheduler $\maxsched$ becomes optimal. 
For this we rely on statement \eqref{magic-constraint} 
where  $\theta/y$ stands for the conditional expectation under $\maxsched$, 
$\zeta/z$ for the conditional expectation
under an arbitrary scheduler $\sched$ and
$\mathsf{C}=\CExp{\ub}$ is an upper bound of $\CExp{\max}$ 
(see Theorem \ref{thm:finiteness}),
while $r=\saturation$ is the wanted value.
More precisely, for $s\in S$,
let $\theta_s=\Exp{\maxsched}{\cM,s}$,
$y_s = \Pr^{\maxsched}_{\cM,s}(\Diamond \goal)
     = \Pr^{\max}_{\cM,s}(\Diamond \goal)$.
To compute a saturation point we determine
the smallest value $\saturation \in \Nat$ such that
\begin{center}
    $\theta_s - (\CExp{\ub}{-}\saturation)\cdot y_s
     \ \ = \ \
     \max\limits_{\sched} \
        \bigl( \ \Exp{\sched}{\cM,s} -
           (\CExp{\ub} {-}\saturation)\cdot
               \Pr^{\sched}_{\cM,s}(\Diamond \goal) \
        \bigr)$
\end{center}
for all states $s$ where $\sched$ ranges over all schedulers for $\cM$.
\CiteAppendix{In Appendix~\ref{appendix:compute-saturation} }%
  we show that
  instead of the maximum over all schedulers $\sched$ it suffices to
  take the local maximum over all ``one-step-variants'' of $\maxsched$.
  That is,
  a saturation point is obtained by
  $\saturation = \max\{ \lceil \CExp{\ub} - D \rceil, 0 \} $
  where
  \begin{center}
    $D \ = \ 
     \min \
       \bigl\{ (\theta_s - \theta_{s,\alpha})/(y_s - y_{s,\alpha})
               \, : \, s \in S, \alpha \in \Act(s),  y_{s,\alpha} < y_s
       \bigr\}$
  \end{center}
  and
  $y_{s,\alpha} =  \sum\limits_{t\in S} P(s,\alpha,t)\cdot y_t$
  and
  $\theta_{s,\alpha} =
       \rew(s,\alpha)\cdot y_{s,\alpha}  +
    \sum\limits_{t\in S} P(s,\alpha,t)\cdot \theta_t$.

\begin{example}
 \label{example:saturation}
{\rm
  The so obtained saturation point for the
  MDP $\cM[\rewparam]$ in Figure~\ref{fig:running-example} is 
  $\saturation = \lceil \CExp{\ub}{+}1\rceil$.
  Note that  only state $s=s_2$ behaves nondeterministically,
  and $\maxsched(s)=\alpha$, 
  $y_s=y_{s,\alpha}=1$, 
  $\theta_s= \theta_{s,\alpha}=0$, while
  $y_{s,\beta}=\theta_{s,\beta}=\frac{1}{2}$.
  This yields $D = (0{-}\frac{1}{2})/(1{-}\frac{1}{2})=-1$.
  Thus,
  $\saturation \geqslant \rewparam {+}2$ 
  as $\CExp{\ub}\geqslant \CExp{\max} > \rewparam$.
\Ende}
\end{example}

The logarithmic length of $\saturation$ is polynomial
in the size of $\cM$.
Thus, the value (i.e., the length of an unary encoding) of $\saturation$
  can be exponential in $\Size(\cM)$.
  This is unavoidable as there are families $(\cM_k)_{k \in \Nat}$
  of MDPs where the size of $\cM_k$ is in $\cO(k)$, while $2^k$ is a
  lower bound for
  the smallest saturation point of $\cM_k$.
  This, for instance, applies to the MDPs $\cM_k = \cM[2^k]$ where
  $\cM[\rewparam]$ is as in Figure~\ref{fig:running-example}.
  Recall from Example \ref{example:running-example} that
  the scheduler $\sched_{\rewparam{+}2}$ that selects $\beta$ by the
  first $\rewparam {+} 2$ visits of $s$ and $\alpha$ for the
  $(\rewparam {+} 3)$-rd visit of $s$ is optimal
  for $\cM[\rewparam]$. 
  Hence, the smallest saturation point for $\cM[2^k]$ is $2^k{+}2$.

\tudparagraph{1ex}{Threshold algorithm.}
The input of the threshold algorithm is an MDP $\cM$ as above
and a non-negative rational
number $\threshold$. The task is to generate
a deterministic reward-based scheduler $\sched$ with
$\residual{\sched}{\saturation}=\maxsched$
(where $\maxsched$ and $\saturation$ are as in
Prop.~\ref{prop:saturation-maxsched})
such that
$\CExp{\sched} > \threshold$ if $\CExp{\max} > \threshold$, and
$\CExp{\sched} = \threshold$ if $\CExp{\max} = \threshold$.
If $\CExp{\max} < \threshold$ then the output of the threshold algorithm
is ``no''.%
\footnote{%
   The threshold algorithm solves all
   four variants of the threshold problem.
   E.g., $\CExp{\max} \leqslant \threshold$ iff 
   $\CExp{\sched}=\threshold$, while
   $\CExp{\max}<\threshold$ iff the threshold algorithm
   returns ``no''.}

The algorithm operates level-wise and determines \emph{feasible} actions
$\action(s,r)$ for all non-trap states $s$ and
$r=\saturation{-}1,\saturation{-}2,\ldots,0$, using the decisions
$\action(\cdot,i)$ for
the levels $i \in \{r{+}1,\ldots,\saturation\}$ that 
have been treated before 
and linear programming techniques to treat zero-reward loops.
In this context, feasibility 
is understood with respect to the following condition: 
If $\CExp{\max} \unrhd \threshold$ where
$\unrhd \in \{>,\geqslant\}$ then
there exists a reward-based scheduler $\sched$ with
$\CExp{\sched} \unrhd \threshold$ and
$\sched(s,R)=\action(s,\min\{\saturation,R\})$ for all $R \geqslant r$.

The algorithm stores for each state-reward pair $(s,r)$
the probabilities $y_{s,r}$
to reach $\goal$ from $s$ and the corresponding partial
expectation $\theta_{s,r}$
for the scheduler given by the decisions in the action table.
The values for $r=\saturation$ are given by
$\action(s,\saturation)=\maxsched(s)$,
$y_{s,\saturation}=\Pr^{\maxsched}_{\cM,s}(\Diamond \goal)$ and
$\theta_{s,\saturation}=\Exp{\maxsched}{\cM,s}$.
The candidates for the decisions at level $r < \saturation$ are given by the
deterministic memoryless schedulers $\psched$ for $\cM$.
We write $\psched_{+}$ for the reward-based scheduler given by
$\psched_{+}(s,0)=\psched(s)$ and
$\psched_{+}(s,i)=\action(s,\min\{\saturation,r{+}i\})$ for $i \geqslant 1$.
Let $y_{s,r,\psched} = \Pr^{\psched_{+}}_{\cM,s}(\Diamond \goal)$
and $\theta_{s,r,\psched} = \Exp{\psched_{+}}{\cM,s}$
be the corresponding partial expectation.

To determine feasible actions for level $r$, the threshold 
algorithm makes use of a variant of
\eqref{magic-constraint} stating that if
$\theta - (\threshold {-} r)y \geqslant \zeta - (\threshold {-}r)z$
and $\mathsf{B} \unrhd \threshold$ then
$\mathsf{A} \unrhd \threshold$,
where $\mathsf{A}$ and $\mathsf{B}$ 
are as in \eqref{magic-constraint} and the requirement
$y>z$ is dropped.
Thus, the aim of the threshold algorithm is to compute a
deterministic memoryless scheduler $\psched^*$ for $\cM$
such that
the following condition \eqref{difference} holds:
\begin{equation*}
  \label{difference}
  \theta_{s,r,\psched^*} - (\threshold {-} r)\cdot y_{s,r,\psched^*}
  \ \ = \ \
  \max\limits_{\psched} \
    \bigl( \
           \theta_{s,r,\psched} - (\threshold {-} r)\cdot y_{s,r,\psched}
           \
    \bigr)
  \tag{*}
\end{equation*}
Such a scheduler $\psched^*$ is computable in time polynomial in the size
of $\cM$ (without the explicit consideration of all schedulers $\psched$
and their extensions $\psched_{+}$)
using the following linear program with one variable $x_s$ for each
state. The objective is to minimize $\sum\limits_{s\in S} x_s$ subject
to the following conditions:
\begin{center}
    \begin{tabular}{ll}
      (1) &
      If $s \in S \setminus \{\goal,\fail\}$ then
      for each action $\alpha \in \Act(s)$ with $\rew(s,\alpha)=0$:
      \\[1ex]
      &
      \qquad \qquad
      $x_s \ \ \geqslant \ \ \sum\limits_{t\in S} P(s,\alpha,t) \cdot x_t$
      \\[1ex]
      (2) &
      If $s \in S \setminus \{\goal,\fail\}$ then
      for each action $\alpha \in \Act(s)$ with $\rew(s,\alpha)>0$:
      \\[1ex]
      &
      \qquad \qquad
      $x_s \ \ \geqslant \ \
       \sum\limits_{t\in S} P(s,\alpha,t) \cdot
            \bigl( \, \theta_{t,R} + \rew(s,\alpha) \cdot y_{t,R} \, - \,
                (\threshold {-} r) \cdot y_{t,R} \, \bigr)$
      \\[1ex]
      &
      where $R=\min\{\saturation,r{+}\rew(s,\alpha)\}$
      \\[1ex]
      (3) &
      For the trap states: \ \
      $x_{\goal} = r-\threshold$ \ \ and  \ \
      $x_{\fail}=0$
  \end{tabular}
\end{center}
This linear program has a unique solution $(x_s^*)_{s\in S}$.
Let $\Act^*(s)$ denote the set of actions $\alpha \in \Act(s)$ such that
the following constraints (E1) and (E2) hold:
\begin{eqnarray*}
     \text{\rm (E1)} & \ \  &
     \text{If $\rew(s,\alpha)=0$ then:}
     \  \ x_s^* \  =  \ \sum_{t\in S} P(s,\alpha,t)\cdot x_t^*
     \\[1ex]
     \text{\rm (E2)} &  &
     \text{If $\rew(s,\alpha)>0$ and
           $R\, =\,
       \min \, \bigl\{\, \saturation,\, r{+}\rew(s,\alpha)\, \bigr\}$ then:}
     \\[1ex]
     & &
     \hspace*{1cm}
     \ x_s^* \  =  \ \sum_{t\in S} P(s,\alpha,t)\cdot
        \bigl( \theta_{t,R} + \rew(s,\alpha)\cdot y_{t,R}
                  \, - \, (\threshold {-} r)\cdot y_{t,R} \bigr)
\end{eqnarray*}
Let $\cM^*=\cM^*_{r,\threshold}$ denote the MDP  with state space $S$
induced by the state-action pairs $(s,\alpha)$ with
$\alpha \in \Act^*(s)$ where the positive-reward actions are redirected to
the trap states. Formally, for $s,t\in S$, $\alpha \in \Act^*(s)$
we let
$P_{\cM^*}(s,\alpha,t)=P(s,\alpha,t)$ if $\rew(s,\alpha)=0$
and
$P_{\cM^*}(s,\alpha,\goal) =
    \sum_{t\in S} P(s,\alpha,t) \cdot y_{t,R}$
and
$P_{\cM^*}(s,\alpha,\fail) =   1 - P_{\cM^*}(s,\alpha,\goal)$
if $\rew(s,\alpha)>0$
and $R=\min \{\saturation,r{+}\rew(s,\alpha)\}$.
The reward structure of $\cM^*$ is irrelevant for our purposes.

A scheduler $\psched^*$ satisfying \eqref{difference} is obtained
by computing a memoryless deterministic scheduler for $\cM^*$
with
$\Pr^{\psched^*}_{\cM^*,s}(\Diamond \goal) =
 \Pr^{\max}_{\cM^*,s}(\Diamond \goal)$ for all states $s$.
This scheduler $\psched^*$ indeed
provides feasible decisions for level $r$, i.e.,
if $\CExp{\max} \unrhd \threshold$ where
$\unrhd \in \{>,\geqslant\}$ then
there exists a reward-based scheduler $\sched$ with
$\CExp{\sched} \unrhd \threshold$,
$\sched(s,r)=\psched^*(s)$ and
$\sched(s,R)=\action(s,\min\{\saturation,R\})$ for all $R > r$.

The threshold algorithm then puts 
$\action(s,r)=\psched^*(s)$
and
computes the values $y_{s,r}$ and $\theta_{s,r}$ as follows.
Let $T$ denote the set of states $s\in S \setminus \{\goal,\fail\}$
where $\rew(s,\psched^*(s))>0$.
For $s\in T$, the values $y_{s,r}=y_{s,r,\psched^*}$ and
$\theta_{s,r}= \theta_{s,r,\psched^*}$
can be derived directly from the results obtained for the previously treated
levels $r{+}1,\ldots,\saturation$ as we have:
\begin{center}
 \begin{tabular}{lll}
  $y_{s,r} =  \sum\limits_{t\in S} P(s,\alpha,t)\cdot y_{t,R}$
  & \ and \ &
  $\theta_{s,r} =
  \rew(s,\alpha)\cdot y_{s,r}  +
    \sum\limits_{t\in S} P(s,\alpha,t)\cdot \theta_{t,R}$
 \end{tabular}
\end{center}
where $\alpha = \psched^*(s)$ and
$R = \min \{\saturation, r{+}\rew(s,\alpha)\}$.
For the states $s\in S \setminus T$:
\begin{center}
   $y_{s,r} =
    \sum\limits_{t\in T}
        \Pr^{\psched^*}_{\cM,s}(\neg T \Until t) \cdot y_{t,r}$
   \ \ and \ \
   $\theta_{s,r} =
    \sum\limits_{t\in T}
        \Pr^{\psched^*}_{\cM,s}(\neg T \Until t) \cdot \theta_{t,r}$
\end{center}
Having treated the last level $r=0$, the output of the algorithm is as follows.
Let $\sched$ be the scheduler given by the action
table $\action(\cdot)$.
For the conditional expectation we have
$\CExp{\sched} = \theta_{\sinit,0}/y_{\sinit,0}$
if $y_{\sinit,0}>0$.
If $y_{\sinit,0}=0$ or $\theta_{\sinit,0}/y_{\sinit,0} < \threshold$
then the algorithm returns the answer ``no''.
Otherwise, the algorithm returns $\sched$, in which case
$\CExp{\sched} > \threshold$ or $\CExp{\sched}=\threshold = \CExp{\max}$.
Proofs for the soundness and the exponential
time complexity are provided in
\CiteAppendix{Appendix~\ref{appendix:threshold}. }%

\begin{example}
 \label{example:threshold-algo}
{\rm
 For the MDP $\cM[\rewparam]$ in Example \ref{example:running-example},
 scheduler $\maxsched$ selects action $\alpha$ for state $s=s_2$.
 Thus, $\action(s,\saturation)=\alpha$ for the computed
 saturation point $\saturation \geqslant \rewparam {+} 2$
 (see Example \ref{example:saturation}).
 The threshold algorithm for each positive rational threshold
 $\threshold$ computes for each level 
 $r=\saturation{-}1, \saturation{-}2,\ldots,1,0$ where
 $\action(s,r{+}1)=\alpha$, the value
 $x_s^*=\max \{ r{-}\threshold, \frac{1}{2}+\frac{1}{2}(r{-}\threshold) \}$
 and the action set $\Act^*(s)=\{\alpha\}$ if $r>\threshold{+}1$,
 $\Act^*(s)=\{\alpha,\beta\}$ if $r=\threshold{+}1$ and
 $\Act^*(s)=\{\beta\}$ if $r<\threshold{+}1$.
 Thus, if $n=\min \{ \saturation, \lceil \threshold {+} 1 \rceil\}$ 
 then
 $\action(s,r)=\alpha$, $y_{s,r}=1$, $\theta_{s,r}=0$
 for $r \in \{n,\ldots,\saturation\}$, while
 $\action(s,n{-}k)=\beta$, $y_{s,n{-}k}=1/2^k$, $\theta_{s,n{-}k}=k/2^k$
 for $k=1,\ldots,n$. That is, 
 the threshold algorithm computes the 
 scheduler $\sched_n$ that selects $\beta$ for the first
 $n$ visits of $s$ and $\alpha$ for the $(n{+}1)$-st visit of $s$. 
 Thus, if $\rewparam \leqslant \threshold < \rewparam {+} 1$
 then $n = \rewparam {+} 2$, in which case the computed scheduler
 $\sched_n$ is optimal (see Example \ref{example:running-example}).
 The returned answer depends on whether $\threshold \leqslant \CExp{\max}$.
 If, for instance, $\threshold = \frac{\rewparam}{2}$ and $\rewparam >0$ 
 is even then the threshold algorithm returns the scheduler
 $\sched_{n}$ where $n=\frac{\rewparam}{2}{+}1$, whose conditional
 expectation is 
 $\rewparam - (\frac{\rewparam}{2}{-}1)/(2^{\frac{\rewparam}{2}+1}{+}1)
  > \frac{\rewparam}{2}=\threshold$.
\Ende}
\end{example}

\tudparagraph{1ex}{{\it MDPs without zero-reward cycles and acyclic MDPs.}}
If $\cM$ does not contain zero-reward cycles then there is no need for
the linear program. Instead we can use a topological sorting of the
states in the graph of the sub-MDP consisting of zero-reward actions
and determine a scheduler $\psched^*$ satisfying \eqref{difference} directly.
For acyclic MDPs, there is even no need for a saturation point.
We can explore $\cM$ using a recursive procedure and determine
feasible decisions for each reachable state-reward pair $(s,r)$
on the basis of \eqref{difference}. This yields a polynomially space-bounded
algorithm to decide whether $\CExp{\max} \unrhd \threshold$ in acyclic
MDPs. 
\CiteAppendix{(See Appendix~\ref{appendix:PSPACE}.) }%

\tudparagraph{1ex}{Construction of an optimal scheduler.}
Let $\SchedThreshold{\threshold}$ denote the scheduler that is generated
by calling the threshold algorithm for the threshold value $\threshold$.
A simple approach is to apply the threshold algorithm iteratively:
\begin{enumerate}
\item []
    let $\sched$ be the scheduler
    $\maxsched$ as in Proposition \ref{prop:saturation-maxsched};
\item []
    {\tt REPEAT}
       \
       $\threshold := \CExp{\sched}$; \
       $\sched := \SchedThreshold{\threshold}$ \ \
    {\tt UNTIL} $\threshold = \CExp{\sched}$;
\item []
    return $\threshold$ and $\sched$
\end{enumerate}
The above algorithm generates a sequence of
deterministic reward-based schedulers that are memoryless from
$\saturation$ on with strictly increasing conditional expectations.
The number of such schedulers is bounded by $\md^{\saturation}$
where $\md$ denotes the
number of memoryless deterministic schedulers for $\cM$.
Hence, the algorithm terminates and
correctly returns $\CExp{\max}$ and an optimal scheduler.
As $\md$ can be exponential in the number of states,
this simple algorithm has double-exponential time complexity.

To obtain a (single) exponential-time algorithm, we seek for
better (larger, but still promising) 
threshold values than the conditional expectation of
the current scheduler.
We propose an algorithm that operates level-wise and
freezes optimal decisions for levels
$r=\saturation,\saturation{-}1,\saturation{-}2,\ldots,1,0$.
The algorithm maintains and successively improves a left-closed and
right-open interval
$I = [A,B[$ with $\CExp{\max}\in I$ and
$\CExp{\sched} \in I$ for the current scheduler $\sched$.

\tudparagraph{0.5ex}{{\it Initialization.}}
The algorithm starts with the scheduler
$\sched= \SchedThreshold{\CExp{\maxsched}}$ where $\maxsched$ is
as above.
If $\CExp{\sched}=\CExp{\maxsched}$ then the algorithm immediately
terminates.
Suppose now that $\CExp{\sched} > \CExp{\maxsched}$.
The initial interval is $I = [A,B[$ where $A = \CExp{\sched}$
and $B = \CExp{\ub}{+}1$ where $\CExp{\ub}$ is as in
Theorem \ref{thm:finiteness}.

\tudparagraph{0.5ex}{{\it Level-wise scheduler improvement.}}
The algorithm successively determines
optimal decisions for the levels
$r=\saturation{-}1,\saturation{-}2,\ldots,1,0$.
The treatment of level $r$ consists of a sequence of scheduler-improvement
steps where at the same time the interval $I$ is replaced with proper
sub-intervals.
The current scheduler $\sched$ has been obtained by the last successful run
of the threshold algorithm, i.e., it has the form
$\sched=\SchedThreshold{\threshold}$ where $\CExp{\sched} > \threshold$.
Besides the decisions of $\sched$ (i.e., the actions
$\sched(s,R)$ for all state-reward pairs
$(s,R)$ where $s\in S \setminus \{\goal,\fail\}$ and
$R\in \{0,1,\ldots,\saturation\}$), the algorithm also stores the values
$y_{s,R}$ and $\theta_{s,R}$
that have been computed in the threshold algorithm.%
\footnote{%
  As the decisions of the already treated levels are optimal,
  the values $y_{s,R}$ and $\theta_{s,R}$
  for $R \in \{r{+}1,\ldots,\saturation\}$ can be reused in the calls
  of the threshold algorithms.
  That is, the calls of the threshold algorithm that are invoked
  in the scheduler-improvement steps at level $r$ can skip levels
  $\saturation,\saturation{-}1,\ldots,r{+}1$ and only need to process
  levels $r,r{-}1,\ldots,1,0$.}
For the current level $r$, the algorithm also computes
for each state $s\in S \setminus \{\goal,\fail\}$ and
each action $\alpha \in \Act(s)$ the values
$y_{s,r,\alpha}  =  \sum_{t\in S} P(s,\alpha,t) \cdot y_{t,R}$
and
$\theta_{s,r,\alpha}  =
   \rew(s,\alpha) \cdot y_{s,r,\alpha} \ + \
   \sum_{t\in S} P(s,\alpha,t) \cdot \theta_{t,R}$
where $R = \min \{ \saturation, r+\rew(s,\alpha)\}$.

\tudparagraph{1ex}{{\it Scheduler-improvement step.}}
Let $r$ be the current level,
$I = [A,B[$ the current interval
and $\sched$ the current scheduler with $\CExp{\max}\in I$.
At the beginning of the scheduler-improvement step we have
$\CExp{\sched} =A$.
Let
$$
\begin{array}{rcl}
  \cI_{\sched,r} & \ = \ &
  \Bigl\{ \
     r + \frac{\theta_{s,r} - \theta_{s,r,\alpha}}{y_{s,r}-y_{s,r,\alpha}}
     \ : \
     s \in S \setminus \{\goal,\fail\}, \
     \alpha \in \Act(s), \
     y_{s,r} > y_{s,r,\alpha}
     \
  \Bigr\}
  \\
  \\[-1ex]
  \cI^{\uparrow}_{\sched,r} & = &
  \bigl\{ \ d \in \cI_{\sched,r} \ : \
            d \ \geqslant \ \CExp{\sched}
  \bigr\}
  \hspace*{1cm}
  \cI^{B}_{\sched,r} \ \ = \ \
   \bigl\{ \ d \in \cI_{\sched,r} \ : \
             d < B \
  \bigr\}
\end{array}
$$
Intuitively, the values in $d \in \cI^B_{\sched,r}$ are the ``most promising''
threshold values, as according to statement \eqref{magic-constraint} 
these are the points where the decision of the
current scheduler $\sched$ for some state-reward pair $(s,r)$ can be
improved, provided that $\CExp{\max} > d$.
(Note that the values in $\cI_{\sched,r}\setminus \cI^B_{\sched,r}$ 
can be discarded as $\CExp{\max} < B$.)

The algorithm proceeds as follows.
If $\cI^{B}_{\sched,r} = \varnothing$ then no further improvements
at level $r$ are possible as the function $\psched^* = \sched(\cdot,r)$
satisfies \eqref{difference} for the (still unknown) value
$\threshold=\CExp{\max}$.
\CiteAppendix{See Lemma~\ref{lemma:freeze-level-r} in the appendix. }%
In this case:
\begin{itemize}
  \item
    If $r=0$ then the algorithm terminates with the
    answer $\CExp{\max}=\CExp{\sched}$ and $\sched$ as an optimal scheduler.
  \item
    If $r > 0$ then the algorithm goes to the next level $r{-}1$
    and performs the scheduler-improvement step for $\sched$ at level
    $r{-}1$.
\end{itemize}
Suppose now that $\cI^B_{\sched,r}$ is nonempty.
Let $\cK = \cI^{\uparrow}_{\sched,r} \cup \{\CExp{\sched}\}$.
The algorithm seeks for the largest
value $\threshold' \in \cK \cap I$
such that $\CExp{\max}\geqslant \threshold'$.
More precisely, it successively
calls the threshold algorithm for the threshold value
$\threshold'=\max (\cK \cap I)$
and performs the following steps
for the generated scheduler
$\sched' = \SchedThreshold{\threshold'}$:
\begin{itemize}
\item
 If the result of the threshold algorithm is ``no''
 and $\Pr^{\sched'}_{\cM,\sinit}(\Diamond \goal)$ is positive
 (in which case
  $\CExp{\sched'} \leqslant \CExp{\max} < \threshold'$), then:
 \begin{itemize}
 \item
   If $\CExp{\sched'} \leqslant A$ then the algorithm
   refines $I$ by putting $B:=\threshold'$.
 \item
    If $\CExp{\sched'} > A$
    then the algorithm refines $I$ by putting $A := \CExp{\sched'}$,
    $B := \threshold'$
    and adds $\CExp{\sched'}$ to $\cK$
    (Note that then $\CExp{\sched'}\in \cK \cap I$, while
     $\CExp{\sched} \in \cK \setminus I$.)
 \end{itemize}
\item
  Suppose now that
  $\CExp{\sched'} \geqslant \threshold '$.
  The algorithm terminates if $\CExp{\sched'} = \threshold'$,
  in which case $\sched'$ is optimal.
  Otherwise, i.e., if $\CExp{\sched'} > \threshold'$, then the algorithm
  aborts the loop by putting $\cK:= \varnothing$,
  refines the interval $I$ by putting $A := \CExp{\sched'}$,
  updates the current scheduler by setting $\sched := \sched'$
  and performs the next scheduler-improvement step.
\end{itemize}
The soundness proof and complexity analysis
can be found in
\CiteAppendix{Appendix~\ref{appendix:cexpmax}, }%
where (among others)
we show that the scheduler-improvement step for schedulers
$\sched$ with $\CExp{\sched}<\CExp{\max}$ terminates
with some scheduler $\sched'$ such that $\CExp{\sched} < \CExp{\sched'}$.
The total number of calls of the threshold algorithm
is in $\cO(\saturation \cdot \md \cdot |S|\cdot |\Act|)$.
This yields an exponential time bound as stated in
Theorem~\ref{thm:computing-cexpmax}.

\begin{example}
{\rm
We regard again the MDP $\cM[\rewparam]$ of 
Example \ref{example:running-example} where we suppose
$\rewparam$ is positive and even.
The algorithm first computes $\CExp{\ub}$ 
(see Section \ref{sec:finiteness}), a saturation point
$\saturation \geqslant \rewparam {+} 2$ (see Example \ref{example:saturation}),
the scheduler $\maxsched$, its conditional expectation
$\CExp{\maxsched}=\frac{\rewparam}{2}$ 
and the scheduler $\sched=\SchedThreshold{\frac{\rewparam}{2}}$.
The initial interval is $I=[A,B[$ where 
$A=\CExp{\sched} = 
 \rewparam - (\frac{\rewparam}{2}{-}1)/(2^{\frac{\rewparam}{2}+1}{+}1)$
(see Example \ref{example:threshold-algo})
and $B = \CExp{\ub}{+}1$.
The scheduler improvement step for $\sched$ at levels 
$r=\saturation {-}1,\ldots,\rewparam {+}1$
determines the set $\cI_{\sched,r}=\{r{-}1\}$ and calls the
threshold algorithm for $\threshold'=r{-}1$.
These calls are not successful for 
$r=\saturation {-}1,\ldots,\rewparam {+}2$. That is, the scheduler
$\sched$ remains unchanged and the upper bound $B$ is successively
improved to $r{-}1$.
At level $r=\rewparam {+}1$, the threshold algorithm is 
called for $\threshold' = \rewparam$, which yields
the optimal scheduler $\sched'=\SchedThreshold{\threshold'}$
(see Example \ref{example:threshold-algo}).
\Ende}
\end{example}

\tudparagraph{0ex}{Implementation and experiments.}
We have implemented the algorithms presented in this paper
as a prototypical extension of the model checker
PRISM~\cite{prism40,PrismWebSite} and carried out initial
experiments to demonstrate the general feasibility of our approach
(see 
\mbox{\url{https://wwwtcs.inf.tu-dresden.de/ALGI/PUB/TACAS17/}}
and
\CiteAppendix{Appendix~\ref{appendix:implementation} }%
 for details).

\section{Conclusion}
\label{sec:conclusion}

Although the switch to conditional expectations appears rather natural
to escape from the limitations of known solutions for unconditional
extremal expected accumulated rewards,
to the best of our knowledge
computation schemes for conditional expected accumulated rewards
have not been addressed before.
Our results show that new techniques are needed to compute
maximal conditional expectations,
as optimal schedulers might need memory
and local reasoning in terms of the past and possible future
is not sufficient (Example \ref{example:running-example}).
The key observations for our algorithms are the existence of a saturation
point $\saturation$ for the reward that has been accumulated so far,
from which on optimal schedulers can behave memoryless, and
a linear correlation between optimal decisions for all
state-reward pairs $(s,r)$ of the same reward level $r$
(see \eqref{difference} and the linear program used in the threshold
algorithm).
The difficulty to reason about conditional expectations
is also reflected in the achieved complexity-theoretic
results stating that all variants of the threshold problem lie
between PSPACE and EXPTIME.
While PSPACE-completeness has been established for acyclic MDPs
\CiteAppendix{(Appendix~\ref{appendix:PSPACE}), }%
the precise complexity for cyclic MDPs is still open.
  In contrast, optimal schedulers for unconditional expected accumulated
  rewards as well as for conditional reachability probabilities
  are computable in polynomial time
  \cite{deAlf99,BKKM14}.

  Using standard automata-based approaches, our method can easily
  be generalized to compute maximal conditional expected rewards for
  regular co-safety conditions (rather than reachability conditions
  $\Diamond G$) and/or where the accumulation of rewards is 
  ``controlled'' by a deterministic finite automaton
  as in the logics considered in \cite{BokChaHenKupf11,BKKW14}
  (rather than $\accdiaplus F$).
 In this paper, we restricted to MDPs with non-negative integer rewards.
 Non-negative rational rewards can be treated by 
 multiplying all reward values with their least common multiple
\CiteAppendix{(Appendix~\ref{appendix:rational-rewards}). }%
 In the case of acyclic MDPs, our methods are even applicable
 if the MDP has negative and positive rational rewards
\CiteAppendix{ (Appendix~\ref{appendix:negative-rewards}). }%
 By swapping the sign of all rewards, this 
 yields a technique to compute minimal conditional expectations
 in acyclic MDPs.
 We expect that minimal conditional
 expectations in cyclic MDPs with non-negative rewards
 can be computed using 
 similar algorithms as we suggested for maximal conditional expectations.
 This as well as 
 MDPs with negative and positive rewards will be addressed in
 future work.

\bibliographystyle{abbrv}  %
\bibliography{lit-compact}

\newpage
\appendix

\section*{APPENDIX}

\tudparagraph{1ex}{Outline of the appendix.}

\begin{itemize}
\item
  Appendix~\ref{appendix:notations}
  (Relevant notations for the appendix, page~\pageref{appendix:notations})
  explains the notations used in the appendix.
\item 
 Appendix~\ref{appendix:reset-Methode}
 (Extremal conditional probabilities, page~\pageref{appendix:reset-Methode})
 sketches the reset mechanism of \cite{BKKM14}
 for computing
 maximal conditional probabilities (used in
 Section \ref{sec:finiteness} to compute an upper bound
 $\CExp{\ub}$). It also provides an example illustrating
 why the reset method fails for conditional
 expectations.
\item
 Appendix~\ref{appendix:finiteness}
 (Finiteness and upper bound, page~\pageref{appendix:finiteness})
 deals with determining finiteness and the computation of an upper
 bound (Theorem~\ref{thm:finiteness}).
\item
 Appendix~\ref{sec:reward-based}
 (Deterministic reward-based schedulers are sufficient, page~\pageref{sec:reward-based})
 shows that deterministic reward-based schedulers are sufficient for
 maximizing the conditional expectation in our setting.
\item
 Appendix~\ref{appendix:saturation}
 (Existence of a saturation point and optimal schedulers, page~\pageref{appendix:saturation})
 provides the proof for Proposition~\ref{prop:saturation-maxsched}.
\item
 Appendix~\ref{appendix:compute-saturation}
 (Computing a saturation point, page~\pageref{appendix:compute-saturation})
 shows the correctness of the computation of a saturation point
 as described in Section~\ref{sec:threshold}.
\item
 Appendix~\ref{appendix:threshold}
 (Threshold algorithm, page~\pageref{appendix:threshold})
 provides additional details for the threshold algorithm
 as well as the proof of the soundness proof and 
 the exponential time bound as stated in
 Theorem~\ref{thm:threshold-problem}.
\item 
 Appendix~\ref{appendix:cexpmax}
 (Computing an optimal scheduler and the maximal conditional expectation, page~\pageref{appendix:cexpmax}) 
 provides additional details for the scheduler improvement algorithm 
 of Section~\ref{sec:threshold}
 and the proof of Theorem~\ref{thm:computing-cexpmax}.
\item 
 Appendix~\ref{appendix:PSPACE}
 (PSPACE-completeness for acyclic MPDs, page~\pageref{appendix:PSPACE}) 
 provides the PSPACE-completeness proof of the threshold problem
 in acyclic MDPs as stated in 
 Theorem~\ref{thm:threshold-problem}.
\item 
 Appendix~\ref{appendix:rational}
 (Rational and negative rewards, page~\pageref{appendix:rational})
 provides details for rational and negative rewards as mentioned
 in the conclusion.
\item 
 Appendix~\ref{appendix:implementation}
 (Implementation and experiments, page~\pageref{appendix:implementation})
 provides details on our prototypical implementation in PRISM and our
 experiments.
\end{itemize}

\section{Relevant notations for the appendix}

\label{appendix:notations}
\label{sec:notations}

\tudparagraph{1ex}{Notations for Markov decision processes.}
A \emph{Markov decision process} (MDP) is a tuple $\cM = (S,\Act,P,\sinit,\rew)$
where $S$ is a finite set of states,
$\Act$ a finite set of actions,
$\sinit \in S$ the initial state,
$P : S \times \Act \times S \to [0,1] \cap \Rational$ is the
transition probability function and
$\rew : S \times \Act \to \Nat$ the reward function.
We require that
$P(s,\act,S) \in \{0,1\}$
for all $(s,\alpha)\in S\times \Act$ where, for $F \subseteq S$,
$P(s,\act,F) = \sum_{t\in F} P(s,\act,t)$.
We write $\Act(s)$ for the set of actions that are enabled in $s$,
i.e., $\act \in \Act(s)$ iff $P(s,\act,\cdot)$ is not the null function.
State $s$ is called a \emph{trap} if $\Act(s)=\varnothing$.

The paths of $\cM$ are finite or
infinite sequences $s_0 \, \act_0 \, s_1 \, \act_1 \, s_2 \, \act_2 \ldots$
where states and actions alternate such that
$P(s_i,\act_i,s_{i+1}) >0$ for all $i\geqslant 0$.
A path $\fpath$ is called \emph{maximal} if it is either infinite or
finite and its last state is a trap.
If $\fpath =
    s_0 \, \act_0 \, s_1 \, \act_1 \, s_2 \, \act_2 \ldots \act_{k-1} \, s_k$
is finite then
\begin{itemize}
\item
  $\rew(\fpath)=
   \rew(s_0,\act_0) + \rew(s_1,\act_1) + \ldots + \rew(s_{k-1},\act_{k-1})$
  denotes the accumulated reward,
\item
  $|\fpath|=k$ its length (number of transitions),
\item
  $\probability(\fpath) =
   P(s_0,\act_0,s_1) \cdot P(s_1,\act_1,s_2)
   \cdot \ldots \cdot P(s_{k-1},\act_{k-1},s_k)$
   its probability,
\item
  $\first(\fpath)=s_0$, $\last(\fpath)=s_k$ its first resp.~last state.
\end{itemize}
The notation $\fpath_1; \fpath_2$ is used to denote the concatenation of paths
$\fpath_1$ and $\fpath_2$ with $\last(\fpath_1)=\first(\fpath_2)$, where
$\fpath_1 ; \fpath_2 = \fpath_2$ if $|\fpath_1|=0$.

The \emph{size} of $\cM$, denoted $\Size(\cM)$,
is the sum of the number of states
plus the total sum of the logarithmic lengths of the non-zero
probability values
$P(s,\alpha,s')$ and the reward values $\rew(s,\alpha)$.

\tudparagraph{1ex}{Scheduler.}
A \emph{(randomized) scheduler} for $\cM$,
often also called policy or adversary,
is a function $\sched$ that assigns to each finite path $\fpath$ where
$\last(\fpath)$ is not a trap
a probability distribution over $\Act(\last(\fpath))$.
$\sched$ is called memoryless if $\sched(\fpath)=\sched(\fpath')$ for
all finite paths $\fpath$, $\fpath'$ with $\last(\fpath)=\last(\fpath')$,
in which case $\sched$ can be viewed as a function
that assigns to each non-trap state $s$ a distribution over $\Act(s)$.
$\sched$ is called deterministic if $\sched(\fpath)$ is a Dirac distribution
for each path $\fpath$,
in which case $\sched$ can be viewed as a function that assigns an action
to each finite path $\fpath$ where $\last(\fpath)$ is not a trap.
Given a scheduler $\sched$,
a \emph{$\sched$-path} is any path that might arise when the nondeterministic
choices in $\cM$ are resolved using $\sched$. Thus,
$\infpath \, = \, s_0 \, \act_0 \, s_1 \, \act_1 \ldots$
is a $\sched$-path iff $\infpath$ is a path and
$\sched(s_0 \, \act_0 \, s_1 \, \act_1 \ldots \act_{k-1} \, s_k)(\act_k)>0$
for all $k \geqslant 0$.
If $\fpath$ is a finite $\sched$-path then
$\residual{\sched}{\fpath}$ denotes the residual scheduler
``$\sched$ after $\fpath$'' given by:
\begin{center}
  $(\residual{\sched}{\fpath})(\finpath)$
  \ $=$ \
  $\sched(\fpath;\finpath)$
  \ \
  if $\finpath$ is a finite path with $\last(\fpath)=\first(\finpath)$
\end{center}
The behavior of $\residual{\sched}{\fpath}$
for paths not starting in $\last(\fpath)$
is irrelevant.

Let  $\sched$ and $\usched$ be schedulers.
If $\fpath$ is a finite $\sched$-path $\fpath$, then
$\redefresidual{\sched}{\fpath}{\usched}$ denotes the unique scheduler
$\tsched$ with $\residual{\tsched}{\fpath}=\usched$
that behaves as $\sched$ for all paths $\finpath$ where $\fpath$
is not a prefix of $\finpath$. That is:
$$
  (\redefresidual{\sched}{\fpath}{\usched})(\finpath)
  \ \ = \ \
  \left\{
    \begin{array}{lcl}
      \sched(\finpath) & : &
      \text{if $\fpath$ is not a prefix of $\finpath$}
      \\
      \usched(\finpath') & : &
      \text{if $\finpath = \fpath ; \finpath'$}
     \end{array}
  \right.
$$
where $\fpath ; s = \fpath$.
Hence,
$(\redefresidual{\sched}{\fpath}{\usched})(\fpath) =
 \usched( \last(\fpath))$.
If $R \in \Nat$ then $\redefresidual{\sched}{R}{\usched}$ denotes the
scheduler given by
$(\redefresidual{\sched}{R}{\usched})(\finpath)=\sched(\finpath)$
if $\rew(\finpath) < R$ and
$(\redefresidual{\sched}{R}{\usched})(\fpath;\finpath')=\usched(\finpath')$
if $\rew(\fpath)\geqslant R$ and $\rew(\fpath') < R$ for all proper
prefixes $\fpath'$ of $\fpath$.

Scheduler $\sched$ is said to be \emph{reward-based} if
$\sched(\fpath)=\sched(\fpath')$ for all finite paths $\fpath$, $\fpath'$
with $\rew(\fpath)=\rew(\fpath')$ and $\last(\fpath)=\last(\fpath')$.
Thus, deterministic reward-based schedulers
can be viewed as function that assign
actions to state-reward pairs.
As stated in the main paper, for reward-based schedulers we use the
notation $\residual{\sched}{R}$ to denote the reward-based scheduler
given by $(\residual{\sched}{R})(s,r) = \sched(s,R{+}r)$.
We use the notation $\residual{\sched}{(s,R)}$
to denote the scheduler $\residual{\sched}{\fpath}$
for each/some finite $\sched$-path $\fpath$ with $\rew(\fpath)=R$ and
$\last(\fpath)=s$.%
\footnote{Note that if $\sched$ is reward-based then
   $\residual{\sched}{\fpath}= \residual{\sched}{\fpath'}$
   for all finite paths $\fpath$, $\fpath'$ with
   $(\last(\fpath),\rew(\fpath))=(\last(\fpath'),\rew(\fpath'))$.}
Clearly, if $\sched$ and $\usched$ are reward-based schedulers and $R\in \Nat$
then $\redefresidual{\sched}{R}{\usched}$
is a reward-based  scheduler too.  In this case,
$(\redefresidual{\sched}{R}{\usched})(s,r) = \sched(s,r)$ for $r < R$
and
$(\redefresidual{\sched}{R}{\usched})(s,r) = \usched(s,r{-}R)$
for $r \geqslant R$. Hence,
$\residual{(\redefresidual{\sched}{R}{\usched})}{R}=\usched$.

\tudparagraph{1ex}{Probability measure.}
We write $\Pr^{\sched}_{\cM,s}$ or briefly $\Pr^{\sched}_{s}$
to denote the probability measure induced by $\sched$ and $s$.
The underlying
sigma-algebra is the one that is generated by the cylinder sets
of the finite paths starting in $s$
where the cylinder set $\Cyl(\fpath)$ of $\fpath$ consists of all
maximal paths $\infpath$ that are extensions of $\fpath$.
Then, $\Pr^{\sched}_{\cM,s}$ is the unique probability measure
such that for each finite path
$\fpath =
  s_0 \, \act_0 \, s_1 \, \act_1 \, s_2 \, \act_2 \ldots \act_{k-1} \, s_k$:
\begin{center}
 $\Pr^{\sched}_{\cM,s}\bigl(\, \Cyl(\fpath)\, \bigr)
  \ \ = \ \
  \probability(\fpath)\cdot
  \prod\limits_{i=0}^{k-1} \sched(\prefix{\fpath}{i})(\act_i)$
\end{center}
where
$\prefix{\fpath}{i} =
s_0 \, \act_0 \, s_1 \, \act_1 \, s_2 \, \act_2 \ldots \act_{i-1} \, s_i$.
Thus, $\Pr^{\sched}_{\cM,s}(\Cyl(\fpath))=0$
if $\fpath$ is not a $\sched$-path.
Given a measurable set $\psi$ of maximal paths, then
$\Pr^{\min}_{\cM,s}(\psi) = \inf_{\sched} \Pr^{\sched}_{\cM,s}(\psi)$
and
$\Pr^{\max}_{\cM,s}(\psi) = \sup_{\sched} \Pr^{\sched}_{\cM,s}(\psi)$
where $\sched$ ranges over all schedulers for $\cM$.
If $\Pr^{\sched}_{\cM,s}(\psi) >0$ then the
conditional probability measure
$\Pr^{\sched}_{\cM,s}(\ \cdot \ | \psi)$ is
given by
$\Pr^{\sched}_{\cM,s}(\varphi | \psi) =
 \Pr^{\sched}_{\cM,s}(\varphi \cap \psi)/\Pr^{\sched}_{\cM,s}(\psi)$.
We write
$\Pr^{\max}_{\cM,s}(\varphi | \psi)$
for the supremum of the values $\Pr^{\sched}_{\cM,s}(\varphi | \psi)$
where $\sched$ ranges over all schedulers with
$\Pr^{\sched}_{\cM,s}(\psi) >0$.

We often use LTL-like notations with
the temporal modalities $\neXt$ (next),
$\Diamond$ (eventually), $\Box$ (always) and
$\Until$ (until)
to specify measurable sets of maximal paths.
For these it is well-known that optimal deterministic schedulers
exists. If $\psi$ is a reachability condition then even optimal deterministic
memoryless schedulers exist.
Let $\varnothing \not= F \subseteq S$.
If $\bowtie$ a comparison operator (e.g. $=$ or $\leqslant$)
and $r\in \Nat$
then $\Diamond^{\bowtie r} F$ denotes the event
``reaching $F$
along some finite path $\fpath$ with $\rew(\fpath)\bowtie r$'',
while $\neXt^{\bowtie n} F$ denotes
``reaching $F$
along some finite path $\fpath$ with $|\fpath|\bowtie r$''.

\tudparagraph{1ex}{(Conditional) expected rewards.}
If $F \subseteq S$ then $\accdiaplus F$ denotes
the random variable that assigns to each maximal
path $\infpath$ in $\cM$ the reward $\rew(\fpath)$ of the shortest prefix
$\fpath$ of $\infpath$ where $\last(\fpath)\in F$.
If $\infpath \not\models \Diamond F$ then $(\accdiaplus F)(\infpath)=\infty$.
If $s\in S$ then $\ExpRew{\sched}{\cM,s}(\accdiaplus F)$ denotes
the expectation of $\accdiaplus F$ in $\cM$ with starting state $s$
under $\sched$, which is infinite if
$\Pr^{\sched}_{\cM,s}(\Diamond F) <1$.
$\ExpRew{\max}{\cM,s}(\accdiaplus F) \in \Real \cup \{\pm\infty\}$ stands for
$\sup_{\sched} \ExpRew{\sched}{\cM,s}(\accdiaplus F)$ where the supremum
is taken over all schedulers $\sched$ with
$\Pr^{\sched}_{\cM,s}(\Diamond F)=1$ and $\sup \varnothing=-\infty$.
If $\psi$ is a measurable set of maximal paths and
$\Pr^{\sched}_{\cM,s}(\psi)>0$
then
$\ExpRew{\sched}{\cM,s}(\accdiaplus F|\psi)$ stands for the expectation
of $\accdiaplus F$ w.r.t.~the conditional probability
measure $\Pr^{\sched}_{\cM,s}(\ \cdot \ | \psi)$.
$\ExpRew{\max}{\cM,s}(\accdiaplus F|\psi) \in \Real \cup \{\pm\infty\}$
denotes the supremum of these
conditional expectations when ranging over all schedulers $\sched$
where $\Pr^{\sched}_{\cM,s}(\psi)>0$ and
$\Pr^{\sched}_{\cM,s}(\Diamond F|\psi)=1$.

\tudparagraph{1ex}{End components, MEC-quotient.}
An \emph{end component} of $\cM$ is a strongly connected sub-MDP. End components
can be formalized as pairs $\cE = (E,\ActEC)$ where $E$ is a nonempty subset
of $S$ and $\ActEC$ a function that assigns to each state $s\in E$ a nonempty
subset of $\Act(s)$ such that the graph induced by $\cE$ is strongly connected.
$\cE$ is called \emph{maximal} if there is no end component
$\cE' = (E',\ActEC')$ with $\cE \not= \cE'$, $E \subseteq E'$
and $\ActEC(s) \subseteq \ActEC'(s)$ for all $s\in E$.
$\cE$ is called \emph{positive} if there exists a state-action pair
$(s,\alpha)$ with $s\in E$, $\alpha\in \ActEC(s)$ and
$\rew(s,\alpha)>0$.

The \emph{MEC-quotient} of an MDP $\cM$ is the MDP $\MEC(\cM)$ arising
from $\cM$
by collapsing all states that belong to the same maximal end component
\cite{CBGK08}. Formally, we consider the equivalence relation $\simMEC$
on the state space $S$ of $\cM$
given by $s \simMEC t$ iff $s=t$ is not contained in some end component
or $s$ and $t$ belong to the same maximal end component.
Let $[s]$ denote the equivalence class of state $s$ with respect to $\simMEC$.
The state space of $\MEC(\cM)$ is $S/\simMEC = \{[s] : s\in S\}$.
The action set in $\MEC(\cM)$ is $(S \times \Act) \cup \{\tau\}$.
Action $(s,\act)$ is enabled in state $E \in S/\simMEC$ of $\MEC(\cM)$ iff
$s\in E$ and $P(s,\act,t)>0$ for at least one state $t \in S \setminus E$.
In this case, the transition probabilities in $\MEC(\cM)$ are given by
$P_{\text{MEC}}(E,(s,\act),F) = P(s,\act,F)$.
If $E$ is a bottom end component then $E$ is trap state in $\MEC(\cM)$.
If $G \subseteq S$ consists of trap states then the maximal probability
to reach $G$ in $\cM$ from  $\sinit$ agrees with
the maximal probability to reach $G$ in $\MEC(\cM)$
(where we identify $[s]$ and $s$ if $s$ is a trap).
Reward functions can be lifted by
$\rew_{\text{MEC}}(E,(s,\act)) = \rew(s,\act)$.
However, for reasoning about reward-bounded properties the switch from $\cM$
to $\MEC(\cM)$ can cause problems. In general it is only justified
if all state-action pairs that belong to some end component have reward 0.

\section{Extremal conditional probabilities}

\label{appendix:reset-Methode}

We provide here a high-level overview of the reset-mechanism
presented in \cite{BKKM14}
for the computation of maximal conditional probabilities for reachability
objectives $\Diamond F$ and conditions $\Diamond G$ where
$F,G \subseteq S$:
$$
  \Pr^{\max}_{\cM,\sinit} (\Diamond F | \Diamond G)
  \ \ = \ \
  \sup_{\sched} \ \Pr^{\sched}_{\cM,\sinit}( \Diamond F | \Diamond G)
$$
where $\sched$ ranges over all
schedulers with $\Pr^{\sched}_{\cM,\sinit}(\Diamond G) > 0$.
The approach of \cite{BKKM14} uses a transformation of $\cM$ to a new
MDP $\cM'$. It relies on the observation that once $G$ has been
reached, the optimal behavior is to maximize reaching $F$. Similarly,
once $F$ has been reached, the optimal behavior is to maximize
reaching $G$. This allows to capture the three relevant outcomes
``goal'' (both $F$ and $G$ have been seen), ``stop'' (``$G$ but not
$F$ have been seen) and ``fail'' (``$G$ has not been seen'') by
special trap states called $\goal$, $\mathit{stop}$ and $\fail$.
More precisely, transition to $\fail$ are inserted for all states $s$ where
$\Pr^{\min}_{\cM,s}(\Diamond G)=0$.
By adding a reset-transition with probability $1$ from $\fail$
back to $\sinit$, the probabilities for the paths that
never visit $G$ are ``redistributed'' to the
paths that eventually enter $G$. This yields:
$$
  \Pr^{\max}_{\cM,\sinit} (\Diamond F | \Diamond G)
  \ \ = \ \
  \Pr^{\max}_{\cM',\sinit} (\Diamond \goal),
$$
Thus, the computation of the maximal conditional probability for
reachability objectives and conditions can be
reduced to the computation of a maximal unconditional reachability
probability, yielding a polynomial time bound.

The following example illustrates why the reset-mechanism
presented in \cite{BKKM14} is not adequate
to compute maximal conditional expected accumulated rewards.

\begin{example}[Reset-mechanism fails for conditional expectations]
Consider the Markov chain $\cM$ consisting of the initial state
$s=\sinit$ and two trap states $\goal$ and $\fail$
with the transition probabilities $P(s,\goal)= P(s,\fail)=\frac{1}{2}$.
Suppose the reward for state $s$ is 1 and that $F=G=\{\goal\}$.%
\footnote{As $\cM$ is a Markov chain,
  action names for the transitions are irrelevant
  and can be omitted. Thus, the states of $\cM$ are decorated with
  reward values.}
Then, the expected accumulated reward for reaching $\goal$
under the condition to reach $\goal$ is simply the reward of the
path $\fpath = s \, \goal$, which is 1.
The reset-mechanism introduced for the computation of conditional
reachability probabilities introduces a reset-transition
from $\fail$ to $\sinit$.
For $\psi=\Diamond \goal$, taking the reset-transition
in the resulting Markov chain $\cM'$
corresponds to discarding all paths that eventually enter $\fail$
and redistributing
their probabilities to the successful paths that eventually reach $\goal$.
Thus, the (only) successful $\fpath$ in $\cM$
is mimicked in $\cM'$ by the paths
$\fpath_n = (s \, \fail)^n \, s \, \goal$.
The total probability of the (cylinder sets spanned by)
paths $\fpath_n$ in $\cM'$ is 1, which agrees with
the conditional probability of $\fpath$ in $\cM$. However, the
rewards of the paths $\fpath_n$ are different from $\rew(\fpath)$.
Indeed we have:
$$
 \begin{array}{lclclclcl}
  \Exp{}{\cM',s}(\accdiaplus \goal)
  & \ \ = \ &
  \sum\limits_{n=1}^{\infty} \, \bigl(\frac{1}{2}\bigr)^n \cdot n
  & \ \ = \ \ & 2
  & \ \ > \ \ & 1 & \ \ = \ \ &
  \Exp{}{\cM,s}(\ \accdiaplus \goal \ | \ \Diamond \goal \ )
 \end{array}
$$
if we assign reward 0 to the reset-transition
(resp.~to state $\fail$).
\Ende
\end{example}

\section{Finiteness and upper bound}
\label{appendix:finiteness}

This section provides the proof for Theorem~\ref{thm:finiteness}
and the details and soundness proofs for the methods
to check finiteness of maximal conditional expectations
and to compute an upper bound as outlined in
Section~\ref{sec:finiteness}.
Throughout this section, we suppose that the given MDP 
$\cM = (S,\Act,P,sinit,\rew)$ has two
distinguished sets $F$ and $G$ of states such that there is at least one
scheduler $\sched$ with 
$\Pr^{\sched}_{\cM,\sinit}(\Diamond G)>0$ and
$\Pr^{\sched}_{\cM,\sinit}(\Diamond F |\Diamond G)=1$.
This condition can be checked in polynomial time using the reset-approach
for conditional probabilities of \cite{BKKM14} 
(see also Section \ref{appendix:reset-Methode}).

\subsection{Finiteness -- preprocessing and normal form transformation}
\label{appendix:finitness-transformation}

We now present the details of the preprocessing and
normal form transformation.
After some cleaning-up (step 1), we describe the transformations 
$\cM \leadsto \tilde{\cM} \leadsto \tilde{\cM}'$ (steps 2 and 3).
Section \ref{appendix:finitness-critical-schedulers} 
will then transform
$\tilde{\cM}'$ into the MDP $\Hut{\cM}$ satisfying properties
(1) and (2) presented in Section \ref{sec:finiteness}, 
which will then be subject for checking finiteness.

\tudparagraph{2ex}{Step 1: cleaning-up and assumptions.}
Obviously, all states not reachable from $\sinit$ can be removed without
affecting the conditional expected accumulated rewards from $\sinit$.
Thus, it is no restriction to suppose that all states $s\in S$
are reachable from $\sinit$.
We can also safely assume that $\sinit \notin F \cup G$.
Note that $\sinit \in F$ would imply that the accumulated reward until $F$
is 0 under each scheduler,
while assumption $\sinit \in G$ would yield that
$\Pr^{\sched}_{\cM,\sinit}(\Diamond G)=1$ for all schedulers $\sched$,
in which case standard linear-programming techniques to compute
(unconditional)
maximal
expected accumulated rewards can be applied.

\tudparagraph{2ex}{Step 2: normal form transformation.}
We first show that there is a transformation $\cM \mapsto \tilde{\cM}$
that permits to assume that $F=G$.
Intuitively, $\tilde{\cM}$ operates in four modes:
``normal mode'', ``after $G$'', ``after $F$'' and ``goal''.
$\tilde{\cM}$ starts in normal mode where it
behaves as $\cM$ as long as neither $F$ nor $G$ have been visited.
\begin{itemize}
\item
  If a $G\setminus F$-state $u$ has been reached in normal mode
  then $\tilde{\cM}$ switches to the mode ``after $G$'' where again
  it simulates $\cM$ and attempts to reach $F$.
\item
  If an $F\setminus G$-state $t$ has been reached in normal mode
  then $\tilde{\cM}$ switches to the mode ``after $F$''
  still simulating $\cM$ and expecting to reach a $G$-state.
\item
  $\tilde{\cM}$ enters the goal mode (consisting of a single trap state)
  as soon as a path fragment
  containing a state in $F$ and a state in $G$ has been generated,
  which is the case if $\cM$ visits an $F$-state in mode ``after $G$''
  or visits a $G$-state in mode ``after $F$'',
  or visits a state in $F \cap G$ in the normal mode.
\end{itemize}
The rewards in the normal mode and in mode ``after $G$''
are precisely as in $\cM$,
while the rewards are 0 in all other cases.
The objective for $\tilde{\cM}$ is then to find a scheduler
$\tsched$ such that
$\Pr^{\tsched}_{\tilde{\cM},\sinit}(\Diamond \goal)$
is positive and that optimizes the accumulated reward to reach
the goal-state under the condition that the goal state will indeed be
reached.\\

\noindent Formally, the state space of $\tilde{\cM}$ is
$\tilde{S} \ = \
 S \cup \afterG{S} \cup \afterF{S} \cup \{\goal\}$
where  $\afterG{S}$ and $\afterF{S}$
consist of pairwise distinct copies of all states in $\cM$.
More generally, for $U \subseteq S$,
$\afterG{U} = \{ \afterG{s} : s \in U \}$ and
$\afterF{U} = \{ \afterF{s} : s \in U \}$ 
with pairwise distinct, fresh states $\afterG{s}$ and $\afterF{s}$.

The action set of $\tilde{\cM}$ is the action set $\Act$ of $\cM$
extended by a fresh action $\tau$.
The transition probability function $\tilde{P}$ and the reward function
$\tilde{\rew} : \tilde{S} \times \Act \to \Nat$
of $\tilde{\cM}$ are defined as follows.
The new state $\goal$ is a trap.
The transition probabilities for the normal mode are as follows
(where $v$ ranges over all states $s\in S$):
\begin{itemize}
\item
   If $s\in S \setminus (F \cup G)$ then
   $\Act_{\tilde{\cM}}(s)=\Act_{\cM}(s)$,
   $\tilde{P}(s,\act,v)  =  P(s,\act,v)$,
   $\tilde{\rew}(s,\act) =  \rew(s,\act)$.
\item
   If $s \in F \cap G$ then $\Act_{\tilde{\cM}}(s)=\{\tau\}$ and
   $\tilde{P}(s,\tau,\goal)=1$,
   $\tilde{\rew}(s,\tau) =0$.
\item
   The switches from normal mode to mode ``after $G$'' resp.
   ``after $F$'' are formalized as follows.
   If $u\in G \setminus F$ and $t\in F \setminus G$ and $v\in S$ then:
   $$
    \begin{array}{l@{\hspace*{0.25cm}}c@{\hspace*{0.25cm}}l%
                  l@{\hspace*{0.25cm}}c@{\hspace*{0.25cm}}l}
       \tilde{P}(u,\act,\afterG{v})
       & = & P(u,\act,v), \ \ \ \
       &
       \tilde{\rew}(u,\act) & = & \rew(u,\act),
       \\[1ex]

       \tilde{P}(t,\act,\afterF{v})
       & = & P(t,\act,v), \ \ \
       &
       \tilde{\rew}(t,\act) & = & 0
    \end{array}
  $$
  and action $\tau$ is not enabled in $u$ or $t$, i.e.,
  $\Act_{\tilde{\cM}}(s)=\Act_{\cM}(s)$ for all
  $s \in (G \setminus F) \cup  (F \setminus G)$.
\end{itemize}
The transition probabilities and rewards in the
modes ``after $G$''  are defined as follows.
\begin{itemize}
\item
    If $s \in S \setminus F$ then
    $\Act_{\tilde{\cM}}(\afterG{s})=\Act_{\cM}(s)$ and
    $\tilde{P}(\afterG{s},\act,\afterG{v})
        =  P(s,\act,v)$,
    $\tilde{\rew}(\afterG{s},\act) =  \rew(s,\act)$.
\item
   If $s\in F$ then $\Act_{\tilde{\cM}}(\afterG{s})= \{\tau\}$,
   $\tilde{P}(\afterG{s},\tau,\goal) =1$
   and $\tilde{\rew}(\afterG{s},\tau) =  0$.
\end{itemize}
The transition probabilities and rewards in the
mode ``after $F$''  are defined analogously, except that
the rewards for all state-action pairs in mode ``after $F$'' are 0.
That is:
\begin{itemize}
\item
   If $s \in S \setminus G$ then
   $\Act_{\tilde{\cM}}(\afterF{s})=\Act_{\cM}(s)$,
   $\tilde{P}(\afterF{s},\act,\afterF{v}) = P(s,\act,v)$ and
   $\tilde{\rew}(\afterF{s},\act)= 0$.
\item
   If $s \in G$ then $\Act_{\tilde{\cM}}(\afterF{s})= \{\tau\}$,
   $\tilde{P}(\afterF{s},\tau,\goal) = 1$ and
   $\tilde{\rew}(\afterF{s},\tau)= 0$.
\end{itemize}

\noindent
Since the mode-switches in $\cM$ are deterministic,
there is a one-to-one correspondence between the finite paths
$\fpath$ in $\cM$ starting in $\sinit$ and
visiting an $F$-state and a $G$-state
(of minimal length with this property)
and the finite paths in $\tilde{\cM}$ that lead from $\sinit$
to $\goal$.
This yields a transformation of a given scheduler $\sched$ for $\cM$ into
a scheduler $\tilde{\sched}$ for $\tilde{\cM}$ such that:
\begin{center}
    \begin{tabular}{lcl}
        $\ExpRew{\sched}{\cM,\sinit}
           \bigl( \, \accdiaplus F \, | \, \Diamond G \, \bigr)$
        &\, = \,&
        $\ExpRew{\tilde{\sched}}{\tilde{\cM},\sinit}
           \bigl(\, \accdiaplus \goal \, | \, \Diamond \goal \, \bigr)$
    \end{tabular}
\end{center}
Here, we suppose that
 $\Pr^{\sched}_{\cM,\sinit}(\Diamond G)>0$
and $\Pr^{\sched}_{\cM,\sinit}(\Diamond F|\Diamond G)=1$.
Vice versa, each scheduler $\tsched$ for $\tilde{\cM}$
with $\Pr^{\tsched}_{\tilde{\cM},\sinit}(\Diamond \goal)>0$
and either 
$\Pr^{\tsched}_{\tilde{\cM},\sinit}(\Diamond \afterG{S})=0$
or
$\Pr^{\tsched}_{\tilde{\cM},\sinit}
  (\Diamond \goal|\Diamond \afterG{S})=1$
induces a scheduler $\sched$ for $\cM$ with
 $\Pr^{\sched}_{\cM,\sinit}(\Diamond G)>0$
and $\Pr^{\sched}_{\cM,\sinit}(\Diamond F|\Diamond G)=1$
and such that
$\tsched = \tilde{\sched}$.

\tudparagraph{1ex}{{\it Cleaning-up after Step 2.}}
The MDP $\tilde{\cM}$ can be further simplified using standard techniques
without affecting the maximal
condition expected accumulated reward.
First, we simplify  the sub-MDPs in the modes ``after $F$'' and
``after $G$'' where $G$ resp.~$F$ will be reached almost surely under all
schedulers as follows.
We introduce a new action symbol $\tau$,
discard the enabled actions in each state $\afterF{s}$
where $s \notin G$ and $\Pr^{\min}_{\cM,s}(\Diamond G)=1$
and add
a new $\tau$-transition with reward 0 from $\afterF{s}$
to $\goal$ with probability 1.
Note that for such states $s$ we have
$\Pr^{\min}_{\cM,s}(\Diamond G)=1$ iff
$\Pr^{\min}_{\tilde{\cM},\afterF{s}}(\Diamond \goal)=1$.

Likewise, for the states $\afterG{s}$ with $s \notin F$ and
$\Pr^{\min}_{\cM,s}(\Diamond F)=1$ we can discard the enabled
actions of $\afterG{s}$, while adding a
$\tau$-transition from $\afterG{s}$
to $\goal$ with probability 1.
The reward of $(\afterG{s},\tau)$ is defined as the unconditional maximal
expected accumulated reward to reach $F$ from state $s$ in $\cM$,

Finally, if $U$ denotes the set of states $u\in \tilde{S}$
with $\Act_{\tilde{\cM}}(u)=\{\tau\}$,
$\tilde{P}(u,\tau,\goal)=1$ and $\tilde{\rew}(u,\tau)=0$
then we can identify all states $u \in U$ in  $\tilde{\cM}$
with $\goal$. Formally, the latter means that we replace $\tilde{\cM}$ with
the MDP $\hat{\cM}$ arising from $\tilde{\cM}$ by removing all states
$u\in U$ and redefining the transition probability function by
$$
  \hat{P}(\tilde{s},\alpha,\goal) \ \ = \ \
  \tilde{P}(\tilde{s},\alpha,\goal) \ + \
    \sum_{u\in U} \tilde{P}(\tilde{s},\alpha,u)
$$
for all states $\tilde{s}\in \tilde{S}\setminus U$ and all actions
$\alpha$. The probabilities of all other transitions and the reward function
remain unchanged. That is,
$\hat{P}(\tilde{s},\alpha,\tilde{t})= \tilde{P}(\tilde{s},\alpha,\tilde{t})$
and $\hat{\rew}(\tilde{s},\alpha) = \tilde{\rew}(\tilde{s},\alpha)$
for all states $\tilde{s},\tilde{t}\in \tilde{S}\setminus U$
and all actions $\alpha$.

\tudparagraph{2ex}{Step 3: auxiliary trap state.}
We now perform a further transformation
$\tilde{\cM}\leadsto \tilde{\cM}'$ where $\tilde{\cM}'$
arises from the normal form MDP $\tilde{\cM}$ generated in Step 2.
The new MDP $\tilde{\cM}'$ arises from $\tilde{\cM}$ by first removing
all states $\tilde{t}$ in the ``after $G$'' mode with
$\Pr^{\max}_{\tilde{\cM},\tilde{t}}(\Diamond \goal) < 1$.
Let $V$ be the smallest set of states and state-action pairs
that contains 
all states $\tilde{t}$ in the ``after $G$'' mode with
$\Pr^{\max}_{\tilde{\cM},\tilde{t}}(\Diamond \goal) < 1$
and such that:
\begin{itemize}
\item
   for all states $\tilde{v}\in V$:
   if $(\tilde{s},\alpha)$ is a state-action pair in $\tilde{\cM}$ with
   $\tilde{P}(\tilde{s},\alpha,\tilde{v}) >0$ 
   then $(\tilde{s},\alpha) \in V$
\item
   if $(\tilde{s},\beta)\in V$ for all 
   $\beta \in \Act_{\tilde{\cM}}(\tilde{s})$
   and $\tilde{s}$ is not a trap state  %
   then
   $\tilde{s}\in V$.
\end{itemize}
Let $\tilde{\cM}_0$ be the sub-MDP of $\tilde{\cM}$ 
consisting of the states $\tilde{s}\in \tilde{S}$ with $\tilde{s}\notin V$
and the action sets: 
$$
  \Act_{\tilde{\cM}_0}(\tilde{s})
  \ \ \ = \ \ \ 
  \Act_{\tilde{\cM}}(\tilde{s})
  \setminus \{ \alpha : (\tilde{s},\alpha)\in V \}
$$
Then, $\tilde{\cM}'$ results from $\tilde{\cM}_0$ by:
\begin{itemize}
\item
  removing all states $\tilde{t}$
  where $\goal$ is not reachable from $\tilde{t}$ in $\tilde{\cM}$
  by redirecting all incoming transitions of $\tilde{t}$
  into $\fail$,
\item
  adding new transitions from the states $\tilde{s}$ with
  $\Pr^{\min}_{\tilde{\cM},\tilde{s}}(\Diamond \goal)=0$ to $\fail$,
  provided $\tilde{s}$ is not in the ``after $G$'' mode.
\end{itemize}
We shall use the action label $\iota$ for transitions to $\fail$.
More precisely, the state space of $\tilde{\cM}'$ is:
$\tilde{S}' \ =  \
  (\tilde{S}_0\cup \{\fail\}) \setminus T$
where $\tilde{S}_0$ is the state space of $\tilde{\cM}_0$ and
$$
  \begin{array}{lcl}
    T & \ = \ & 
    \bigl\{ \ \tilde{t}\in \tilde{S}_0 \ : \ 
            \Pr^{\max}_{\tilde{\cM}_0,\tilde{t}}(\Diamond \goal)=0 \ 
    \bigr\}
  \end{array}
$$
Note that $\sinit \notin T$ as we require
$\Pr^{\max}_{\cM,\sinit}(\Diamond F | \Diamond G) =1$.
The action set of $\tilde{\cM}'$ extends the action set $\Act_{\tilde{\cM}}$
of $\tilde{\cM}$ by a fresh action $\iota$.
Then, the transition probability function $\tilde{P}'$ of
$\tilde{\cM}'$ for the states 
$\tilde{s},\tilde{s}'\in \tilde{S}_0\setminus T$ and actions
$\alpha \in \Act_{\tilde{\cM}}(\tilde{s})$ is given by:
 \begin{eqnarray*}
  \tilde{P}'(\tilde{s},\alpha,\tilde{s}') \ \ = \ \ 
  \tilde{P}(\tilde{s},\alpha,\tilde{s}'),
  \ \ \ & &
  \tilde{P}'(\tilde{s},\alpha,\fail) \ \ = \ \
  \sum\limits_{\tilde{t}\in T} \tilde{P}(\tilde{s},\alpha,\tilde{t})
\end{eqnarray*}
That is, $\tilde{\cM}'$ collapses all states in $T$ into the
single trap state $\fail$.
The reward of the state-action pairs $(\tilde{s},\alpha)$ with
$\tilde{s}\in \tilde{S}$ and 
$\alpha \in \Act_{\tilde{\cM}}(\tilde{s})$ in $\tilde{\cM}'$ is
the same as in $\tilde{\cM}$.
Finally, we add transitions from every state $\tilde{s}$ 
in $\tilde{\cM}_0$ with $\tilde{s} \notin \afterG{S}$ and
$\Pr^{\min}_{\tilde{\cM},\tilde{s}}(\Diamond \goal)=0$
to $\fail$ with action label $\iota$.
That is, in all those states $\tilde{s}$, $\iota$ is an additional enabled
action with the transition probability $\tilde{P}'(\tilde{s},\iota,\fail)=1$
and reward 0 for the state-action pair $(\tilde{s},\iota)$.

\begin{lemma}[Soundness of the transformation]
 \label{lemma:soundness-norma-form-transformation}
 \
 Let $\cM$ denote the original MDP and
 $\tilde{\cM}'$ the MDP resulting from the transformations
 in steps 1,2 and 3. Then:
 \begin{center}
   \begin{tabular}{lcl@{\hspace*{0.2cm}}l}
     $\ExpRew{\max}{\cM,\sinit}
           \bigl( \, \accdiaplus F \, | \, \Diamond G \, \bigr)$
     &\, = \,&
     $\sup\limits_{\tsched'}$ &
     $\ExpRew{\tsched}{\tilde{\cM}',\sinit}
           \bigl(\, \accdiaplus \goal \, | \, \Diamond \goal \, \bigr)$
  \end{tabular}
 \end{center}
 where the supremum %
 on the right ranges over all
 schedulers $\tsched'$ for $\tilde{\cM}'$ such that
 $\Pr^{\tsched'}_{\tilde{\cM}',\sinit}(\Diamond \goal)>0$ and
 $\Pr^{\tsched'}_{\tilde{\cM}',\sinit}(\Diamond (\goal \vee \fail))=1$.
\end{lemma}
Note that this does not imply the finiteness of
$\ExpRew{\max}{\cM,\sinit}
  \bigl( \, \accdiaplus F \, | \, \Diamond G \, \bigr)$, which will be
  checked with the methods of the next section
(Section \ref{appendix:finitness-critical-schedulers}).

\begin{proof}
Recall that
$\ExpRew{\max}{\cM,\sinit}( \accdiaplus F  |  \Diamond G  )$
has been defined as the supremum of the conditional expectations
$\ExpRew{\sched}{\cM,\sinit}(  \accdiaplus F  |  \Diamond G  )$
where $\sched$ ranges over all schedulers for $\cM$ such that
$\Pr^{\sched}_{\cM,\sinit}(\Diamond G)>0$ and
$\Pr^{\sched}_{\cM,\sinit}(\Diamond F |\Diamond G)=1$.
Let $\Sched(\cM,F,G)$ denote this class of schedulers.

Let $\Sched(\tilde{\cM}',\goal)$ denote the class of schedulers
$\tsched'$ for $\tilde{\cM}'$ such that
 $\Pr^{\tsched'}_{\tilde{\cM}',\sinit}(\Diamond \goal)>0$ and
 $\Pr^{\tsched'}_{\tilde{\cM}',\sinit}(\Diamond (\goal \vee \fail))=1$.

Each scheduler $\sched$ in $\Sched(\cM,F,G)$ naturally induces
a scheduler $\tsched'$ in $\Sched(\tilde{\cM}',\goal)$ such that
$\tsched'$  mimics the $\sched$-paths
satisfying $\Diamond F \wedge \Diamond G$ by $\tsched$-paths
to $\goal$ and such that:
$$
  \ExpRew{\sched}{\cM,\sinit}
           \bigl( \, \accdiaplus F \, | \, \Diamond G \, \bigr)
  \ \ = \ \
   \ExpRew{\tsched'}{\tilde{\cM}',\sinit}
           \bigl(\, \accdiaplus \goal \, | \, \Diamond \goal \, \bigr)
$$
To see this, let $\tsched$ denote the lifting of $\sched$ to a scheduler
for $\tilde{\cM}$. Then:
$$
  \Pr^{\tsched}_{\cM,\sinit}(\Diamond \goal)
  \ = \ \Pr^{\sched}_{\cM,\sinit}(\Diamond F \wedge \Diamond G)
  \ > \ 0 
$$ 
Scheduler $\tsched'$ for $\tilde{\cM}'$ 
behaves as $\tsched$ for all finite paths 
$\tilde{\fpath}$ where 
$\Pr^{\residual{\tsched}{\tilde{\fpath}}}_{\tilde{\cM},\sinit}
  (\Diamond \goal)>0$.
As soon as $\tsched'$ has generated a path $\tilde{\fpath}$ with
$\Pr^{\residual{\tsched}{\tilde{\fpath}}}_{\tilde{\cM},\sinit}
   (\Diamond \goal)=0$.
then $\tsched'$ schedules action $\iota$. 
Then, $\tsched$ and $\tsched'$ have the same paths from $\sinit$ to $\goal$,
and these correspond to the $\sched$-paths 
satisfying $\Diamond F \wedge \Diamond G$.
The probabilities and 
accumulated rewards of these paths in $\tilde{\cM}'$, $\tilde{\cM}$ and
$\cM$ are the same. 
This yields:
$$
  \ExpRew{\max}{\cM,\sinit}
           \bigl( \, \accdiaplus F \, | \, \Diamond G \, \bigr)
  \ \ \leqslant \ \
  \sup_{\tsched'} \
   \ExpRew{\tsched'}{\tilde{\cM}',\sinit}
           \bigl(\, \accdiaplus \goal \, | \, \Diamond \goal \, \bigr)
$$
Vice versa, we show that for
each $\tsched'$ in $\Sched(\tilde{\cM}',\goal)$
there exists a scheduler $\sched$ in $\Sched(\cM,F,G)$
such that:
$$
  \ExpRew{\sched}{\cM,\sinit}
           \bigl( \, \accdiaplus F \, | \, \Diamond G \, \bigr)
  \ \ \geqslant \ \
   \ExpRew{\tsched'}{\tilde{\cM}',\sinit}
           \bigl(\, \accdiaplus \goal \, | \, \Diamond \goal \, \bigr)
$$
Let us first suppose that $\tilde{\cM}$ and $\tilde{\cM}'$ have
a positive end component $\tilde{\cE}$ in the ``after $G$'' mode
such that $\goal$ is reachable from $\cE$ 
and $\cE$ is reachable from $\sinit$.
We show that in this case the maximal conditional expectation of
$\cM$ is infinite.
Let $\usched$ be a scheduler that ``realizes'' this end component.
That is, if $\psi$ denotes the event 
``take any state-action pair of $\cE$ infinitely often''
then $\Pr^{\usched}_{\tilde{\cM},\tilde{t}}(\psi)=1$ for each state
$\tilde{t}$ of $\tilde{\cE}$.
Obviously, $\tilde{\cE}$ corresponds to some positive end component of $\cM$
and $\usched$ can also be viewed as a scheduler for $\cM$
that ``realizes'' $\cE$.
Furthermore, we pick some memoryless scheduler $\vsched$ for $\cM$
and 
a shortest path $\finpath$ in $\cM$ from $\sinit$ to $\cE$
as well as a shortest path $\fpath$ from $\cE$ to some state in $\afterG{F}$.
(Such a path exists as $\goal$ is reachable from $\cE$ in $\tilde{\cM}$.)
Let now $R\in \Nat$ and let $\sched_R$ be the following scheduler.
In its first mode, $\sched_R$ attempts to generate the path $\finpath$.
If it fails then it switches mode and behaves as $\vsched$ from then on.
As soon as $\finpath$ has been generated then $\sched_R$ behaves
as $\usched$ as long as the accumulated rewards is smaller than $R$.
As soon as the accumulated reward exceeds $R$ then $\sched_R$ switches mode
again, leaves $\cE$ and attempts to reach $\afterG{F}$ along $\fpath$.
If it fails then it behaves as $\vsched$.
For the conditional expectation of $\sched_R$ we have:
$$
  \CExp{\sched_R}(\accdiaplus F |\Diamond G)
  \ \ \ \geqslant \ \ \ \frac{\rho + p \cdot R}{x+p}
$$
where $\rho,x\geqslant 0$ and 
$p=\probability(\finpath)\cdot \probability(\fpath)$.
The value $x$ stands for the probability under $\sched_R$ to reach $\goal$
via paths that do not have the form
$\finpath; \cycle ; \fpath$ where $\cycle$ is a $\usched$-path, and
$\rho$ stands for the corresponding partial expectation.
Thus, there exists some $R\in \Nat$ with
$$
  \CExp{\sched_R}(\accdiaplus F |\Diamond G)
  \ \ \ \geqslant \ \ \
  \CExp{\tsched'}(\accdiaplus \goal |\Diamond \goal)
$$
Let us now suppose that there is no positive end component in the
``after $G$'' mode of $\tilde{\cM}'$ that is reachable from $\sinit$
and from which $\goal$ is reachable.
Let $\tsched' \in \Sched(\tilde{\cM}',\goal)$. 
Let $\tilde{U}$ denote the set of states $\tilde{u}$ in $\tilde{\cM}$ with 
$\Pr^{\min}_{\tilde{\cM},\tilde{u}}(\Diamond \goal)=0$.
Thus, $\tilde{u}\in \tilde{U}$ iff $\iota \in \Act_{\tilde{\cM}'}(\tilde{u})$.
Then, $\tilde{U} \cap \afterG{S}=\varnothing$ by definition of
$\tilde{\cM}'$.
We pick a memoryless scheduler $\usched$ for $\tilde{\cM}$ such that
$\Pr^{\usched}_{\tilde{\cM},\tilde{u}}(\Diamond \goal)=0$ for all states
$\tilde{u}\in \tilde{U}$.
Then, $\usched$ can also be viewed as a scheduler for $\cM$.
Scheduler $\sched$ for an input path $\fpath$ in $\tilde{\cM}$ 
behaves as $\tsched'$
for the corresponding path $\tilde{\fpath}$ in $\tilde{\cM}'$
provided that
$\tsched'(\tilde{\fpath}) \not= \iota$.
As soon as $\tsched'$ schedules $\iota$ then 
$\tilde{u}\eqdef \last(\tilde{\fpath})\in \tilde{U}$.
Then, 
$\sched$ behaves as $\usched$ from then on.
Up to the mode-annotations of the states in $\tilde{\cM}'$,
$\sched$ and $\tsched'$ have the same ``successful'' paths 
(i.e., satisfying $\Diamond F \wedge \Diamond G$ resp.~$\Diamond \goal$)
with the same probabilities and rewards.
In particular, this yields $\Pr^{\sched}_{\cM,\sinit}(\Diamond G)>0$.
Furthermore, we have
$\Pr^{\sched}_{\cM,\sinit}(\Diamond F|\Diamond G)=1$.
The latter is a consequence of the fact that
$\tilde{U}\cap \afterG{S}=\varnothing$.
This yields $\sched \in \Sched(\cM,F,G)$ and 
that $\sched$ and $\tsched'$ have the same maximal conditional expectations.
\Ende
\end{proof}

\noindent
In summary, with the three preprocessing steps, we can transform
$\cM$ into an MDP $\tilde{\cM}'$ with two trap states $\goal$ and $\fail$
such that $\sinit \notin \{\goal,\fail\}$.
In the next section, we will describe a further transformation
MDP $\tilde{\cM}' \leadsto \Hut{\cM}$ such that 
$\Hut{\cM}$ satisfies conditions (1) and (2) of Section \ref{sec:finiteness}. 
$\Hut{\cM}$ will then be used for
deciding finiteness (Section \ref{summary:check-finiteness}),
computing an upper bound $\CExp{\ub}$ for the $\CExp{\max}$
(Section \ref{sec:upper-bound})
and the subsequent threshold algorithm and the computation of an optimal
scheduler as outlined in Section \ref{sec:threshold}.

\subsection{Finiteness -- critical schedulers}
\label{appendix:finitness-critical-schedulers}

In the sequel, we suppose that $\cM$ is the result of
the preprocessing and normal form transformation presented in the
previous subsection. 
Thus, the task is to compute the maximal 
expected accumulated reward until reaching the trap state
$\goal$ under the condition that $\goal$ will be reached where the 
supremum is taken over all schedulers $\sched$ satisfying the following
requirement \eqref{assumption:SR}
(see Lemma \ref{lemma:soundness-norma-form-transformation}):
\begin{equation}
   \label{assumption:SR}
   \Pr^{\sched}_{\cM,\sinit}(\Diamond \goal)>0
   \ \ \ \text{and} \ \ \
   \Pr^{\sched}_{\cM,\sinit}(\Diamond (\goal \vee \fail))=1
  \tag{SR}
\end{equation}
Furthermore, for all states $s$ in $\cM$:
\begin{equation}
   \label{assumption:A1}
   \text{$s \not\models \exists \Diamond \goal$
         \ \ \ iff \ \ \
         $s=\fail$}
   \tag{A1}
\end{equation}

Before we start into this section, we will briefly repeat and 
summarize notations relevant for this section.

\begin{definition}[Shortform notations for (conditional) expectations]
\label{def:CExp}
{\rm
As before, we often write $\Pr^{\sched}_{s}$ for $\Pr^{\sched}_{\cM,s}$.
If $\sched$ is a scheduler for $\cM$ with
$\Pr^{\sched}_{\sinit}(\Diamond \goal) >0$
then we shortly
write $\CExp{\sched}$ for the conditional expected accumulated reward until
reaching the goal state under scheduler $\sched$ under the condition that
the goal state will indeed be reached.
That is:
$$
  \CExp{\sched}
  \ \ = \ \
  \ExpRew{\sched}{\sinit}(\ \accdiaplus \goal \ | \ \Diamond \goal \ )
$$
We often refer to $\CExp{\sched}$ as the conditional expectation under
$\sched$.
Furthermore, let
$$
   \CExp{\max} \ \ = \ \ \sup_{\sched} \ \CExp{\sched}
   \qquad \text{and} \qquad
   \CExp{\min} \ \ = \ \ \inf_{\sched} \ \CExp{\sched}
$$
where $\sched$ ranges over all
schedulers for $\cM$ satisfying the scheduler requirement
\eqref{assumption:SR}.
We also often use the notation $\Exp{\sched}{s}$ as a shortform for
$$
  \Exp{\sched}{s} \ \ = \ \ \Exp{\sched}{\cM,s} \ \ = \ \
  \sum_{r=0}^{\infty} \ r \cdot \Pr^{\sched}_{\cM,s}(\Diamond^{=r} \goal)
$$
Here, $\sched$ is an arbitrary scheduler in $\cM$ and $s$ a state of $\cM$.
Clearly, we then have
$$
  \CExp{\sched} \ \ = \ \
  \frac{\Exp{\sched}{\sinit}}{\Pr^{\sched}_{\cM,\sinit}(\Diamond \goal)}
$$
for each scheduler $\sched$ satisfying \eqref{assumption:SR}.
If $s \in S \setminus \{\goal,\fail\}$ and
$\sched$ satisfies  \eqref{assumption:SR}
then:
  $$
   \CExpState{\sched}{s} \ \ = \ \
   \ExpRew{\tsched}{s}(\ \accdiaplus \goal \ | \ \Diamond \goal \ )
   \ \ = \ \
   \frac{\Exp{\sched}{s}}
        {\Pr^{\sched}_s(\Diamond \goal)}
  $$
Hence, $\CExp{\sched} \ = \ \CExpState{\sched}{\sinit}$.
  }
\Ende
\end{definition}

\noindent

We will now present a criterion to decide whether
$\CExp{\max} < \infty$ and consider the unconditional case first.

\begin{remark}[Unconditional maximal expected accumulated reward]
\label{remark:sched-infinite-exp-reward}
It is well-known that if
$\Pr^{\min}_{\cM,\sinit}(\Diamond \goal)=1$ then the unconditional
expected accumulated reward
$\Exp{\sched}{\sinit} = \ExpRew{\sched}{\sinit}(\accdiaplus \goal)$
is finite for all schedulers $\sched$.
Furthermore, there exists a memoryless deterministic scheduler maximizing the
unconditional expected accumulated reward.
In particular, the supremum of the unconditional expected accumulated rewards
under all schedulers is finite.

However, if $\Pr^{\min}_{\cM,\sinit}(\Diamond \goal)<1$ then the
expected accumulated reward to reach $\goal$ can be infinite
for (infinite-memory) schedulers $\sched$
with $\Pr^{\sched}_{\cM,\sinit}(\Diamond \goal)=1$.
To illustrate this phenomenon, consider an MDP $\cM$ with two states
$s=\sinit$ and $\goal$, the transition probabilities
$P(s,\alpha,\goal)= P(s,\beta,s)=1$ and the reward
$\rew(s,\alpha)=\rew(s,\beta)=1$.
We pick an infinite sequence $(q_n)_{n \geqslant 1}$
of rational numbers in $]0,1[$ such that
$\sum_n q_n =1$ and $\sum_n n \cdot q_n$ diverges
(e.g., $q_n = x/n^2$ where $1/x$ is the value of the series $\sum_n 1/n^2$).
Furthermore, we put $x_1 = 1$ and
$x_n= p_1 \cdot \ldots \cdot p_{n-1}$ for $n > 1$.
Let $\sched$ be the randomized scheduler for $\cM$
that schedules action $\beta$ with probability $p_n = 1- q_n/x_n$
and action $\alpha$ with probability $q_n/x_n$
for the $n$-th visit of state $s$.
Let $\fpath_n$ be the path $(s \, \beta)^{n-1} \, s \, \alpha \, \goal$.
The probability of $\fpath_n$ under $\sched$ is
$p_1 \cdot p_2 \ldots \cdot p_{n-1} \cdot (1-p_n) = x_n(1-p_n)=q_n$.
Hence:
$$
  \Pr^{\sched}_{\cM,s}(\Diamond \goal)
  \ \ = \ \
  \sum_{n=1}^{\infty} q_n
  \ \ = \ \ 1
$$
Since $\rew(\fpath_n)=n$,
the expected accumulated reward under $\sched$ for reaching $\goal$
is $\sum_n  n \cdot x_n \cdot (1-p_n) = \sum_n n \cdot q_n = \infty$.
\Ende
\end{remark}

Let $\mathit{PosEC}$ be the set of all states $s$ in $\cM$ that belong to
some end component $\cE = (E,\ActEC)$
with $\{\goal,\fail\} \cap E = \varnothing$
and $\rew(t,\act)\geqslant 1$ for some
state-action pair $(t,\act)$ with $t\in E$ and $\act \in \ActEC(t)$.
Such end components are said to be \emph{positive}.

\begin{lemma}
   \label{lemma:PosEC}
   If $\mathit{PosEC}\not= \varnothing$ then
   $\CExp{\max} = \infty$.
\end{lemma}

\begin{proof}
By assumption \eqref{assumption:A1} we have
$t \models \exists \Diamond \goal$ for all $t \in \mathit{PosEC}$.
Furthermore, for each infinite path $\infpath$:
$\infpath \models \Diamond \mathit{PosEC}$ iff
$\infpath \models (\neg \goal ) \Until \mathit{PosEC}$.

\noindent Let $R \in \Nat$. We pick some state $s \in \mathit{PosEC}$
and schedulers $\sched_s$ and $\sched_g$
with
\begin{center}
   $p_s = \Pr^{\sched_s}_{\cM,\sinit}(\Diamond s)>0$
   \ \ and \ \
   $p_g = \Pr^{\sched_g}_{s}(\Diamond \goal)>0$.
\end{center}
Furthermore, let $\sched_{\mathit{EC}}$ be a scheduler
such that the limit of almost all $\sched_{\mathit{EC}}$-paths starting in $s$
is a positive end component containing $s$.
We construct a scheduler $\tsched_R$ as follows.
For input paths starting in $\sinit$, scheduler $\tsched_R$ first behaves as
$\sched_s$ until state $s$ has been reached
(this happens with probability $p_s$).
It then switches its mode and behaves as $\sched_{\mathit{EC}}$
until a path fragment $\fpath$ has been generated such that
$\last(\fpath) = s$ and $\rew(\fpath) > R$ (this happens with  probability 1).
Having generated such a path fragment $\fpath$, $\tsched_R$ switches its mode
again and behaves as $\sched_g$ from then on.
Then, the expected accumulated reward to reach $\goal$ under $\tsched_R$
has the form:
$$
  \CExp{\tsched_R}
  \ \ \geqslant \ \
  \begin{array}{r@{\hspace*{0.2cm}}c@{\hspace*{0.2cm}}l}
        \ \rho & + & p_s  \cdot p_g \cdot R \
        \\[0.3ex]
        \hline
        \\[-2.4ex]
        x & + & p_s  \cdot p_g
  \end{array}
$$
where
$$
  \begin{array}{lcl}
   \rho \ =
   \sum\limits_{r=0}^{\infty}
     r \cdot \Pr^{\sched_s}_{\sinit}(\neg s \Until^{=r} \goal)
   & \ \ \text{and} \ \ &
   x \ = \ \Pr^{\sched_s}_{\sinit}(\neg s \Until \goal)
 \end{array}
$$
Since $\rho$ and $x$ do not depend on $R$, we have:
$$
  \lim_{R \to \infty}
    \CExp{\tsched_R}
  \ = \ \infty
$$
Thus,
$\CExp{\max} = \infty$ if $\mathit{PosEC}$ is nonempty.
\Ende
\end{proof}

Obviously, $\mathit{PosEC}=\varnothing$ if there is no
positive maximal end component. Hence, the criterion of
Lemma \ref{lemma:PosEC} can be checked in time polynomial
in the size of the MDP $\cM$ using the same techniques
as proposed by de Alfaro \cite{deAlf99}
for computing maximal unconditional expectations.%
\footnote{More precisely, Section 4 of \cite{deAlf99} addresses
  the computation of minimal expected accumulated reward in
  non-positive MDPs where the minimum is taken over all schedulers
  that reach the goal state almost surely.}
In what follows, we suppose that $\cM$ has no positive end component.
That is, $\rew(s,\act)=0$ for all state-action pairs $(s,\act)$
that belong to some (maximal) end component.
But then we can collapse all maximal end components into a single state
and discard all actions $\act \in \Act(s)$ where $(s,\act)$ belongs
to an end component
(i.e., we identify $\cM$ with its MEC-quotient,
see Appendix \ref{appendix:notations}),
without affecting the maximal expected accumulated reward.

\begin{lemma}[Maximal unconditional expectations (see \cite{deAlf99})]
   \label{lemma:PosEC-uncond}
   Suppose $\mathit{PosEC} = \varnothing$. Then,
   $
    \ExpRew{\max}{\cM,\sinit}
       (\accdiaplus \goal ) < \infty$
   and there exists a memoryless deterministic scheduler
   $\sched$ with $\Pr^{\sched}_{\cM,\sinit}(\Diamond \goal)=1$
   and
   $\ExpRew{\sched}{\cM,\sinit} (\accdiaplus \goal) =
    \ExpRew{\max}{\cM,\sinit} (\accdiaplus \goal)$.
\end{lemma}

We now return to the case of conditional expected accumulated rewards.
Assuming $\mathit{PosEC}=\varnothing$, the transformation
that collapses all maximal end components into a single state (see above)
permits the following additional assumption:
\begin{equation}
   \label{assumption:A2}
   \text{$\cM$ has no end component}
   \tag{A2}
\end{equation}
Under assumption \eqref{assumption:A2} we have
$\Pr^{\min}_{\cM,s}\bigl(\, \Diamond (\goal \vee \fail)\, \bigr)=1$
for all states $s$.
Hence, the scheduler requirement \eqref{assumption:SR} reduces to
$\Pr^{\sched}_{\cM,s}(\Diamond \goal) >0$.
Note that after this transformation the MDP $\cM=\Hut{\cM}$ 
satisfies conditions (1) and (2) presented in Section \ref{sec:finiteness}.

The nonemptiness of $\mathit{PosEC}$ is, however,
not sufficient to cover all cases where
the conditional expected accumulated reward is infinite.
This is illustrated in the following example.

\begin{example}
\label{example:PosEC-not-sufficient}
Let $\cM$ be the MDP $\cM[\rewparam]$ 
shown in Figure~\ref{fig:running-example},
but with initial state $\sinit = s_2$.
The parameter $\rewparam$ is of no concern for this example.
For the scheduler $\sched_n$ that chooses action $\beta$ for
the first $n$ visits of $s$ and then $\alpha$, there is a single
$\sched_n$-path reaching $\goal$,
namely $\fpath_n = (s\, \beta)^n \, s \, \alpha \, \goal$.
The accumulated reward of $\fpath_n$ is
$n$, %
and so is the conditional expectation of $\cM$ under $\sched_n$.
Thus, $\CExp{\max}$  is infinite,
although $\mathit{PosEC}$ is empty.
\Ende
\end{example}

\begin{definition}[Positive, zero-reward, simple cycle; critical scheduler]
\label{def:critical-scheduler}
{\rm
Let $\cM$ be an MDP as before satisfying 
\eqref{assumption:A1} and \eqref{assumption:A2}.
A cyclic path
$\cycle = t_0 \, \beta_0 \, t_1 \, \beta_1 \ldots \beta_{n-1}\, t_n$
in $\cM$ is said to be \emph{positive} if $\rew(t_i,\beta_i) >0$ for at least
one index $i\in \{0,1,\ldots,n{-}1\}$.
Otherwise $\cycle$ is called a \emph{zero-reward} cycle.
$\cycle$ is said to be \emph{simple}
if $t_i \not= t_j$ for $0 \leqslant i < j <n$.

Scheduler $\usched$ is called \emph{critical} if
$\Pr^{\usched}_{\cM,\sinit}(\Box \neg \goal) =1$ and
there is a reachable positive $\usched$-cycle, i.e.,
there is a finite $\usched$-path
$s_0 \, \alpha_0 \, s_1 \, \alpha_1 \, \ldots \alpha_{m-1}\, s_m$
starting in $s_0 = \sinit$
such that for some $k < m$ the  suffix
$s_k \, \alpha_k \, s_{k+1} \, \alpha_{k+1} \, \ldots \alpha_{m-1}\, s_m$
is a positive cycle.
  }
\Ende
\end{definition}

Note that $\Pr^{\usched}_{\sinit}(\Diamond \fail)=1$
for each critical scheduler by assumption (A2).
Thus, if $\Pr^{\min}_s(\Diamond \goal)>0$ for all states
$s\in S \setminus \{\fail\}$ then $\cM$ has no critical scheduler.

\begin{proposition}[Infinite maximal conditional expected rewards]
\label{proposition:exprew-infinite-critical-scheduler}
  \ \\
  With the notations and assumptions \eqref{assumption:A1} and \eqref{assumption:A2} as above,
  the following three statements are equivalent:
\begin{center}
 \begin{tabular}{ll}
    {\rm (i)} &
    $\CExp{\max} = \infty$
    \\[1ex]

    {\rm (ii)} &
    $\cM$ has a memoryless deterministic critical scheduler
    \hspace*{2cm}
    \\[1ex]

    {\rm (iii)} &
    $\cM$ has a critical scheduler
  \end{tabular}
 \end{center}
\end{proposition}

\begin{proof}
We first show the equivalence of statements (ii) and (iii).
The implication (ii) $\Longrightarrow$ (iii) is trivial.
For the implication (iii) $\Longrightarrow$ (ii), we suppose
that we are given a (possibly randomized history-dependent)
critical scheduler $\usched$.
A memoryless deterministic critical scheduler $\usched'$ is obtained
by picking a simple positive $\usched$-cycle
$\cycle = s_0 \, \alpha_0 \, s_1 \, \alpha_1 \ldots \alpha_{n-1}\, s_n$.
We then put $\usched'(s_i)=\act_i$ for $0 \leqslant i < n$.
For each state $s\in S \setminus \{s_0,\ldots,s_{n-1}\}$
we pick an arbitrary action $\act$ with $\usched(s)(\act)>0$
and define $\usched'(s) = \act$.

\tudparagraph{1ex}{\rm (ii) $\Longrightarrow$ (i):}
Suppose $\cM$ has a memoryless deterministic critical
scheduler $\usched$.
The argument is similar to the proof of Lemma \ref{lemma:PosEC}.
Let $\cycle = t_0 \, \beta_0 \, t_1 \, \beta_1 \ldots \beta_{n-1}\, t_n$
be positive $\usched$-cycle and $s=t_0=t_n$.
Furthermore, let $\sched$ be a memoryless deterministic
scheduler under which $s$ reaches the goal-state with positive
probability.
Then, $p_g = \Pr^{\sched}_{s}(\Diamond \goal) >0$.
Let $p_s = \Pr^{\usched}_{\sinit}(\Diamond s)$
and $p_c = \Pr^{\usched}_s(\Cyl(\cycle))$.
Then, $p_s>0$ and $p_c = \probability(\cycle)>0$.
For $R \in \Nat$, we construct a scheduler $\tsched_R$ with two modes
as follows.
\begin{itemize}
\item
  In its first mode, $\tsched_R$ behaves as
  the critical scheduler $\usched$ for all input paths $\fpath$
  where $\cycle^R$ is not a fragment of $\fpath$.
\item
  If the input path $\fpath$ has a suffix that runs $R$-times through
  the cycle $\cycle$ and has not visited state $s$ before entering $\cycle$
  (in which case $\last(\fpath)=s$ and $\rew(\fpath) \geqslant R$),
  then $\tsched_R$ switches mode and behaves
  as $\sched$ from now on.
\end{itemize}
Note that the switch from the first mode
(simulation of $\usched$) to the second mode (simulation of $\sched$)
appears with probability $p_s \cdot p_c^R$.
As $\Pr^{\usched}_{\sinit}(\Box \neg \goal)=1$, we have
$\Pr^{\tsched_R}_{\sinit}(\neg s \Until \goal)=0$.
For the conditional
expected accumulated reward to reach $\goal$ under $\tsched_R$,
we get:
$$
  \CExp{\tsched_R}
  \ \ \geqslant \ \
  \begin{array}{l}
        p_s \cdot p_c^R \cdot p_g \cdot R \
        \\[0.3ex]
        \hline
        \\[-2.4ex]
        p_s \cdot p_c^R \cdot p_g
  \end{array}
  \ \ = \ \ R
$$
Hence, the limit of $\CExp{\tsched_R}$ is infinite
if $R$ tends to infinity.

\tudparagraph{1ex}{\rm (i) $\Longrightarrow$ (iii):}
We now suppose that $\cM$ has no critical scheduler and show that
the conditional expected accumulated reward is finite.
For this, we adapt the reset-mechanism that has been introduced
for computing conditional reachability probabilities in 
\cite{BKKM14} (see Appendix \ref{appendix:reset-Methode})
and rely on Lemma \ref{lemma:PosEC-uncond}.

The reset-mechanism is, however, not directly applicable since
the resulting MDP $\cN$ with reset-transitions from the fail-state
to the initial state might have positive end components, in which case
the maximal unconditional expected accumulated reward can be infinite.
For this reason,
we first transform $\cM$ into a new MDP $\cM'$
that has the same probabilistic structure
and the same schedulers, but uses a different reward structure.
The maximal conditional expected accumulated rewards to reach $\goal$
are the same in $\cM$ and $\cM'$.
Having constructed $\cM'$ we then can implement the reset-mechanism
and switch from $\cM'$ to a new MDP $\cN$ that has no positive end component.
The maximal unconditional expected accumulated
reward to reach $\goal$ in $\cN$ is finite
(by Lemma \ref{lemma:PosEC-uncond})
and yields an upper bound $\CExp{\ub}$ for the maximal conditional
expected accumulated reward to reach $\goal$ in $\cM$ (or $\cM'$).

\tudparagraph{1ex}{\it The new MDP $\cM'$.}
We first provide an informal explanation of the behavior
of the MDP $\cM'$. The idea is that $\cM'$ behaves as $\cM$,
but postpones the assignment of rewards until it is clear
that $\goal$ will be reached with positive probability.
Note that we can deal with $\cM'=\cM$ if
   $\Pr^{\min}_s(\Diamond \goal)>0$ for all states $s \not= \fail$.
   The following construction only serves to deal with the case
   where $\cM$ contains some non-trap states $s$ with
   $\Pr^{\min}_s(\Diamond \goal)=0$.

\noindent The definition of $\cM'$ relies on the following observation.
Let
$$
  R \ \ = \ \ \sum_{s\in S} \rew^{\max}(s)
$$
where $\rew^{\max}(s)=0$ if $s=\goal$ or $s=\fail$ and
$\rew^{\max}(s) = \max \{ \rew(s,\alpha) :\alpha \in \Act(s)\}$
for $s\in S \setminus \{\goal,\fail\}$.
Clearly, if $\fpath$ is a finite path in $\cM$ with
$\rew(\fpath) >R$ then $\fpath$ contains a positive cycle.
As $\cM$ has no critical schedulers, we have
$\Pr^{\sched}_{\sinit}(\Diamond \goal)>0$ for each scheduler
$\sched$ where $\rew(\fpath) >R$ for some
$\sched$-path $\pi$ starting in $\sinit$.

The new MDP $\cM'$ simulates $\cM$ while operating in two modes.
In its first mode, $\cM'$ augments the states with the
information on the reward that has been accumulated
in the past. That is, the states of $\cM'$ in its first mode
are pairs $\<s,r\> \in S \times \Nat$
where $r=\rew(\fpath) \leqslant R$
for some finite path $\fpath$ in $\cM$  from $\sinit$ to $s$.
Thus, as soon as $\cM'$ takes an action $\alpha$ in state $\<s,r\>$
where $r+ \rew(s,\alpha)$ exceeds $R$ then $\cM'$ switches to the second
mode where it behaves exactly as $\cM$ without any reward annotations of the states.
We assign reward 0 to all state-action pairs in the first mode.
The reward of the switches from the first to the second mode, say
from state $\<s,r\>$ via firing action $\alpha$, is
the total reward accumulated so far (value $r$)
plus the reward of the taken action (the value $\rew(s,\alpha)$).
The rewards for the state-action pairs in the second mode
are the same as in $\cM$.

The definition of the new MDP $\cM'$ is as follows.
The state space of $\cM'$ is:
$$
  S' \ \ = \ \ S \times \{0,1,\ldots,R\} \ \cup \ S
$$
The action set of $\cM'$ is $\Act' = \Act$.
The computation of $\cM'$ starts in state $\sinit' = \<\sinit,0\>$.
The transition probability function $P'$ and reward function $\rew'$ are
defined as follows.
\begin{itemize}
\item
  Let $s \in S \setminus \{ \goal\}$, $r \in \{0,1,\ldots,R\}$,
  $\act \in \Act(s)$ and $r' = r + \rew(s,\alpha)$.
  \begin{itemize}
  \item
    If $r' \leqslant R$ then
    $P'(\<s,r\>,\act,\<t,r'\>) = P(s,\act,t)$ for all states $t\in S$ and
    $\rew'(\<s,r\>,\act)=0$.
  \item
    If $r'>R$ then
    $P'(\<s,r\>,\act,t) = P(s,\act,t)$ for all states $t\in S$ and
    $\rew'(\<s,r\>,\act)=0$.
 \end{itemize}
\item
  If the current state of $\cM'$ is a state $\<\goal,r\>$ then $\cM'$
  switches to the goal state in the second mode while earning reward $r$.
  That is,
  $P'(\<\goal,r\>,\tau,\goal)=1$ and
  $\rew'(\<\goal,r\>,\tau)=r$
  for some distinguished action name $\tau \in \Act$.
\item
  If the current state in the new MDP $\cM'$ is a state $s\in S$
  then $\cM'$ behaves as $\cM$, i.e.,
  $\Act_{\cM'}(s)=\Act_{\cM}(s)$,
  $P'(s,\act,t) = P(s,\act,t)$ and $\rew'(s,\act) = \rew(s,\act)$
  if $s,t\in S$ and $\act \in \Act$.
\end{itemize}
The states $\goal$ and $\fail$ and $\<\fail,r\>$ for $r\in \{0,1,\ldots,R\}$
are trap states.

Obviously, there is a one-to-one correspondence between the paths
in $\cM$ and the paths in $\cM'$ and between the schedulers of $\cM$ and
$\cM'$.
Given a finite path $\fpath'$ in $\cM'$, let
$\fpath'|_{\cM}$ denote the path in $\cM$ resulting from $\fpath'$ by dropping
the annotations for the first mode. Vice versa, for a given path $\fpath$
in $\cM$ we add appropriate annotations for the states in
some prefix of $\fpath$ to obtain
a path $\mathit{lift}(\fpath)$ in $\cM'$ starting in
$\sinit' = \<\sinit,0\>$
with  $\mathit{lift}(\fpath)|_{\cM}=\fpath$.
We have:
\begin{itemize}
\item
  If $\fpath'$ consists of states in the first mode,
  i.e., all states of $\fpath'$ are annotated states
  in $S \times \{0,1,\ldots,R\}$,
  then $\rew'(\fpath') =0$ and $\rew(\fpath'|_{\cM}) = r$
  where  $\last(\fpath') = \<s,r\>$.
\item
  If $\last(\fpath')\in S$ is a state in second mode
  then $\rew'(\fpath')=\rew(\fpath'|_{\cM})$.
\item
  If the last transition in $\fpath'$ is a mode-switch,
  say $\<s,r\> \stackrel{\act}{\longrightarrow} t$, then
  $\Pr^{\sched}_{\cM,\sinit}(\Diamond \goal)>0$
  for each scheduler $\sched$ of $\cM$ where $\fpath'|_{\cM}$
  is a $\sched$-path
  (otherwise $\sched$ would be a critical scheduler).
\end{itemize}
Obviously,
$\cM$ and $\cM'$ have the same conditional expected
accumulated reward to reach $\goal$ under corresponding schedulers.
In particular:
$$
  \CExp{\max} \ \ = \ \
  \ExpRew{\max}{\cM,\sinit}
         (\ \accdiaplus \goal \ | \ \Diamond \goal \ )
  \ \ = \ \
  \ExpRew{\max}{\cM',\<\sinit,0\>}
          (\ \accdiaplus \goal \ | \ \Diamond \goal \ )
$$

\tudparagraph{1ex}{\it The MDP $\cN$}
results from $\cM'$ by adding reset-transitions from the states
$\fail$ and $\<\fail,\cG\>$. Let
$$
  \Fail \ \ \ = \ \ \
  \bigl\{\, \fail \, \bigr\} \ \cup \
  \bigl\{\ \<\fail,r\> \ : \ r\in \{0,1,\ldots,R\} \ \bigr\}
$$
For all states $s_f \in \Fail$
we deal with $\Act_{\cN}(s_f)=\{\tau\}$ for some distinguished
action $\tau$
and define $P_{\cN}(s_f,\tau,\sinit') = 1$ and
$\rew_{\cN}(s_f,\tau)=0$.
The transition probabilities of all other states and the
rewards of all other state-action pairs are the same as in $\cM'$.

While $\cM'$ has no end components (by assumption (A2)),
$\cN$ might have end components containing $\sinit'$
and at least one of the fail-states.
The above shows that the reward of each finite path in $\cM'$ from
$\sinit'$ to some state $\<\fail,r\>$ is 0,
while the reward of paths from $\sinit'$ to $\fail$
can be positive.
We now show that the end components of $\cN$
cannot contain any of the states of $\cM$'s second mode.
In particular, there is no end component containing state $\fail$.

\tudparagraph{1ex}{{\it Claim.}}
  $\cN$ has no positive end components.

\tudparagraph{1ex}{{\it Proof of the claim.}}
  By the definition of the reward function in $\cN$, it suffices to show
  that none of the states $t\in S\setminus \{\goal\}$
  belongs to an end component of $\cN$.

  Suppose by contradiction that there is an end component $\cE = (E,\ActEC)$
  of $\cN$ containing some state $t\in S\setminus \{\goal\}$.
  Since $\cM$ and $\cM'$ do not contain  end components
  (assumption (A2)), $\cE$ must contain one of the reset-transitions
  from $\fail$ or some state $\<\fail,r\>$ to the
  initial state $\sinit'=\<\sinit,0\>$.
  As $\goal$ is a trap, $\cE$ does not contain
  $\goal$ and none of its copies $\<\goal,r'\>$.
  We pick a finite-memory scheduler $\sched'$ for $\cN$ such that
  \begin{center}
    $\infpath \models \Box \Diamond t$ \ and \
    $\infpath \models \Box E$ \ for all infinite $\sched'$-paths $\infpath$
    starting in $\sinit'$.
  \end{center}
  Thus, $\Pr^{\sched'}_{\cN,\sinit'}(\Box E)=1$
  and therefore
  $\Pr^{\sched'}_{\cN,\sinit'}(\Diamond \goal)=0$.

  Let $\fpath'$ be a shortest $\sched'$-path from $\sinit'$ to $t$.
  Clearly, $\fpath'$ contains a mode-switch. Hence, it must
  contain a positive cycle.
  As $\cM$ does not have critical schedulers,
  the scheduler $\sched$ for $\cM$ that behaves as $\sched'$
  as long as none of the fail-states has been visited enjoys
  the property $\Pr^{\sched}_{\cM,\sinit}(\Diamond \goal)>0$.
  As all $\sched$-paths are $\sched'$-paths we get
  $$
    \Pr^{\sched'}_{\cN,\sinit'}(\Diamond \goal) \ > \ 0
  $$
  Contradiction.
This completes the proof of the claim.

Using the results of \cite{BKKM14}, %
each scheduler
$\sched'$ for $\cM'$ induces a corresponding scheduler $\sched$ for $\cN$
such that the conditional probability for some reachability objective
$\varphi$ under the assumption $\Diamond \goal$ with respect to scheduler
$\sched'$ in $\cM'$ agrees with the unconditional probability for
$\varphi$ in $\cN$ with respect to scheduler $\sched$.
If $\fpath'$ is a $\sched'$-path from
$\sinit'=\<\sinit,0\>$ to $\goal$ then
$\fpath'$ is a suffix of all ``corresponding'' $\sched$-paths $\fpath$
in $\cN$ and therefore $\rew'(\fpath') \leqslant \rew_{\cN}(\fpath)$.
With Lemma \ref{lemma:PosEC-uncond} applied to $\cN$ we get:
$$
  \ExpRew{\max}{\cM',\sinit'}
    (\ \accdiaplus \goal \ | \ \Diamond \goal \ )
  \ \ \leqslant \ \
  \ExpRew{\max}{\cN,\<\sinit,0\>}(\ \accdiaplus \goal \ )
  \ \ < \ \ \infty
$$
This completes the proof of 
Proposition \ref{proposition:exprew-infinite-critical-scheduler}.
\Ende
\end{proof}

\begin{corollary}
   If $\cM$ is an MDP where \eqref{assumption:A1} and \eqref{assumption:A2} hold and
   $\Pr^{\min}_{s}(\Diamond \goal)>0$ for all states
   $s\in S \setminus \{\fail\}$ then $\CExp{\max} < \infty$.
\end{corollary}

\begin{proof}
 The proof is obvious as there is no scheduler $\usched$ with
 $\Pr^{\usched}_{\sinit}(\Box \neg \goal)=1$. Hence, there is no critical
 scheduler.
\Ende
\end{proof}

\subsection{Algorithm for checking finiteness of $\CExp{\max}$}
\label{summary:check-finiteness}

By the obtained results,
the finiteness of the maximal conditional expected accumulated reward
can be checked as follows.
We first apply standard algorithms to compute the maximal end components.
If one of them is positive then the maximal conditional expected
accumulated reward is infinite (see Lemma \ref{lemma:PosEC}).
Otherwise we switch from $\cM$ to its MEC-quotient, i.e.,
we collapse all maximal end components into a single state
and identify $\cM$ with the resulting MDP.
The remaining task is to check the absence of critical schedulers in $\cM$
(see Definition \ref{def:critical-scheduler}
 and Proposition \ref{proposition:exprew-infinite-critical-scheduler}).
For this, we consider the largest sub-MDP $\cM_{\setminus \goal}$
that results by iteratively removing
states and actions. We first remove state $\goal$ and
for each state $s$, we remove all actions $\act$ from
$\Act(s)$ with $P(s,\act,\goal)>0$.
If some state arises where $\Act(s)$ is empty then we remove state $s$
using the same technique.
For all states $s$ in the resulting sub-MDP $\cM_{\setminus \goal}$
we have $\Pr^{\min}_{\cM_{\setminus \goal},s}(\Diamond \fail)=1$
(by (A2)) and
$\Pr^{\max}_{\cM,s}(\Diamond \fail)=1$
(as $\Pr^{\sched}_{\cM,s}(\Diamond \fail)=1$ for each scheduler for
$\cM_{\setminus \goal}$ viewed as a scheduler of $\cM$ with initial state $s$).
Furthermore,
$\cM_{\setminus \goal}$ has a positive cycle
if and only if $\cM$ has a critical scheduler.
The existence of a positive cycle can be checked using a
nested depth-first search \cite{CVWY92}
as it is known for checking the existence of a path satisfying a
B\"uchi (repeated reachability) condition in ordinary transition systems.
Together with the preprocessing explained in
Appendix \ref{appendix:finitness-transformation}
we get:

\begin{corollary}
  The task to check whether $\CExp{\max}$
  is finite is solvable in time polynomial in the size of $\cM$.
\end{corollary}

\subsection{Computing an upper bound}

\label{sec:upper-bound}

Suppose now that $\cM$ has no critical schedulers, in which case
$\CExp{\max}$ is finite.
The maximal unconditional expected accumulated
reward until goal
in the MDP $\cN$ constructed
in the proof of 
Proposition \ref{proposition:exprew-infinite-critical-scheduler} 
yields an upper bound $\CExp{\ub}$ for the maximal conditional
expected accumulated reward until reach $\goal$ in $\cM$:
$$
  \CExp{\ub} \ \ \eqdef \ \ 
  \ExpRew{\max}{\cN,\<\sinit,0\>}(\accdiaplus \goal)
  \ \ \geqslant \ \ 
  \CExp{\max}
$$
We now address the complexity bounds for the computation of
$\CExp{\ub}$ stated in Theorem~\ref{thm:finiteness}.
The logarithmic length of $R = \sum_{s} \rew^{\max}(s)$ where
$\rew^{\max}(s)=\max\{\rew(s,\alpha):\alpha \in \Act(s)\}$
is linear in $\Size(\cM)$.
The size of the MDPs $\cN$ is in $\cO(R\cdot \Size(\cM))$.
The computation of the upper bound $\CExp{\ub}$ then corresponds to 
the computation of the maximal unconditional expected accumulated
reward to reach $\goal$ in $\cN$, which has a polynomial time bound 
in the size of $\cN$. Consequently, $\CExp{\ub}$
can be computed in time polynomial in $R$ and the size of $\cM$,
which gives the pseudo-polynomial time bound as stated in
Theorem~\ref{thm:finiteness}.

As noted in the proof of 
Prop.~\ref{proposition:exprew-infinite-critical-scheduler},
for the special case where $\Pr^{\min}_{\cM,s}(\Diamond \goal)>0$ for all
states $s \in S \setminus \{\fail\}$,
we can avoid the annotation of the states with
the accumulated reward up to $N$. In this case, the size of $\cN$
is polynomial in $\Size(\cM)$, which leads to the polynomial bound on the
computation time of $\CExp{\ub}$ as stated in
Theorem~\ref{thm:finiteness}.

\section{Deterministic reward-based schedulers are sufficient}
\label{sec:reward-based}

In this section, we show that deterministic
reward-based schedulers are sufficient for $\CExp{\max}$
where we suppose that the given MDP $\cM = (S,\Act,P,\sinit,\rew)$ 
has two trap states $\goal$ and $\fail$ and
satisfies \eqref{assumption:A1}, \eqref{assumption:A2} and
$\CExp{\max} < \infty$.
Recall that \eqref{assumption:A1} asserts that $\goal$ is reachable from
all states $s\in S \setminus \{\fail\}$, while
\eqref{assumption:A2} asserts that $\cM$ has no end components
and therefore
$\Pr^{\min}_{\cM,s}(\Diamond (\goal \vee \fail))=1$ for all
states $s\in S$.

Recall that deterministic reward-based schedulers can be viewed as functions
$\sched : S \times \Nat \to \Act$.

\begin{proposition}
\label{prop:det-reward-based}
\begin{align*}
   \CExp{\max} \ \ = \ \
   \sup \
   \bigl\{ \ \CExp{\sched} \ : & \
            \text{$\sched$ is a deterministic reward-based scheduler}
              \\
              &
              \ \ \text{for $\cM$ such that} \ \
               \Pr^{\sched}_{\cM,\sinit}(\Diamond \goal)>0 \
   \bigr\}
\end{align*}
\end{proposition}

\begin{proof}
We consider the MDPs $\cM'$ and $\cN$  that have been introduced
in the proof of
Proposition \ref{proposition:exprew-infinite-critical-scheduler}.
Recall that $\cM$ and $\cM'$ have the same maximal conditional expectation
and there is a one-to-one-correspondence between the
reward-based schedulers for $\cM$ and $\cM'$.%
\footnote{If a deterministic reward-based scheduler $\sched$ for $\cM$
  chooses action $\alpha$
  for $(s,r)$ with $r \leqslant R$ then the lifted scheduler for $\cM'$
  chooses $\alpha$ for the state-action pair $(\<s,r\>,0)$ in $\cM'$,
  and vice versa.
  Recall that the accumulated reward of each path in $\cM'$ from
  $\sinit'=\<\sinit,0\>$ to some state $\<s,r\>$ in the first
  mode has reward 0.
  For the state-reward pairs $(s',r)$ where $s'$ is a state of $\cM'$
  in its second mode (i.e., $s'\in S$),
  the lifted scheduler selects the same action
  as $\sched$.}
For the following arguments we may identify $\cM$ and $\cM'$.
This means that we may assume whenever $\sched$ is a scheduler for $\cM$ with
$\Pr^{\sched}_{\cM,\sinit}(\Diamond \fail)=1$ then
$\rew(\fpath)=0$ for all $\sched$-paths from $\sinit$ to $\fail$.
Stated differently, by identifying $\cM$ and $\cM'$ we have
$\Pr^{\sched}_{\cM,\sinit}(\Diamond \goal)>0$ for
each scheduler $\sched$ for $\cM$ where
$\rew(\fpath)>0$ for some $\sched$-path from $\sinit$ to $\fail$.
With these assumptions, the MDP $\cN$ arises from $\cM$ by adding
a reset transition from $\fail$ to $\sinit$ with reward 0.
As shown in the proof of
Proposition \ref{proposition:exprew-infinite-critical-scheduler},
$\cN$ has no positive end components,
and therefore $\ExpRew{\max}{\cN,\sinit}(\accdiaplus \goal)$ is finite.
Moreover, as $\goal$ is the only trap state in $\cN$ and
as all state-action pairs contained in some end component of $\cN$
have reward 0, $\ExpRew{\max}{\cN,\sinit}(\accdiaplus \goal)$ agrees with
the maximal total reward in $\cN$.

We now switch from $\cN$ to a new MDP $\acc{\cN}$ that arises from
$\cN$ by attaching the accumulated reward to all states
and modifying $\cN$'s reward structure such that the reward of the
reset transition after a path $\fpath$  from $\sinit$ to a fail state
is $-\rew(\fpath)$.
Thus, $\acc{\cN}$ has infinitely many states and
positive and negative rewards.
Formally, the definition of $\acc{\cN}$ is as follows.
The state space of $\acc{\cN}$ is $\acc{S} = S \times \Nat$.
The action set remains unchanged, i.e.,
$\Act_{\acc{\cN}} = \Act_{\cN}=\Act \cup \{\tau\}$ with $\tau$ being a
symbol for the reset transitions.
The transition probabilities and reward structure in $\acc{\cN}$
are given by:
$$
  \acc{P}(\<s,r\>,\alpha,\<s',r'\>)
  \ \ = \ \
  \left\{
   \begin{array}{lcl}
      P(s,\alpha,s') & : & \text{if $r'=r{+}\rew(s,\alpha)$} \\
      0 & : & \text{otherwise}
   \end{array}
  \right.
$$
and
$$
  \acc{\rew}(\<s,r\>,\alpha) \ \ = \ \ \rew(s,\alpha)
$$
for all states $s\in S \setminus \{\goal,\fail\}$, $s'\in S$,
actions $\alpha \in \Act(s)$ and $r,r'\in \Nat$.
The reset action $\tau$ is the only enabled action of the fail states
with
$$
  \acc{P}(\<\fail,r\>,\tau,\<\sinit,0\>)=1
  \qquad \text{and} \qquad
  \acc{\rew}(\<\fail,r\>,\tau) \ \ = \ \ -r
$$
and $\acc{P}(\cdot)=0$ in all remaining cases.
The starting state of $\acc{\cN}$ is $\<\sinit,0\>$.

Each path in the original MDP $\cM$ can be lifted to a path
$\lift{\fpath}$ in $\acc{\cN}$ by
augmenting the states with reward values.
Vice versa,
for each finite path
$\fpath' = \<s_0,0\> \, \alpha_0 \,
                \<s_1,r_1\> \alpha_1 \ldots \alpha_{n-1}\, \<s_n,r_n\>$
in $\acc{\cN}$ where $\fail \notin \{s_0,s_1,\ldots,s_n\}$, the sequence
$\fpath = s_0\, \alpha_0 \, s_1 \, \alpha_1 \ldots \alpha_{n-1}\, s_n$ is a
path in $\cM$ with $\rew(\fpath)=r_n$.
Thus, $\acc{\rew}(\fpath') =0$ for all paths $\fpath'$ in $\acc{\cN}$
from $\<\sinit,0\>$ to $\<\sinit,0\>$
where the last transition is a reset-transition from some fail state
$\<\fail,r\>$ to $\<\sinit,0\>$.
If now $\fpath'$ is an arbitrary path in $\acc{\cN}$ from $\<\sinit,0\>$
to some state $\<s,r\>$ then $\acc{\rew}(\fpath')=r$.
We define $\fpath'|_{\cM}$ as the unique path $\fpath$ from $\sinit$ to $s$
in $\cM$
that arises from $\fpath'$
by (i) removing the longest prefix $\finpath$ of $\fpath'$
where the last transition of $\finpath$
is a reset transition (i.e., $\fpath'$ has the form $\finpath;\fpath''$
where $\fpath''$ is a path from $\<\sinit,0\>$ to some state $\<s,r\>$
in $\acc{\cN}$ that does not contain a reset transition)
and (ii) erasing the
reward annotations of  remaining path $\fpath''$.
Then, $\acc{\rew}(\fpath')=\rew(\fpath'|_{\cM})$ and the lifting of
$\fpath'|_{\cM}$ is the suffix $\fpath''$ of $\fpath'$.

Each deterministic memoryless scheduler $\sched$ for $\acc{\cN}$
can be viewed as a deterministic reward-based scheduler for $\cM$.
That is, if $\sched'$ is a deterministic memoryless scheduler for $\acc{\cN}$
then the corresponding deterministic reward-based scheduler
$\sched'|_{\cM}$ for $\cM$
is given by $\sched'|_{\cM}(s,r)=\sched'(\<s,r\>)$.
Using arguments as in \cite{BKKM14}
(see also Section~\ref{appendix:reset-Methode}), we get:
$$
  \ExpRew{\sched'}{\acc{\cN},\sinit}(\accdiaplus \goal)
  \ \ = \ \
  \ExpRew{\sched'|_{\cM}}{\cM,\sinit}
     (\ \accdiaplus \goal \ | \ \Diamond \goal \ )
$$
Vice versa, each (possibly history-dependent and randomized)
scheduler $\sched$  for $\cM$ induces a scheduler
$\lift{\sched}$ for $\acc{\cN}$ given by
$\lift{\sched}(\fpath') = \sched(\fpath'|_{\cM})$.
Thus, the residuals of the lifted scheduler enjoy the property
$\residual{\lift{\sched}}{\fpath'} = \lift{\sched}$ for each finite
path $\fpath'$ in $\acc{\cN}$ where the last transition in $\fpath'$
is the reset transition from some fail state $\<\fail,r\>$ to
$\<\sinit,0\>$. Again, using \cite{BKKM14}, we obtain:
$$
  \ExpRew{\lift{\sched}}{\acc{\cN},\sinit}(\accdiaplus \goal)
  \ \ = \ \
  \ExpRew{\sched}{\cM,\sinit}(\ \accdiaplus \goal \ | \ \Diamond \goal \ )
$$
The remaining task is to show that
the supremum of the values
$\ExpRew{\sched'}{\acc{\cN},\sinit}(\accdiaplus \goal)$
when ranging over all schedulers $\sched'$ for $\acc{\cN}$ agrees with
the supremum over the deterministic memoryless schedulers for $\acc{\cN}$.
For this, we rely on known results
for infinite-state MDPs with
positive and negative rewards as stated e.g.~in \cite{Puterman}.
Let $\acc{\cN}^+$ be the MDP resulting from $\acc{\cN}$
by replacing the reward function
$\acc{\rew}$ with the (non-negative) reward function $\acc{\rew}^+$
defined by
$$
  \acc{\rew}^+(\<s,r\>,\alpha)
  \ \ = \ \
  \max \ \bigl\{\, \acc{\rew}(s,\alpha), \, 0 \, \bigr\}
  \ \ = \ \ \rew_{\cN}(s,\alpha)
$$
for all states $s \in S \setminus \{\goal\}$,
actions $\alpha\in \Act_{\cN}(s)$ and $r\in \Nat$.
In particular, $\acc{\rew}^+(\<\fail,r\>,\tau)=0$.
It is well-known (see Section 7.1 in \cite{Puterman})
that if the maximal (unconditional) expected
total reward in $\acc{\cN}^+$ is finite, then
the supremum of the expected total reward in $\acc{\cN}$
when ranging over all schedulers $\sched'$ for $\acc{\cN}$ agrees with
the supremum over the deterministic memoryless schedulers for $\acc{\cN}$.
The expected total reward in $\acc{\cN}$ under some scheduler
$\sched'$ agrees with the value
$\ExpRew{\sched'}{\acc{\cN},\<\sinit,0\>}(\accdiaplus \mathit{Goal})$
where $\mathit{Goal}=\{\<\goal,r\> : r \in \Nat\}$.
The analogous statement holds for $\acc{\cN}^+$.
So, the remaining task is to show that the maximal expected reward
until reaching a goal state in
$\acc{\cN}^+$ is finite.
The relation $\cR = \{ (s,\<s,r\>): r \in \Nat\}$
is a reward-preserving bisimulation for the MDPs $\acc{\cN}^+$ and
the MDP $\cN$ resulting from $\cM$ by adding a reset-transition
of reward 0 from $\fail$ to $\sinit$ (see above).
Hence:
$$
  \ExpRew{\max}{\acc{\cN}^+,\<\sinit,0\>}(\accdiaplus \mathit{Goal})
  \ \ = \ \
  \ExpRew{\max}{\cN,\sinit}(\accdiaplus \goal) \ \ < \ \ \infty
$$
As stated above, using classical results about countable MDPs
with positive and negative rewards (see e.g.~\cite{Puterman}),
we obtain:
\begin{eqnarray*}
  & & \ExpRew{\max}{\acc{\cN},\sinit}(\accdiaplus \mathit{Goal})
  \\[2ex]
  = & &
  \sup \
  \bigl\{ \ \ExpRew{\sched'}{\acc{\cN},\<\sinit,0\>}
                (\accdiaplus \mathit{Goal}) \ : \
            \text{$\sched'$ is a det.~memoryless scheduler
             for $\acc{\cN}$}   \
  \bigr\}
  \\[2ex]
  = & &
  \sup \
  \bigl\{ \ \ExpRew{\sched}{\cM,\sinit}(\accdiaplus \goal) \ : \
            \text{$\sched$ is a det.~reward-based scheduler
             for $\cM$}   \
  \bigr\}
\end{eqnarray*}
This completes the proof of Proposition \ref{prop:det-reward-based}.
\Ende
\end{proof}

\section{Existence of a saturation point and optimal schedulers}

\label{appendix:saturation}

The main obligation in this section is to provide
a proof for Proposition~\ref{prop:saturation-maxsched} in the main paper,
asserting the existence of an optimal reward-based scheduler that is
memoryless for sufficiently large accumulated rewards. 

In this and the remaining sections of the appendix, we suppose
that MDP $\cM = (S,\Act,P,\sinit,\rew)$ 
has two trap states $\goal$ and $\fail$ and 
satisfies assumptions
\eqref{assumption:A1}, \eqref{assumption:A2} 
stated in Appendix \ref{appendix:finitness-critical-schedulers}
and
$\CExp{\max} < \infty$.
Recall that \eqref{assumption:A1} asserts that $\goal$ is reachable from
all states $s\in S \setminus \{\fail\}$, while
\eqref{assumption:A2} asserts that $\cM$ has no end components
and therefore
$\Pr^{\min}_{\cM,s}(\Diamond (\goal \vee \fail))=1$ for all
states $s\in S$.

In what follows, we will often use the residual notations 
$\residual{\sched}{R}$, $\residual{\sched}{(s,R)}$
and the redefine operator $\redefresidual{\sched}{R}{\tsched}$
that have been introduced for reward-based schedulers.
See Section \ref{sec:notations}.

\begin{proposition}[See Proposition~\ref{prop:saturation-maxsched}]
 \label{prop:saturation-maxsched-appendix}
  There exists
  a natural number $\saturation$ (called saturation point of $\cM$)
  and a deterministic memoryless scheduler $\maxsched$
  such that:
  \begin{enumerate}
  \item [(a)]
     $\CExp{\tsched} \ \leqslant \
      \CExp{\redefresidual{\tsched}{\saturation}{\maxsched}}$
     \
     for each scheduler $\tsched$ with
     $\Pr^{\tsched}_{\cM,\sinit}(\Diamond \goal)>0$.
  \item [(b)]
     $\CExp{\sched}=\CExp{\max}$ for some
     deterministic reward-based scheduler $\sched$ such that
     $\Pr^{\sched}_{\cM,\sinit}(\Diamond \goal)>0$ and
     $\residual{\sched}{\saturation}=\maxsched$.
  \end{enumerate}
\end{proposition}

\noindent 
In this section the statement (a) is captured in 
Lemma \ref{lemma:maxsched-after-threshold-is-optimal},
whereas statement~(b) corresponds to  
Lemma \ref{lemma:threshold-theorem}.
 In order to prove Prop.~\ref{prop:saturation-maxsched-appendix}, we first
 show the existence of a saturation point $\saturation$ which will be derived
 -- among others -- from the convergence rate for reaching one of the trap
 states (see Lemma \ref{lemma:abschaetzung-prob-reward-bounded-until} below). 
 The so obtained saturation point is, however, very large.
 Later (see Section \ref{appendix:compute-saturation}), 
 we show that there is a smaller and easily computable
 saturation point as outlined in Section \ref{sec:threshold}.

\subsection{Some technical statements}

\label{appendix:saturation-technical-statements}

The following lemmas are trivial observations about quotients of sums
that will be used at various places for comparing the conditional expectations
under different schedulers.

\begin{lemma}
   \label{lemma:basic-abschaetzung}
For all real numbers $K,L,k,l$ with $L > 0$ and
$l>0$,
one of the following three cases applies:
\begin{eqnarray*}
  & &
  \frac{K}{L} \ \ < \ \ \frac{K+k}{L+l} \ \ < \ \ \frac{k}{l}
  \\
  \\[0ex]
  \text{or} & \ \ \ \ &
  \frac{K}{L} \ \ > \ \ \frac{K+k}{L+l} \ \ > \ \ \frac{k}{l}
  \\
  \\[0ex]
  \text{or} & &
  \frac{K}{L} \ \ = \ \ \frac{K+k}{L+l} \ \ = \ \ \frac{k}{l}
\end{eqnarray*}
\end{lemma}

\noindent
The proof of Lemma \ref{lemma:basic-abschaetzung} is straightforward
and omitted here.
Note, however, if $K/L < k/l \leqslant k'/l'$ then
$(K{+}k')/(L{+}l') < (K{+}k)/(L{+}l)$ is possible.
Consider, for example, $K=1$, $L=2$ and $k=l=2$ and $k'=l'=1$.
Then, $(K{+}k')/(L{+}l') = 2/3 < 3/4 = (K{+}k)/(L{+}l)$.

\begin{lemma}
\label{lemma:rho-theta-x-y}
Let $\rho,\theta,\zeta,x,y,z$ be real numbers such that $x$, $y$ and
$z$ are non-negative and $x+y >0$, $x+z>0$ and $y > z$.
Then, one of the following three cases holds:
\begin{eqnarray*}
   & &
   \frac{\rho + \zeta}{x+z} \ \ < \ \
   \frac{\rho + \theta}{x+y} \ \ < \ \
   \frac{\theta - \zeta}{y-z}
   \\
   \\[0ex]
   \text{or} \ \ & \ \ \ \ &
   \frac{\rho + \zeta}{x+z} \ \ > \ \
   \frac{\rho + \theta}{x+y} \ \ > \ \
   \frac{\theta - \zeta}{y-z}
   \\
   \\[0ex]
   \text{or} \ \ & &
   \frac{\rho + \zeta}{x+z} \ \ = \ \
   \frac{\rho + \theta}{x+y} \ \ = \ \
   \frac{\theta - \zeta}{y-z}
\end{eqnarray*}
\end{lemma}

\begin{proof}
The claim follows from Lemma \ref{lemma:basic-abschaetzung}
with $(K,L) = (\rho + \zeta,x+z)$ and
$(k,l) = (\theta-\zeta,y-z)$.
\Ende
\end{proof}

We often make use of Lemma \ref{lemma:rho-theta-x-y} in the following form.
If $y > z$ then
$$
   \frac{\rho + \theta}{x+y} \ \leqslant \ \frac{\rho + \zeta}{x+z}
   \qquad \text{iff} \qquad
   \frac{\theta - \zeta}{y-z} \ \leqslant \ \frac{\rho + \theta}{x+y}
   \qquad \text{iff} \qquad
   \frac{\theta - \zeta}{y-z} \ \leqslant \ \frac{\rho + \zeta}{x+z}
$$
and the analogous statements for other comparison operators
$<$, $>$, $\geqslant$ and $=$ rather than  $\leqslant$.

\subsection{Convergence rate}

\label{appendix:saturation-bounds}

\begin{lemma}[Fast convergence for reaching a trap]
   \label{lemma:abschaetzung-prob-reward-bounded-until}
There exists $\lambda \in ]0,1[$ and $R_0\in \Nat$
such that for each state
$s\in S \setminus \{\goal,\fail\}$ and $r\in \Nat$ with $r \geqslant R_0$:
$$
  \Pr^{\min}_{\cM,s}
    \bigl( \ \Diamond^{\leqslant r} (\goal \vee \fail) \ \bigr)
  \ \ \geqslant \ \
  1 - \lambda^r
$$
\end{lemma}

\begin{proof}
The proof relies on a calculation similar to the one of \cite{UB13}.
The parameter $\lambda$ depends
on the minimal positive transition probability,
the maximal reward assigned to a state-action pair and
the number of states in $\cM$.
Let $N = |S|$ and
\begin{eqnarray*}
  q & \ = \ &
  \min \ \bigl\{ \ P(s,\act,t) \ : \
             (s,\act,t)\in S \times \Act \times S,\, P(s,\act,t)>0 \ \bigr\}
  \\[1ex]
  R & = &
  \max \ \bigl\{ \ \rew(s,\act) \ : \ (s,\act)\in S \times \Act \ \bigr\}
\end{eqnarray*}
\eqref{assumption:A2} yields that the probability to reach a trap state $\goal$ or $\fail$
from any state $s$ within $N$ or fewer steps is at least $p^N$
under each scheduler.
The accumulated reward $\rew(\fpath)$ of paths $\fpath$ with
$|\fpath|\leqslant N$ is bounded by $N \cdot R$.
Thus, for all schedulers $\sched$
and all $r, k\in \Nat$ with $r=k \cdot N \cdot R$:
\begin{center}
  $\Pr^{\sched}_{s}
  \bigl(\, \neg (\Diamond^{\leqslant r} (\goal \vee \fail)) \, \bigr)
  \ \ \leqslant \ \ (1-q^N)^k$
\end{center}
Let $R_0 = NR$ and $\lambda = (1-q^N)^{1/2NR}$ if $q < 1$.
Then,  $0 < \lambda < 1$ and for all $r \geqslant R_0$:
$$
  \lambda^r \ \ = \ \
  \bigl(1-q^N\bigr)^{r/2NR}
  \ \ \geqslant \ \
  \bigl(1-q^N\bigr)^{\lfloor r/NR \rfloor}
$$
Hence, for each scheduler $\sched$ and $r \geqslant R_0$:
\begin{eqnarray*}
  \Pr^{\sched}_{\cM,s}
    \bigl( \ \Diamond^{\leqslant r} (\goal \vee \fail) \ \bigr)
  & \ \geqslant \ &
  \Pr^{\sched}_{\cM,s}
    \bigl( \
     \Diamond^{\leqslant \lfloor r/NR \rfloor NR}
        (\goal \vee \fail) \ \bigr)
  \\[2ex]
  & = &
  1 -
  \Pr^{\sched}_{\cM,s}
    \bigl( \ \neg
     \Diamond^{\leqslant \lfloor r/NR \rfloor NR}
          (\goal \vee \fail) \ \bigr)
  \\[2ex]
  & \geqslant &
  1 - \bigl( 1-q^N \bigr)^{\lfloor r/NR \rfloor}
  \ \ \ \geqslant \ \ \
  1 - \lambda^r
\end{eqnarray*}
If $q=1$ then MDP can be viewed as a nondeterministic
transition system, in which case we can deal with any
$\lambda \in \ ]0,1[$.
\Ende
\end{proof}

As Remark \ref{remark:sched-infinite-exp-reward} shows,
assumption \eqref{assumption:A2} is crucial for
Lemma \ref{lemma:abschaetzung-prob-reward-bounded-until}.
As a consequence of Lemma \ref{lemma:abschaetzung-prob-reward-bounded-until}
we get that if $\cM$ satisfies assumptions \eqref{assumption:A1} and \eqref{assumption:A2} then
$\CExp{\sched} < \infty$
for each scheduler $\sched$ where
$\Pr^{\sched}_{\sinit}(\Diamond \goal)$ is positive,
but the supremum over all schedulers can still be infinite
(see Proposition \ref{proposition:exprew-infinite-critical-scheduler}).

\begin{lemma}[Arrearage of expected accumulated rewards]
   \label{lemma:abschaetzung-exp-reward}
Let $\lambda$ and $R_0$ be as in
Lemma \ref{lemma:abschaetzung-prob-reward-bounded-until}.
Then, for each $R \geqslant R_0$, each scheduler
$\sched$ and each state $s$ of $\cM$ we have:
$$
  \sum_{r=R}^{\infty}
     r \cdot \Pr^{\sched}_{\cM,s}(\Diamond^{=r}\goal)
  \ \ \leqslant \ \
  C \cdot R \cdot \lambda^R
  \qquad \text{where} \qquad
  C \ = \ \frac{1}{(1-\lambda)^2}
$$
\end{lemma}

\begin{proof}
As $0 < \lambda <1$, the infinite series
$\sum_{r=0}^{\infty} r \cdot \lambda^r$
converges to $\lambda/(1-\lambda)^2$.
More precisely, for each $R\in \Nat$ with $R \geqslant R_0$:
\begin{eqnarray*}
  \sum_{r=R}^{\infty} r \cdot \lambda^r
  & \ \ = \ \ &
  \lambda^R \cdot \sum_{r=0}^{\infty} \, (r+R)\cdot \lambda^r
  \ \ \ = \ \ \
  \lambda^R \cdot \sum_{r=0}^{\infty} r \cdot \lambda^r
  \ + \
  R \cdot \lambda^R \cdot \sum_{r=0}^{\infty}  \lambda^r
  \\
  \\[0ex]
  & \ \ = \ \ &
  \frac{\lambda^{R+1}}{(1-\lambda)^2}
  \ \ + \ \
  \frac{R \cdot \lambda^R}{1-\lambda}
  \\
  \\[0ex]
  & \ \ = \ \ &
  R \cdot \lambda^R \cdot
  \Bigl( \frac{1}{R} \cdot \frac{\lambda}{(1-\lambda)^2}
         \  \ + \ \ \frac{1}{1-\lambda} \ \Bigr)
  \\
  \\[0ex]
  & \ \ \leqslant \ \  &
  R \cdot \lambda^R \cdot
  \Bigl( \frac{\lambda}{(1-\lambda)^2}
         \  +\  \frac{1}{1-\lambda} \ \Bigr)
  \ \ \ = \ \ \
  R \cdot \lambda^R \cdot \frac{1}{(1-\lambda)^2}
  \ \ \ = \ \ \
  C \cdot R \cdot \lambda^R
\end{eqnarray*}
This yields the claim.
\Ende
\end{proof}

Recall (see Definition \ref{def:CExp})
that if $s \in S \setminus \{\goal,\fail\}$ and $\tsched$ is
a scheduler with $\Pr^{\tsched}_s(\Diamond \goal) >0$ then:
  $$
   \CExpState{\tsched}{s} \ \ = \ \
   \ExpRew{\tsched}{s}(\ \accdiaplus \goal \ | \ \Diamond \goal \ )
   \ \ = \ \
   \frac{\Exp{\tsched}{s}}
        {\Pr^{\tsched}_s(\Diamond \goal)}
  $$
  where $\Exp{\tsched}{s}$ is a shortform notation for
  $\ExpRew{\tsched}{s}(\ \accdiaplus \goal \ )$ given by
  $$
    \Exp{\tsched}{s}
    \ \ = \ \
    \sum_{r=0}^{\infty} r \cdot \Pr^{\tsched}_s(\Diamond^{=r} \goal)
  $$
\begin{corollary}
  \label{cor:upper-bound-cond-exp-scheduler}
  Assumptions and notations as in Lemma \ref{lemma:abschaetzung-exp-reward}.
  For each scheduler $\tsched$,
  each $R \in \Nat$ with $R\geqslant R_0$ and each state
  $s\in S \setminus \{\goal,\fail\}$ we have:
  $$
    \Exp{\tsched}{s}  \ \ \leqslant \ \
    (R{-}1) \cdot \Pr^{\tsched}_s(\Diamond^{< R} \goal) \ \ + \ \
     C \cdot R \cdot \lambda^R
  $$
  Moreover, if \, $\Pr^{\tsched}_s(\Diamond \goal) >0$ \, and
  \, $\CExpState{\tsched}{s} \geqslant 2R$ \, then \,
  $\Pr^{\tsched}_s(\Diamond \goal) < C \cdot \lambda^R$.
\end{corollary}%

\begin{proof}
Clearly, for each scheduler $\tsched$ and state $s$ we have:
$$
  \sum_{r=0}^{R-1} r \cdot \Pr^{\tsched}_s(\Diamond^{=r} \goal)
  \ \ \ \leqslant \ \ \
  (R-1) \cdot \Pr^{\tsched}_s(\Diamond^{< R} \goal)
$$
Thus, the first statement is
a consequence of Lemma \ref{lemma:abschaetzung-exp-reward}.
For the second statement, we suppose
$\CExpState{\tsched}{s} \geqslant 2R$. But then:
$$
  2R \cdot \Pr^{\tsched}_s(\Diamond \goal)
  \ \ \ \leqslant \ \ \
  \Exp{\tsched}{s} \ \ \ < \ \ \
  R \cdot \Pr^{\tsched}_s(\Diamond \goal) \ + \ C \cdot R \cdot \lambda^R
$$
and therefore:
$R \cdot \Pr^{\tsched}_s(\Diamond \goal)
  \, < \, C \cdot R \cdot \lambda^R$.
This yields $\Pr^{\tsched}_s(\Diamond \goal) \, < \, C \cdot \lambda^R$.
\Ende
\end{proof}

Recall that we use the notation $\Diamond^{\bowtie n}$ for reward-bounded
reachability, while $\neXt^{\bowtie n}$ denotes a step-bound.

\begin{proposition}[Continuity of conditional expectations]
  \label{proposition:continuity-exp-reward}
  Let $\sched$ be a scheduler of $\cM$ with
  $\Pr^{\sched}_{\cM,\sinit}(\Diamond \goal)>0$
  and $\varepsilon >0$.
  Then, there exists $n_{\varepsilon} \in \Nat$ such that
  $\Pr^{\sched}_{\cM,\sinit}(\neXt^{\leqslant n_{\varepsilon}} \goal)>0$
  and for each scheduler $\tsched$ with $\sched(\fpath)=\tsched(\fpath)$
  for all finite paths $\fpath$ with $|\fpath|\leqslant n_{\varepsilon}$
  and $\first(\fpath)=\sinit$ we have:
  $$
      \big| \ \CExp{\sched} \ - \ \CExp{\tsched} \ \big|
      \ \ < \ \ \varepsilon
  $$
  Note that under the above assumption
  $\Pr^{\tsched}_{\cM,\sinit}(\Diamond \goal)$ is positive.
\end{proposition}

\begin{proof}
Let $\Pi_n$ denote the set of all finite paths 
$s_0\, \act_0 \, s_1 \ldots \act_{n-1}\, s_n$ of length $n$
from $s_0=\sinit$ to $s_n=\goal$ with
$\goal \notin \{s_0,\ldots,s_{n-1}\}$
and $\Pi_{\bowtie N}=\bigcup_{n \bowtie N} \Pi_n$, e.g.,
$\Pi_{\leqslant N} = \Pi_0 \cup \Pi_1 \cup \ldots \cup \Pi_N$
and
$\Pi_{> N} = \Pi_{N+1} \cup \Pi_{N+2} \cup \ldots$.

We first suppose that $\CExp{\sched}>0$.
Then, $z= \CExp{\sched} \cdot \Pr^{\sched}_{\sinit}(\Diamond \goal)>0$.
We pick some $n_0 \in \Nat$ such that
$$
  x_0 \ \ = \ \
  \Pr^{\sched}_{\cM,\sinit}(\neXt^{\leqslant n_0} \goal)
  \ \ = \ \
  \Pr^{\sched}_{\cM,\sinit}(\Pi_{\leqslant n_0})
  \ \ > 0
$$
and such that there is at least one finite $\sched$-path
$\fpath \in \Pi_{\leqslant n_0}$ with $\rew(\fpath)>0$.
Lemma \ref{lemma:abschaetzung-prob-reward-bounded-until}
(applied to $\cM$ with the unit-reward function) and
Lemma \ref{lemma:abschaetzung-exp-reward} yield a step bound
$n_{\varepsilon} \geqslant n_0$ such that
$$
 \begin{array}{rcl}
    \Pr^{\tsched}_{\cM,\sinit}
          (\, \neXt^{\leqslant n_{\varepsilon}} (\goal \vee \fail) \, )
    & \ \ \geqslant \ \ &
    1 - \frac{1}{4} \cdot x_0^2 \cdot \varepsilon \cdot
    \frac{1}{z}
    \\[1ex]

    \sum\limits_{n=n_{\varepsilon}+1}^{\infty}
    \sum\limits_{\fpath \in \Pi_n}
      \rew(\fpath) \cdot \probability^{\tsched}(\fpath)
    & \leqslant &
    \frac{1}{4} \cdot x_0^2 \cdot \varepsilon
\end{array}
$$
for all schedulers $\tsched$.
Let $x = \Pr^{\sched}_{\cM,\sinit}(\neXt^{\leqslant n_{\varepsilon}} \goal)$
and
\begin{center}
   $\rho \ \ = \ \
    \sum\limits_{n=0}^{n_{\varepsilon}}
    \sum\limits_{\fpath \in \Pi_{n_{\varepsilon}}}
         \rew(\fpath) \cdot \probability^{\sched}(\fpath)
    \ \ = \ \
    z \ - \!
    \sum\limits_{n=n_{\varepsilon}+1}^{\infty}
    \sum\limits_{\fpath \in \Pi_{> n_{\varepsilon}}}
         \rew(\fpath) \cdot \probability^{\sched}(\fpath)$
\end{center}
Then, $x_0 \leqslant x$ and $\rho \leqslant z$.
Moreover, $\rho>0$ by the choice of $n_0$ and the requirement
$n_{\varepsilon} \geqslant n_0$.
In particular:
$$
 \begin{array}{lclcl}
     y^{\tsched} & \ \eqdef \ & \ 
     \Pr^{\tsched}_{\cM,\sinit}(\Pi_{> n_{\varepsilon}}) \ 
     &\  < & \ 
     \frac{1}{4} \cdot x^2 \cdot \varepsilon \cdot
     \frac{1}{\rho}
 \end{array}
$$
For all non-negative real number $y$, $\theta$ with
$y <
 \frac{1}{4} \cdot x^2 \cdot \varepsilon \cdot \frac{1}{\rho}$
and $\theta < \frac{1}{4} \cdot x^2 \cdot \varepsilon$ we have:
$$
    \big| \ \frac{\rho+\theta}{x+y} - \frac{\rho}{x} \ \big|
    \ \ = \ \
    \frac{ | \, \theta \cdot x \, - \, \rho \cdot y \, | }{(x+y)\cdot x}
    \ \ \leqslant \ \
    \frac{  \theta \cdot x  \, + \, \rho \cdot y }{x^2}
    \ \ < \ \
    \frac{\varepsilon}{2}
$$
Suppose now that $\tsched$ agrees with $\sched$ for all paths up to length
$n_{\varepsilon}$. Then:
$$
  \CExp{\sched} \ = \
  \frac{\rho + \theta^{\sched}}{x+y^{\sched}}
  \hspace*{0.86cm}
  \CExp{\tsched} \ = \
  \frac{\rho + \theta^{\tsched}}{x+y^{\tsched}}
  \hspace*{0.86cm}
  \text{where} \ \
  \text{$\theta^{\tsched} =
         \sum\limits_{n > n_{\varepsilon}} \sum\limits_{\fpath \in \Pi_n}
            \rew(\fpath) \cdot \probability^{\sched}(\fpath)$}
$$
The definition of $\theta^{\sched}$ is analogous.
By the choice of $n_{\varepsilon}$ we have
$\theta^{\sched},\theta^{\tsched} <
 \frac{1}{4} \cdot x^2 \cdot \varepsilon$.
We obtain:
$$
    \big| \ \CExp{\sched} - \CExp{\tsched} \ \big|
    \ \ \leqslant \ \
    \big| \ \CExp{\sched} - \frac{\rho}{x} \ \big|
    \ + \
    \big| \ \frac{\rho}{x} - \CExp{\tsched}\ \big|
    \ \ < \ \
    \frac{\varepsilon}{2} + \frac{\varepsilon}{2}
    \ = \
    \varepsilon
$$
It remains to consider the case $\CExp{\sched}=0$.
In this case we pick some $n_{\varepsilon} \in \Nat$ such that
$x = \Pr^{\sched}_{\sinit}(\neXt^{\leqslant n_{\varepsilon}}\goal)>0$
and
$$
 \begin{array}{lcrcl}
    \theta^{\tsched} & = &
    \sum\limits_{n=n_{\varepsilon}+1}^{\infty}
    \sum\limits_{\fpath \in \Pi_n}
      \rew(\fpath) \cdot \probability^{\tsched}(\fpath)
    & < &
    x \cdot \varepsilon
\end{array}
$$
for all schedulers $\tsched$ (Lemma \ref{lemma:abschaetzung-exp-reward}).
If $\sched$ and $\tsched$ agree on all paths
up to length $n_{\varepsilon}$ then:
$$
    \CExp{\tsched} \ \ = \ \
    \frac{\theta^{\tsched}}{x+y^{\tsched}} \ \ \leqslant \ \
    \frac{\theta^{\tsched}}{x}
    \ \ < \ \
    \varepsilon
$$
where $y^{\tsched}$ is as above.
\Ende
\end{proof}

\subsection{Optimal reward-based eventually memoryless schedulers}
\label{appendix:saturation-memory-optimal}

\begin{proposition}[Turning point]
  \label{prop:turning-point}
  There exists $\turning \in \Nat$ such that
  for each scheduler $\sched$ the following statement holds.
  If $\fpath$ is a finite $\sched$-path from $\sinit$ to some state
  $s \in S \setminus \{\goal,\fail\}$
  such that $\rew(\fpath) \geqslant \CExp{\sched}+1$ and
  $\CExpState{\residual{\sched}{\fpath}}{s} \geqslant \turning$ then
  $$
    \CExp{\sched} \ \ < \ \
    \CExp{\redefresidual{\sched}{\fpath}{\usched}}
  $$
  where $\usched$ is an arbitrary
  scheduler that maximizes the probability to reach $\goal$ from
  $s$.
  The value $\turning$ will be called a turning point of $\cM$.
\end{proposition}

\begin{proof}
Recall that $p_s^{\max} = \Pr^{\max}_s(\Diamond \goal)$.
Let
$$
   p \ \ = \ \ \min_{\stackrel{s\in S}{s\not=\fail}}  p_s^{\max}
$$
The default assumption stating that $\goal$ is reachable from all states
in $S \setminus \{\fail\}$ yields $p > 0$.
Let $\lambda \in ]0,1[$ and $R_0 \in \Nat$ be as in
Lemma \ref{lemma:abschaetzung-prob-reward-bounded-until}
and $C$ as in Lemma \ref{lemma:abschaetzung-exp-reward}.
Recall that $C = 1/(1{-}\lambda)^2$,
which yields $C > 1 \geqslant p$.
We pick some $\turning \geqslant R_0$ such that $\turning$ is even and
$$
  C \cdot \turning \cdot \lambda^{\frac{\turning}{2}}
  \ \ < \ \ \frac{p}{2}
$$
Let $\fpath$ be a finite $\sched$-path from $\sinit$ to some state
$s\in S \setminus \{\goal,\fail\}$ such that $\rew(\fpath)$ is positive and
$\CExpState{\residual{\sched}{\fpath}}{s} \geqslant \turning$.
We proceed as follows. We first establish upper bounds for
$\Pr^{\tsched}_s(\Diamond \goal)$ and
$\ExpRew{\tsched}{s}(\Diamond \goal)$.
These bounds will be used in the second part where we prove
$\CExp{\sched} < \CExp{\redefresidual{\sched}{\fpath}{\usched}}$.

Let $\tsched = \residual{\sched}{\fpath}$ be the residual scheduler
and $y = \Pr^{\tsched}_s(\Diamond \goal)$.
The first part of Corollary \ref{cor:upper-bound-cond-exp-scheduler}
yields:
$$
  \Exp{\tsched}{s}
  \ \ \ \leqslant \ \ \
  \frac{\turning}{2} \cdot y
  \ \ + \ \
  C \cdot \frac{\turning}{2} \cdot \lambda^{\frac{\turning}{2}}
$$
By assumption we have $\CExpState{\tsched}{s} \geqslant \turning$.
By the choice of $\turning$ we have
$C \cdot \turning \cdot \lambda^{\frac{\turning}{2}} \  <  \ p/2$.
Hence:
$$
  y
  \ \ \ \leqslant \ \ \
  C \cdot \lambda^{\frac{\turning}{2}}
  \ \ \ < \ \ \
 \frac{p}{2} \cdot \lambda^{\frac{\turning}{2}}
$$
by the second part of
Corollary \ref{cor:upper-bound-cond-exp-scheduler}.
The fact that $C > p$ implies $C > p/2$. By the choice of $\turning$
we obtain:
$$
  \Exp{\tsched}{s}
  \ \ \leqslant \ \
  \frac{p}{2} \cdot \frac{\turning}{2} \cdot \lambda^{\frac{\turning}{2}}
  \ + \
  C \cdot \frac{\turning}{2} \cdot \lambda^{\frac{\turning}{2}}
  \ \ \leqslant \ \
  C \cdot \turning \cdot \lambda^{\frac{\turning}{2}}
  \ \ < \ \ \frac{p}{2}
$$
We now compare the conditional expectations of the schedulers
$\sched$ and $\redefresidual{\sched}{\fpath}{\usched}$
where $\usched$ is a memoryless scheduler maximizing the probabilities
to reach the goal state from each state. That is,
$p^{\max}_t = \Pr^{\usched}_t(\Diamond \goal) > 0$.

Let $r = \rew(\fpath)$ and $z = \probability(\fpath)$.
By assumption $r \geqslant \CExp{\sched}{+}1$.
We define:
$$
    \rho \ \ = \ \
    \Exp{\sched}{\sinit} \ - \ yr
    \qquad \text{and} \qquad
    x \ = \ \Pr^{\sched}_{\sinit}(\Diamond \goal) \ - \ y
$$
Recall that $s = \last(\fpath)$.
Then:
$$
  \CExp{\sched} \ \ = \ \
  \frac{\rho + z (yr + \Exp{\tsched}{s})}{x + zy}
$$
and
$$
  \CExp{\redefresidual{\sched}{\fpath}{\usched}}
  \ \ \ = \ \ \
  \frac{\rho + z(p_s^{\max}r + \Exp{\usched}{s})}{x + zp_s^{\max}}
  \ \ \ \geqslant \ \ \
  \frac{\rho + zp_s^{\max}r}{x + zp_s^{\max}}
$$
Thus, to prove
$\CExp{\sched} < \CExp{\redefresidual{\sched}{\fpath}{\usched}}$,
it suffices to show:
$$
  \frac{\rho + z (yr + \Exp{\tsched}{s})}{x + zy}
  \ \ \ < \ \ \
  \frac{\rho + zp_s^{\max}r}{x + zp_s^{\max}}
$$
We now use the bounds $\Exp{\tsched}{s} < p/2$ and
$y < p/2 \cdot \lambda^{\frac{\turning}{2}}$ (which yields $y < p/2$) that
have been established above and obtain:
\begin{eqnarray*}
  \frac{zp_s^{\max}r \ - \ z(yr + \Exp{\tsched}{s})}{zp_s^{\max} - zy}
   & \ = \ &
  \frac{p_s^{\max}r - (yr + \Exp{\tsched}{s})}{p_s^{\max} - y}
  \\
  \\[0ex]
  & = &
  \frac{p_s^{\max}r-yr}{p_s^{\max}-y} \ - \
  \frac{\Exp{\tsched}{s}}{p_s^{\max} - y}
  \\
  \\[0ex]
  & = &
  r \ - \ \frac{\Exp{\tsched}{s}}{p_s^{\max} - y}
  \\
  \\[0ex]
  & \geqslant &
  r \ - \ \frac{\Exp{\tsched}{s}}{p - y}
  \\
  \\[0ex]
  & > &
  r \ - \ \frac{\frac{p}{2}}{p - \frac{p}{2}}
  \ \ \ = \ \ \ r - 1
\end{eqnarray*}
By assumption we have $r = \rew(\fpath) \geqslant \CExp{\sched}+1$.
Thus, $r{-}1 \geqslant \CExp{\sched}$. Therefore:
$$
  \frac{zp_s^{\max}r \ - \ z(yr + \Exp{\tsched}{s})}{zp_s^{\max} - zr}
  \ \ \ > \ \ \ r-1
  \ \ \ \geqslant \ \ \ \CExp{\sched}
$$
By Lemma \ref{lemma:rho-theta-x-y} we get:
$$
  \CExp{\sched} \ \ < \ \
  \frac{\rho + zp_s^{\max}r}{x+zp_s^{\max}}
  \ \ < \ \
  \frac{zp_s^{\max}r \ - \ z(yr + \Exp{\tsched}{s})}{zp_s^{\max} - zy}
$$
But then
$\CExp{\sched} \ < (\rho + zp_s^{\max}r)/(x+zp_s^{\max}) \ \leqslant \
 \CExp{\redefresidual{\sched}{\fpath}{\usched}}$.
\Ende
\end{proof}

Given a reward-based scheduler $\sched$
and a state-reward pair $(s,r)\in S \times \Nat$
with $r \geqslant \CExp{\sched}{+}1$ and
$\CExpState{\residual{\sched}{(s,r)}}{s} \geqslant \turning$,
we may applying Proposition \ref{prop:turning-point} repeatedly
to obtain
  $$
    \CExp{\sched} \ \ < \ \
    \CExp{\redefresidual{\sched}{(s,r)}{\usched}}
  $$
where $\usched$ is any
scheduler that maximizes the probability to reach $\goal$ from $s$.
Hence,
by Proposition \ref{prop:det-reward-based},
$\CExp{\max}$ is the supremum of the values $\CExp{\sched}$ where
$\sched$ ranges over all deterministic reward-based
schedulers with $\CExp{\residual{\sched}{r}} < \turning$
for all $r\geqslant \CExp{\sched}{+}1$.

\begin{definition}[Eventually memoryless]
\label{def:eventually-memoryless}
\label{def:reward-threshold}
Let $\sched$ be a reward-based scheduler. $\sched$ is called
eventually memoryless if there exists $\saturation \in \Nat$ such that
$\sched(s,r) = \sched(s,\saturation)$ for all $r \geqslant \saturation$.
\end{definition}

We will show that there exists an optimal
reward-based eventually memoryless scheduler $\sched$
such that for all states $s\in S \setminus \{\goal,\fail\}$ and
each $r\geqslant \saturation$ we have
$\residual{\sched}{r} = \maxsched(s)$
where $\maxsched$ is a deterministic memoryless scheduler that maximizes
the probability to reach $\goal$ from all states and the 
conditional expectations
for all those schedulers.
We will see that such a scheduler $\maxsched$ is computable in polynomial time
using linear programming techniques.
See Lemma \ref{lemma:Sched-max-exp} below.

\begin{definition}[Additional notations %
]
{\rm
As before, let $p_s^{\max} = \Pr^{\max}_s(\Diamond \goal)$.
Let $\Act^{\max}(s)$ denote the set of actions $\alpha \in \Act(s)$ where
$$
  p_s^{\max} \ \ = \ \
  \sum_{t\in S} P(s,\alpha,t) \cdot p_t^{\max}
$$
Let $\Sched^{\max}$ denote the class of deterministic
schedulers $\usched$ such that
$\Pr^{\usched}_s(\Diamond \goal)=p_s^{\max}$ for all states $s$.
\Ende
  }
\end{definition}

It is well-known (see e.g.~\cite{Puterman}) that
$\Act^{\max}(s)$ is nonempty and that
$\usched(\fpath) \in \Act^{\max}(\last(\fpath))$ for each $\usched$-path
starting in $s$ and each scheduler $\usched \in \Sched^{\max}$.
This justifies to regard the sub-MDP $\cM^{\max}$ of $\cM$ that arises
by eliminating all state-action pairs $(s,\beta)$ with $s\in S$
and $\beta \notin \Act^{\max}(s)$. That is, the enabled action of $s$
as a state of $\cM^{\max}$ are exactly the actions in $\Act^{\max}(s)$.
Clearly:
$$
  \Pr^{\vsched}_{\cM,s}(\Diamond \goal) \ \ = \ \
  \Pr^{\vsched}_{\cM^{\max},s}(\Diamond \goal)
$$
for each scheduler $\vsched$ for $\cM^{\max}$
and each scheduler $\usched \in \Sched^{\max}$ is also a scheduler
for $\cM^{\max}$.
The reverse direction does not hold in general as
$\cM^{\max}$ can have end components that do not contain the goal state,
i.e., $\Pr^{\vsched}_{s}(\Diamond \goal) < p_s^{\max}$
for some scheduler $\vsched$ for $\cM^{\max}$ is possible.
However, such scenarios are impossible because of
assumptions \eqref{assumption:A1} and \eqref{assumption:A2}.

\begin{lemma}
\label{lemma:Sched-max}
$\Sched^{\max}$ agrees with the set of schedulers for $\cM^{\max}$.
That is,
for each scheduler $\usched$ for $\cM^{\max}$ we have
$\Pr^{\usched}_{\cM,s}(\Diamond \goal) =p_{s}^{\max}$ for all states $s$.
\end{lemma}

\begin{proof}
Let $\sched$
be a deterministic memoryless schedulers for $\cM^{\max}$
where $\Pr^{\sched}_{s}(\Diamond \goal)$ is minimal for all states $s$.
Let $q_s = \Pr^{\sched}_s(\Diamond \goal)$ and $\beta_s = \sched(s)$.
Let $S'=\{s\in S : q_s>0\}$. Then, the vector $(q_s)_{s\in S'}$ is the
unique solution of the following linear equation system with variables
$x_s$ for $s\in S'$:
$$
  \begin{array}{l}
     x_s \ =  \
     \sum\limits_{t\in S} P(s,\beta_s,t) \cdot x_t
     \qquad \text{for $s\in S'\setminus \{\goal\}$}
     \\[2ex]
     x_{\goal} \, = \, 1
 \end{array}
$$
But the vector $(p^{\max}_s)_{s\in S'}$ also solves the above linear
equation system. Hence, $q_s=p_s^{\max}$ for all states $s\in S'$.

It remains to show that $S \setminus \{\fail\} = S'$.
For all states $s\in S \setminus S'$
we have $\Pr^{\sched}_s(\Diamond \fail)=1$ by assumption \eqref{assumption:A2}
and the vector $(w_s)_{s\in S \setminus S'}$ with $w_s=1$ for all
$s\in S \setminus S'$
is the unique solution of the following
linear equation system with variables $y_s$ for $s\in S \setminus S'$:
$$
  \begin{array}{l}
     y_s \ =  \
     \sum\limits_{t\in S} P(s,\beta_s,t) \cdot y_t
     \qquad \text{for $s\in S \setminus (S' \cup \{\fail\})$}
     \\[2ex]
     y_{\fail} \, = \, 1
 \end{array}
$$
However, the vector $(1{-}p_s^{\max})_{s\in S \setminus S'}$
also solves the above
linear equation system. Hence, $1{-}p^{\max}_s =1$ and therefore
$p^{\max}_s=0$ for all $s \in S \setminus S'$.
But then $S \setminus S' = \{\fail\}$ by assumption \eqref{assumption:A1}.
\Ende
\end{proof}

\begin{remark}%
\label{remark:sched-max}
Obviously, for each scheduler $\sched \in \Sched^{\max}$ we have:
$$
  \CExp{\sched} \ \ = \ \
  \frac{\ \Exp{\sched}{\sinit} \ }{p_{\sinit}^{\max}}
$$
Hence, if $\sched, \usched \in \Sched^{\max}$ then
$$
   \CExp{\sched} \ \geqslant \ \CExp{\usched}
   \qquad \text{iff} \qquad
   \Exp{\sched}{\sinit} \ \geqslant \ \Exp{\usched}{\sinit}
$$
For each $\sched \in \Sched^{\max}$
and each $\sched$-path $\fpath$ from $\sinit$,
the residual schedulers $\residual{\sched}{\fpath}$ maximize the
probabilities to reach $\goal$ from $\last(\fpath)$.
Hence, we may suppose $\residual{\sched}{\fpath} \in \Sched^{\max}$.
Thus, $\CExp{\sched}$ is maximal under all schedulers
in $\Sched^{\max}$ iff
$$
  \Exp{\residual{\sched}{\fpath}}{s}
  \ \  = \ \
  \sup \
  \bigl\{ \
        \Exp{\usched}{s} \ : \
        \usched \in \Sched^{\max} \
  \bigr\}
$$
for all $\sched$-paths 
$\fpath$ from $\sinit$ with $s=\last(\fpath) \not= \fail$.
This follows by the fact that
$$
  \frac{\rho + \theta}{x+p} \ \ \geqslant \ \
  \frac{\rho + \zeta}{x+p}
  \qquad \text{iff} \qquad
  \theta \geqslant \zeta
$$
for all real numbers $\rho,\theta,\zeta,x,p$ with $x+p > 0$.
In this case, we deal with $p = \probability(\fpath)$,
$x = p_{\sinit}^{\max}-p$,
$\theta = \Exp{\residual{\sched}{\fpath}}{s}$
and
$\rho = \Exp{\sched}{\sinit}- p \theta$.
Thus, $\CExp{\sched} = (\rho + \theta)/(x+p)$.
The value $\zeta$ stands for the possible values
$\Exp{\usched}{s}$ for $\usched \in \Sched^{\max}$.
\Ende
\end{remark}

\begin{lemma}
\label{lemma:maxsched-Theta}
Let
$\Theta_s \ = \ \sup \
  \bigl\{ \
        \Exp{\usched}{s} \ : \
        \usched \in \Sched^{\max} \
  \bigr\}$.
Then:
$$
  \Theta_s
  \ \ = \ \
  \max \
  \Bigl\{ \
     \rew(s, \alpha) \cdot p_{s}^{\max} \ + \
     \sum_{t\in S} P(s,\alpha,t) \cdot \Theta_t \ : \
      \alpha \in \Act^{\max}(s) \
  \Bigr\}
$$
\end{lemma}

\begin{proof}
Clearly, for each state
$s\in S \setminus \{\fail\}$ and $\alpha \in \Act^{\max}(s)$
and each deterministic scheduler $\usched$ we have:
$$
   \Exp{\usched}{s} \ \ = \ \
   \rew(s, \alpha) \cdot p_{s}^{\max} \ + \
   \sum_{t\in S} P(s,\alpha,t) \cdot
      \Exp{\residual{\usched}{(s \,\alpha \, t)}}{t}
$$
where $\alpha = \usched(s)$.
This yields $\Theta_s \geqslant \Xi_s$ for all states $s$
where
$$
  \Xi_s
  \ \ = \ \
  \max \
  \Bigl\{ \
     \rew(s, \alpha) \cdot p_{s}^{\max} \ + \
     \sum_{t\in S} P(s,\alpha,t) \cdot \Theta_t \ : \
      \alpha \in \Act^{\max}(s) \
  \Bigr\}
$$
It remains to show that $\Xi_s \leqslant \Theta_s$ for all states $s$.
Suppose by contradiction that
$\Theta_s > \Xi_s$ for some state $s$.
Let $\varepsilon = (\Theta_s -\Xi_s)/2$. We pick some deterministic
scheduler $\sched \in \Sched^{\max}$ such that
$\Exp{\sched}{s} > \Theta_s - \varepsilon$. Hence, $\Exp{\sched}{s} > \Xi_s$.
For $\alpha = \sched(s)$, we get:
$$
  \Xi_s \ \  < \ \
  \Exp{\sched}{s} \ \ \leqslant \ \
  \rew(s, \alpha) \cdot p_{s}^{\max} \ + \
  \sum_{t\in S} P(s,\alpha,t) \cdot \Theta_t
  \ \ \leqslant \ \ \Xi_s
$$
Contradiction.
\Ende
\end{proof}

\begin{lemma}
\label{lemma:Sched-max-exp}
There exists a deterministic memoryless scheduler
$\maxsched \in \Sched^{\max}$ that maximizes the
partial expected total reward until reaching $\goal$
for all states $s\in S \setminus \{\fail\}$
under all schedulers $\usched \in \Sched^{\max}$, i.e.,
$$
  \Exp{\maxsched}{s} \ \ = \ \
  \max \
  \bigl\{ \ \Exp{\usched}{s} \ : \ \usched \in \Sched^{\max} \ \bigr\}
$$
Such a scheduler $\maxsched$ and the values $\Exp{\maxsched}{s}$
are computable in time polynomial in the size of $\cM^{\max}$,
using the linear program with variables $\theta_s$ for $s\in S$
shown in Figure \ref{fig:LP-for-maxsched}.
\end{lemma}

\noindent
By Lemma \ref{lemma:Sched-max-exp} and
Remark \ref{remark:sched-max}:
$\CExpState{\maxsched}{s} \ = \
  \max \
  \bigl\{ \ \CExpState{\usched}{s} \ : \ \usched \in \Sched^{\max} \ \bigr\}$.

\begin{figure}[t]
$$
 \begin{array}{l}
    \text{Minimize $\sum\limits_{s\in S} \theta_s$ subject to}
    \\[2ex]
    \begin{array}{ll}
      \mathrm{(1)} \ \ &
      \theta_s  \ \geqslant \
      \rew(s,\alpha) \cdot p_s^{\max} \, + \,
      \sum\limits_{t\in S} P(s,\alpha,t) \cdot \theta_t
      \quad \text{for $s \in S \setminus \{\goal,\fail\}$,
                   $\alpha \in \Act^{\max}(s)$}
       \\[2ex]
       \mathrm{(2)} \ \ &
       \theta_{\goal} \, = \,  \theta_{\fail}  \, = \, 0
       \ \ \text{and} \ \
       \theta_s \geqslant 0 \ \text{ for $s \in S \setminus \{\goal,\fail\}$}
    \end{array}
  \end{array}
$$
\caption{Linear program for
      $\max \,
       \bigl\{ \, \Exp{\usched}{s} \, : \,
                  \usched \in \Sched^{\max} \, \bigr\}$}
\label{fig:LP-for-maxsched}
\end{figure}

\begin{proof}
The linear program in Figure \ref{fig:LP-for-maxsched} is the same
as the one for the maximal (unconditional) total expectation
of the MDP $\cM'$ that agrees with $\cM^{\max}$, but uses
the (rational-valued) reward function
$\rew'(s,\alpha) = \rew(s,\alpha) \cdot p_s^{\max}$.
Using standard results for finite MDPs (see e.g.~\cite{Puterman}),
we get that the linear
program has a unique solution.
Thus, one proof obligation is to show that
$\Exp{\max}{\cM^{\max},s}(\accdiaplus \goal) =
 \ExpRew{\max}{\cM',s}(\text{``total reward''})$
for all states $s$ and to show that optimal schedulers for $\cM'$
(w.r.t.~to the total expected reward) are
optimal for $\cM^{\max}$ (w.r.t.~the partial expectation until reaching
the goal state).
We follow here a different approach and present a direct proof
that adapts the soundness of the linear program for total expected accumulated
rewards in finite MDPs.

Clearly, the vector $(\Theta_s)_{s\in S}$
defined as Lemma \ref{lemma:maxsched-Theta}
provides a solution for the constraints (1) and (2)
in Figure \ref{fig:LP-for-maxsched}.
Moreover, Lemma \ref{lemma:maxsched-Theta} implies that for
each state $s\in S \setminus \{\goal,\fail\}$ there is an action
$\beta_s \in \Act^{\max}(s)$ such that
$$
   \Theta_s \ \ = \ \
   \rew(s,\beta_s) \cdot p_s^{\max} \  +  \
   \sum\limits_{t\in S} P(s,\beta_s,t) \cdot \Theta_t
$$
Let $\maxsched$ be the deterministic memoryless scheduler for $\cM^{\max}$
given by $\maxsched(s)=\beta_s$ for all states
$s \in S \setminus \{\goal,\fail\}$.
By Lemma \ref{lemma:Sched-max} we get $\maxsched \in \Sched^{\max}$,
i.e.,
$\Pr^{\maxsched}_s(\Diamond \goal) = p_s^{\max}$ for all $s$.

The vectors $(\Exp{\maxsched}{s})_{s\in S}$
and $(\Theta_s)_{s\in S}$ solve the following
linear equation system with variables $\zeta_s$ for all states $s\in S$:
$$
 \begin{array}{ll}
   \text{(3)} \ \ &
   \zeta_s \ = \
   \rew(s,\beta_s) \cdot p_s^{\max} \  +  \
   \sum\limits_{t\in S} P(s,\beta_s,t) \cdot \zeta_t
   \qquad \text{for $s\in S \setminus \{\goal,\fail\}$}
   \\
   \\[-2ex]
   \text{(4)} &
   \zeta_{\goal} \, = \, \zeta_{\fail} \, = \, 0
 \end{array}
$$
By applying standard arguments for the Markov chain induced by $\maxsched$
and using assumption \eqref{assumption:A2}, 
we obtain that the above linear equation system
has a unique solution.%
\footnote{Note that the linear equation system (3), (4)
 can be written in the form
 $(I{-}A)\zeta = b$ where $A$ is the probability matrix of the
 Markov chain induced by $\maxsched$ restricted to the states
 $s\in S \setminus \{\goal,\fail\}$ and $I$ the  identity matrix.
 The vector $b$ contains the values $\rew(s,\beta_s) \cdot p_{s^{\max}}$, 
 $s\in S \setminus \{\goal,\fail\}$.
 Assumption \eqref{assumption:A2} ensures that $I{-}A$ is non-singular. 
 This implies the existence of a unique solution of equations (3) and (4).}
This yields:
$\Theta_s = \Exp{\maxsched}{s}$ for all states $s\in S$.

It remains to show that
$\sum_{s\in S} \Theta_s \, \leqslant \, \sum_{s\in S} \rho_s$
for each solution $(\rho_s)_{s\in S}$ of (1) and (2).
We pick a solution $(\rho_s)_{s\in S}$ of (1) and (2).
We first observe that then also the vector with the elements
$\min\{\Theta_s,\rho_s\}$ is a solution of (1) and (2).
Hence, we may assume that $\rho_s \leqslant \Theta_s$ for all $s\in S$.

We now define $\rho_s^{(0)} = \rho_s$ and for $n \in \Nat$:
$$
   \rho_s^{(n+1)} \ \ = \ \
   \rew(s,\beta_s) \cdot p_s^{\max} \ + \
     \sum_{t\in S} P(s,\alpha,t)\cdot \rho_t^{(n)}
$$
By induction on $n$, we get
$\rho_s^{(0)} \geqslant \rho_s^{(1)} \geqslant \rho_s^{(2)} \geqslant \ldots$
for all $s\in S$ and $n \geqslant 0$.
Let
$$
   \rho_s^* \ \ = \ \ \lim_{n\to \infty} \rho_s^{(n)}
$$
Clearly, we have $\rho_s \geqslant \rho_s^*$ and
$$
     \rho_s^{*} \ \ = \ \
     \rew(s,\beta_s) \cdot p_s^{\max} \ + \
     \sum_{t\in S} P(s,\alpha,t)\cdot \rho_t^{*}
$$
for all states $s$.
But then the vector $(\rho_s^*)_{s\in S}$ solves the linear equation system
(3) and (4) of above. Again, we can rely on the fact that (3) and (4)
have a unique solution, which yields $\rho_s^* = \Theta_s$ for all states $s$.
But then $\rho_s \, \geqslant \, \rho_s^* \, \geqslant \, \Theta_s$
for all $s$.

The above shows that the vector $(\Theta_s)_{s\in S}$ is the unique
solution of the linear program shown in Figure \ref{fig:LP-for-maxsched}
and coincides with the vector $(\Exp{\maxsched}{s})_{s\in S}$.
\Ende
\end{proof}

\begin{lemma}[Existence of optimal
                eventually memoryless schedulers]
\label{lemma:threshold-theorem}
$\CExp{\sched}=\CExp{\max}$ for some
deterministic reward-based scheduler $\sched$ such that
$\Pr^{\sched}_{\cM,\sinit}(\Diamond \goal)>0$ and
$\residual{\sched}{\saturation}=\maxsched$ for some saturation point
$\saturation$.
\end{lemma}

\begin{proof}
We define
$$
  \delta_s \ \ = \ \
  \min \
  \Bigl\{ \
     p_s^{\max} \ - \ \sum_{t\in S} P(s,\beta,t) \cdot p_t^{\max}  \ : \
     \beta \in \Act(s)\setminus \Act^{\max}(s) \
  \Bigr\}
$$
and with $S_{\delta} = \{ s \in S : \delta_s > 0\}$:
$$
  \delta \ \ = \ \ \min_{s\in S_{\delta}} \delta_s
$$
If $S_\delta = \varnothing$ then $\Act(s) = \Act^{\max}(s)$ for all states
$s$. In this case, the deterministic memoryless scheduler
$\maxsched$ as in Lemma \ref{lemma:Sched-max-exp}
is an optimal scheduler as it maximizes the conditional expectation
from every state (see Remark \ref{remark:sched-max}).

In what follows, we suppose that $S_{\delta}$ is nonempty, in which case
$\delta$ is positive.
Let $\turning$ be the turning point of $\cM$ as in
Proposition \ref{prop:turning-point}.
We now define the saturation point $\saturation$
as any natural number satisfying the
following constraint:
$$
  \saturation
  \ \ \geqslant \ \
  \CExp{\max} \ + \ \frac{\turning}{\delta}
$$
As $0 < \delta \leqslant 1$ we have $\turning/\delta \geqslant 1$.
Moreover, we pick a deterministic memoryless
scheduler $\maxsched$ as in Lemma \ref{lemma:Sched-max-exp}.
Lemma \ref{lemma:maxsched-after-threshold-is-optimal} (see below)
shows that for each partial deterministic reward-based scheduler
$$
  \psched :
  S \times \{ r \in \Nat : 0 \leqslant r < \saturation\}
  \, \to \, \Act
$$
the scheduler
$\redefresidual{\psched}{\saturation}{\maxsched}$
is optimal among all schedulers
$\redefresidual{\psched}{\saturation}{\tsched}$ where
$\tsched$ ranges over all schedulers and where optimality
is understood with respect to conditional expectations.
Note that the scheduler
$\redefresidual{\psched}{\saturation}{\maxsched}$
is reward-based eventually memoryless with saturation point $\saturation$.
Since there are only finitely many partial schedulers $\psched$,
this completes the proof of
Lemma \ref{lemma:threshold-theorem}.
\Ende
\end{proof}

\begin{lemma}
\label{lemma:maxsched-after-threshold-is-optimal}
  $\CExp{\redefresidual{\psched}{\saturation}{\maxsched}}
   \, \geqslant  \,
   \CExp{\redefresidual{\psched}{\saturation}{\tsched}}
  $
  for each scheduler $\tsched$ and
  each partial scheduler
  $\psched :
    S \times \{ r \in \Nat : 0 \leqslant r < \saturation\}
    \, \to \, \Act$.
\end{lemma}

\begin{proof}
We first prove the following claim.

\tudparagraph{1ex}{{\it Claim.}}
Suppose we are given two schedulers $\sched$ and $\tsched$
that agree for all but the extensions
of some finite path $\fpath$ starting in $\sinit$
where $\fpath$ is both a $\sched$-path and a $\tsched$-path from $\sinit$,
i.e., $\residual{\sched}{\finpath}= \residual{\tsched}{\finpath}$
for each finite path $\finpath$ where $\fpath$ is not a prefix of $\fpath$.
Let $s= \last(\fpath)$ and suppose  $\rew(\fpath) \geqslant \saturation$.
Then:
\begin{center}
   $\Pr^{\residual{\sched}{\fpath}}_{s}(\Diamond \goal)
    =  p^{\max}_s  > 
    \Pr^{\residual{\tsched}{\fpath}}_{s}(\Diamond \goal)$
   \ implies \
   $\CExp{\redefresidual{\psched}{\saturation}{\sched}} \geqslant
    \CExp{\redefresidual{\psched}{\saturation}{\tsched}}$
\end{center}
\tudparagraph{1ex}{{\it Proof of the claim.}}
We provide the proof of the claim for deterministic schedulers.
The argument for randomized schedulers is similar and omitted here
as randomized schedulers are irrelevant for our purposes
by Proposition \ref{prop:det-reward-based}.
Let $s=\last(\fpath)$, $r=\rew(\fpath)$,
$w = \probability(\fpath)$ and
$\vsched = \residual{\tsched}{\fpath}$ and
$\usched = \residual{\sched}{\fpath}$.
The assumption $\Pr^{\usched}_s(\Diamond \goal)=p_s^{\max}$
yields that $\Pr^{\residual{\usched}{\finpath}}_t(\Diamond \goal)=p_t^{\max}$
for each finite $\usched$-path with $s=\first(\finpath)$ and
$t=\last(\finpath)$.
We may assume w.l.o.g. that $\fpath$ is minimal with
respect to the above property,
i.e., $\usched(s) \not= \vsched(s)$.
Then, $\CExp{\sched}$ and $\CExp{\tsched}$ have the following form:
$$
  \CExp{\sched}
  \ \ = \ \
  \frac{\rho + w p r + w \Exp{\usched}{s}}{x + wp}
  \qquad \text{and} \qquad
  \CExp{\tsched}
  \ \ = \ \
  \frac{\rho + w y r + w \Exp{\vsched}{s}}{x + w y}
$$
where
$p = p_s^{\max}$ and $y = \Pr^{\vsched}_s(\Diamond \goal)$.
Then, $y < p$ by assumption.
The values $\rho$ and $x$ are given by
$$
  \rho \ \ = \ \
  \sum_{\finpath} \rew(\finpath) \cdot \probability(\finpath)
  \qquad \text{and} \qquad
  x \ \ = \ \
  \sum_{\finpath} \probability(\finpath)
$$
where $\finpath$ ranges over all $\sched$-paths from $\sinit$ to $\goal$
where $\fpath$ is not a prefix of $\finpath$. These paths are also
$\tsched$-paths.
We now rely on Lemma \ref{lemma:rho-theta-x-y}, which yields
$$
  \CExp{\sched} \ \geqslant \ \CExp{\tsched}
  \qquad \text{iff} \qquad
  \frac{(pr + \Exp{\usched}{s}) - (yr +\Exp{\vsched}{s})}{p-y}
  \ \ \geqslant \ \ \CExp{\sched}
$$
Thus, the task is to show that
$$
  r (p-y) \ + \
  \Exp{\usched}{s} - \Exp{\vsched}{s}
  \ \ \ \geqslant \ \ \
  \CExp{\sched}(p-y)
$$
As $r \, \geqslant \, \saturation$ and
$\CExp{\sched}\leqslant \CExp{\max}$ and by the choice of
$\saturation$
we have:
$$
  r-\CExp{\sched}
  \ \ \ \geqslant \ \ \
  \saturation - \CExp{\max}
  \ \ \ \geqslant \ \ \
  \frac{ \turning}{\delta}
$$
Recall that $\turning$ denotes the turning point of $\cM$.
We may assume w.l.o.g. that
$$
  \frac{\Exp{\vsched}{s}}{y} \ \ = \ \
  \CExpState{\vsched}{s} \ \ = \ \
  \CExpState{\residual{\tsched}{\fpath}}
  \ \ < \ \ \turning
$$
as otherwise the claim follows immediately by
Proposition \ref{prop:turning-point}.
But this yields:
$$
  \Exp{\vsched}{s}
  \ \ < \ \
  \turning  y \ \ < \ \ \turning p
$$
As $\usched(s) \not= \vsched(s)$ and $y < p = p_s^{\max}$
we have $p{-}y \geqslant \delta$.
(Here, we use the assumption that $\usched$ and $\vsched$ are deterministic.)
But then $(p{-}y)/\delta \geqslant 1$ and therefore:
\begin{eqnarray*}
  \Exp{\vsched}{s} - \Exp{\usched}{s}
  \ \ \ \leqslant \ \ \
  \Exp{\vsched}{s}
 & \leqslant &
  \turning p
  \ \ \ \leqslant  \ \ \
  \turning
  \ \ \ \leqslant  \ \ \
  \frac{p-y}{\delta} \cdot   \turning
  \\
  \\[0ex]
  & \ = \ &
  \frac{\turning}{\delta} (p-y)
  \ \ \ \leqslant \ \ \
  (r-\CExp{\sched}) (p-y)
\end{eqnarray*}
Hence,
$r (p{-}y) + \Exp{\usched}{s} - \Exp{\vsched}{s}
 \ \geqslant \  \CExp{\sched}(p{-}y)$.
This completes the proof of the claim.

\tudparagraph{2ex}{{\it Proof of
    Lemma \ref{lemma:maxsched-after-threshold-is-optimal}.}}
Let $S'$ denote the set of states $s\in S$ such that $s$
is the last state of some finite $\psched$-path $\fpath$
such that $r =\rew(\fpath)< \saturation$ and
$\rew(s,\psched(s,r)) \geqslant \saturation$.
Lemma \ref{lemma:Sched-max-exp} yields:
  $$
    \CExp{\redefresidual{\psched}{\saturation}{\maxsched}}
    \ \ \geqslant \ \
    \CExp{\redefresidual{\psched}{\saturation}{\tsched}}
  $$
for each scheduler $\tsched$ where
$\Pr^{\tsched}_s(\Diamond \goal) = p_s^{\max}$
for all $s\in S'$.
In the sequel, we address the case where $\tsched$ is a scheduler with
$\Pr^{\tsched}_s(\Diamond \goal) < p_s^{\max}$ for some state $s\in S'$.
Let $\fpath_1,\fpath_2,\fpath_3,\ldots $ be an enumeration of all
$\tsched$-paths such that
(i) $\rew(\fpath_i) \geqslant \saturation$ and
(ii) $\Pr^{\residual{\tsched}{\fpath}}_{s_i}(\Diamond \goal)
   < p_{s_i}^{\max}$ where
$s_i=\last(\fpath_i)$ and such that no proper prefix of $\fpath_i$
enjoys these two properties (i) and (ii).
Furthermore, we suppose that
$|\fpath_1| \leqslant |\fpath_2| \leqslant \ldots$.
We successively apply the claim to obtain a sequence of schedulers
$\tsched_0=\tsched, \tsched_1,\tsched_2,\ldots$ such that
$\tsched_{i+1} \ = \ 
  \redefresidual{\tsched_i}{\fpath_i}{\maxsched}$
and
$$
  \CExp{\redefresidual{\psched}{\saturation}{\tsched_0}}
  \ \ \leqslant \ \
  \CExp{\redefresidual{\psched}{\saturation}{\tsched_1}}
  \ \ \leqslant \ \
  \CExp{\redefresidual{\psched}{\saturation}{\tsched_2}}
  \ \ \leqslant \ \
  \ldots
$$
Moreover, the limit of the schedulers
$\CExp{\redefresidual{\psched}{\saturation}{\tsched_i}}$
is $\redefresidual{\psched}{\saturation}{\maxsched}$.
Proposition \ref{proposition:continuity-exp-reward} then yields
$\CExp{\redefresidual{\psched}{\saturation}{\tsched}}
 \ \ \leqslant \ \
 \CExp{\redefresidual{\psched}{\saturation}{\maxsched}}$.
\Ende
\end{proof}

\noindent
Obviously, 
Lemma \ref{lemma:maxsched-after-threshold-is-optimal} implies
  $\CExp{\redefresidual{\sched}{\saturation}{\maxsched}}
   \, \geqslant  \,
   \CExp{\sched}$
  for each scheduler $\sched$
as stated in part (a) of Proposition \ref{prop:saturation-maxsched-appendix}.

\section{Computing a saturation point}

\label{appendix:compute-saturation}

Although the proof presented in Appendix \ref{appendix:saturation} 
is constructive,
the constructed saturation point can be very large.
We now present a simple method for generating a smaller
saturation point.
The rough idea is make use of the observation made in
Appendix \ref{appendix:saturation} stating that optimal
schedulers eventually behave as the scheduler $\maxsched$.
Here and in what follows, $\maxsched$ is a deterministic memoryless
scheduler for $\cM$ that maximizes the probability to reach $\goal$ from
each state and whose conditional expectation is maximal under all those
schedulers (see Appendix \ref{appendix:saturation}).
The idea is now to compute a saturation point $\saturation$
as the smallest reward value from which on $\maxsched$ 
is better than other schedulers.

  Let $\theta_s=\Exp{\maxsched}{\cM,s}$ and
  $y_s=\Pr^{\maxsched}_{\cM,s}(\Diamond \goal)$. 
  Thus, $y_s=p_s^{\max}=\Pr^{\max}_{\cM,s}(\Diamond \goal)$.
  For each state-action pair $(s,\alpha)$ with $\alpha \in \Act(s)$ we define:
  \begin{center}
    \begin{tabular}{lll}
       $y_{s,\alpha} =  \sum\limits_{t\in S} P(s,\alpha,t)\cdot y_t$
       & \ and \ &
       $\theta_{s,\alpha} =
       \rew(s,\alpha)\cdot y_{s,\alpha}  +
    \sum\limits_{t\in S} P(s,\alpha,t)\cdot \theta_t$
 \end{tabular}
\end{center}
 By the choice of $\maxsched$, $y_{s,\alpha}\leqslant y_s$, and
 $\theta_{s,\alpha}\leqslant \theta_s$ if $y_{s,\alpha}=y_s$. 

 In what follows, we suppose that there is at least one state-action pair
 $(s,\alpha)$ with $y_{s,\alpha}<y_{s}$ as otherwise
 the scheduler $\maxsched$ maximizes the conditional expectation in $\cM$
 (see Lemma~\ref{lemma:Sched-max-exp}).
 We now define: 
  $$
    \saturation \ \ = \ \ 
    \max \ \bigl\{ \ \lceil \CExp{\ub} - D \rceil, \ 0 \ \bigr\} 
  $$
  where 
  $$
    D \ \ = \ \ 
    \min \
       \Bigl\{ \ \frac{\theta_s - \theta_{s,\alpha}}{y_s - y_{s,\alpha}} 
               \ : \ s \in S, \ \alpha \in \Act(s),  \ y_{s,\alpha} < y_s \
       \Bigr\}
  $$
  and $\CExp{\ub}$ is an upper bound for $\CExp{\max}$
  (e.g., the one computed by the algorithm presented in
  Section \ref{sec:upper-bound}).

\begin{proposition}
  \label{prop:improved-saturation}
    The computed value $\saturation$ is a saturation point for
    $\cM$, i.e.,
    $\CExp{\tsched} \ \leqslant \
    \CExp{\redefresidual{\tsched}{\saturation}{\maxsched}}$
    \
    for each scheduler $\tsched$ with
    $\Pr^{\tsched}_{\cM,\sinit}(\Diamond \goal)>0$.
\end{proposition}

\noindent
The remainder of this section is concerned with the proof of
Prop.~\ref{prop:improved-saturation}.
Let $\saturationNaive$ be some other saturation
point, e.g., the one obtained using
Lemma~\ref{lemma:threshold-theorem} and
Lemma~\ref{lemma:maxsched-after-threshold-is-optimal}. 
If $\saturation \geqslant \saturationNaive$, 
then $\saturation$ is obviously a saturation point as well.
In what follows, we suppose $\saturation < \saturationNaive$. 
As $\redefresidual{(\redefresidual{\tsched}{\saturationNaive})}
         {\saturation}{\maxsched}
    \ = \ \redefresidual{\tsched}{\saturation}{\maxsched}$,
it suffices to consider schedulers $\tsched$ that behave as $\maxsched$
for all paths $\fpath$ with $\rew(\fpath)\geqslant \saturationNaive$.
Furthermore, it suffices to consider reward-based schedulers.
In the sequel, let $\Sched'$ denote the set of reward-based 
schedulers $\tsched$ for
$\maxsched$ such that $\tsched(\fpath)=\maxsched(\last(\fpath))$ for
all paths $\fpath$ with $\rew(\fpath)\geqslant \saturationNaive$.
So, the task is to show that
$\CExp{\tsched} \ \leqslant \
    \CExp{\redefresidual{\tsched}{\saturation}{\maxsched}}$
    \
    for each scheduler $\tsched \in \Sched'$ with
    $\Pr^{\tsched}_{\cM,\sinit}(\Diamond \goal)>0$.

\begin{lemma}
 \label{diff-Sched-maxsched}
  $\theta_{s} - (\CExp{\ub}{-}\saturation)\cdot y_s
   \  \geqslant \ 
   \theta_{s,\alpha} - (\CExp{\ub}{-}\saturation)\cdot y_{s,\alpha}$
  for all states $s\in S\setminus \{\goal,\fail\}$ 
  and all actions $\alpha \in \Act(s)$.
\end{lemma}

\begin{proof}
By the definition of $\saturation$, 
for all $s\in S \setminus \{\goal,\fail\}$,
$\alpha \in \Act(s)$ with $y_{s,\alpha} < y_s$ we have:    
$$
     \saturation \ \ \ \geqslant \ \ \ \CExp{\ub}-D
     \ \ \ \geqslant \ \ \ 
     \CExp{\ub} \ - \ 
     \frac{\theta_s-\theta_{s,\alpha}}{y_s-y_{s,\alpha}} 
$$
and therefore:
$$ 
    \theta_s - (\CExp{\ub}{-}\saturation)\cdot y_s 
    \ \ \geqslant \ \ 
    \theta_{s,\alpha} - (\CExp{\ub}{-}\saturation)\cdot y_{s,\alpha}
$$
If $y_s = y_{s,\alpha}$ then 
$\theta_s \geqslant \theta_{s,\alpha}$ (by the choice of $\maxsched$).
The case $y_s < y_{s,\alpha}$ is impossible as $\maxsched$ maximizes
the probabilities to reach $\goal$.
This yields the claim.
\Ende
\end{proof}

\begin{lemma}
 \label{diff-Sched-all}
  $\theta_{s} - (\CExp{\ub}{-}\saturation)\cdot y_s
    \ \geqslant \ 
   \Exp{\tsched}{\cM,s} - 
      (\CExp{\ub}{-}\saturation)\cdot \Pr^{\tsched}_{\cM,s}(\Diamond \goal)$
  for all schedulers $\tsched \in \Sched'$.
\end{lemma}

\begin{proof}
The idea is to define a new MDP $\cN$ that simulates $\cM$ in such a way 
that the value
$\Exp{\tsched}{\cM,s} - 
      (\CExp{\ub}{-}\saturation)\cdot \Pr^{\tsched}_{\cM,s}(\Diamond \goal)$
equals the expected accumulated reward until reaching 
$\final$ from $s$ in $\cN$ under scheduler $\tsched \in \Sched'$.
The new MDP $\cN$ operates in two modes
and extends $\cM$ by a new trap state $\final$. 
It tracks the accumulated reward
until the moment the accumulated reward surpasses $\saturationNaive$. 
From that point on, $\cN$ behaves according to $\maxsched$.
The accumulated reward of a path in $\cM$ that has surpassed
$\saturationNaive$ in the last step is bound by:
$$
  N \ \ = \ \ \saturationNaive \ + \ \max_{s,\alpha} \rew(s,\alpha)
$$
That is, if $\fpath$ is a finite path in $\cM$ with
$\rew(\fpath) \geqslant \saturationNaive$ and
$\rew(\fpath') < \saturationNaive$ for all proper prefixes of
$\fpath$ then $\rew(\fpath) < N$.

Formally, $\cN$ is a MDP with negative and positive reward.
Its state space is:
$$
   S_{\cN} \ \ = \ \ S_1 \cup S_2 \cup \{\final\}
$$
where $S_1 = S \times \{0, \ldots, \saturationNaive {-}1\}$ (first mode)
and   $S_2 = S \times S \times \{\saturationNaive, \ldots, N\}$ (second mode).
Intuitively, the pairs $\<s,r\> \in S_1$ in the first mode
represent the current state 
$s$ and the accumulated reward $r$, while the triples
$\<s,t,r\>$ used for the states in the second mode
represent the current state $s$, the state $t$ where the
  switch to the second mode occurred and the accumulated reward $r$ until
  the switch.
The auxiliary state $\final$ is a trap.
The initial state of $\cN$ is $\<\sinit,0\>$.

The action set is $\Act_{\cN}=\Act \cup \{\tau\}$.
The transition probabilities for the states in 
the first mode are as follows. Let $s \in S \setminus \{\goal,\fail\}$,
$r \in \{0,1,\ldots,\saturationNaive{-}1\}$, $\alpha \in \Act(s)$ 
and $r'= r{+}\rew(s,\alpha)$. Then:
$$
 \begin{array}{lcll}
   P_{\cN}(\<s,r\>,\alpha,\<t,r'\>) & \ = \ & P(s,\alpha,t) \ \
   & 
   \text{ if } r'  < \saturationNaive
   \\[1ex]
   P_{\cN}(\<s,r\>,\alpha,\<t,t,r'\>) & \  = \ & P(s,\alpha,t) 
   & 
   \text{ if } r' = r  \geqslant \saturationNaive
 \end{array}
$$
and $\rew_{\cN}(\<s,r\>,\alpha)  \ = \ \rew(s,\alpha)$.
In the second mode $\cN$ behaves according to $\maxsched$. That is,
if $s \in S \setminus \{\goal,\fail\}$, 
$r\in \{\saturationNaive, \ldots, N\}$
 and $\alpha = \maxsched(s)$ then:
$$
   P_{\cN}(\<s,s',r\>,\alpha,\<t,s',r\>) \ \ = \ \ P(s,\alpha,t)
$$
and $\rew_{\cN}(\<s,s',r\>,\alpha) \  = \ \rew(s,\alpha)$.
Thus, $\Act_{\cN}(\<s,r\>)=\Act(s)$ for the states in the first mode, while 
$\Act_{\cN}(\<s,s',r\>)=\{\maxsched(s)\}$ for the states in the second mode.

The goal and fail states in both modes have $\tau$-transitions
to the final state, and no other actions is enabled in the
goal and fail states.
We first consider the goal states: 
$$
   P_{\cN}(\<\goal,r\>,\tau,\final) \  = \ 
   P_{\cN}(\<\goal,s,r\>,\tau,\final) \  = \ 1
$$
and   
$$
   \rew_{\cN}(\<\goal,r\>,\tau) \ \ = \ \
   \rew_{\cN}(\<\goal,s,r\>,\tau) \ \ = \ \
   - (\CExp{\ub}{-}\saturation)
$$
Thus, if $\fpath$ is a path from $s$ to $\goal$ in $\cM$ then
for the lifted path $\fpath_{\cN}$ from $\<s,0\>$ in $\cN$ we have
$\rew_{\cN}(\fpath_{\cN}) \ = \ 
    \rew(\fpath) - (\CExp{\ub} {-}\saturation )$.
For the fail states, we want to make sure that the partial expectation
for the accumulated reward from the initial state to $\final$ via a fail state
is 0 under all schedulers. For this purpose, we define:
$$
   P_{\cN}(\<\fail,r\>,\tau,\final) \  = \ 1,
   \qquad
   \rew_{\cN}(\<\fail,r\>,\tau) \ = \ -r
$$
This ensures that all paths $\fpath$ in $\cN$ from some state $\<s,0\>$ to 
$\final$ via a fail state in the first mode have reward 0.
For the fail states in the second mode, we define:
$$
   P_{\cN}(\<\fail,s,r\>,\tau,\final) \  = \ 1,
   \qquad
   \rew_{\cN}(\<\fail,s,r\>,\tau) \ = \
   -r - \frac{\Exp{\maxsched}{\cM,s}(\accdiaplus \fail)}
             {\Pr^{\maxsched}_{\cM,s}(\Diamond \fail)}
$$
provided that $\Pr^{\maxsched}_{\cM,s}(\Diamond \fail) >0$.
If $\Pr^{\maxsched}_{\cM,s}(\Diamond \fail) =1$ then state
$\<\fail,s,r\>$ is not reachable and the transition probabilities
and reward for $\<\fail,s,r\>$ are irrelevant.

Obviously, there is a one-to-one correspondence between the schedulers
$\tsched$ for $\cM$ that belong to $\Sched'$ and the schedulers for $\cN$.
For all schedulers $\tsched\in \Sched'$ we have
$\Pr^{\tsched}_{\cN,\tilde{s}}(\Diamond \final)=1$.
For all states $s\in S$ and all schedulers $\tsched \in \Sched'$:
$$
   \Exp{\tsched}{\cN,\<s,s',r\>}(\accdiaplus \mathit{Fail}) 
   \ \ = \ \ 
   \Exp{\maxsched}{\cM,s}(\accdiaplus \fail)
   \ \ = \ \ 
   \sum_{\fpath \in \Pi_s} \rew(\fpath)\cdot \probability(\fpath)
$$
where $\Pi_s$ denotes the set of finite $\maxsched$-paths 
$\fpath$ in $\cM$ with
$\first(\fpath)=s$ and $\last(\fpath)=\fail$.
Clearly, 
$\sum_{\fpath \in \Pi_s} \probability(\fpath) = 
 \Pr^{\maxsched}_{\cM,s}(\Diamond \fail)$.
The partial expectation of all paths from
$\<s,s',r\>$ to $\final$ via some fail state 
under each scheduler $\tsched \in \Sched'$ is:
$$
 \begin{array}{ll}
   & 
   \Exp{\tsched}{\cN,\<s,s',r\>}
     (\accdiaplus \text{``$\final$ via $\mathit{Fail}$''}) 
   \\
   \\
   = \ \ \ & 
   \sum\limits_{\fpath \in \Pi_s} 
     \Bigl(\ \rew(\fpath)-r-
      \frac{\Exp{\maxsched}{\cM,s'}(\accdiaplus \fail)}
           {\Pr^{\maxsched}_{\cM,s'}(\Diamond \fail)} \ \Bigr) 
     \cdot \probability(\fpath)
   \\[0ex]
   \\
   = \ \ \ &
   \Exp{\maxsched}{\cM,s}(\accdiaplus \fail) \ - \
   \Bigl(\ r+
      \frac{\Exp{\maxsched}{\cM,s'}(\accdiaplus \fail)}
           {\Pr^{\maxsched}_{\cM,s'}(\Diamond \fail)} \ \Bigr)
   \cdot \Pr^{\maxsched}_{\cM,s}(\Diamond \fail)
 \end{array}
$$
where $\mathit{Fail}$ denotes the set of all fail states (in either mode)
and
where we suppose $\Pr^{\maxsched}_{\cM,s'}(\Diamond \fail)>0$.
With $s=s'$ we get:
$$
   \Exp{\tsched}{\cN,\<s,s,r\>}
     (\accdiaplus \text{``$\final$ via $\mathit{Fail}$''}) 
   \ \ = \ \
   -r \cdot \Pr^{\maxsched}_{\cM,s}(\Diamond \fail)
$$
Therefore:
$$
 \begin{array}{ll}
   &
   \Exp{\tsched}{\cN,\<s,s,r\>}
     (\accdiaplus \final )
   \\[2ex]
   = \ \ \ &
   \Exp{\maxsched}{\cM,s}
   \ - \ 
   (\CExp{\ub}{-}\saturation) \cdot \Pr^{\maxsched}_{\cM,s}(\Diamond \goal)
   - \ r \cdot \Pr^{\maxsched}_{\cM,s}(\Diamond \fail)
  \\[2ex]
  = \ \ \ &
  \theta_s - (\CExp{\ub}{-}\saturation) \cdot y_s
   - \ r \cdot (1{-}y_s)
 \end{array}
$$
For the states in the first mode:
$$
\begin{array}{lcl}
  \ExpRew{\tsched}{\cN,\<s,r\>}(\accdiaplus \final)
  & \ \ = \ \ & 
  \Exp{\tsched}{\cM,s} \ - \ 
   (\CExp{\ub}{-}\saturation) \cdot \Pr^{\tsched}_{\cM,s}(\Diamond \goal)
  ] - \ r \cdot \Pr^{\tsched}_{\cM, s}(\Diamond \fail)
\end{array}
$$
We get for the special case $r=0$:
$$
  \ExpRew{\tsched}{\cN,\<s,0\>}(\accdiaplus \final)
  \ \ = \ \ 
  \Exp{\tsched}{\cM,s} \ - \ 
     (\CExp{\ub}{-}\saturation) \cdot \Pr^{\tsched}_{\cM,s}(\Diamond \goal)
$$
for all states $s\in S$ and all schedulers $\tsched \in \Sched'$. 
Moreover, for $\tsched=\maxsched$, we get for the states in the first mode:
$$
 \begin{array}{lcl}
  \ExpRew{\maxsched}{\cN,\<s,r\>}(\accdiaplus \final)
  & \ \ = \ \ &
  \Exp{\maxsched}{\cM,s} \ - \ 
     (\CExp{\ub}{-}\saturation) \cdot \Pr^{\maxsched}_{\cM,s}(\Diamond \goal)
  \ - \ r \cdot \Pr^{\maxsched}_{\cM,s}(\Diamond \fail)
  \\[2ex]
  & \ \ = \ \ &
  \theta_s \ - \ (\CExp{\ub}{-}\saturation) \cdot y_s \ - \ r \cdot (1{-}y_s)
 \end{array}
$$
for all states $s\in S$.
Note that $\Pr^{\max}_{\cM,s}(\Diamond \final) =
      1-\Pr^{\maxsched}_{\cM,s}(\Diamond \goal) = 1{-}y_s$.

Let $f : \Real^{S_{\cN}} \to \Real^{S_{\cN}}$ denote the
fixed point operator for the maximal (unconditional) expected accumulated
reward
until $\final$ in $\cN$. If
$\phi = (\phi_{\tilde{s}})_{\tilde{s}\in S_{\cN}}$ then
$$
     f(\phi) \ \ = \ \ 
     \bigl(\, f_{\tilde{s}}(\phi)\ \bigr)_{\tilde{s}\in S_{\cN}}
$$
where $f_{\final}(\phi)=0$ and for $\tilde{s}\in S_{\cN}\setminus \{\final\}$:
$$
     f_{\tilde{s}}(\phi) \ \ = \ \ 
     \max \ \bigl\{ \ f_{\tilde{s},\alpha}(\phi) \ : \ 
                      \alpha \in \Act_{\cN}(\tilde{s}) \ 
            \bigr\}
$$
where the function $f_{\tilde{s},\alpha} : \Real^{S_{\cN}} \to \Real$
is given by:
$$
     f_{\tilde{s},\alpha}(\phi) \ \ = \ \
     \rew_{\cN}(\tilde{s},\alpha) \ + \ 
               \sum_{\tilde{t}\in S_{\cN}} 
                 P_{\cN}(\tilde{s},\alpha,\tilde{t})\cdot \phi_{\tilde{t}} 
$$
We now consider the vector   
$\phi^* = (\phi^*_{\tilde{s}})_{\tilde{s}\in S_{\cN}}$
where 
$$
  \phi^*_{\tilde{s}} \ \ = \ \ 
  \ExpRew{\maxsched}{\cN,\tilde{s}}(\accdiaplus \final)
$$
As $\cN$ in the second mode 
has no nondeterministic choices and behaves
as $\maxsched$, we have 
$f_{\tilde{s}}(\phi^*) = \phi^*_{\tilde{s}}$ for all states
$\tilde{s}$ in the second mode of $\cN$.
For the states $\<s,r\>$ of the first mode, we have (see above):
$$
  \phi^*_{\<s,r\>} \ \ = \ \ 
  \theta_s - (\CExp{\ub}{-}\saturation) \cdot y_s
  \ - \ r \cdot (1{-}y_s)
$$
We now show that $\phi_{\<s,r\>}^* \geqslant f_{\<s,r\>,\alpha}(\phi^*)$
for each action  $\alpha \in \Act(s) =\Act_{\cN}(\<s,r\>)$.
Let $k=\rew(s,\alpha) = \rew_{\cN}(\<s,r\>,\alpha)$. 
In the following calculation, we suppose $r{+}k <\saturationNaive$. 
The calculation for mode switches (i.e., $r{+}k \geqslant \saturationNaive)$)
is similar and omitted here.
\begin{eqnarray*}
   f_{\<s,r\>,\alpha}(\phi^*) &  \ =  \ &
   k \ + \ 
   \sum_{t\in S} 
     P_{\cN}(\<s,r\>,\alpha,\<t,r{+}k\>) \cdot \phi^*_{\<t,r{+}k\>}
   \\
   \\
   & = &
   k \ + \ 
   \sum_{t\in S} 
     P(s,\alpha,t)
     \cdot \bigl( \ \theta_t - (\CExp{\ub}{-}\saturation) \cdot y_t
  \ - \ (r{+}k) \cdot (1{-}y_t) \ \bigr)
   \\
   \\
   & = &
   ky_s \ + \ 
   \sum_{t\in S} P(s,\alpha,t) \cdot \theta_t
   \ - \ 
   (\CExp{\ub}{-}\saturation) \cdot \sum_{t\in S} P(s,\alpha,t)\cdot y_t
   \\
   \\
   & & 
   \ + \ 
   k(1{-}y_s) 
   \ - \ 
   (r{+}k) \cdot \sum_{t\in S} P(s,\alpha,t)\cdot (1{-}y_t) 
   \\
   \\
   & = &
   \theta_{s,\alpha}
   \ - \ 
   (\CExp{\ub}{-}\saturation) \cdot y_{s,\alpha}
   \ + \ k(1{-}y_{s,\alpha}) 
   \ - \ (r{+}k)(1{-}y_{s,\alpha})
   \\
   \\
   & = &
   \theta_{s,\alpha}
   \ - \ 
   (\CExp{\ub}{-}\saturation) \cdot y_{s,\alpha}
   \ - \ r(1{-}y_{s,\alpha})
\end{eqnarray*}
If $\alpha = \maxsched(s)$ then 
$\theta_s=\theta_{s,\alpha}$ and $y_s=y_{s,\alpha}$.
Hence, $f_{\<s,r\>,\maxsched(s)}(\phi^*)=\phi_{\<s,r\>}^*$.
If $\alpha \in \Act(s) \setminus \{\maxsched(s)\}$ then
$y_s \geqslant y_{s,\alpha}$ 
and 
$$
   \theta_{s,\alpha}
   \ - \ 
   (\CExp{\ub}{-}\saturation) \cdot y_{s,\alpha}
   \ \ \leqslant \ \ 
   \theta_{s}
   \ - \ 
   (\CExp{\ub}{-}\saturation) \cdot y_{s}
$$
by Lemma \ref{diff-Sched-maxsched}. Hence:
$$
  f_{\<s,r\>}(\phi^*) 
  \ \ = \ \ 
  f_{\<s,r\>,\maxsched(s)}(\phi^*) \ \ = \ \ \phi^*_{\<s,r\>}
$$
This yields $f(\phi^*)=\phi^*$.
By the results of \cite{BerTsi91},
the vector
$$
  \bigl( \ 
   \ExpRew{\max}{\cN,\tilde{s}}(\accdiaplus \final) \ 
  \bigr)_{\tilde{s} \in S_\cN}
$$
is the unique fixpoint of $f$.%
\footnote{\cite{BerTsi91} considers the fixed point operator 
   for minimal (unconditional) expected accumulated rewards in MDPs.
   However, \cite{BerTsi91} treats MDPs that might have negative and positive
   reward values.
   Thus, by multiplying all rewards with ${-}1$ the results of
   \cite{BerTsi91} carry over to maximal expectations.}
Hence,
$\phi^*_{\tilde{s}} = 
 \ExpRew{\max}{\cN,\tilde{s}}(\accdiaplus \final)$
for all states $\tilde{s} \in S_{\cN}$.
That is, scheduler $\maxsched$ maximizes the (unconditional)
accumulated reward until reaching the final state in $\cN$.
That is,
$\ExpRew{\maxsched}{\cN,\tilde{s}}(\accdiaplus \final) \
 \geqslant \ 
 \ExpRew{\tsched}{\cN,\tilde{s}}(\accdiaplus \final)$
for all schedulers $\tsched$ for $\cN$ and all states
$\tilde{s}$ in $\cN$.
But then:
$$
 \begin{array}{lcl}
  \theta_s - (\CExp{\ub}{-}\saturation) \cdot y_s 
  & = &
  \ExpRew{\maxsched}{\cN,\<s,0\>}(\accdiaplus \final)
  \\[2ex]  
  & \ \ \ \geqslant \ \ \ &
  \ExpRew{\tsched}{\cN,\<s,0\>}(\accdiaplus \final)
  \\[2ex]
  & \ \ = \ \ &
  \Exp{\tsched}{\cM,s} \ - \  
    (\CExp{\ub}{-}\saturation) \cdot \Pr^{\tsched}_{\cM,s}(\Diamond \goal)
 \end{array}
$$
for all states $s$ in $\cM$ and all schedulers $\tsched \in \Sched'$.
\Ende
\end{proof}

\begin{lemma}
  If $\tsched\in \Sched'$ with $\Pr^{\tsched}_{\cM,\sinit}(\Diamond \goal)>0$
  then
  $\CExp{\tsched} \ \leqslant \
    \CExp{\redefresidual{\tsched}{\saturation}{\maxsched}}$.
\end{lemma}

\begin{proof}
Let $\Gamma$ denote the set of $\tsched$-paths $\fpath$ in $\cM$ from
$\sinit$ to $\goal$ with $\rew(\fpath) < \saturation$.
Let
$$
  x \ = \sum_{\fpath \in \Gamma} \probability(\fpath), \ \ \ 
  \rho \ = \ 
  \sum_{\fpath \in \Gamma} \rew(\fpath)\cdot \probability(\fpath)
$$
Then, 
$x = \Pr^{\tsched}_{\cM,\sinit}(\Diamond^{<\saturation} \goal)
  = \Pr^{\redefresidual{\tsched}{\saturation}{\maxsched}}_{\cM,\sinit}  
   (\Diamond^{<\saturation} \goal)$
and
$$ 
  \rho
   \ \ = \ \ 
   \sum_{r=0}^{\saturation-1}
   \Pr^{\tsched}_{\cM,\sinit}(\Diamond^{=r} \goal) \cdot r
   \ \ = \ \ 
   \sum_{r=0}^{\saturation-1}
   \Pr^{\redefresidual{\tsched}{\saturation}{\maxsched}}_{\cM,\sinit}
      (\Diamond^{=r} \goal) \cdot r
$$
Thus, the conditional expectations $\CExp{\tsched}$ and 
$\CExp{\redefresidual{\tsched}{\saturation}{\maxsched}}$ have the form
$$
  \CExp{\tsched} \ \ = \ \ \frac{\rho + \zeta}{x+z}
  \qquad \text{and} \qquad
  \CExp{\redefresidual{\tsched}{\saturation}{\maxsched}}
  \ \ = \ \ \frac{\rho + \theta}{x+y}
$$
where
$z = \Pr^{\tsched}_{\cM,\sinit}(\Diamond^{\geqslant \saturation} \goal)$
and
$y = \Pr^{\redefresidual{\tsched}{\saturation}{\maxsched}}_{\cM,\sinit}  
   (\Diamond^{\geqslant \saturation} \goal)$
and 
$\zeta$ and $\theta$ are the corresponding partial expectations.

For $r \in \Nat$, $r \geqslant \saturation$, 
let $\Pi_{s,r}$ denote the set of $\tsched$-paths $\fpath$ from $\sinit$ 
to $s$ with $\saturation \leqslant \rew(\fpath)$
and such that $\saturation > \rew(\fpath')$ for all proper prefixes
$\fpath'$ of $\fpath$. Thus, $\Pi_{s,r}=\varnothing$ if
 $r > N \eqdef \saturation + \max\limits_{s,\alpha} \rew(s,\alpha)$.
Let $p_{s,r} = \sum\limits_{\fpath \in \Pi_{s,r}} \probability(\fpath)$.
Then:
$$
   y \ \ = \ \ 
   \sum_{s\in S} \sum_{r=\saturation}^N 
      p_{s,r} \cdot \Pr^{\maxsched}_{\cM,s}(\Diamond \goal)
   \ \ = \ \ 
   \sum_{s\in S} \sum_{r=\saturation}^N 
      p_{s,r} \cdot y_s
$$
Similarly:
$$
   z \ \ = \ \ 
   \sum_{s\in S} \sum_{r=\saturation}^N 
      p_{s,r} \cdot z_{s,r}
   \qquad \text{where} \qquad
   z_{s,r} \ = \
      \Pr^{\residual{\tsched}{r}}_{\cM,s}(\Diamond \goal)
$$
and 
$$
   \theta \ \ = \ \ 
   \sum_{s\in S} \sum_{r=\saturation}^N 
      p_{s,r} \cdot \bigl(\Exp{\maxsched}{\cM,s} + r \cdot y_s\bigr)
   \qquad
   \zeta 
   \ \ = \ \ 
   \sum_{s\in S} \sum_{r=\saturation}^N 
      p_{s,r} \cdot 
      \bigl(\Exp{\residual{\tsched}{r}}{\cM,s} + r \cdot z_{s,r} \bigr)
$$
Note that $y_s \geqslant z_{s,r}$ as
$y_s=\Pr^{\maxsched}_{\cM,s}(\Diamond \goal)=
     \Pr^{\max}_{\cM,s}(\Diamond \goal)$.
Using Lemma \ref{diff-Sched-all} we obtain:
\begin{equation}
  \label{diff-max-residual}
  \Exp{\maxsched}{\cM,s} \ - \ \Exp{\residual{\tsched}{r}}{\cM,s}
  \ \ \ \geqslant \ \ \
  (\CExp{\ub}{-}\saturation) \cdot (y_s-z_{s,r})
  \tag{+}
\end{equation}
This yields:
\begin{eqnarray*}
  \theta - \zeta
  & \ \ = \ \  &
  \sum_{s,r}
    p_{s,r} \cdot 
    \bigl( \Exp{\maxsched}{\cM,s}-\Exp{\residual{\tsched}{r}}{\cM,s} \ \bigr)
  \ + \ 
  \sum_{s,r}
    p_{s,r} \cdot r \cdot (y_s -z_{s,r})
  \\[0ex]
  \\  
  & \geqslant &
  (\CExp{\ub}{-}\saturation) \cdot 
  \sum_{s,r}
    p_{s,r} \cdot (y_s-z_{s,r})  
  \ + \ 
  \sum_{s,r}
    p_{s,r} \cdot r \cdot (y_s -z_{s,r})
  \\[0ex]
  \\  
  & = &
 (\CExp{\ub}{-}\saturation) \cdot (y-z)
 \ + \ 
  \sum_{s,r}
    p_{s,r} \cdot r \cdot (y_s -z_{s,r}) 
\end{eqnarray*}
where the sum ranges over all pairs $(s,r)$ with $s\in S$ and
$r\in \{\saturation,\ldots,N\}$.
As 
$$
  \sum_{s\in S} \sum_{r=\saturation}^N
    p_{s,r} \cdot r \cdot (y_s -z_{s,r}) 
  \ \ \geqslant \ \
  \saturation \cdot \sum_{s\in S} \sum_{r=\saturation}^N
    p_{s,r} \cdot (y_s -z_{s,r})
  \ \ = \ \ 
  \saturation \cdot (y-z)
$$
we obtain:
\begin{eqnarray*}
  \theta - \zeta
  & \ \ \geqslant \ \ &
 (\CExp{\ub}{-}\saturation) \cdot (y-z)
 \ + \ 
 \saturation \cdot (y-z)
 \ \ = \ \ \CExp{\ub}\cdot (y-z)
\end{eqnarray*}
As $y_s \geqslant z_{s,r}$ we have $y \geqslant z$.
Let us first consider the case $y=z$.
By the choice of $\maxsched$, 
  $y=z$ implies $y_s=z_{s,r}$ for all $s,r$ and
  therefore
  $\Exp{\maxsched}{\cM,s} \geqslant
  \Exp{\residual{\tsched}{r}}{\cM,s}$
  (see \eqref{diff-max-residual}).
  Consequently,
  $\theta \geqslant \zeta$ and thus
  $\CExp{\tsched} \leqslant 
   \CExp{\redefresidual{\tsched}{\saturation}{\maxsched}}$.
If $y > z$ then we get:
\begin{eqnarray*}
  \frac{\theta - \zeta}{y-z}
  & \ \ \geqslant \ \ &
  \CExp{\ub} \ \ \geqslant \ \ \CExp{\max} \ \ \geqslant \ \ 
  \CExp{\tsched}
\end{eqnarray*}
But then $\CExp{\tsched} \leqslant 
   \CExp{\redefresidual{\tsched}{\saturation}{\maxsched}}$
by Lemma~\ref{lemma:rho-theta-x-y}.
\Ende
\end{proof}

\noindent
As the cost of computing $\saturation$ is dominated by
the computation of $\CExp{\ub}$, which has a pseudo-polynomial time
bound in the size of $\cM$ 
(see Appendices \ref{appendix:finitness-critical-schedulers}
 and \ref{sec:upper-bound}), 
we obtain
a pseudo-polynomial time bound for the computation of $\saturation$ as
well. 
As the logarithmic length of $\CExp{\ub}$ is polynomially bounded
in the size of $\cM$, so is the logarithmic length of $\saturation$.

\section{Threshold algorithm}
\label{appendix:threshold}

The algorithms for the threshold problem as well as the algorithm to compute
the maximal conditional expectation will rely on the following simple
observation (Lemma \ref{lemma:decision-for-s-r}).
Among others, we will use it as
a semi-local criterion on the best choice for
a given state reward pair $(s,r)$ among two options,
say schedulers $\tsched$ and $\usched$.
Let $y=\Pr^{\tsched}_s(\Diamond \goal)$ and
$\theta = \Exp{\tsched}{s}$. The pair $(z,\zeta)$ has analogous
meaning for scheduler $\usched$ where we suppose $y > z$.
As illustrated by Example \ref{example:running-example},
the best choice might depend on still unknown
decisions for other state-reward pairs.
In Lemma \ref{lemma:decision-for-s-r} these unknown decisions are represented
by the parameters $x, \rho, p$ where $x$ stands for the probability
to reach $\goal$ from $\sinit$ via path that has no prefix
$\fpath$ with $\rew(\fpath)=r$ and $\last(\fpath)=s$ and $\rho$ for
the corresponding expectation. The meaning of $p$ is the probability
to reach $s$ from $\sinit$ via a path $\fpath$ with $\rew(\fpath)=r$.
Then, Lemma \ref{lemma:decision-for-s-r} states
that $\threshold = r + (\theta-\zeta)/(y-z)$ is a threshold for the decision
which of the schedulers $\tsched$ or $\usched$ is better for $(s,r)$:
$\tsched$ is better than $\usched$ if $\CExp{\max} < \threshold$
and
$\usched$ is better than $\tsched$ if $\CExp{\max} > \threshold$.
Note that given $\tsched$ and $\usched$, we can compute $\threshold$,
while $\CExp{\max}$ might still be unknown.

\begin{lemma}
\label{lemma:decision-for-s-r}
  Let $\rho,\theta,\zeta,r,x,y,z,p$ be real numbers
  such that $p>0$, $x,y,z \geqslant 0$ and
  $x{+}y >0$ and $x{+}z>0$.
  If $y > z$ then one of the following three cases holds:
  \begin{eqnarray*}
    & &
    r \ + \
    \frac{\theta - \zeta}{y-z}
    \ \ \ >  \ \ \
    \frac{\rho + p(r y + \theta)}{x+py}
    \ \ \ > \ \ \
    \frac{\rho + p(rz + \zeta)}{x+pz}
    \\
    \\[0ex]
    \text{or} \ \ \ & &
    r \ + \
    \frac{\theta - \zeta}{y-z}
    \ \ \ <  \ \ \
    \frac{\rho + p(r y + \theta)}{x+py}
    \ \ \ < \ \ \
    \frac{\rho + p(rz + \zeta)}{x+pz}
    \\
    \\[0ex]
    \text{or} \ \ \ & &
    r \ + \
    \frac{\theta - \zeta}{y-z}
    \ \ \ =  \ \ \
    \frac{\rho + p(r y + \theta)}{x+py}
    \ \ \ = \ \ \
    \frac{\rho + p(rz + \zeta)}{x+pz}
 \end{eqnarray*}
\end{lemma}

\begin{proof}
immediate by Lemma \ref{lemma:rho-theta-x-y}.
\Ende
\end{proof}

In what follows, we often use Lemma \ref{lemma:decision-for-s-r} in the
following form. If $y> z$ then:
\begin{eqnarray*}
    \frac{\rho + p(r y + \theta)}{x+py}
    \ \  >  \ \
    \frac{\rho + p(rz + \zeta)}{x+pz}
    & \ \ \ \ \text{iff} \ \ \ \ &
    r \ + \ \frac{\theta - \zeta}{y - z}
    \ \  >  \ \
    \frac{\rho + p(r y + \theta)}{x+py}
    \\
    \\[0ex]
    & \text{iff} &
    r \ + \ \frac{\theta - \zeta}{y - z}
    \ \ >  \ \
    \frac{\rho + p(r z + \zeta)}{x+pz}
\end{eqnarray*}
and the analogous statement for $\geqslant$ rather than $>$.

\subsection{Algorithms for the threshold problem}

\label{algo:threshold-cexp}

The input of the threshold problem for maximal conditional expectations
is a positive rational number $\threshold$ (called threshold) and
an MDP $\cM$ with non-negative integer rewards and distinguished  states
$\sinit$, $\goal$ and $\fail$.
The goal is to check whether the maximal conditional expectation in
$\cM$ meets the bound specified by the threshold value
$\threshold$ either as a strict or non-strict lower or strict or non-strict
upper bound:
\begin{center}
    \begin{tabular}{l}
        does
        $\CExpState{\max}{\cM,\sinit}(\accdiaplus \goal | \Diamond \goal)
         \, \geqslant \, \threshold$ hold ?
        \\[1ex]

        does
        $\CExpState{\max}{\cM,\sinit}(\accdiaplus \goal | \Diamond \goal)
         \, > \, \threshold$ hold ?
        \\[1ex]

        does
        $\CExpState{\max}{\cM,\sinit}(\accdiaplus \goal | \Diamond \goal)
         \, \leqslant \, \threshold$ hold ?
        \\[1ex]

        does
        $\CExpState{\max}{\cM,\sinit}(\accdiaplus \goal | \Diamond \goal)
         \, < \, \threshold$ hold ?
       \end{tabular}
\end{center}
Throughout this section, we suppose that $\cM$
satisfies conditions \eqref{assumption:A1}, \eqref{assumption:A2} and has no critical schedulers.
Thus, $\CExp{\max}$ is finite
(see Proposition \ref{proposition:exprew-infinite-critical-scheduler}).
As before, we write $\CExp{\max}$ rather than
$\CExpState{\max}{\cM,\sinit}(\accdiaplus \goal | \Diamond \goal)$.

Obviously, the threshold problem for strict (resp.~non-strict)
upper thresholds is dual to the threshold problem for non-strict
(resp.~strict) thresholds. Thus, it suffices to consider lower
thresholds.

As described in Section~\ref{sec:threshold},
we provide a threshold algorithm that,
given an MDP $\cM$ and a rational threshold $\threshold$,
generates
a deterministic reward-based scheduler $\sched$ with
$\residual{\sched}{\saturation}=\maxsched$
(where $\maxsched$ and $\saturation$ are as in
Prop.~\ref{prop:saturation-maxsched})
such that
$\CExp{\sched} > \threshold$ if $\CExp{\max} > \threshold$, and
$\CExp{\sched} = \threshold$ if $\CExp{\max} = \threshold$.
If $\CExp{\max} < \threshold$ then the output of the threshold algorithm
is ``no''.
It is easy to see how this algorithm can be used in a decision
procedure for deciding whether $\CExp{\sched} > \threshold$ or
$\CExp{\sched} \geqslant \threshold$.

In a preprocessing step, we compute the saturation point
$\saturation$ (see Section \ref{appendix:compute-saturation}).
The threshold algorithm 
operates level-wise and attempts to construct a reward-based
deterministic scheduler that is memoryless from the last level
$\saturation$ and that satisfies the threshold condition
$\CExp{\sched} \geqslant \threshold$, provided
$\CExp{\max} \geqslant \threshold$. Otherwise the algorithm returns ``no''.

\tudparagraph{1ex}{{\it Initialization of the threshold algorithm.}}
The treatment of level $\saturation$ is obvious as
optimal decisions are known
by the results of Appendix~\ref{appendix:saturation}.
Let $\maxsched$ be a deterministic memoryless scheduler
that maximizes the probability to reach $\goal$ from each state
and whose conditional expectation is maximal under all those
schedulers (by Lemma \ref{lemma:Sched-max-exp}).
Let $\action(s,\saturation) \in \Act(s)$
be the action that $\maxsched$ chooses for state $s$.
Furthermore, we define $y_{s,\saturation} = p^{\max}_s$ and
$\theta_{s,\saturation} = \Exp{\maxsched}{s}$.

\tudparagraph{1ex}{{\it Level-wise computation of feasible actions.}}
For $r = \saturation{-}1,\saturation{-}2,\ldots,2,1,0$
and each state $s\in S \setminus \{\goal,\fail\}$, the algorithm
computes actions $\action(s,r) \in \Act(s)$ and rational
values $y_{s,r}$, $\theta_{s,r}$ for the probability to reach $\goal$ from
$s$ and the corresponding expectation under the residual scheduler defined
by the action table $\action(\cdot)$.
The values for the trap states are the trivial ones:
$y_{\goal,r}=1$, $y_{\fail,r}=0$ and $\theta_{\goal,r}=\theta_{\fail,r}=0$.

Suppose now that $r \in \Nat$ with $0 \leqslant r < \saturation$
and that levels $r{+}1,r{+}2,\ldots,\saturation$ have been treated before
and the triples $(\action(t,R),y_{t,R},\theta_{t,R})$ have been computed
for all states $t\in S$ and for all $R \in \{r{+}1,\ldots,\saturation\}$.
Before describing the steps that the threshold algorithm performs to
treat level $r$, we explain the demands on the actions that the
threshold algorithm assigns to the state-reward pairs $(s,r)$.

\tudparagraph{1ex}{{\it Most feasible actions at level $r$.}}
The goal of the treatment of level $r$ is
to find the most feasible way to combine zero-reward actions
with positive-reward actions as decisions of a deterministic
reward-based scheduler for paths where the accumulated reward is $r$.
Here, ``most feasible'' is understood with respect to the task to
assign actions to the state-reward pairs $(s,r)$ that
any scheduler whose conditional expectation is at least or larger
than the given threshold $\threshold$ may take.
More precisely, given a function
$\psched : S \setminus \{\goal,\fail\} \to \Act$
such that $\psched(s)\in \Act(s)$ for all states $s$,
let
$$
  T_{\psched} \ \ \ = \ \ \
  \bigl\{\goal,\fail \bigr\} \ \cup \
  \bigl\{\ s\in S\setminus \{\goal,\fail\} \ : \
           \rew(s,\psched(s))>0 \ \bigr\}
$$
For the non-trap states $s \in T_{\psched} \setminus \{\goal,\fail\}$,
the values
$y_{s,r,\psched}$ and $\theta_{s,r,\psched}$ are defined by:
\begin{eqnarray*}
     y_{s,r,\psched} & \ = \ &
         \sum_{t\in S} P(s,\psched(s),t)\cdot y_{t,R}
     \\
     \\[0ex]
     \theta_{s,r,\psched} & = &
         \rew(s,\alpha) \cdot y_{s,r,\psched} \ \ + \ \
         \sum_{t\in S} P(s,\psched(s),t)\cdot \theta_{t,R}
\end{eqnarray*}
where $R = \min \{ \saturation, r+\rew(s,\alpha) \}$.
Furthermore,
$y_{\fail,r,\psched}= \theta_{\fail,r,\psched}=\theta_{\goal,r,\psched}=0$
and $y_{\goal,r,\psched}=1$.
For the states $s\in S \setminus T_{\psched}$ we define:
\begin{eqnarray*}
  y_{s,r,\psched} & \ \ = \ \ &
  \sum_{t\in T_{\psched}} \
     \Pr^{\psched}_{s}\bigl(\ \neg T_{\psched} \Until t \ \bigr)
      \cdot y_{t,r,\psched}
  \\
  \\[0ex]
  \theta_{s,r,\psched} & \ \ = \ \ &
  \sum_{t\in T_{\psched}} \
     \Pr^{\psched}_{s}\bigl(\ \neg T_{\psched} \Until t \ \bigr)
      \cdot \theta_{t,r,\psched}
\end{eqnarray*}
The task of the threshold algorithm
is now to find a function $\psched^*$ satisfying
the following constraint.
If $\CExp{\max} \geqslant \threshold$ %
then there exists an
eventually memoryless, reward-based scheduler $\tsched$ with
$\Pr^{\tsched}_{\sinit}(\Diamond \goal) >0$ and
$\CExp{\tsched}\geqslant \threshold$
that schedules
\begin{itemize}
\item
   $\maxsched(s)$ for all state-reward pairs $(s,R)$ with
   $R \geqslant \saturation$ and
\item
   the actions $\action(t,R)$ for all state-reward pairs $(t,R)$ with
   $r < R < \saturation$
\item
   $\psched^*(s)$ for the state-reward pairs $(s,r)$.
\end{itemize}
To do so, we present a procedure that is based on
linear programming techniques and the following observation.

\begin{lemma}
   \label{lemma:rho-theta-LP}
   Let $\rho,\theta,\zeta,\threshold,p,x,y,z$ be rational numbers with
   $p>0$, $x,y,z \geqslant 0$, $x{+}y, x{+}z>0$ and $r\in \Nat$
   such that
   \begin{equation}
     \label{equation:*-theta-zeta}
     \theta - (\threshold - r) \cdot y \ \ \geqslant \ \
     \zeta - (\threshold - r) \cdot z
     \tag{*}
   \end{equation}
   Then:
   \begin{eqnarray*}
        \text{\rm (a)} \ \ \ \ \ \
        \frac{\rho + p(rz + \zeta)}{x+pz} \ \geqslant \ \threshold
        & \ \ \text{implies} \ \ &
        \frac{\rho + p(ry + \theta)}{x+py} \ \geqslant \ \threshold
        \\
        \\[0ex]
        \text{\rm (b)} \ \ \ \ \ \
        \frac{\rho + p(rz + \zeta)}{x+pz} \ > \ \threshold
        & \ \ \text{implies} \ \ &
        \frac{\rho + p(ry + \theta)}{x+py} \ > \ \threshold
   \end{eqnarray*}
\end{lemma}

\noindent
Before presenting the proof of Lemma \ref{lemma:rho-theta-LP}, let us first
state its role for the treatment of level $r$ in the threshold algorithm.
Intuitively, the value $x$ stands for the
probability to reach $\goal$ along some path that has no prefix $\fpath$ with
$\rew(\fpath)=r$ and $\rho$ for the corresponding expectation, while
$p$ denotes the probability to generate a finite path $\fpath$
from $\sinit$ to some state $s$ with $\rew(\fpath)=r$.
The pairs $(\theta,y)$ and $(\zeta,z)$ stand for the expected reward and
probability to reach $\goal$ from $s$ under some scheduler.
When treating the states at level $r$, the values $\rho,x,p$ are unknown,
while the pairs $(\theta_{s,r,\psched},y_{s,r,\psched})$ are candidates
for $(\theta,y)$ and $(\zeta,z)$.
Thus, Lemma \ref{lemma:rho-theta-LP} asserts that the most promising candidates
for $\psched$ are the ones where
$\theta_{s,r,\psched} - (\threshold - r) \cdot y_{s,r,\psched}$ is maximal
for all states $s$.
Here, ``most promising'' means that
if $\CExp{\sched}\geqslant \threshold$
for some scheduler $\sched$
that extends the already made decisions for levels
$r{+}1,\ldots,\saturation$ then
$\CExp{\tsched}\geqslant \threshold$
for some scheduler $\sched$
that extends the decisions for levels $r{+}1,\ldots,\saturation$
and behaves as $\psched$ for level $r$.

\begin{proof}
We first consider statement (a) and
suppose  $(\rho + p(rz + \zeta))/(x+pz) \geqslant \threshold$.
The task is to show that \eqref{equation:*-theta-zeta}  implies
$(\rho + p(ry + \theta))/(x+py) \ \geqslant \ \threshold$.
The claim is clear if $y=z$, in which case \eqref{equation:*-theta-zeta}
implies $\theta \geqslant \zeta$.
Suppose now $y > z$. Then, \eqref{equation:*-theta-zeta}  implies
$$
  \frac{\theta - \zeta}{y-z} \ \ \ \geqslant \ \ \ \threshold - r
$$
and therefore
$$
  r \ + \ \frac{\theta - \zeta}{y-z} \ \ \ \geqslant \ \ \ \threshold
$$
Suppose by contradiction that
$$
  \frac{\rho + p(ry + \theta)}{x+py} \ \ \ < \ \ \ \threshold
$$
Then:
$$
  \frac{\rho + p(ry + \theta)}{x+py} \ \ \ < \ \ \
  r \ + \ \frac{\theta - \zeta}{y-z}
$$
We now apply Lemma \ref{lemma:decision-for-s-r} and obtain:
$$
  \frac{\rho + p(ry+\theta)}{x+py}
  \ \ \ > \ \ \
  \frac{\rho + p(rz+\zeta)}{x+pz}
  \ \ \ \geqslant \ \ \ \threshold
$$
Contradiction. Hence,
$(\rho + p(ry + \theta))/(x+py) \ \geqslant \ \threshold$ if $y \geqslant z$.
The remaining case $y < z$ can be handled by
analogous arguments.
This completes the proof of Lemma \ref{lemma:rho-theta-LP}.
\Ende
\end{proof}

The idea for the treatment of each level $r<\saturation$
is now to compute the values
$\max_{\psched}
 (\theta_{s,r,\psched} - (\threshold {-} r) \cdot y_{s,r,\psched})$
for all states
by solving
the linear program shown in Figure \ref{LP-threshold}.
The latter has one variable $x_s$ for each state
$s\in S$ and one linear constraint for each state-action pair $(s,\alpha)$
with $\alpha \in \Act(s)$.

\begin{figure}[t]
\begin{tabular}{l}
    Minimize $\sum\limits_{s\in S} x_s$ subject to:
    \\[2ex]
    \qquad
    \begin{tabular}{ll}
      (1) &
      If $s \in S \setminus \{\goal,\fail\}$ then
      for each action $\alpha \in \Act(s)$ with $\rew(s,\alpha)=0$:
      \\[2ex]
      &
      \qquad \qquad
      $x_s \ \ \geqslant \ \ \sum\limits_{t\in S} P(s,\alpha,t) \cdot x_t$
      \\[2.5ex]
      (2) &
      If $s \in S \setminus \{\goal,\fail\}$ then
      for each action $\alpha \in \Act(s)$ with $\rew(s,\alpha)>0$:
      \\[2ex]
      &
      \qquad \qquad
      $x_s \ \ \geqslant \ \
       \sum\limits_{t\in S} P(s,\alpha,t) \cdot
            \bigl( \, \theta_{t,R} + \rew(s,\alpha) \cdot y_{t,R} \, - \,
                (\threshold {-} r) \cdot y_{t,R} \, \bigr)$
      \\[2ex]
      &
      where $R=\min\{\saturation,r{+}\rew(s,\alpha)\}$
      \\[2.5ex]
      (3) &
      For the trap states: \ \
      $x_{\goal} = r-\threshold$ \ \ and  \ \
      $x_{\fail}=0$
  \end{tabular}
\end{tabular}
\caption{Linear program  for the treatment of level $r$
         in the threshold algorithm
    \label{LP-threshold}}
\end{figure}

Lemma \ref{lemma:soundness-LP-threshold} (see below) will
show the existence of a unique solution
of the linear program in Figure \ref{LP-threshold}.
Let $(x_s^*)_{s\in S}$ be the solution
of the linear program in Figure \ref{LP-threshold}.
Let $\Act^*(s)$ denote the set of actions $\alpha \in \Act(s)$ such that
the following constraints (E1) and (E2) hold:
\begin{eqnarray*}
     \text{\rm (E1)} & \ \  &
     \text{If $\rew(s,\alpha)=0$ then:}
     \  \ x_s^* \  =  \ \sum_{t\in S} P(s,\alpha,t)\cdot x_t^*
     \\[2ex]
     \text{\rm (E2)} &  &
     \text{If $\rew(s,\alpha)>0$ and
           $R\, =\,
       \min \, \bigl\{\, \saturation,\, r{+}\rew(s,\alpha)\, \bigr\}$ then:}
     \\[2ex]
     & &
     \hspace*{1cm}
     \ x_s^* \  =  \ \sum_{t\in S} P(s,\alpha,t)\cdot
         \bigl( \theta_{t,R} + \rew(s,\alpha)\cdot y_{t,R}
                  \, - \, (\threshold {-} r)\cdot y_{t,R} \bigr)
\end{eqnarray*}
Let $\cM^*=\cM^*_{r,\threshold}$ denote the MDP  with state space $S$
induced by the state-action pairs $(s,\alpha)$ with
$\alpha \in \Act^*(s)$ where the positive-reward actions are redirected to
the trap states. More precisely, if $s,t\in S$ and $\alpha \in \Act^*(s)$ 
and $\rew(s,\alpha)=0$ then
$P_{\cM^*}(s,\alpha,t)=P(s,\alpha,t)$.
For $\alpha \in \Act^*(s)$ and $\rew(s,\alpha)>0$:
$$
  P_{\cM^*}(s,\alpha,\goal) \ \ = \ \
  \sum_{t\in S} P(s,\alpha,t) \cdot y_{t,R},
  \qquad
  P_{\cM^*}(s,\alpha,\fail) \ \ = \ \ 1 - P_{\cM^*}(s,\alpha,\goal)
$$
where $R=\min \{\saturation,r+\rew(s,\alpha)\}$.
The reward structure of $\cM^*$ is irrelevant for our purposes.

\tudparagraph{1ex}{{\it Treatment of level $r$ in the threshold algorithm.}}
The threshold algorithm solves the linear program
of Figure \ref{LP-threshold}%
\footnote{%
  Note that the values
  $\theta_{t,R}$ and $y_{t,R}$ for $t\in S$ and
  $r < R \leqslant \saturation$
  used in the constraints (2) of Figure \ref{LP-threshold}
  have been computed before in the treatment of level $R$.}
and then computes the action sets $\Act^*(s)$
and a deterministic memoryless scheduler
$\psched^* : S \setminus \{\goal,\fail\} \to \Act$ for the MDP $\cM^*$ with
\begin{center}
   $\psched^*(s)\in \Act^*(s)$ \ \ and \ \
   $\Pr^{\psched^*}_{\cM^*,s}(\Diamond \goal) \ = \
    \Pr^{\max}_{\cM^*,s}(\Diamond \goal)$
\end{center}
for all $s\in S \setminus \{\goal,\fail\}$. It then defines:
\begin{center}
   $\action(s,r)=\psched^*(s)$, \ \
   $y_{s,r} = y_{s,r,\psched^*}$ \ \ and \ \
   $\theta_{s,r}=\theta_{s,r,\psched^*}$.
\end{center}
This completes the treatment of level $r$.

\tudparagraph{1ex}{{\it Output of the threshold algorithm.}}
Having reached the last level $r=0$, the output of the algorithm
is as follows.
The generated reward-based scheduler $\sched$, given by
$\sched(s,r)=\action(s,r)$ for $r < \saturation$ and
$\sched(s,r) = \maxsched(s)$ for $r \geqslant \saturation$,
satisfies the equations
$y_{\sinit,0}=\Pr^{\sched}_{\sinit}(\Diamond \goal)$ and
$\theta_{\sinit,0}=\Exp{\sched}{\sinit}$.
Thus, if $y_{\sinit,0}>0$ then
$\CExp{\sched} = \theta_{\sinit,0}/y_{\sinit,0}$.
Hence, the 
threshold algorithm returns $\sched$  
if $y_{\sinit,0}>0$ 
and $\theta_{\sinit,0}/y_{\sinit,0}\geqslant \threshold$.
The correctness is obvious, as $\CExp{\sched} \geqslant \threshold$.
Otherwise, i.e., if $y_{\sinit,0}=0$ or
$\theta_{\sinit,0}/y_{\sinit,0} < \threshold$,
the algorithm terminates with the answer ``no''.
The correctness of the answer ``no'' is a consequence of
Lemma \ref{lemma:soundness-threshold-algorithm} (see below).

\begin{lemma}[Soundness of the LP of Figure \ref{LP-threshold}]
   \label{lemma:soundness-LP-threshold}
   The linear program in Figure \ref{LP-threshold} has a unique solution
   $(x_s^*)_{s\in S}$. The action-sets $\Act^*(s)$ are non-empty for
   all $s\in S \setminus \{\goal,\fail\}$ and
   for each function
   $\psched : S \setminus \{\goal,\fail\} \to \Act$
   with $\psched(s)\in \Act^*(s)$ we have:
   $$
     x_s^* \ \ \ = \ \ \
     \theta_{s,r,\psched} \ - \ (\threshold {-}r)\cdot y_{s,r,\psched}
   $$
   Moreover, whenever
   $\psched : S \setminus \{\goal,\fail\} \to \Act$
   is a function with $\psched(s)\in \Act(s)$ then:
   $$
       x_s^* \ \ \ \geqslant \ \ \
        \theta_{s,r,\psched} \ - \ (\threshold {-}r)\cdot y_{s,r,\psched}
   $$
\end{lemma}

\begin{proof}
 The solvability and the uniqueness of the solution of the linear program
 in Figure \ref{LP-threshold} follows by the fact that the linear program
 agrees with the one that is known to represent the expected total reward
 in the MDP $\cN$ with state space $S_{\cN} = S \cup \{\final\}$
 and action set $\Act_{\cN} = \Act \cup \{\tau\}$ such that for all states
 $s\in S$ and
 all actions $\alpha \in \Act(s)$:
 \begin{center}
  \begin{tabular}{l}
    $P_{\cN}(s,\alpha,t) = P(s,\alpha,t)$ and $\rew_{\cN}(s,\alpha)=0$
    if $\rew(s,\alpha)=0$
    \\[1ex]

    $P_{\cN}(s,\alpha,\final)=1$ and
    $\rew_{\cN}(s,\alpha) =
     \theta_{s,r,\alpha} - (\threshold {-}r) \cdot y_{s,r,\alpha}$
    if $\rew(s,\alpha)>0$
    \\[1ex]

    $P_{\cN}(\goal,\tau,\final)=1$ and $\rew_{\cN}(\goal,\tau)=r-\threshold$
    \\[1ex]

    $P_{\cN}(\fail,\tau,\final)=1$ and $\rew_{\cN}(\fail,\tau)=0$
  \end{tabular}
 \end{center}
 and $P_{\cN}(\cdot)=\rew_{\cN}(\cdot)=0$ in all remaining cases.
 In particular, state $\final$ is a trap in $\cN$ and there are no
 other traps in $\cN$.
 Assumption \eqref{assumption:A2} yields
 $\Pr^{\min}_{\cN,s}(\Diamond \final)=1$ for all states $s \in S_{\cN}$.
 Using standard results for finite-state MDPs
 (see e.g.~Theorem 4.20 in \cite{Kallenberg}),
 the values $\Exp{\max}{\cN,s}(\accdiaplus \final)$ are finite
 and are computable as the unique solution
 of the linear program shown in Figure \ref{LP-threshold}.
 That is, if $(x_s^*)_{s\in S}$ is the unique solution of
 the linear program in Figure \ref{LP-threshold} then
 \begin{center}
     $x_s^* \ \ = \ \ \Exp{\max}{\cN,s}(\accdiaplus \final)$ \ \
     for all states $s\in S$.
 \end{center}
 Let now $\psched$ be a deterministic memoryless scheduler
 for $\cN$ that maximizes the expected total reward in $\cN$, i.e.,
 $x_s^* = \Exp{\psched}{\cN,s}(\accdiaplus \final)$ for all states $s\in S$.
 Clearly, $\psched$ can be viewed as a function
 $S \setminus \{\goal,\fail\} \to \Act$ with $\psched(s)\in \Act^*(s)$
 for all $s\in S$.
 This yields that the sets $\Act^*(s)$ are non-empty for the
 non-trap states $s$ for $\cM$.
 Vice versa, each function
 $\psched : S \setminus \{\goal,\fail\} \to \Act$ with
 $\psched(s)\in \Act^*(s)$ for all $s \in S \setminus \{\goal,\fail\}$
 can be viewed as a deterministic memoryless scheduler
 for $\cN$ that maximizes the expected total reward in $\cN$.

 Suppose now that $\psched : S \setminus \{\goal,\fail\} \to \Act$
 is a function with $\psched(s) \in \Act(s)$ for all
 $s\in S \setminus \{\goal,\fail\}$. Let $T$ denote the set of states
 $s \in S$ such that either $s\in \{\goal,\fail\}$ or
 $\rew(s,\psched(s))>0$. Clearly, $\psched$ can be viewed as a scheduler
 for $\cN$ that schedules the unique action $\tau$ for the trap-states
 $\goal$ and $\fail$ of $\cM$, and we have:
 $$
    \Pr^{\psched}_{\cN,s}(\Diamond T) \ = \ 1
 $$
 With $\theta_{\fail,r,\psched}=\theta_{\fail,r}=0$,
 $\theta_{\goal,r,\psched}=\theta_{\goal,r}=0$,
 $y_{\fail,r,\psched}=y_{\fail,r}=0$ and
 $y_{\goal,r,\psched}=y_{\goal,r}=1$ we obtain:
 $$
  \rew_{\cN}(t,\psched(t)) \ \ \ = \ \ \
  \theta_{t,r,\psched} \ - \ (\threshold {-}r)\cdot y_{t,r,\psched}
 $$
 for all states $t\in T$.
 This yields that for all states $s\in S$:
 \begin{eqnarray*}
     \Exp{\psched}{\cN,s}(\accdiaplus \final)
     & \ = \ &
     \sum_{t\in T}
         \Pr^{\psched}_{\cN,s}(\, \neg T \Until t \, ) \cdot
         \rew_{\cN}(t,\psched(t))
     \\
     \\[0ex]
     & \ = \ &
     \sum_{t\in T}
         \Pr^{\psched}_{\cN,s}(\, \neg T \Until t \, ) \cdot
         \theta_{t,r,\psched}
     \ \ - \ \
     (\threshold - r) \cdot
     \sum_{t\in T}
         \Pr^{\psched}_{\cN,s}(\, \neg T \Until t \, ) \cdot
         y_{t,r,\psched}
     \\
     \\[0ex]
     & \ = \ &
     \theta_{s,r,\psched} \ - \ (\threshold {-}r)\cdot y_{s,r,\psched}
 \end{eqnarray*}
 As $x_s^* \ = \ \Exp{\max}{\cN,s}(\accdiaplus \final)
     \ \geqslant \
     \Exp{\psched}{\cN,s}(\accdiaplus \final)$
 we obtain:
 $$
   x_s^*
   \ \ \ \geqslant \ \ \
   \theta_{s,r,\psched} \ - \ (\threshold {-}r)\cdot y_{s,r,\psched}
 $$
 Moreover, if $\psched(s)\in \Act^*(s)$ for all states $s$ then
 $x_s^* \ = \ \Exp{\max}{\cN,s}(\accdiaplus \final)
     \ = \
     \Exp{\psched}{\cN,s}(\accdiaplus \final)$ (see above)
 and therefore:
 $$
     x_s^* \ \ \ = \ \ \
     \Exp{\max}{\cN,s}(\accdiaplus \final) \ \ \ = \ \ \
     \Exp{\psched}{\cN,s}(\accdiaplus \final) \ \ \ = \ \ \
     \theta_{s,r,\psched} \ - \ (\threshold {-}r)\cdot y_{s,r,\psched}
 $$
 This completes the proof of Lemma \ref{lemma:soundness-LP-threshold}.
\Ende
\end{proof}

It remains to show that if the algorithm returns ``no'' then there is no
scheduler meeting the bound for its conditional expectation.
We prove this by showing that if $\CExp{\max} \geqslant \threshold$
then after the treatment of each level $r$ there exists a reward-based
scheduler $\tsched_r$ using the decisions that have been stored in the
action-table $\action(\cdot)$ for all level $> r$ and the decisions
of $\psched^*$ at level $r$ and satisfying
$\CExp{\tsched_r} \geqslant \threshold$.

\begin{lemma}[Soundness of the answer ``no'']
  \label{lemma:soundness-threshold-algorithm}
   If $\CExp{\max} \geqslant \threshold$ then
   the threshold algorithm generates a scheduler $\sched$ with
   $\CExp{\sched} \geqslant \threshold$.
\end{lemma}

\begin{proof}
 The task is to show that the algorithm indeed returns 
 a scheduler $\sched$ with $\CExp{\sched} \geqslant \threshold$
 if $\CExp{\max} \geqslant \threshold$.
 For this, we use an inductive argument to prove the following statement.
 If $\CExp{\max} \geqslant \threshold$ then
 for each $r\in \{\saturation,\saturation{-}1,\ldots,1,0\}$,
 there exists
 a reward-based scheduler $\tsched_r$ for $\cM$ with
 $\Pr^{\tsched_r}_{\cM,\sinit}(\Diamond \goal) >0$ and
 $\CExp{\tsched_r} \geqslant \threshold$ such that:
 \begin{itemize}
   \item
     $\tsched_r(s,R) = \maxsched(s)$ for all state-reward pairs $(s,R)$ with
     $R \geqslant \saturation$ and
   \item
     $\tsched_r(s,R) = \action(t,R)$ for all state-reward pairs $(t,R)$ with
     $r < R < \saturation$
   \item
     $\tsched_r(s,r)=\psched^*(s)$ for the state-reward pairs $(s,r)$
 \end{itemize}
 where $\psched^*$ is the function as explained in the treatment of
 level $r$. For $r=0$ we obtain $\tsched_0=\sched$ and therefore
 $\CExp{\sched}\geqslant \threshold$.

 The claim is obvious for $r=\saturation$ as then we can deal with
 $\tsched_{\saturation} = \tsched$ where $\tsched$ is any scheduler
 with $\CExp{\tsched}=\CExp{\max}$.
(This follows from the fact
    that $\saturation > \turning$
    for the turning point $\turning$
    of Proposition \ref{prop:turning-point}.)
 Suppose now that $r < \saturation$. By induction hypothesis
 there exists a reward-based scheduler $\tsched_{r+1}= \usched$
 such that
   \begin{itemize}
   \item
     $\usched(s,R) = \maxsched(s)$ for all state-reward pairs $(s,R)$ with
     $R \geqslant \saturation$ and
   \item
     $\usched(s,R) = \action(t,R)$ for all state-reward pairs $(t,R)$ with
     $r < R < \saturation$
   \end{itemize}
   and $\Pr^{\usched}_{\cM,\sinit}(\Diamond \goal) >0$ and
   $\CExp{\usched} \geqslant \threshold$.
   Let
   $$ p_s \ \ = \ \ \Pr^{\usched}_{\cM,\sinit}(\Diamond^{=r} s),
      \qquad
      p \ = \ \sum_{s \in S} p_s
   $$
   Furthermore, we define
   $\psched : S \setminus \{\goal,\fail\}\to \Act$ by
   $\psched(s)=\usched(s,r)$.
   Then:
   $$
     \zeta_s \ \ \eqdef \ \
     \theta_{s,r,\psched} \ \ = \ \
     \Exp{\residual{\usched}{r}}{\cM,s},
     \qquad
     z_s \ \ \eqdef \ \ y_{s,r,\psched} \ \ = \ \
     \Pr^{\residual{\usched}{r}}_{\cM,s}(\Diamond \goal)
   $$
   Likewise, we define
   $$
     \theta_s \ \ = \ \ \theta_{s,r,\psched^*}
     \qquad \text{and} \qquad
     y_s \ \ = \ \ \theta_{s,r,\psched^*}
   $$
 Let $\tsched = \tsched_r$ denote the scheduler for $\cM$ that agrees with
 $\usched$ except that $\tsched(s,r)=\psched^*(s)$ for all states $s$ of $\cM$.
 We then have
 \begin{center}
   $\theta_s \ = \ \Exp{\residual{\tsched}{r}}{\cM,s}(\accdiaplus \goal)$
   \qquad and \qquad
   $y_s \ = \ \Pr^{\residual{\psched^*}{r}}_{\cM,s}(\Diamond \goal)$.
 \end{center}
 Let
 \begin{eqnarray*}
    \theta \ \ = \ \ \sum\limits_{s\in S} \frac{p_s}{p} \cdot \theta_s, \ \ \ \
    & &
    y \ \ = \ \ \sum\limits_{s\in S} \frac{p_s}{p} \cdot y_s
    \\
    \\[0ex]
    \zeta \ \ = \ \ \sum\limits_{s\in S} \frac{p_s}{p} \cdot \zeta_s, \ \ \ \
    & &
    z \ \ = \ \ \sum\limits_{s\in S} \frac{p_s}{p} \cdot z_s
 \end{eqnarray*}
 Then, $py$ is the probability of the infinite
 $\tsched$-paths from $\sinit$ that have a prefix $\fpath$
 with $\rew(\fpath)=r$ and $p\theta$ the corresponding expectation.
 The values $pz$ and $p\zeta$ have analogous meaning for scheduler
 $\usched$.

 We now use Lemma \ref{lemma:soundness-LP-threshold}.
 As
 $x_s^* =
  \theta_{s,r,\psched^*} \ - \ (\threshold {-}r)\cdot y_{s,r,\psched^*}$
 and
 $x_s^* \ \geqslant \
  \theta_{s,r,\psched} \ - \ (\threshold {-}r)\cdot y_{s,r,\psched}$,
 we have:
 \begin{eqnarray*}
    \theta_{s} - (\threshold {-}r) \cdot y_{s}
    & \ \ = \ \ &
    \theta_{s,r,\psched^*} - (\threshold {-}r) \cdot y_{s,r,\psched^*}
    \\[2ex]
    & \ \geqslant \ &
    \theta_{s,r,\psched} \ - \ (\threshold {-}r)\cdot y_{s,r,\psched}
    \ \ \ = \ \ \
    \zeta_s - (\threshold {-}r)\cdot z_s
 \end{eqnarray*}
 for all states $s \in S$. But this yields:
   $$
     \theta - (\threshold {-} r) \cdot y
     \ \ \ \geqslant \ \ \
     \zeta - (\threshold {-} r) \cdot z
   $$
   There exists non-negative rational numbers $\rho$ and $x$
   (namely, $x$ is the probability of the maximal $\usched$-paths
    from $\sinit$ to $\goal$ that do not have a prefix $\fpath$ with
    $\rew(\fpath)=r$ and $\rho$ is the corresponding expectation)
   such that
   $$
     \CExp{\usched} \ \ = \ \ \frac{\rho + p(rz + \zeta)}{x+pz}
     \qquad \text{and} \qquad
     \CExp{\tsched} \ \ = \ \ \frac{\rho + p(ry + \theta)}{x+py}
   $$
   Let us check that $x+py$ is indeed positive. This is clear if $x>0$.
   Suppose now that $x=0$. The goal is to show that $y>0$.
   As $x=0$ we have $\rho=0$ and therefore:
   $$
     \threshold \ \ \leqslant \ \ \CExp{\usched} \ \ = \ \
      \frac{p(rz + \zeta)}{pz} \ \ = \ \ r + \frac{\zeta}{z}
   $$
   Hence, $(\threshold - r)\cdot z \ \leqslant \ \zeta$.
   Suppose by contradiction that $y=0$. Then, $\theta=0$.
   Hence,
    $$
      0 \ \ = \ \ \theta - (\threshold {-}r)y \ \
      \geqslant \ \ \zeta - (\threshold {-}r)z \ \ \geqslant  \ \ 0
    $$
    This yields
    $0 \ = \ \theta - (\threshold {-}r)y \ = \ \zeta - (\threshold {-}r)z$.
    As
    $\theta_s - (\threshold {-}r)y_s \ \geqslant \
     \zeta_s - (\threshold {-}r)z_s$ for all states $s$, we have
    $\theta_s - (\threshold {-}r)y_s \ = \
     \zeta_s - (\threshold {-}r)z_s$.
    This implies that $\psched$ viewed as a scheduler for $\cN$
    maximizes the expected total reward from every state $s$.
    In particular, $\psched(s)\in \Act^*(s)$ for all states $s$.
    Thus, $\psched$ can also be viewed as a scheduler for the MDP
    $\cM^*$.
    As $x+pz >0$ and $x=0$ (by assumption) we have $z >0$.
    Thus, $z_s > 0$ for at least one state $s$ with $p_s >0$.
    By the choice of $\psched^*$,
    $$
      y_s \ \ = \ \ \Pr^{\psched^*}_{\cM^*,s}(\Diamond \goal)
         \ \ = \ \ \Pr^{\max}_{\cM^*,s}(\Diamond \goal)
         \ \ \geqslant \ \ \Pr^{\psched}_{\cM^*,s}(\Diamond \goal)
         \ \ = \ \ z_s \ \ > \ \ 0
    $$
    This yields $y>0$, and therefore $x+py>0$.

    We are now in the position to apply Lemma \ref{lemma:rho-theta-LP}.
    Recall that we have
    $$
     \threshold \ \ \leqslant \ \
     \CExp{\usched} \ \ = \ \ \frac{\rho + p(rz + \zeta)}{x+pz}
     \qquad \text{and} \qquad
     \CExp{\tsched} \ \ = \ \ \frac{\rho + p(ry + \theta)}{x+py}
   $$
   and
   $$
     \theta - (\threshold {-} r) \cdot y \ \ \
     \geqslant \ \
     \zeta - (\threshold {-} r) \cdot z
   $$
   Part (a) of Lemma \ref{lemma:rho-theta-LP} yields
   $\CExp{\tsched} \geqslant \threshold$.
\Ende
\end{proof}

\begin{corollary}[Optimality of the generated scheduler]
  \label{corollary:threshold-algorithm-exact}
  If the threshold algorithm returns a scheduler $\sched$ with
  $\CExp{\sched}=\threshold$ then $\CExp{\max}=\threshold$ and
  $\sched$ is a reward-based scheduler that maximizes the conditional
  expectation.
\end{corollary}

The above corollary as well as the following observations about the
scheduler $\sched$ generated by the threshold algorithm will be crucial
for the computation of the maximal conditional expectation
$\CExp{\max}$.

Lemma \ref{lemma:soundness-LP-threshold} yields
for the function $\psched^*$ used to define the decisions of the
generated scheduler at level $r$:
$$
     \theta_{s,r,\psched^*} \ - \ (\threshold {-}r)\cdot y_{s,r,\psched^*}
     \ \ \ = \ \ \
     \max_{\psched} \
     \bigl( \
        \theta_{s,r,\psched} \ - \ (\threshold {-}r)\cdot y_{s,r,\psched}
     \bigr)
$$
where $\psched$ ranges over all functions
$\psched : S \setminus \{\goal,\fail\} \to \Act$ with
$\psched(s)\in \Act(s)$ for all states $s\in S$.
This implies the first part of the following lemma:

\begin{lemma}[Difference-property of the generated scheduler]
\label{lemma:property-sched-threshold-algorithm}
  Notations as before.
  The scheduler $\sched$ generated by the threshold algorithm for
  the threshold $\threshold$
  enjoys the following property. For each $r\in \{0,1,\ldots,\saturation\}$,
  each state $s\in S \setminus \{\goal,\fail\}$ and each
  function $\psched : S \setminus \{\goal,\fail\} \to \Act$ with
  $\psched(t)\in \Act(t)$ for all $t$ we have:
  $$
    \theta_{s,r} - (\threshold {-} r)\cdot y_{s,r}
    \ \ \ \geqslant \ \ \
    \theta_{s,r,\psched} - (\threshold {-} r)\cdot y_{s,r,\psched}
  $$
  Moreover, if $r\in \{0,1,\ldots,\saturation\}$ and
  $$
   \threshold_r \ \ = \ \
   \min \
   \Bigl\{ \
       r + \frac{\theta_{s,r}-\theta_{s,r,\alpha}}{y_{s,r}-y_{s,r,\alpha}}
       \ : \
       s \in S \setminus \{\goal,\fail\}, \ y_{s,r} > y_{s,r,\alpha}
       \
   \Bigr\}
  $$
  (where we put $\min \varnothing = +\infty$)
  then $\threshold_r \geqslant \threshold$ and
  for each value $\threshold^*$ with $\threshold^* \geqslant \threshold$
  we have: $\threshold^*$ satisfies the following condition
  \eqref{threshold-difference-property}
  for all states $s$ and functions $\psched$ if and only if
  $\threshold^* \leqslant \threshold_r$.
  \begin{equation}
     \label{threshold-difference-property}
    \theta_{s,r} - (\threshold^* {-} r)\cdot y_{s,r}
    \ \ \ \geqslant \ \ \
    \theta_{s,r,\psched} - (\threshold^* {-} r)\cdot y_{s,r,\psched}
    \tag{*}
 \end{equation}
  Furthermore,
  if $\threshold < \threshold^* \leqslant
      \min \{\threshold_R : r \leqslant R \leqslant \saturation\}$
  and $\sched^*$ is the scheduler that has been generated by the
  threshold algorithm for the lower
  bound $\threshold^*$
  then $(y_{s,R},\theta_{s,R}) = (y_{s,R}^*,\theta_{s,R}^*)$
  for all states $s\in S$ and $R\in \{r,\ldots,\saturation\}$
  where
  $$
     y_{s,R}^* \ \ = \ \ \Pr^{\residual{\sched^*}{r}}_s(\Diamond \goal),
     \quad
     \theta_{s,R}^* \ \ = \ \ \Exp{\residual{\sched^*}{r}}{s}
  $$
  These properties hold, no matter whether $\CExp{\sched} < \threshold$
  or $\CExp{\sched} = \threshold$ or $\CExp{\sched} > \threshold$.
\end{lemma}

\begin{proof}
The first statement is obvious. To prove the second statement, we pick
some $r\in \{0,1,\ldots,\saturation\}$. Obviously, we have
$\threshold_r \geqslant \threshold$.

Recall that $y_{s,r,\alpha} = \sum_{t\in S} P(s,\alpha,t) \cdot y_{t,R}$
and $\theta_{s,r,\alpha}= \sum_{t\in S} P(s,\alpha,t) \cdot \theta_{t,R}$
where $R=\min\{\saturation,r{+}\rew(s,\alpha)\}$.

Let now $\threshold^*$ be any value
with $\threshold \leqslant \threshold^* \leqslant \threshold_r$.
We prove that \eqref{threshold-difference-property} holds.
As $\threshold \leqslant \threshold^*$ we obtain:
\begin{equation*}
  \label{theta-difference-max}
  \theta_{s,r} - (\threshold^*{-}r) \cdot y_{s,r}
  \ \ \ \geqslant \ \ \
  \theta_{s,r,\alpha} - (\threshold^*{-}r) \cdot y_{s,r,\alpha}
\end{equation*}
for all states $s \in S \setminus \{\goal,\fail\}$ and
actions $\alpha \in \Act(s)$.

As in the proof of Lemma \ref{lemma:soundness-LP-threshold}, we consider
the MDP $\cN$ with state space $S_{\cN}=S \cup \{\final\}$
and action set $\Act_{\cN} = \Act \cup \{\tau\}$ such that for all states
$s\in S$ and all actions $\alpha \in \Act(s)$:
\begin{center}
 \begin{tabular}{l}
    $P_{\cN}(s,\alpha,t) = P(s,\alpha,t)$ and $\rew_{\cN}(s,\alpha)=0$ \ \
    if $\rew(s,\alpha)=0$
    \\[1ex]

    $P_{\cN}(s,\alpha,\final)=1$ and
    $\rew_{\cN}(s,\alpha) =
     \theta_{s,\alpha} - (\threshold^* {-}r) \cdot y_{s,\alpha}$ \ \
    if $\rew(s,\alpha)>0$
    \\[1ex]

    $P_{\cN}(\goal,\tau,\final)= 1$ and
    $\rew_{\cN}(\goal,\tau)=r-\threshold^*$
    \\[1ex]

    $P_{\cN}(\fail,\tau,\final)=1$ and $\rew_{\cN}(\fail,\tau)=0$
 \end{tabular}
\end{center}
and $P_{\cN}(\cdot)=\rew_{\cN}(\cdot)=0$ in all remaining cases.
Note that state $\final$ is a trap in $\cN$ and there are no
other traps in $\cN$. Assumption \eqref{assumption:A2} yields
$\Pr^{\min}_{\cN,s}(\Diamond \final)=1$ for all states $s \in S_{\cN}$.
For $s\in S$ let
$$
  x_s^{\sched} \ \ = \ \ \theta_{s,r} - (\threshold^*-r) \cdot y_{s,r}
$$
and let $x_{\final}^{\sched}=0$.
For $\alpha \in \Act(s)$ with $\rew(s,\alpha)>0$
and $R=\min \{\saturation,r{+}\rew(s,\alpha)\}$ we have:
\begin{eqnarray*}
  x_s^{\sched} & \ \geqslant \ &
  \theta_{s,r,\alpha} - (\threshold^*-r) \cdot y_{s,r,\alpha}
  \\
  \\[0ex]
  & \ = \ &
  \rew_{\cN}(s,\alpha) \ + \
  P_{\cN}(s,\alpha,\final) \cdot x_{\final}^{\sched}
  \\
  \\[0ex]
  & = &
  \rew_{\cN}(s,\alpha) \ + \
  \sum_{t\in S_{\cN}} P(s,\alpha,t) \cdot x_t^{\sched}
\end{eqnarray*}
Thus, the vector $(x_s^{\sched})_{s\in S_{\cN}}$ provides a solution of the
following constraints:
$$
  x_s \ \ \geqslant \ \
  \rew_{\cN}(s,\alpha) + \sum_{s\in S_{\cN}} P_{\cN}(s,\alpha,t) \cdot x_t
  \qquad
  \text{for $s \in S$ and $\alpha \in \Act_{\cN}(s)$}
$$
and $x_{\final}=0$.
It is well known that the vector $(x_s^*)_{s\in S_{\cN}}$
where $x_s^*=\Exp{\sched}{\cN,s}(\accdiaplus \final)$
provides the unique solution of the above linear constraints
that minimizes $\sum_{s\in S_{\cN}} x_s$.
Hence:
$$
  x_s^{\sched} \ \ \geqslant \ \ \Exp{\sched}{\cN,s}(\accdiaplus \final)
$$
for all states $s$.
On the other hand, for $\sched(\cdot,r)$ viewed as a scheduler for
$\cN$ we have
$x_s^{\sched} = \Exp{\sched(\cdot,r)}{\cN,s}(\accdiaplus \final)$.
Thus,
$$
  x_s^{\sched} \ \ = \ \ \Exp{\sched}{\cN,s}(\accdiaplus \final)
$$
for all states $s$.
But then for each function $\psched : S \setminus \{\goal,\fail\} \to \Act$
with $\psched(t)\in \Act(t)$ for all $t$ (viewed as a scheduler for $\cN$)
we have:
$$
  x_s^{\sched} \ \ = \ \ \Exp{\psched}{\cN,s}(\accdiaplus \final)
$$
This yields:
$$
  \theta_{s,r} - (\threshold^*{-}r) \cdot y_{s,r}
  \ \ \ \geqslant \ \ \
  \theta_{s,r,\psched} - (\threshold^*{-}r) \cdot y_{s,r,\psched}
$$
for all states $s \in S \setminus \{\goal,\fail\}$, all
actions $\alpha \in \Act(s)$ and all functions $\psched$.
Hence, \eqref{threshold-difference-property} holds for any
value $\threshold^*$ with
$\threshold \leqslant \threshold^* \leqslant \threshold_r$.

It remains to show that if $\threshold^* > \threshold_r$
then \eqref{threshold-difference-property} does not
holds for at least one pair $(s,\psched)$.
This, however, is obvious as we can pick a pair $(s,\alpha)$
such that $y_{s,r} > y_{s,r,\alpha}$ and
$$
   \threshold_r \ \ \ = \ \ \
   r + \frac{\theta_{s,r}-\theta_{s,r,\alpha}}{y_{s,r}-y_{s,r,\alpha}}
$$
If $\threshold^* > \threshold_r$ then
$$
   \threshold^* \ \ \ > \ \ \
   r + \frac{\theta_{s,r}-\theta_{s,r,\alpha}}{y_{s,r}-y_{s,r,\alpha}}
$$
But then
$$
  \theta_{s,r} - (\threshold^*{-}r) \cdot y_{s,r}
  \ \ \ < \ \ \
  \theta_{s,r,\psched} - (\threshold^*{-}r) \cdot y_{s,r,\psched}
$$
This completes the proof of the second statement in
Lemma \ref{lemma:property-sched-threshold-algorithm}.
To prove the last statement of
Lemma \ref{lemma:property-sched-threshold-algorithm},
we suppose
$$
  \threshold \ \ < \ \ \threshold^* \ \ \leqslant \ \
  \min \ \bigl\{\ \threshold_R \ : \
                   r \leqslant R \leqslant \saturation \
  \bigr\}
$$
\eqref{threshold-difference-property} yields that for each level
$R\in \{r,\ldots,\saturation\}$,
the unique solution of the linear program in Figure \ref{LP-threshold}
is the vector $(x_s^R)_{s\in S}$ where
$x_s^R = \theta_{s,R}-(\threshold^*{-}r)y_{s,R}$.
But then the calls of the threshold algorithm for $\threshold$
and $\threshold^*$ deal level-wise with the same MDP
$\cM^*_R \eqdef \cM^*_{R,\threshold}=\cM^*_{R,\threshold^*}$ to derive
$y_{s,R}=y_{s,R}^*$ as the maximal probability to reach $\goal$ from $s$
in $\cM^*_R$ and $\theta_{s,R}=\theta_{s,R}^*$ as the expected total
reward under each scheduler for $\cM^*_R$.
\Ende
\end{proof}

\begin{remark}[MDP without zero-reward cycles]
\label{threshold-problem-MDP-without-zero-reward-cycles}
For the special case of an MDP without zero-reward cycles, the presented
algorithm for the threshold problem can be simplified as follows.
The initialization phase remains unchanged, but in the
treatment of the level $r=\saturation{-}1,\saturation{-}2,\ldots,1,0$,
the solution of the linear program in Figure \ref{LP-threshold}
can be computed directly without linear programming techniques.
For this, we consider an enumeration $s_1,s_2,\ldots,s_N$ of the states
in $S \setminus \{\goal,\fail\}$
such that $P(s_i,\alpha,s_j)>0$ and $\rew(s_i,\alpha)=0$ implies $i > j$.
Then, for $i=1,2,\ldots,N$, and each action $\alpha \in \Act(s_i)$ we put
\begin{eqnarray*}
  y_{s_i,r,\alpha} & \ = \ &
  \sum_{t\in S} P(s_i,\alpha,t) \cdot y_{t,R}
  \\
  \\[0ex]
  \theta_{s_i,r,\alpha} & = &
  \rew(s_i,\alpha) \cdot y_{s_i,r,\alpha} \ \ + \ \
      \sum_{t\in S} P(s_i,\alpha,t) \cdot \theta_{t,R}
\end{eqnarray*}
where $R = \min \, \bigl\{ \saturation, r+\rew(s_i,\alpha)\}$,
$y_{\goal,R}=1$ and
$y_{\fail,R}=\theta_{\goal,R}=\theta_{\fail,R}=0$.
Note that if $\rew(s_i,\alpha)=0$ then $R=r$ and
$P(s_i,\alpha,t) >0$ implies $t=s_j$ for some $j < i$.
Hence, the relevant values $y_{t,R}$ and $\theta_{t,R}$ have been computed
before.
Let
$$
  \Delta_{s_i,r} \ \ = \ \
  \max_{\alpha \in \Act(s_i)} \Delta_{s_i,r,\alpha}
  \qquad \text{where} \qquad
  \Delta_{s_i,r,\alpha} \ = \
  \theta_{s_i,r,\alpha} \ - \
              (\threshold {-}r ) \cdot y_{s_i,r,\alpha}
$$
We then pick an action $\alpha \in \Act(s_i)$ where
$\Delta_{s_i,r,\alpha} = \Delta(s_i,r)$ and
$y_{s_i,r,\alpha}\geqslant y_{s_i,r,\beta}$ for each
action $\beta \in \Act(s_i)$ with
$\Delta_{s_i,r,\beta} = \Delta(s_i,r)$.
We then define $(\action(s_i,r),y_{s_i,r},\theta_{s_i,r})$ as the triple
$(\alpha,y_{s_i,r,\alpha},\theta_{s_i,r,\alpha})$.
Then, the values $x_s^*= \Delta_{s,r,\alpha}$ obtained after treating all
states at level $r$ constitute the unique solution of the linear program in
Figure \ref{LP-threshold}.
Hence, the threshold problem in MDPs without zero-reward cycles is solvable
in $\cO(|S|\cdot |\Act| \cdot \saturation)$ steps.
\Ende
\end{remark}

Analogous techniques are applicable to establish a PSPACE upper bound
for the threshold problem in acyclic MDPs.
This will be shown in Lemma \ref{lemma:CExp-in-PSPACE}.

\newcommand{\mypoly}{g}
\newcommand{\thepoly}{f}
\newcommand{\Denom}{Denom}
\newcommand{\LogLength}{LogLen}

\subsection{Complexity of the threshold algorithm}

\label{appendix:size}
\label{appendix:complexity-threshold}

The time complexity of the threshold algorithm is dominated by
(i) the computation of the saturation point $\saturation$
as described in Section \ref{appendix:compute-saturation}
and (ii) the linear programs to compute feasible actions for the levels
$r=\saturation{-}1,\saturation{-}2,\ldots,1,0$.
The time complexity of step (i) is pseudo-polynomial in the size
of $\cM$ as outlined in Section \ref{appendix:compute-saturation}.
To prove the exponential time bound as stated in
Theorem \ref{thm:threshold-problem}, we show that the time
complexity of step (ii) is exponential in the size of $\cM$
and polynomial in the logarithmic length of the threshold value $\threshold$.
As linear programs are solvable in time polynomial in the number of variables
and the total logarithmic lengths of the coefficients in the linear
constraints, it suffices to establish a exponential bound for
the logarithmic lengths of the
probability values $y_{s,r}$ and the partial expectations $\theta_{s,r}$
that are computed in the threshold algorithm and used in the linear program
in Figure \ref{LP-threshold} in Appendix~\ref{appendix:threshold}.

Given a non-zero-rational number $x$, we refer to
the unique coprime integers $n$, $d$ with
$x=n/d$ and $d > 0$
as \emph{the} numerator ($n$) and \emph{the} denominator ($d$) of $x$.
For $x=0$ we say \emph{the} denominator is 1 and \emph{the} numerator is 0.

\begin{lemma}
  \label{integer-LES}
  Let $A = (a_{i,j})_{i,j=1,\ldots,m}$
  be a non-singular $m\times m$-matrix
  with integer values $a_{i,j}$
  whose logarithmic length is bounded by $K$.
  Let $b = (b_{i})_{i=1,\ldots,m}$  be an integer vector
  where the logarithmic length of the values $b_i$ is bounded by $L$.
  Let $x = (x_j)_{j=1,\ldots,n}$ be
  the unique solution of the linear equation
  system $Ax=b$.
  Then,
  the $x_j$'s are rational numbers and
  the logarithmic length of the
  least common multiple $\textsf{lcm}$
  of the denominators of the values $x_1,\ldots,x_m$
  is bounded by $m\log m + Km$.
  Moreover, for all $j\in \{1,\ldots,m\}$,
  the logarithmic length of $x_j \cdot \textsf{lcm} \in \Integer$
  is at most $m\log m + K(m{-}1)+L$.
\end{lemma}

\begin{proof}
 Let $A_i$ denote the matrix resulting from $A$ by replacing the $i$-th column
 with the vector $b$. By Cramer's rule we have:
 $$
   x_i \ \ = \ \ \frac{\textrm{det}(A_i)}{\textrm{det}(A)}
 $$
 where $\textrm{det}(B)$ denotes the determinant of matrix $B$.
 By the Leibniz formula for determinants we have:
 $$
   \textrm{det}(A)
   \ \ = \ \
   \sum_{\sigma \in \textit{Perm}_m}
     \textit{sign}(\sigma) \cdot
     a_{1,\sigma(1)} \cdot a_{2,\sigma(2)} \cdot \ldots \cdot
     a_{m,\sigma(m)}
 $$
 where $\textit{Perm}_m$ denotes the set of permutations of the values
 $1,\ldots,m$ and $\textit{sign}(\sigma)$ the sign of $\sigma$.
 As $|a_{i,j}|\leqslant 2^K{-}1$ we get:
 $$
   |\textrm{det}(A)| \ \ < \ \
   m! \cdot 2^{Km}
   \ \ \leqslant \ \ 2^{m\log m + Km}
 $$
 where we use the fact that $m! \leqslant m^m = 2^{m \log m}$.
 Likewise, we get for the absolute value of the determinant of $A_i$:
 $$
   |\textrm{det}(A_i)| \ \ < \ \
   m! \cdot 2^{K(m-1)} \cdot 2^L
   \ \ \leqslant \ \ 2^{m\log m + K(m-1)+L}
 $$
 This yields the claim.
\Ende
\end{proof}

\begin{lemma}
  \label{rational-LES}
  Notations as in Lemma \ref{integer-LES}, except that
  the values $a_{i,j}$ and $b_{i}$ are rational
  and the logarithmic lengths of the numerators (resp.~denominators)
  of the values $a_{i,j}$ are bounded by $K_n$ (resp.~$K_d$),
  while the logarithmic lengths of the numerators (resp.~denominators)
  of the values $b_{i}$ are bounded by $L_n$ (resp.~$L_d$).
  Then, the values $x_i$  are rational numbers and the logarithmic length
  of the least common multiple $\textsf{lcm}$ of the denomimators
  is bounded by
  $m (\log m + K_n  + mK_d + L_d)$.
  The logarithmic length of $x_j \cdot \textsf{lcm}$
  is at most
  $m (\log m + K_n  + mK_d + L_d) + L_n + mK_d + L_d$.
\end{lemma}

\begin{proof}
    We multiply the $i$-th row of $A$ and $b_i$
    with the least common multiple
    $\textsf{lcm}_i$  of
    the denominators of the values $a_{i,j}$ and the denominator of
    $b_i$.
    Let $A'x = b'$ be the resulting equation system.
    Obviously, $\textsf{lcm}_i < 2^{mK_d + L_d}$.
    Thus, the values $a_{i,j}'$ of the matrix $A'$
    are integers whose logarithmic length is bounded by
    $K = K_n + mK_d + L_d$.
    Likewise, $b = (b_i')_{i=1,\ldots,m}$ is an integer vector and the
    logarithmic length of the values $b_i'$ is bounded by
    $L = L_n + mK_d + L_d$.

    The unique solution $x$ of $Ax=b$ is also the unique
    solution of $A'x = b'$.
    Thus, the claim follows by Lemma \ref{integer-LES}.
\Ende
\end{proof}

\noindent
We now consider the probability values $y_{s,r}$ computed in
the threshold algorithm and show that their logarithmic lengths are
polynomially bounded in $\saturation$ and the size of $\cM$.
Analogous arguments can be provided for the partial expectations
$\theta_{s,r}$.
Let $y_{s,r}=y'_{s,r}/d_r$  where $d_r$  is the least
common multiple of the denominators of $y_{s,R}$
where $s$ ranges over all states in $S$ and $R$ over the levels in
$\{r,r{+}1,\ldots,\saturation\}$.
Hence, the values $y'_{s,r}$ are non-negative integers
and $d_r$ is a multiple of $d_{R}$ for all $r < R \leqslant \saturation$.
Let $k_r$  denote the maximal logarithmic length of the
values $y'_{s,r}$.
The goal is show that $k_r$ and the logarithmic lengths of the values $d_r$
are polynomially bounded in $\saturation$ and $\Size(\cM)$
for all $r \in \{0,1,\ldots,\saturation\}$.

For $r=\saturation$, the values
$y_{s,\saturation}=\Pr^{\maxsched}_{\cM,s}(\Diamond \goal)$
can be characterized as the unique solution of a
linear equation system $Ax=b$ where $A$ is a $m\times m$-matrix
of the form $I-P'$ with $m \leqslant |S|$.
Here $I$ is the identity matrix and
$P'$ arises from the transition probability matrix
of the Markov chain induced by $\maxsched$
by deleting certain columns and rows.
Using Lemma \ref{rational-LES} we get that the logarithmic length
of the values $y'_{s,\saturation}$ is in $\cO(\Size(\cM)^3)$.
Thus, $k_{\saturation}$ is in $\cO(\Size(\cM)^3)$.

Let now $r < \saturation$ and let $\psched_r^*$ denote the
memoryless deterministic scheduler chosen by the threshold algorithm
for level $r$.
For $t\in S \setminus \{\goal,\fail\}$, let
$\alpha_t = \psched_r^*(t)$ and
$R_t=\min \{\saturation, r{+}\rew(t,\alpha_t)\}$.
Let $T_r$ denote the set of states $t\in S$ such that either
$t\in \{\goal,\fail\}$ or $\rew(t,\alpha_t)>0$.
For the states $t\in T_r \setminus \{\goal,\fail\}$:
\begin{eqnarray*}
  y_{t,r} & \ =  \ &
  \sum_{u\in S} P(t,\alpha_t,u) \cdot y_{u,R_t}
\end{eqnarray*}
while for the states $s\in S \setminus T$:
$$
  y_{s,r} \ \ = \ \
  \sum_{t\in T_r}
    \Pr^{\psched_r^*}_{\cM,s}(\neg T_r \Until t) \cdot y_{t,r}
$$
As a consequence of Lemma \ref{rational-LES} we obtain a polynomial
bound for the logarithmic lengths of the reachability probabilities
$\Pr^{\psched_r^*}_{\cM,s}(\neg T_r \Until t)$.
More precisely:

\begin{corollary}
 \label{Until-log-length}
  There exists a polynomial~$\mypoly$ of degree $3$
  such that for all triples $(\sched,T,t)$ where $\sched$
  is a memoryless deterministic scheduler for $\cM$,
  $T\subseteq S$ and $t\in T$, there exists
  a positive integer $d$ and non-negative integers $n_s$ for $s\in S$ with
  $\Pr^{\sched}_{\cM,s}(\neg T \Until t) = n_s/d$
  and such that the logarithmic lengths of $n_s$ and $d$ are
  bounded by $\mypoly(\Size(\cM))$.
\end{corollary}

\begin{proof}
  The claim follows by Lemma \ref{rational-LES} as the non-zero
  values $\Pr^{\sched}_{\cM,s}(\neg T \Until t)$
  can be obtained as the unique solution of a linear equation system
  of the form $Ax=b$ where $A = I- P'$ where $P'$ arises from the
  transition probability matrix $(P(s,\sched(s),t))_{s,t,\in S}$ by removing
  certain columns and rows.
  Note that $\Size(\cM)$ is an upper bound for the values
  $m,K_d,K_n,L_d, L_n$ in Lemma \ref{rational-LES},
  no matter how $\sched$, $T$ and $t$ are chosen.
\Ende
\end{proof}

\noindent
In what follows, let $A$ denote the least common multiple of the
denominators $P_d(s,\alpha,t)$
of the transition probabilities $P(s,\alpha,t)$ when
ranging over all triples $(s,\alpha,t)\in S \times \Act \times S$
where $P(s,\alpha,t)>0$.
Obviously, $A$ is bounded by the product of the values $P_d(s,\alpha,t)$.
Hence, the logarithmic length of $A$ is bounded by $\Size(\cM)$.
Let $B_r$ denote the least common multiple of the
denominators of the reachability probabilities
$\Pr^{\psched_r^*}_{\cM,s}(\neg T_r \Until t)$ when ranging over all
triples $(s,t,r)$ with $s\in S \setminus T_r$.
The logrithmic length of $B_r$ is bounded by
$|S|^2 \cdot \mypoly(\Size(\cM))$ where
$\mypoly$ is as in Corollary \ref{Until-log-length}.

For the states $t\in T_r$ we have:
$$
  y_{t,r} \ \ = \ \
  \sum_{u\in S} P(t,\alpha_t,u) \cdot \frac{y'_{u,R_t}}{d_{R_t}}
  \ \ = \ \
  \frac{Y_{t,r}}{A \cdot d_{r+1}}
$$
where
$$
   Y_{t,r} \ \ = \ \
   \sum_{u\in S}
        A \cdot P(t,\alpha_t,u) \cdot
        y'_{u,R_t} \cdot \frac{d_{r+1}}{d_{R_t}}
$$
Note that $d_{r+1}/d_{R_t}$ is an integer whose logarithmic length
is bounded by the logarithmic length of $d_{r+1}$.
Obviously, the values $A \cdot P(t,\alpha_t,u)$ are integers
with the logarithmic lengths at most $2\cdot \Size(\cM)$.
Hence, $Y_{t,r}\in \Nat$.
As $R_t \geqslant r{+}1$, the logarithmic length of the values
$y'_{u,R_t}$ is bounded by $k_{r+1}$.
Hence:
$$
  Y_{t,r} \ \ \leqslant \ \
  |S| \cdot  2^{2\Size(\cM)} \cdot 2^{k_{r+1}} \cdot d_{r+1}
$$
For the states $s\in S \setminus T_r$ we have:
$$
  y_{s,r} \ \ = \ \
  \frac{Y_{s,r}}{A \cdot B_r \cdot d_{r+1}}
  \qquad \text{where} \qquad
  Y_{s,r}
  \ \ = \ \
  \sum_{t\in T_r}
    B_r \cdot \Pr^{\psched_r^*}_{\cM,s}(\neg T_r \Until t)
   \cdot Y_{t,r}
$$
The values $B_r \cdot \Pr^{\psched_r^*}_{\cM,s}(\neg T_r \Until t)$
are integers with
$$
  \log\bigl( \ B_r \cdot \Pr^{\psched_r^*}_{\cM,s}(\neg T_r \Until t) \ \bigr)
  \ \ \leqslant \ \
  |S|^2 \cdot \mypoly\bigl(\Size(\cM)\bigr) \ + \
  \mypoly\bigl(\Size(\cM)\bigr)
$$
Here, we use Corollary \ref{Until-log-length}.
We get $Y_{s,r} \in \Nat$ and
$d_r \ \leqslant  \ A \cdot B_r \cdot d_{r+1}$.
Thus, there exists a polynomial $f$  such that:
$$
  \log(d_r) \ \ \leqslant \ \ f\bigl(\Size(\cM)\bigr) \ + \ \log(d_{r+1})
$$
and $d_{\saturation}\leqslant f(\Size(\cM))$.
Hence, the logarithmic lengths of the $d_r$'s are polynomially bounded
in $\saturation$ and the size of $\cM$ as we have:
$$
  \log(d_r) \ \ \leqslant \ \
  (\saturation {-} r{+}1)\cdot f\bigl(\Size(\cM)\bigr)
$$
The above shows that $y'_{s,r}\leqslant Y_{s,r}$ for all states $s\in S$
and there exists a bivariate polynomial $h$ such that
$$
  k_r \ \ \leqslant \ \ k_{r+1} \ + \ h\bigl(\saturation,\Size(\cM)\bigr)
$$
Recall that $k_r$ denotes the maximal logarithmic length of the
values $y'_{s,r}$.
We may assume w.l.o.g.~that
$k_{\saturation} \leqslant h(\saturation,\Size(\cM))$.
Then:
$$
  k_r \ \ \leqslant \ \
  (\saturation {-} r{+}1)\cdot h \bigl( \saturation,\Size(\cM) \bigr)
$$
Hence, $k_r$ and therefore the logarithmic lengths of the values $y_{s,r}$
are polynomially bounded in $\saturation$ and $\Size(\cM)$.
Analogously, we obtain that the logarithmic lengths of the
values $\theta_{s,r}$ are polynomially bounded in
$\saturation$ and $\Size(\cM)$.

We conclude that the representation of the linear programs
(in particular the bit-presentation of the coefficients) used in
the threshold algorithm
(Figure~\ref{LP-threshold} in Appendix~\ref{appendix:threshold})
are polynomially bounded in $\saturation$, $\Size(\cM)$ and the logarithmic
length of the threshold value $\threshold$.
As the logarithmic length of the saturation
point $\saturation$ is polynomial in $\Size(\cM)$
(see Appendix~\ref{appendix:compute-saturation}),
the time complexity of the threshold algorithm
(Section~\ref{sec:threshold} and Appendix~\ref{appendix:threshold})
is exponential.

\section{Computing an optimal scheduler and 
   the maximal conditional expectation}

\label{appendix:cexpmax}

We now address the task to compute $\CExp{\max}$.
As before, we suppose that \eqref{assumption:A1}, (A2) holds and that $\cM$ has no critical
schedulers.
By the results of the previous sections, we can formulate a simple
algorithm that successively calls the threshold algorithm
(see Section \ref{algo:threshold-cexp}) for computing the maximal conditional
expectation and an optimal scheduler.
The preprocessing is the same as for the threshold
algorithm, i.e., the algorithm computes the saturation point $\saturation$
and the deterministic memoryless scheduler $\maxsched$ that is known
to provide optimal decisions for all paths $\fpath$ with
$\rew(\fpath) \geqslant \saturation$.
Let $\SchedThreshold{\threshold}$ denote the scheduler that is generated
by calling the threshold algorithm for the threshold value $\threshold$.
\begin{enumerate}
\item []
    $\sched := \maxsched$ where
    $\maxsched$ is as in Lemma \ref{lemma:Sched-max-exp};
\item []
    {\tt REPEAT}
    \begin{enumerate}
    \item []
       $\threshold := \CExp{\sched}$; \ \
       $\sched := \SchedThreshold{\threshold}$; \
    \end{enumerate}
    {\tt UNTIL} $\threshold = \CExp{\sched}$;
\item []
    return $\threshold$ as the maximal conditional expectation
    and an optimal scheduler $\sched$
\end{enumerate}
Let $\sched_i$ and $\threshold_i$ denote the scheduler resp.~threshold value
at the end of
the $i$-th iteration of the repeat-loop.
Then, $\sched_{i}=\SchedThreshold{\threshold_{i-1}}$
and $\threshold_i = \CExp{\sched_i}$
where $\threshold_0=\CExp{\maxsched}$.
Hence, all calls of the threshold algorithm are successful in the sense
that $\CExp{\sched_{i}} \geqslant \threshold_{i-1}$.

Using Corollary
\ref{corollary:threshold-algorithm-exact} we get the following.
The above algorithm
generates a strictly increasing sequence of threshold values
that are the conditional expectations of some deterministic
reward-based scheduler that is memoryless from
$\saturation$ on. As the set of the latter is finite and
bounded by $K = |\Act|^{\saturation |S|}$
there is some $k \leqslant K$ such that
$$
  \threshold_0 \ < \ \threshold_1 \ < \ \threshold_2 \ < \
   \ldots \ < \threshold_{k-1} \ = \ \threshold_k \ = \ \CExp{\max}
$$
Hence, the algorithm terminates after at most $K$ iterations and
correctly returns $\CExp{\max}$ and an optimal scheduler.
Together with the exponential time complexity of the
threshold algorithm (see Section \ref{appendix:complexity-threshold}),
this yields a double-exponentially time bounded algorithm for the computation
of $\CExp{\max}$.

In the sequel, we present a more efficient algorithm for computing
$\CExp{\max}$ that runs in single exponential time.
The idea is an iterative scheduler-improvement approach that
relies on Lemma \ref{lemma:decision-for-s-r}
and is interleaved with calls of the
threshold algorithm to maintain and successively improve a left-closed and
right-open interval
$I = [A,B[$ with $\CExp{\max}\in I$ together
with a scheduler $\sched$ such that
$\CExp{\sched} \in I$.
Similar to the threshold algorithm,
the proposed algorithm for computing $\CExp{\max}$
operates level-wise and freezes optimal decisions for levels
$r=\saturation,\saturation{-}1,\saturation{-}2,\ldots,1,0$.
More precisely, when level $r$ has been treated then the current scheduler
$\sched$ is strongly optimal for all level $\geqslant r$ in the sense of
Definition \ref{def:strongly-optimal}.

Thanks to Corollary \ref{corollary:threshold-algorithm-exact},
it is possible that an optimal scheduler is found when treating some level
$r >0$, in which case the explicit treatment of levels $r{-}1,r{-}2,\ldots,0$
is skipped. (However, the levels $<r$  have been treated implicitly
in the calls of the threshold algorithms.)

\begin{definition}[Strong $r$-optimality, $r$-optimality]
\label{def:strongly-optimal}
{\rm
Let $\sched$ be a reward-based scheduler.
Scheduler $\sched$ is said to be strongly optimal from level
$r$ on, briefly called strongly $r$-optimal,
if for all states $s\in S$, all schedulers $\tsched$
and all $R\in \Nat$, $R \geqslant r$ we have:
$$
  \Exp{\residual{\sched}{R}}{s} -
  (\CExp{\max}-R) \cdot \Pr^{\residual{\sched}{R}}_s(\Diamond \goal)
  \ \ \ \geqslant \ \ \
  \Exp{\tsched}{s} -
  (\CExp{\max}-R) \cdot \Pr^{\tsched}_s(\Diamond \goal)
$$
Recall that $\residual{\sched}{R}$ denotes the residual scheduler
given by $(\residual{\sched}{R})(s,i) = \sched(s,R{+}i)$.
Hence, $\residual{(\residual{\sched}{R})}{i} = \residual{\sched}{(R{+}i)}$.

The notion of ``optimal from level $r$ on'' is used as before.
That is, scheduler $\sched$ is called optimal from level $r$ on,
briefly called $r$-optimal, if
$\CExp{\tsched}=\CExp{\max}$ implies $\CExp{\usched}=\CExp{\max}$ where
$\usched(\fpath)=\tsched(\fpath)$ for each finite path $\fpath$ with
$\rew(\fpath) < r$ and
$\usched(\fpath)=\sched(\fpath)$ for each finite path $\fpath$ with
$\rew(\fpath) \geqslant r$.

The notion ``strongly optimal'' will be used instead of strongly 0-optimal.
Likewise, a scheduler is briefly called optimal if it is 0-optimal.
\Ende
  }
\end{definition}

Clearly, $\sched$ is strongly $r$-optimal if and only if the following
conditions hold for all states $s\in S$, all schedulers $\tsched$
and all $R\in \Nat$, $R \geqslant r$:
\begin{eqnarray*}
  R + \frac{\Exp{\residual{\sched}{R}}{s} -\Exp{\tsched}{s}}
           {\Pr^{\residual{\sched}{R}}_s(\Diamond \goal)-
            \Pr^{\tsched}_s(\Diamond \goal)}
  \ \ \geqslant \ \ \CExp{\max}
  & \ \ &
  \text{if $\Pr^{\residual{\sched}{R}}_s(\Diamond \goal) >
            \Pr^{\tsched}_s(\Diamond \goal)$}
  \\
  \\[0ex]
  R + \frac{\Exp{\residual{\sched}{R}}{s} -\Exp{\tsched}{s}}
           {\Pr^{\residual{\sched}{R}}_s(\Diamond \goal)-
            \Pr^{\tsched}_s(\Diamond \goal)}
  \ \ \leqslant \ \ \CExp{\max}
  & \ \ &
  \text{if $\Pr^{\residual{\sched}{R}}_s(\Diamond \goal) <
            \Pr^{\tsched}_s(\Diamond \goal)$}
  \\
  \\[0ex]
  \Exp{\residual{\sched}{R}}{s} \ \geqslant \ \Exp{\tsched}{s}
  \hspace*{2cm}
  & &
  \text{if $\Pr^{\residual{\sched}{R}}_s(\Diamond \goal)
            = \Pr^{\tsched}_s(\Diamond \goal)$}
\end{eqnarray*}
Hence, by Lemma \ref{lemma:decision-for-s-r},
each strongly $r$-optimal scheduler $\sched$ is
$r$-optimal. The reverse direction does not hold in general as
some pairs $(s,r)$ might not be reachable under $\sched$.

The existence of strongly optimal schedulers is ensured by the
following lemma:

\begin{lemma}
  \label{lemma:SchedThreshold-CExpmax}
  The scheduler $\SchedThreshold{\CExp{\max}}$ is strongly optimal.
  In particular, if $\CExp{\SchedThreshold{\threshold}} = \threshold$ then
  $\SchedThreshold{\threshold}$ is strongly optimal.
\end{lemma}

\begin{proof}
 The first part
 follows by the fact that strong $r$-optimality of a reward-based
 scheduler $\sched$
 is equivalent to the statement
 that for all states $s$, $R\in \Nat$ with $R \geqslant r$ and
 all schedulers $\tsched$:
 $$
   \Exp{\residual{\sched}{R}}{s} -
   (\CExp{\max}-R)\cdot \Pr^{\residual{\sched}{R}}_{s}(\Diamond \goal)
   \ \ \ \geqslant \ \ \
   \Exp{\tsched}{s} -
   (\CExp{\max}-R)\cdot \Pr^{\tsched}_{s}(\Diamond \goal)
 $$
 The claim
 then follows by
 Lemma \ref{lemma:property-sched-threshold-algorithm} using an induction
 on $i=\saturation-R$.

 To see the second part, we suppose
  $\sched = \SchedThreshold{\threshold}$ and
  $\CExp{\sched}=\threshold$. Then, $\threshold = \CExp{\max}$
 by Corollary \ref{corollary:threshold-algorithm-exact}.
 Thus, $\sched$ is strongly optimal.
\Ende
\end{proof}

\tudparagraph{1ex}{{\it Initialization.}}
The algorithm starts with the scheduler
$\sched= \SchedThreshold{\CExp{\maxsched}}$ where $\maxsched$ is
as in Lemma \ref{lemma:Sched-max-exp}.
If $\CExp{\sched}=\CExp{\maxsched}$ then the algorithm immediately
terminates on the basis of Corollary \ref{corollary:threshold-algorithm-exact}.
Suppose now that $\CExp{\sched} > \CExp{\maxsched}$.
The initial interval is $I = [A,B[$ where $A = \CExp{\sched}$
and $B$ is a strict upper bound for $\CExp{\max}$,
e.g.,~the upper bound $\CExp{\ub}$ (plus some small constant) 
obtained by the preprocessing explained
in Section \ref{sec:finiteness} and 
Appendix~\ref{sec:upper-bound}.

\tudparagraph{1ex}{{\it Level-wise scheduler improvement.}}
The algorithm successively determines
optimal decisions for the levels
$r=\saturation{-}1,\saturation{-}2,\ldots,1,0$.
It maintains a left-closed and right-open interval $I = [A,B[$
and a reward-based scheduler $\sched$
satisfying the invariance $\CExp{\max}\in I$ and $\CExp{\sched}\in I$.
The scheduler $\sched$ has been obtained by the last successful run
of the threshold algorithm for the strict bound ``$> \threshold$''
applied to some threshold $\threshold \leqslant A$ in the previous interval.
Thus, $\sched=\SchedThreshold{\threshold}$ and $\CExp{\sched} > \threshold$.

The treatment of level $r$ consists of a sequence of scheduler improvement
steps where at the same time the interval $I$ is replaced with proper
sub-intervals. The scheduler improvements are obtained by calls of the
threshold algorithm ``does $\CExp{\max} \geqslant \threshold$ hold?''
for some appropriate value $\threshold \in I$.
If the scheduler $\sched$ generated by the threshold algorithm
enjoys the property
$\CExp{\sched}=\threshold$ then
the algorithm terminates
and returns $\threshold$ as the maximal condition expectation and $\sched$
as an optimal scheduler
(see Corollary \ref{corollary:threshold-algorithm-exact}).

Besides the decisions of $\sched$ (i.e., the actions
$\sched(s,R)$ for all state-reward pairs
$(s,R)$ where $s\in S \setminus \{\goal,\fail\}$ and
$R\in \{0,1,\ldots,\saturation\}$) we shall also need the values
\begin{eqnarray*}
   y_{s,r} & \ \ = \ \ &
   \Pr^{\residual{\sched}{r}}_{\cM,s}(\Diamond \goal)
   \ \ \ = \ \ \
   \Pr^{\max}_{\cM^*,s}(\Diamond \goal)
   \\[1.5ex]
  \theta_{s,r} & \ \ = \ \ &
  \Exp{\residual{\sched}{r}}{\cM,s}(\accdiaplus \goal)
  \ \ \ = \ \ \
  x_s^* \ + \ (\threshold {-} r) \cdot y_{s,r}
\end{eqnarray*}
that have been computed in the threshold algorithm
(where $x_s^*$ refers to the unique solution of the linear program
in Figure \ref{LP-threshold} and $\cM^* = \cM^*_{r,\threshold}$ is the
MDP defined as in Section \ref{algo:threshold-cexp}).
The algorithm also stores the optimal actions and the values
$\theta_{s,R}$ and $y_{s,R}$ for $s\in S$ and all levels
$R \in \{r{+}1,\ldots,\saturation\}$ that have been treated before.
These values
can be reused in the calls of the threshold algorithms.
That is, the calls of the threshold algorithm that are invoked
in the scheduler-improvement steps at level $r$ can skip levels
$\saturation,\saturation{-}1,\ldots,r{+}1$ and only need to process
levels $r,r{-}1,\ldots,1,0$.

For the current level $r$, the algorithm also computes
for each state $s\in S \setminus \{\goal,\fail\}$ and
each action $\alpha \in \Act(s)$ the values:
\begin{eqnarray*}
   y_{s,r,\alpha} & = & \sum_{t\in S} P(s,\alpha,t) \cdot y_{t,R}
   \\
   \\[0ex]
   \theta_{s,r,\alpha} & = &
   \rew(s,\alpha) \cdot y_{s,r,\alpha} \ + \
   \sum_{t\in S} P(s,\alpha,t) \cdot \theta_{t,R}
\end{eqnarray*}
where $R = \min \{ \saturation, r+\rew(s,\alpha)\}$.
Thus, $R=r$ if $\rew(s,\alpha)=0$.%
\footnote{Note that
  $y_{s,r,\alpha} \ = \
   \Pr^{\sched_{s,r,\alpha}}_{\cM,s}(\Diamond \goal)$ and
  $\theta_{s,r,\alpha}  \ = \
   \Exp{\sched_{s,r,\alpha}}{\cM,s}(\accdiaplus \goal)$,
 where $\sched_{s,r,\alpha}$ denotes the unique scheduler
 that agrees with $\residual{\sched}{r}$, except that it assigns
 action $\alpha$ to state $s$ (viewed as a path of length 0).
 That is, $\sched_{s,r,\alpha}(s)=\alpha$ and
 $\sched_{s,r,\alpha}(\fpath)=\residual{\sched}{r}(\fpath)$
 for all finite paths of length at most 1.}

\tudparagraph{1ex}{{\it Scheduler-improvement step.}}
Let $r$ be the current level,
$I = [A,B[$ the current interval
and $\sched$ the current scheduler with $\CExp{\max}\in I$.
At the beginning of the scheduler-improvement step we have
$\CExp{\sched} =A$.
Let
\begin{eqnarray*}
  \cI_{\sched,r} & \ = \ &
  \Bigl\{ \
     r + \frac{\theta_{s,r} - \theta_{s,r,\alpha}}{y_{s,r}-y_{s,r,\alpha}}
     \ : \
     s \in S \setminus \{\goal,\fail\}, \
     \alpha \in \Act(s), \
     y_{s,r} > y_{s,r,\alpha}
     \
  \Bigr\}
  \\
  \\[0ex]
  \cI^{\uparrow}_{\sched,r} & = &
  \bigl\{ \ d \in \cI_{\sched,r} \ : \
            d \ \geqslant \ \CExp{\sched}
  \bigr\}
  \hspace*{1cm}
  \cI^{B}_{\sched,r} \ \ = \ \
   \bigl\{ \ d \in \cI_{\sched,r} \ : \
             d < B \
  \bigr\}
\end{eqnarray*}
If $\cI^{B}_{\sched,r} = \varnothing$ then no further scheduler-improvements
at level $r$ are possible, i.e., $\sched$ is strongly $r$-optimal
(see Lemma \ref{lemma:freeze-level-r}).
In this case:
\begin{itemize}
  \item
    If $r=0$ then the algorithm terminates with the
    answer $\CExp{\max}=\CExp{\sched}$ and $\sched$ as an optimal scheduler.
  \item
    If $r > 0$ then the algorithm goes to the next level $r{-}1$
    and performs the scheduler-improvement step for $\sched$ at level
    $r{-}1$.
\end{itemize}
Suppose now that $\cI^B_{\sched,r}$ is nonempty.
Let $\cK = \cI^{\uparrow}_{\sched,r} \cup \{\CExp{\sched}\}$.
The algorithm seeks for the largest
value $\threshold' \in \cK \cap I$
such that $\CExp{\max}\geqslant \threshold'$.
More precisely, it successively
calls the threshold algorithm for the threshold value
$\max (\cK \cap I)$ and performs the following steps after each call
of the threshold algorithm.
Let $\sched' = \SchedThreshold{\threshold'}$
be the scheduler that has been generated by the
threshold algorithm $\threshold' = \max (\cK \cap I)$.
\begin{itemize}
\item
 Suppose the result of the threshold algorithm is ``no''.
 If $\Pr^{\sched'}_{\cM,\sinit}(\Diamond \goal)$ is positive and
 $\CExp{\sched'} \leqslant \CExp{\max} < \threshold'$, then:
 \begin{itemize}
 \item
   If $\CExp{\sched'} \leqslant A$ then the algorithm
   refines $I$ by putting $B:=\threshold'$.%
\footnote{%
Note
    that the case $\CExp{\sched'} < \CExp{\sched}$ is possible,
    although $\sched'$ has been generated by the threshold algorithm
    for a threshold value $\threshold'$ that is larger than the
    threshold $\threshold$ used for the generation of $\sched$.}
 \item
    If $\CExp{\sched'} > A$
    then the algorithm refines $I$ by putting $A := \CExp{\sched'}$,
    $B := \threshold'$
    and adds $\CExp{\sched'}$ to $\cK$
    (Note that then $\CExp{\sched'}\in \cK \cap I$, while
     $\CExp{\sched} \in \cK \setminus I$.)
 \end{itemize}
\item
  Suppose now that the result of the threshold algorithm is ``yes'',
  i.e., $\CExp{\sched'} \geqslant \threshold '$.
  The algorithm terminates if $\CExp{\sched'} = \threshold'$,
  in which case $\sched'$ is optimal
  (Corollary \ref{corollary:threshold-algorithm-exact}).
  Otherwise, i.e., if $\CExp{\sched'} > \threshold'$, then the algorithm
  aborts the loop by putting $\cK:= \varnothing$,
  refines the interval $I$ by putting $A := \CExp{\sched'}$,
  updates the current scheduler by setting $\sched := \sched'$
  and repeats the scheduler-improvement step.
\end{itemize}
The scheduler $\sched'$ of the last case
($\CExp{\sched'} \geqslant \threshold '$)
will be called the \emph{outcome} of the scheduler-improvement step 
for $\sched$.
Lemma \ref{lemma:scheduler-improvement} (see below)
shows that the scheduler-improvement step
indeed terminates and finds some scheduler $\sched'$ such that:
\begin{center}
    $\CExp{\sched}<\CExp{\sched'}$
    \ \ or \ \ $\CExp{\sched}=\CExp{\sched'} = \CExp{\max}$
\end{center}
Let $\sched_{1},\sched_{2},\sched_{3}\ldots$ be the sequence of
schedulers $\sched$ where the algorithm executes
a scheduler-improvement step for $\sched$.
With $\threshold_1=\CExp{\maxsched}$ and
$\sched_1=\SchedThreshold{\threshold_1}$,
$\sched_{i+1}$ is the outcome of the scheduler-improvement step for
$\sched_{i}$ for $i \geqslant 1$. Furthermore, let $\threshold_{i+1}$
denote the threshold value such that $\sched_{i+1}$ is generated by
calling the threshold algorithm for $\threshold_{i+1}$.
By the choice of the threshold values, we have:
$$
   \CExp{\sched_{i+1}}  \ > \ \threshold_{i+1}  \ \geqslant \
   \CExp{\sched_i}
   \qquad \text{or} \qquad
   \CExp{\max} \ =  \ \CExp{\sched_{i+1}} \ = \ \threshold_{i+1}
   \ \geqslant  \
   \CExp{\sched_i}
$$
All schedulers generated by the algorithm have the form
$\SchedThreshold{\threshold}$ for some threshold value $\threshold$.
In particular, they are deterministic and reward-based
schedulers with fixed values for the last level $\saturation$.
The total number of such schedulers is bounded by
$\md^{\saturation}$ where
$\md$ denotes the total number of memoryless deterministic
schedulers of $\cM$. (Thus, $\md \leqslant |\Act|^{|S|}$.)
Hence, there is some $k\in \Nat$ with
$1\leqslant k \leqslant \md^{\saturation}{+}1$
such that:
$$
  \CExp{\sched_{1}}   \ < \
  \CExp{\sched_{2}}   \ < \
  \CExp{\sched_{3}}   \ < \  \ldots  \ < \
  \CExp{\sched_{k-1}} \ < \ \
  \CExp{\sched_{k}}   \ \leqslant  \
  \CExp{\sched_{k+1}} \ = \  \CExp{\max}
$$
This argument yields a proof sketch for the termination and soundness.
We now provide a more careful analysis of the algorithm with respect to
the correctness and the complexity.
Indeed, we will show that the total number of
scheduler-improvement steps is in
$\cO(\saturation \cdot \md \cdot |S|\cdot |\Act|)$.

\begin{theorem}[Soundness and complexity]
  \label{thm:soundness-CExpmax-algo}
  The above algorithm correctly computes
  $\CExp{\max}$ and a scheduler $\sched$ with $\CExp{\sched}=\CExp{\max}$.
  Its time complexity is exponential in the size of the
  MDP.
\end{theorem}

The proof of
Theorem \ref{thm:soundness-CExpmax-algo} is splitted into two parts.
Partial correctness will be shown in
Proposition \ref{partial-correctness-CEmax-algo},
while the statement on the complexity will be shown in
Proposition \ref{complexity-CExp-max-algo}.

\begin{lemma}[Correctness of the scheduler-improvement step]
  \label{lemma:scheduler-improvement}
    Let $\sched$ be a scheduler.
    \begin{enumerate}
    \item [{\rm (a)}]
      If $\sched'$ is the outcome
      of the scheduler-improvement step for $\sched$ then
      either $\CExp{\sched}<\CExp{\sched'}$
      or $\CExp{\sched}=\CExp{\sched'}=\CExp{\max}$
      and $\sched'$ is strongly optimal.
    \item [{\rm (b)}]
      The scheduler-improvement step for $\sched$
      terminates after at most $|S|\cdot |\Act|$
      calls of the threshold algorithm.
   \end{enumerate}
\end{lemma}

\begin{proof}
We first show statement (a).
Obviously, $\sched'$ is strongly optimal if the
scheduler-improvement step terminates on the basis of
Corollary \ref{corollary:threshold-algorithm-exact}
(see also Lemma \ref{lemma:SchedThreshold-CExpmax}).
Suppose now that $\CExp{\sched'} < \CExp{\max}$.

It is obvious that the refinements of the interval $I = [A,B[$
in the scheduler-improvements are safe in the sense that
$\CExp{\max}\in I$ holds
at any moment of the execution of the scheduler-improvement step.
When the scheduler-improvement step for $\sched$ is called then we have
$A = \CExp{\sched}$.

The set $\cK$ is modified during
the execution of the scheduler-improvement step for
$\sched$. However, at any moment the set $\cK$ only contains elements
of $\cI^{\uparrow}_{\sched,r}$ and values of the form
$\CExp{\tsched}$ for some scheduler $\tsched$
with $\CExp{\tsched}\geqslant \CExp{\sched}$.
Furthermore, if $\cK$ is nonempty then $\cK \cap I$ contains exactly
one value $\CExp{\tsched}$ for some scheduler $\tsched$
with $\CExp{\tsched}\geqslant \CExp{\sched}$.
Note that initially we have $\CExp{\sched}\in \cK \cap I$
and whenever an element $\CExp{\tsched}$ is added to $\cK$ then
$\tsched$ is the best scheduler found so far and
the left border $A$ of $I$ is updated to $A = \CExp{\tsched}$.
Furthermore, $d \geqslant \CExp{\sched}$ for
all $d\in \cI^{\uparrow}_{\sched,r}$.
Thus, as long as $\cK$ is nonempty,
$\CExp{\sched} \leqslant \min |\cK \cap I|$ and $\cK \cap I$ contains
at least one value $\threshold'$ with $\CExp{\max} \geqslant \threshold'$.

The refinements of $I$'s right border $B$ ensure that
each element of $\cI^{\uparrow}_{\sched,r}$ is
selected as the maximal element of $\cK \cap I$ at most once.

Hence, eventually some threshold value $\threshold'\in \cK \cap I$
with $\CExp{\max} \geqslant \threshold'$
will be picked as the maximal element of $\cK \cap I$.
Let $\sched'=\SchedThreshold{\threshold'}$.
But then $\CExp{\sched'} \geqslant \threshold'$.
Hence, $\sched'$ is the outcome of the scheduler-improvement step for
$\sched$. Moreover:
\begin{itemize}
\item
  If $\CExp{\sched'} = \threshold'$
  then the scheduler-improvement step for $\sched$
  returns $\sched'$ as a strongly optimal scheduler.
\item
  If $\CExp{\sched'} > \threshold'$
  then $\CExp{\sched'} > \threshold' = \CExp{\tsched} \geqslant \CExp{\sched}$.
\end{itemize}
Hence, the outcome of the scheduler-improvement step for
$\sched$ is a scheduler $\sched'$ with $\CExp{\sched'} > \CExp{\sched}$
or $\CExp{\sched'}=\CExp{\max}$.

We now turn to the proof of statement (b).
The number of calls of the threshold
algorithm is bounded by $|\cK \cap I|$ for the initial set
$\cK=\cI_{\sched,r}^{\uparrow} \cup \{\CExp{\sched}\}$ and the
current interval $I=[A,B[$ at the beginning of the scheduler-improvement
step for $\sched$. (Thus, $A=\CExp{\sched}$.)
Then:
$$
  |\cK \cap I| \ \ \leqslant \ \ |\cK|
  \ \ \leqslant \ \
  |\cI_{\sched,r}|+1
  \ \ \leqslant \ \
  |S| \cdot (|\Act|{-}1) + 1
  \ \ \leqslant \ \
  |S|\cdot |\Act|
$$
This completes the proof of Lemma \ref{lemma:scheduler-improvement}.
\Ende
\end{proof}

\begin{remark}[Behavior for
               strongly optimal schedulers]
 From the moment on where the algorithm for computing
 $\CExp{\max}$ has generated a strongly optimal
 scheduler $\sched$ on level $r$ there are two possible cases:
 \begin{itemize}
 \item
    If $\cI_{\sched,i}^B = \varnothing$ for $i=0,1,\ldots,r$
    then the algorithm terminates without any further calls of the
    threshold algorithm.
  \item
    If $i\in \{0,1,\ldots,r\}$ is the maximal index
    such that  $\cI_{\sched,i}^B \not= \varnothing$
    then the algorithm first freezes the decisions for levels
    $r,r{-}1,\ldots,i{+}1$ and then generates the final scheduler
    $\sched'$ in the scheduler-improvement step for $\sched$ at level
    $i$ by calling the threshold algorithm for the threshold
    $\threshold'=\CExp{\sched}=\CExp{\max}$.
    Thus, $\CExp{\sched'}=\threshold'$ and
    the algorithm terminates on the basis of
    Corollary \ref{corollary:threshold-algorithm-exact}.
 \end{itemize}
 Hence,
 the scheduler-improvement step for
 a strongly optimal scheduler terminates after at most
 $|S|\cdot |\Act|$ calls of the threshold algorithm
 (see part (b) of Lemma \ref{lemma:scheduler-improvement}).
\Ende
\end{remark}

\begin{remark}[Behavior for
               strongly $r$-optimal schedulers]
  If $\sched$ is a strongly $r$-optimal scheduler
  and $\min \cI_{\sched,r}=\CExp{\max}$ then
  the outcome of the scheduler-improvement step for $\sched$ at level $r$
  is the strongly optimal scheduler
  $\sched' = \SchedThreshold{\CExp{\max}}$.
  However, for $r > 0$ there might be strongly $r$-optimal schedulers
  $\sched$ with
  $$
     \min \cI_{\sched,r} \ \ > \ \ \CExp{\max} %
  $$
  In this case, if $B \leqslant \min \cI_{\sched,r}$
  then $\cI_{\sched,r}^B=\varnothing$ and
  the scheduler-improvement algorithm for $\sched$ at level $r$
  freezes the values $\sched(\cdot,r)$ directly without
  any further calls of the threshold algorithm by switching to
  level $r{-}1$.
  If $B > \min \cI_{\sched,r}$ then  $\cI_{\sched,r}^B$
  is nonempty and the outcome of
  the scheduler-improvement step for $\sched$ at level $r$
  improves $B$ to some $B'$ with
  $\CExp{\max} < B' \leqslant \min \cI_{\sched,r}$.
  However, it does not modify the values $y_{s,r}$ and $\theta_{s,r}$.
  (The latter is a consequence of the second part
  of Lemma \ref{lemma:property-sched-threshold-algorithm}.)
  Thus, if $\sched$ is strongly $r$-optimal and
  $B > \min \cI_{\sched,r}$ then
  the outcome the scheduler-improvement step for $\sched$
  is a scheduler $\sched'$
  with $\cI_{\sched,r} =\cI_{\sched',r}$
  and $\cI_{\sched',r}^{B'}=\varnothing$.
\Ende
\end{remark}

\begin{lemma}
  \label{lemma:freeze-level-r}
  Let $r\in \{0,1,\ldots,\saturation{-}1\}$ and
  $\threshold$ a rational value such that $\threshold \leqslant \CExp{\max}$
  and $\sched = \SchedThreshold{\threshold}$
  is strongly $(r{+}1)$-optimal.
  Then (with $\min \varnothing = \infty$):
  \begin{center}
     $\min \cI_{\sched,r} \geqslant \CExp{\max}$
     \quad iff \quad
     $\sched$ is strongly $r$-optimal
  \end{center}
  In particular, if $\cI_{\sched,r}^{B} = \varnothing$
  at the beginning of the scheduler-improvement step for $\sched$ at
  level $r$
  then $\sched$ is strongly $r$-optimal.
\end{lemma}

\begin{proof}
It is obvious that $\min \cI_{\sched,r} \geqslant \CExp{\max}$
for each strongly $r$-optimal scheduler $\sched$.
We now suppose $\min \cI_{\sched,r} \geqslant \CExp{\max}$
and show the strong $r$-optimality of $\sched$.
Let $y_{s,r}=\Pr^{\residual{\sched}{r}}_s(\Diamond \goal)$
and $\theta_{s,r}=\Exp{\residual{\sched}{r}}{s}$.
As $\threshold \leqslant \CExp{\max}$ we have
 $\threshold \leqslant \CExp{\sched}$.
Using Lemma \ref{lemma:property-sched-threshold-algorithm}, we get
that for all states $s$ and
for all functions $\psched : S \setminus \{\goal,\fail\} \to \Act$
with $\psched(t)\in \Act(t)$ for all $t$:
\begin{eqnarray*}
    r + \frac{\theta_{s,r}-\theta_{s,r,\psched}}{y_{s,r}-y_{s,r,\psched}}
    \ \ \geqslant \ \ \threshold
    & \ \ & \text{if $y_{s,r} > y_{s,r,\psched}$}
    \\
    \\[0ex]
    r + \frac{\theta_{s,r}-\theta_{s,r,\psched}}{y_{s,r}-y_{s,r,\psched}}
    \ \ \leqslant \ \ \threshold
    & \ \ & \text{if $y_{s,r} < y_{s,r,\psched}$}
    \\
    \\[0ex]
    \theta_{s,r} \ \geqslant \ \theta_{s,r,\psched}
    \ \ \ \ \ \ \
    & \ \ & \text{if $y_{s,r} = y_{s,r,\psched}$}
\end{eqnarray*}
By assumption:
  $$
     r + \frac{\theta_{s,r} - \theta_{s,r,\alpha}}{y_{s,r}-y_{s,r\alpha}}
     \ \ \ \geqslant \ \ \ \CExp{\max}
  $$
for all states $s$ and actions $\alpha \in \Act(s)$ with
$y_{s,r} > y_{s,r,\alpha}$.
With the notations of Lemma \ref{lemma:property-sched-threshold-algorithm}
we obtain $\threshold_r \geqslant \CExp{\max} > \threshold$.
Hence, we can rely on the second part of
Lemma \ref{lemma:property-sched-threshold-algorithm}
with $\threshold^* = \CExp{\max}$ to obtain:
\begin{equation*}
  \theta_{s,r} - (\CExp{\max}{-}r) \cdot y_{s,r}
  \ \ \ \geqslant \ \ \
  \theta_{s,r,\psched} - (\CExp{\max}{-}r) \cdot y_{s,r,\psched}
\end{equation*}
for all states $s \in S \setminus \{\goal,\fail\}$ and
all functions $\psched$.
Hence, $\sched$ is strongly $r$-optimal.

If $\cI^{B}_{\sched,r}$ is empty then
$\min \cI_{\sched,r} \geqslant B > \CExp{\max}$ where we use the
invariance $\CExp{\max}\in I=[A,B[$.
\Ende
\end{proof}

\begin{proposition}[Partial correctness]
  \label{partial-correctness-CEmax-algo}
  If the algorithm returns the value $\threshold$ and the scheduler
  $\sched$ then $\threshold = \CExp{\max}$ and $\sched$ is a strongly optimal
  scheduler.
\end{proposition}

\begin{proof}
The statement of Proposition \ref{partial-correctness-CEmax-algo}
is clear in case of an early termination where the algorithm returns a
scheduler $\sched=\SchedThreshold{\threshold}$ with
$\CExp{\sched}=\threshold$
(see Corollary \ref{corollary:threshold-algorithm-exact}).
Let us now suppose that the algorithm treats all levels and
returns $\CExp{\sched}$ for some scheduler $\sched$ considered at level 0
satisfying the constraint $\cI_{\sched,0}^{B}=\varnothing$.
We prove by induction on $i=\saturation {-} r$ that
each scheduler considered in a scheduler-improvement step
at level $r$ is strongly $(r{+}1)$-optimal and that the final scheduler
is strongly optimal.
This yields  $\CExp{\sched}=\CExp{\max}$.

For the first level $r=\saturation$ (basis of induction)
the claim is clear as $\maxsched$ is
strongly $\saturation$-optimal by
Lemma \ref{lemma:maxsched-after-threshold-is-optimal}.
Let now $r < \saturation$ and
$\sched$ be the current scheduler when the algorithm
switches  from level $r$ to $r{-}1$, provided $r>0$, resp.~when
the treatment of level 0 is completed if $r=0$.
Then, $\cI^B_{\sched,r}=\varnothing$.
The task is to show that $\sched$ is strongly $r$-optimal.
We may rely on the induction hypothesis
stating that at the beginning of the treatment of level $r$,
the current scheduler is strongly $(r{+}1)$-optimal.
As the decisions at levels $r{+}1,\ldots,\saturation$
remain unchanged when treating level $r$,
all schedulers where a scheduler-improvement step is
executed at level $r$ are strongly $(r{+}1)$-optimal.
Hence, the remaining task is to prove that for each state $s\in S$
and each function $\psched: S \setminus \{\goal,\fail\} \to \Act$
with $\psched(t)\in \Act(t)$ we have:
\begin{eqnarray*}
  r + \frac{\theta_{s,r} -\theta_{s,r,\psched}}
           {y_{s,r}-y_{s,r,\psched}}
  \ \ \geqslant \ \ \CExp{\max}
  & \ \ & \text{if $y_{s,r}> y_{s,r,\psched}$}
  \\
  \\[0ex]
  r + \frac{\theta_{s,r} -\theta_{s,r,\psched}}
           {y_{s,r}-y_{s,r,\psched}}
  \ \ \leqslant \ \ \CExp{\max}
  & \ \ & \text{if $y_{s,r}< y_{s,r,\psched}$}
  \\
  \\[0ex]
  \theta_{s,r} \ \geqslant \ \theta_{s,r,\psched} \ \ \ \ \
  & & \text{if $y_{s,r}= y_{s,r,\psched}$}
\end{eqnarray*}
The above statement is a consequence of
Lemma \ref{lemma:freeze-level-r}.
Hence, the algorithm correctly freezes the values for level $r$
if $\cI_{\sched_i}^{B} = \varnothing$.
With $r=0$ we get that the final scheduler $\sched$ is strongly optimal
and the returned value is the maximal conditional expectation.
\Ende
\end{proof}

We now address the termination and the complexity of the proposed algorithm.
We start with statements about the scheduler-improvement steps.

\begin{definition}[Indices for the variables of the algorithm]
{\rm
In what follows, we will use the enumeration of the schedulers
and threshold values as explained above.
I.e., $\threshold_1 = \CExp{\maxsched}$,
$\sched_1 = \SchedThreshold{\threshold_1}$
and
$\sched_{i+1}=\SchedThreshold{\threshold_{i+1}}$
is the outcome of the scheduler-improvement step for $\sched_i$.
Moreover, $r_i$ and $I_i = [A_i,B_i[$
denote the current value of the level variable $r$ resp.~the interval
$I=[A,B[$ when the scheduler-improvement
step for $\sched_i$ starts. We often use the fact that
$r_1=\saturation{-}1$, $r_{i+1} \in \{r_i,r_i{-}1\}$,
$A_i=\CExp{\sched_i} < B_i$,
and $B_1 \geqslant B_2 \geqslant B_3 \geqslant \ldots \geqslant \CExp{\max}$.
The latter follows by a careful inspection of the scheduler-improvement step
and the soundness of the threshold algorithm.
Furthermore, let
$y_{s,r,i}$ and $\theta_{s,r,i}$ stand for the current values
of $y_{s,r}$ and $\theta_{s,r}$ when the scheduler-improvement for
$\sched_i$ at level $r=r_i$ starts.
We then have
$y_{s,r,i} = \Pr^{\residual{\sched_i}{r}}_s(\Diamond \goal)
 = y_{s,r,\sched_i(\cdot,r)}$
and
$\theta_{s,r,i}=\Exp{\residual{\sched_i}{r}}{s}
  = \theta_{s,r,\sched_i(\cdot,r)}$
where we use the notations of Section \ref{algo:threshold-cexp}.
  }
\Ende
\end{definition}

The decisions of $\sched_i$ and $\sched_{i+1}$ at level $r_i$
might be the same if $\threshold_{i+1} > \min \cI^{\uparrow}_{\sched_i,r_i}$,
In this case, however, the next scheduler
$\sched_{i+2}$ will differ at level $r_i$.
Moreover, if $\threshold_{i+1} > \min \cI^{\uparrow}_{\sched_i,r_i}$ then
the $\sched_i$ and $\sched_{i+1}$ do not agree at level $r_i$.
This will be shown in the following lemma.

\begin{lemma}
  \label{lemma:last-scheduler-improvement-at-level-r}
    Suppose $r=r_i=r_{i+1}$ and
    $\CExp{\sched_{i}} > \threshold_{i}$.
    Then:
    \begin{enumerate}
    \item [{\rm (a)}]
      $y_{s,r,i} \geqslant y_{s,r,i+1}$ for all states $s$.
    \item [{\rm (b)}]
      If $\cI^{\uparrow}_{\sched_i,r} \cap I_i$ is empty or
      $\threshold_{i+1} > \min \cI^{\uparrow}_{\sched_i,r}$
      then there is at least one state $t$ such that
      $y_{t,r,i} > y_{t,r,i+1}$.
    \item [{\rm (c)}]
      Suppose
      $\threshold_{i+1} < \min \cI^{\uparrow}_{\sched_i,r}$
      and $\sched_i(\cdot,r)=\sched_{i+1}(\cdot,r)$.
      Then, $\cI^{B_{i+1}}_{\sched_{i+1},r}=\varnothing$.
    \item [{\rm (d)}]
      Suppose $\threshold_{i+1} = \min \cI^{\uparrow}_{\sched_i,r}$ and
      $\sched_i(\cdot,r)=\sched_{i+1}(\cdot,r)$ and
      $\CExp{\sched_{i+1}} > \threshold_{i+1}$.
      Then,
      there is at least one state $t$
      where
      $y_{t,r,i} > y_{t,r,i+2}$.
    \end{enumerate}
\end{lemma}

\begin{proof}
We first observe that $\cI^{B_i}_{\sched_i,r}$ is nonempty as
otherwise $\sched_i$ would be
strongly $r$-optimal by Lemma \ref{lemma:freeze-level-r}.
In what follows, we simply write $y_{s,r}$ and $\theta_{s,r}$
rather than $y_{s,r,i}$ and $\theta_{s,r,i}$.
Similarly, $y_{s,r}'$ and $\theta_{s,r}'$ stand for
$y_{s,r,i+1}$ and $\theta_{s,r,i+1}$.
The notations $y_{s,r,\alpha}$ and $\theta_{s,r,\alpha}$
have the same meaning as in the previous sections, i.e.,
$y_{s,r,\alpha}=\sum_{t\in S} P(s,\alpha,t)\cdot y_{t,R}$
and
$\theta_{s,r,\alpha}= \rew(s,\alpha)\cdot y_{s,r} +
  \sum_{t\in S} P(s,\alpha,t)\cdot \theta_{t,R}$
where $R=\min \{\saturation,r{+}\rew(s,\alpha)\}$.
Furthermore, let $\threshold=\threshold_i$ and
$\threshold'=\threshold_{i+1}$.

Part (a) follows immediately from
Lemma \ref{lemma:property-sched-threshold-algorithm}
as we have:
$$
  (\threshold {-}r) \cdot (y_{s,r}-y_{s,r}')
  \ \ \ \leqslant \ \ \
  \theta_{s,r} - \theta_{s,r}'
  \ \ \ \leqslant \ \ \
  (\threshold' {-}r) \cdot (y_{s,r}-y_{s,r}')
$$
As $\threshold < \threshold'$ we obtain
$y_{s,r} \geqslant y_{s,r}'$ for all states $s$.

To prove part (b), we suppose
$\threshold' > \min \cI^{\uparrow}_{\sched,r}$ or
$\cI^{\uparrow}_{\sched,r} \cap I= \varnothing$.
As $\cI^{B_i}_{\sched,r}$ is nonempty
there is at least one state $s$ and action $\alpha \in \Act(s)$ such that
$$
  r + \frac{\theta_{s,r}-\theta_{s,r,\alpha}}{y_{s,r}-y_{s,r,\alpha}}
  \ \ \ < \ \ \
  \threshold'
$$
and $y_{s,r} > y_{s,r,\alpha}$.
Hence:
$$
  \theta_{s,r} - (\threshold'{-}r)\cdot y_{s,r}
  \ \ < \ \
  \theta_{s,r,\alpha} - (\threshold'{-}r)\cdot y_{s,r,\alpha}
$$
As $\sched'=\SchedThreshold{\threshold'}$ and using
Lemma \ref{lemma:property-sched-threshold-algorithm}
we get that there is at least one state $t\in S$ such that
$$
  (y_{t,r},\theta_{t,r})\ \ \ \not= \ \ \
  (y_{t,r}',\theta_{t,r}')
$$
There is some rational number $\threshold$  with
$\sched = \SchedThreshold{\threshold}$
and $\CExp{\sched} > \threshold$.
By Lemma \ref{lemma:property-sched-threshold-algorithm}:
$$
  (\threshold {-}r) \cdot (y_{s,r}-y_{s,r}')
  \ \ \ \leqslant \ \ \
  \theta_{s,r} - \theta_{s,r,i+1}
  \ \ \ \leqslant \ \ \
  (\threshold' {-}r) \cdot (y_{s,r}-y_{s,r}')
$$
for all states $s$.
As $\threshold' > \threshold$ we obtain
$y_{s,r} \geqslant y_{s,r}'$ for all $s\in S$.
Moreover, $y_{s,r}=y_{s,r}'$ implies
$\theta_{s,r}=\theta_{s,r}'$.
But then $y_{t,r} > y_{t,r}'$.

For statement (c), we suppose
$\threshold' < \min \cI^{\uparrow}_{\sched,r}$
and $\sched(\cdot,r)=\sched'(\cdot,r)$.
Clearly, then $\cI_{\sched,r}=\cI_{\sched',r}$ and
$\cI_{\sched,r}^{\uparrow}=\cI_{\sched',r}^{\uparrow}$.
By Lemma \ref{lemma:property-sched-threshold-algorithm}:
\begin{eqnarray*}
    r + \frac{\theta_{s,r}-\theta_{s,r,\alpha}}{y_{s,r}-y_{s,r,\alpha}}
    \ \ \geqslant \ \ \threshold'
    & \ \ & \text{if $y_{s,r} > y_{s,r,\alpha}$}
    \\
    \\[0ex]
    r + \frac{\theta_{s,r}-\theta_{s,r,\alpha}}{y_{s,r}-y_{s,r,\alpha}}
    \ \ \leqslant \ \ \threshold'
    & \ \ & \text{if $y_{s,r} < y_{s,r,\alpha}$}
    \\
    \\[0ex]
    \theta_{s,r} \ \geqslant \ \theta_{s,r,\alpha}
    \ \ \ \ \ \ \
    & \ \ & \text{if $y_{s,r} = y_{s,r,\alpha}$}
\end{eqnarray*}
for all states $s$ and actions $\alpha \in \Act(s)$.
As $\threshold' \geqslant \CExp{\max}$ we get
$\cI_{\sched,r} = \cI_{\sched,r}^{\uparrow}$.

We now use the fact that the
scheduler-improvement step attempts to find the largest value in
$\threshold'' \in \cI^{\uparrow}_{\sched,r}$ such that
$\CExp{\max}\geqslant \threshold''$ by successively running the
threshold algorithm for the values in
$\cI_{\sched,r}^{\uparrow} = \cI_{\sched,r}$
that are still contained in the current interval $I$.
As $\sched'$ has been generated by the threshold algorithm for the
threshold $\threshold'$, which is strictly less than
$\min \cI_{\sched,r}$,
we conclude:
$$
  \CExp{\max} \ \ < \ \ B_{i+1} \ \  < \ \ \min \cI_{\sched,r}
$$
Hence, for each state-action pair $(s,\alpha)$ with
$\alpha \in \Act(s)$ and $y_{s,r} > y_{s,r,\alpha}$ we have:
$$
    r + \frac{\theta_{s,r}-\theta_{s,r,\alpha}}{y_{s,r}-y_{s,r,\alpha}}
    \ \ \ \geqslant \ \ \ \CExp{\max}
$$
For state-action pair $(s,\alpha)$ with
$y_{s,r} < y_{s,r,\alpha}$ we have:
$$
    r + \frac{\theta_{s,r}-\theta_{s,r,\alpha}}{y_{s,r}-y_{s,r,\alpha}}
    \ \ \ \leqslant \ \ \ \threshold' \ \ \ <  \ \ \ \CExp{\max}
$$
This yields:
$$
  \theta_{s,r} - (\CExp{\max}{-}r)\cdot y_{s,r}
  \ \ \geqslant \ \
  \theta_{s,r,\alpha} - (\CExp{\max}{-}r)\cdot y_{s,r,\alpha}
$$
for all states $s$ and actions $\alpha \in \Act(s)$.
But then $\sched_i$ and $\sched_{i+1}$ are strongly $r$-optimal.
As $B_{i+1} <  \min \cI_{\sched,r}$ (see above)
and $\cI_{\sched,r} =  \cI_{\sched',r}$
we get $\cI_{\sched',r}^{B_{i+1}}=\varnothing$.

For statement (d), we first observe that
$\sched_i(\cdot,r)=\sched_{i+1}(\cdot,r)$ implies
$\cI_{\sched_i,r}=\cI_{\sched_{i+1},r}$.
Although $\cI^{\uparrow}_{\sched_{i+1},r}$ can be a proper superset of
$\cI^{\uparrow}_{\sched_i,r}$, a call of the threshold algorithm
for $\threshold_{i+1}=\min \cI^{\uparrow}_{\sched_i,r}$ is only possible
if the calls of the threshold algorithm for the values
$d\in \cI^{\uparrow}_{\sched,r}\setminus \{\threshold_{i+1}\}$
were not successful or have been dropped as $d$ was known to be larger
than $\CExp{\max}$.
Thus, we have $d \geqslant B_{i+1}$ for all
$d\in \cI^{\uparrow}_{\sched,r}\setminus \{\threshold_{i+1}\}$
Hence:
$$
  \cI^{\uparrow}_{\sched_{i+1},r} \cap [A_{i+1},B_{i+1}[ \ \ \ = \ \ \
  \{\threshold_{i+1}\}
$$
Recall that the interval $I$ at the beginning of the
scheduler-improvement step for $\sched_{i+1}$ is
$I_{i+1}=[A_{i+1},B_{i+1}[$
where $A_{i+1}=\CExp{\sched_{i+1}}$.
We now can rely on statement (b) for $\sched_{i+1}$ rather than $\sched_i$
to derive statement (d).
\Ende
\end{proof}

Recall that $\md$ denotes the number of memoryless deterministic
schedulers in $\cM$.

\begin{lemma}
  \label{lemma:anzahl-scheduler-improvement-steps}
    The algorithm performs at most $2 \cdot \saturation \cdot \md$
    scheduler-improvement steps.
\end{lemma}

\begin{proof}
It suffices to show that there are at most $2\cdot \md$
scheduler-improvement steps at each level
$r\in \{0,1,\ldots,\saturation{-}1\}$.

By part (a) of Lemma \ref{lemma:last-scheduler-improvement-at-level-r}
we get that if $r=r_i$ then:
$$
   y_{s,r,i} \ \geqslant \
   y_{s,r,i+1} \ \geqslant \
   y_{s,r,i+2} \ \geqslant \ \ldots \
$$
for each state $s\in S$.
We define the following relation $\lhd = \lhd_r$ on
the set of schedulers
$\sched_i$ where $r_i=r$.
\begin{center}
  $\sched_i \ \lhd \ \sched_j$
  \qquad iff \qquad
  there exists some state $t\in S$ with $y_{t,r,i} > y_{t,r,j}$
\end{center}
Clearly, $\lhd$ is transitive and irreflexive.
Moreover, $\sched_i \lhd \sched_{j}$ implies
$\sched_i(\cdot,r)\not=\sched_{j}(\cdot,r)$.
As a consequence of parts (b), (c) and (d) of
Lemma \ref{lemma:last-scheduler-improvement-at-level-r}
we get that
for each $i$ with $r=r_i=r_{i+1}$:
\begin{center}
 \begin{tabular}{ll}
    &
    $\sched_i \ \lhd \ \sched_{i+1}$
    \\[1ex]
    or \ \ \ &
    $\sched_i(\cdot,r)=\sched_{i+1}(\cdot,r)$ \ and \
    $r=r_{i+2}$ \ and \
    $\sched_i \ \lhd \ \sched_{i+2}$
    \\[1ex]
    or \ \ \ &
    $\sched_i(\cdot,r)=\sched_{i+1}(\cdot,r)$ \ and \
    $r_{i+2}=r{-}1$ (if $r>0$) resp.~the algorithm
    \\
    &
    terminates
    with $\sched_{i+1}$ as a strongly optimal scheduler (if $r=0$)
 \end{tabular}
\end{center}
Each of the function $\sched_i(\cdot,r)$ can be viewed as a
memoryless deterministic scheduler for $\cM$.
The above shows that each memoryless deterministic scheduler for $\cM$
appears at most twice in the sequence
$\sched_{i}(\cdot,r), \sched_{i+1}(\cdot,r), \sched_{i+2}(\cdot,r),
   \ldots$ induced by schedulers
$\sched_j$ where the algorithm performs a scheduler-improvement step
at level $r$.
This yields the claim.
\Ende
\end{proof}

\begin{proposition}[Complexity]
  \label{complexity-CExp-max-algo}
  The algorithm terminates after at most
  $2\cdot \saturation \cdot |S|\cdot |\Act|^{|S|+1}$
  calls of the threshold algorithm.
  The time complexity is exponential in the size of the MDP.
\end{proposition}

\begin{proof}
The statement follows by a combination of
Lemma \ref{lemma:anzahl-scheduler-improvement-steps},
part (b) of Lemma \ref{lemma:scheduler-improvement} and
the results of Section \ref{appendix:complexity-threshold}
and using the fact $\md \leqslant |\Act|^{|S|}$.
\Ende
\end{proof}

\section{PSPACE completeness for acyclic MDPs}

\label{appendix:PSPACE}

We now address the complexity of the four variants of the
threshold problem for maximal conditional expectations in MDPs.
The exponential algorithms for the threshold problems presented
in Section~\ref{sec:threshold} and
 Appendix~\ref{appendix:threshold}
yield an exponential upper bound.
The purpose of this section is to show that the threshold problems are
PSPACE-complete for acyclic MDPs.

We start with the observation that even for acyclic MDPs
history-dependent schedulers can be more powerful to maximize or minimize
conditional expected accumulated rewards.
Obviously, if $\cM$ is acyclic then $\cM$ enjoys conditions (A1) and (A2)
and has no critical scheduler. Thus, the maximal conditional expected
accumulated reward for reaching $\goal$ is finite.

\begin{figure}[ht]
  \includeGastex{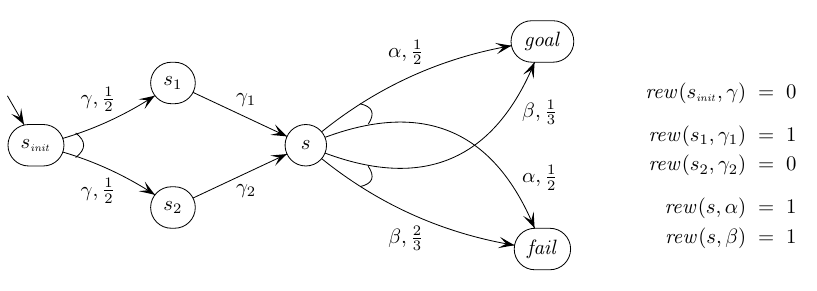}
\caption{MDP $\cM$ for Example \ref{example:HD-vs-det-schedulers}}
\label{fig:MDP-HD-scheduler}
\end{figure}

\begin{example}[History needed in acyclic MDPs]
 \label{example:HD-vs-det-schedulers}
We regard the acyclic MDP shown in Figure \ref{fig:MDP-HD-scheduler}.
The memoryless schedulers that select always $\alpha$ resp.~$\beta$
for state $s$ have the conditional expectation $3/2$:
\begin{eqnarray*}
      \CExp{\sched_{\alpha}} & \ \ = \ \ &
      \frac{\frac{1}{4} \cdot 2 + \frac{1}{4} \cdot 1}{\frac{1}{4}+\frac{1}{4}}
      \ \ \ = \ \ \
      \frac{\ \frac{3}{4} \ }{\frac{1}{2}}
      \ \ \ = \ \ \
      \frac{3}{2}
      \\
      \\[0ex]
      \CExp{\sched_{\beta}} & \ \ = \ \ &
      \frac{\frac{1}{6} \cdot 2 + \frac{1}{6} \cdot 1}{\frac{1}{6}+\frac{1}{6}}
      \ \ \ = \ \ \
      \frac{\ \frac{1}{2} \ }{\frac{1}{3}}
      \ \ \ = \ \ \
      \frac{3}{2}
\end{eqnarray*}
The scheduler $\tsched$ that chooses $\alpha$ for the path
$\fpath_1 = \sinit \, \gamma \, s_1 \, \gamma_1 \, s$
and $\beta$ for the path
$\fpath_2 = \sinit \, \gamma \, s_2 \, \gamma_2 \, s$
has the conditional expectation $8/5$ as we have:
\begin{eqnarray*}
      \CExp{\tsched} & \ \ = \ \ &
      \frac{\frac{1}{4} \cdot 2 + \frac{1}{6} \cdot 1}{\frac{1}{4}+\frac{1}{6}}
      \ \ \ = \ \ \
      \frac{\ \frac{8}{12} \ }{\frac{5}{12}}
      \ \ \ = \ \ \
      \frac{8}{5}
\end{eqnarray*}
Finally, we regard the scheduler $\usched$ that chooses
$\beta$ for $\fpath_1$ and $\alpha$ for $\fpath_2$.
Its conditional expectation is:
\begin{eqnarray*}
      \CExp{\tsched} & \ \ = \ \ &
      \frac{\frac{1}{6} \cdot 2 + \frac{1}{4} \cdot 1}{\frac{1}{6}+\frac{1}{4}}
      \ \ \ = \ \ \
      \frac{\ \frac{7}{12} \ }{\frac{5}{12}}
      \ \ \ = \ \ \
      \frac{7}{5}
\end{eqnarray*}
As the maximal conditional expectation is achieved by a deterministic scheduler
(see Section~\ref{sec:reward-based}) we get:
$$
     \frac{7}{5} \ \ = \ \ \CExp{\min} \ \ = \ \ \CExp{\usched}
     \ \ < \ \
     \underbrace{\CExp{\sched_{\alpha}} \ \ = \ \
                 \CExp{\sched_{\beta}}}_{=\frac{3}{2}}
     \ \ < \ \ \CExp{\tsched} \ \ = \ \ \CExp{\max} \ \ = \ \ \frac{8}{5}
$$
Thus, the history-dependent schedulers are superior when the
task is to maximize or minimize the conditional expectations.
\Ende
\end{example}

\begin{theorem}[PSPACE-completeness for acyclic MDPs]
  \label{theorem:PSPACE-complete-CExp-acyclic}
  All four variants of the
  threshold problem for maximal conditional expectations in acyclic MDPs
  are PSPACE-complete.
\end{theorem}

As PSPACE is closed under complementation, it suffices to consider the
cases where $\threshold$ serves as a strict or non-strict lower bound.
The proof of Theorem \ref{theorem:PSPACE-complete-CExp-acyclic} for lower
bounds ``$\geqslant \threshold$'' resp.~``$> \threshold$''
is splitted into two parts.
The proof for the PSPACE-hardness will be provided  in
Lemma \ref{lemma:acyclic-CExp-PSPACE-hardness}.
Membership to PSPACE will be shown in
Lemma \ref{lemma:CExp-in-PSPACE}.

\begin{lemma}[PSPACE-hardness for acyclic MDPs]
  \label{lemma:acyclic-CExp-PSPACE-hardness}
  The threshold problem for maximal conditional expectations
  in acyclic MDPs is PSPACE-hard.
\end{lemma}

\begin{proof}
We first address the case where the given threshold value $\threshold$
is a non-strict lower bound for the maximal conditional expectations.
We provide a polynomial reduction from the problem
\begin{center}
 \begin{tabular}{lp{10cm}}
   given: &
   an acyclic MDP $\cN$ with initial state $s_0$ and a trap state $\final$
   such that
   $\Pr^{\min}_{\cN,s_0}(\Diamond \final)=1$
   and a natural number $R$
   \\[1ex]

   question: &
   does $\Pr^{\max}_{\cN,s_0}(\Diamond^{\geqslant R}\final)
        \, \geqslant \, \frac{1}{2}$ hold ?
 \end{tabular}
\end{center}
PSPACE-completeness of the above problem has been shown by
Haase and Kiefer (Theorem 7 in \cite{HaaseKiefer15}).
In the following, we provide a polynomial reduction that transforms $\cN$ into
an acyclic MDP $\cM$ with non-negative rational rewards
and threshold $\threshold \in \Rational$
such that
$$
  \Pr^{\max}_{\cN,s_0}(\Diamond^{\geqslant R}\final) \ \geqslant \ \frac{1}{2}
  \qquad \text{iff} \qquad
  \CExpState{\max}{\cM,\sinit}
       (\, \accdiaplus \goal \, | \, \Diamond \goal \, )
  \ \geqslant \ \threshold
$$
In Section \ref{section:rational-rewards} we explain a general approach
for transforming
MDPs with rational rewards into MDPs with integer rewards of the same
asymptotic size.
An alternative approach will be sketched at the end of the proof.

In the sequel, let $S_{\cN}$ be the state space of $\cN$,
$\Act_{\cN}$ the action set,
$P_{\cN} : S_{\cN}\times \Act_{\cN}\times S_{\cN} \to [0,1]$
the transition probability function,
$\rew_{\cN} : S_{\cN}\times \Act_{\cN} \to \Nat$ the reward function
and $s_0\in S_{\cN}$ the initial state of $\cN$.
Furthermore, there is a distinguished trap state $\final \in S_{\cN}$
(i.e., $\Act_{\cN}(\final)=\varnothing$) with
$\Pr^{\min}_{\cN,s_0}(\Diamond \final)=1$, i.e.,
$\Pr^{\tsched}_{\cN,s_0}(\Diamond \final)=1$ for all schedulers
$\tsched$ for $\cN$. As $\cN$ is acyclic this means that all maximal
paths in $\cN$ end in state $\final$.

The idea of the reduction is to define $\cM$ as the MDP that results from
the given MDP $\cN$ by adding a fresh starting state $\sinit$,
an auxiliary state $t$ and two trap states $\goal$ and $\fail$.
In the initial state $\sinit$, $\cM$ behaves purely probabilistically and moves
with probability $p$ to the new state $t$ and with probability $1{-}p$ to
the initial state $s_0$ of $\cN$. The reward of the initialization step is 0.
From state $t$, $\cM$ moves deterministically to state $\goal$ while
earning some reward $T$.
The probability value $p$ will be chosen in such a way that
the conditional expectation of each scheduler for $\cM$ is between
$T{-}\frac{1}{4}$ and $T{+}\frac{1}{4}$.
As soon as $\cN$ has reached $s_0$, $\cM$ behaves as $\cN$.
For the final state $\final$ of $\cN$, $\cM$ offers
$K{+}1$ actions, called $\reject$ and $\accept_0,\accept_1,\ldots,\accept_K$.
The number $K$ will be chosen in such a way that
$2^{K-1} \leqslant \rew_{\cN}(\fpath) -R < 2^K$
for all paths $\fpath$ in $\cN$.
With action $\reject$, $\cM$ moves with probability 1 and
reward 0 to state $\fail$. Intuitively, the reject action should be the
optimal one for all paths from $\sinit$ to $\final$ with reward
less than $R$.
The actions $\accept_i$ for $0 \leqslant i < K$ are probabilistic and
lead from $\final$ to $\goal$ with probability $\lambda^{K-i}$
and with probability $1{-}\lambda^{K-i}$ to state $\fail$ for some rational
value $\lambda \in ]0,1[$.
The reward of $\accept_i$ is some value $X_i$.
When selecting $\accept_K$ in state $\final$,
$\cM$ moves deterministically to state $\goal$, while earning reward $X_K$.
The parameter $\lambda$
and the reward values $X_0,\ldots,X_K$ will be chosen in such a way
that action $\accept_i$ is optimal for all paths $\fpath$
from $\sinit$ to $\final$ with
$R{-}1{+}2^i \leqslant \rew_{\cM}(\fpath) < R{-}1{+}2^{i+1}$.

Assuming that appropriate values $p,\lambda,T,X_0,\ldots,X_K$
have been defined, the formal definition of $\cM$ is as follows.
The state space of $\cM$ is
$$
  S_{\cM} \ \ = \ \ S_{\cN} \cup \bigl\{ \goal,\fail,\sinit,t\bigr\}
$$
The action set of $\cM$ is
$$
  \Act_{\cM} \ \ = \ \
  \Act_{\cN} \cup \bigl\{ \tau, \reject, \accept_0,\ldots,\accept_K\bigr\}
$$
We have $\Act(\sinit)=\Act(t)=\{\tau\}$ and
$$
 \begin{array}{ll}
   P_{\cM}(\sinit,\tau,t) \ = \ p, \ \ \ \ \
   P_{\cM}(\sinit,\tau,s_0) = 1{-}p, \ \ \ \ \ &
   \rew_{\cM}(\sinit,\tau) \ = \ 0
 \end{array}
$$
and
$$
   P_{\cM}(t,\tau,\goal) \ = \ 1, \qquad
   \rew_{\cM}(t,\tau) \ = \ T
$$
For the states $s\in S_{\cN} \setminus \{\final\}$ we have
$\Act_{\cM}(s)=\Act_{\cN}(s)$ and $P_{\cM}(s,\alpha,t)=P_{\cN}(s,\alpha,t)$
and $\rew_{\cM}(s,\alpha)=\rew_{\cN}(s,\alpha)$
for all states $t\in S_{\cN}$ and actions $\alpha \in \Act_{\cN}(s)$.
States $\goal$ and $\fail$ are trap states in $\cM$, i.e.,
$\Act_{\cM}(\goal)=\Act_{\cM}(\fail)=\varnothing$.
For the state $\final$ we have
$$
  \Act_{\cM}(\final) \ \ = \ \
  \bigl\{\reject \bigr\} \ \cup \
  \bigl\{\accept_0,\accept_1,\ldots,\accept_K\bigr\}
$$
and $P_{\cM}(\final,\reject,\fail)=1$, $\rew_{\cM}(\final,\reject)=0$
and
$$
 \begin{array}{lll}
   \begin{array}{lcl}
      P_{\cM}(\final,\accept_i,\goal) & = & \lambda^{K-i} \\[0.5ex]
      P_{\cM}(\final,\accept_i,\fail) & = & 1{-}\lambda^{K-i}
   \end{array}
  & \ & 
    \begin{array}{lcl}
       \rew_{\cM}(\final,\accept_i) & = & X_i  
    \end{array}
 \end{array}
$$
for $i=0,1,\ldots,K$.
In all remaining cases, we have $P_{\cM}(\cdot)=0$.

\tudparagraph{2ex}{{\it Some auxiliary notations and choice of $\lambda$.}}
For $i \in \Nat$ we define
$$
  R_i \ \ = \ \ R-1+2^i
$$
Note that $R_0=R$ and $R_{i+1} \, = \, R_0 + 2^{i+1}\, = \, R_i+2^i$.
Let
$$
  E \ \ = \ \
  \sum_{\stackrel{s\in S}{s\not= \final}} \rew^{\max}_{\cN}(s)
$$
where
$$
  \rew^{\max}_{\cN}(s) \ \ = \ \
  \max \ \bigl\{ \
           \rew_{\cN}(s,\alpha) \ : \ \alpha \in \Act_{\cN}(s) \
         \bigr\}
$$
Clearly, $\rew_{\cN}(\fpath) \leqslant E$ for all paths $\fpath$ in $\cN$.
W.l.o.g. $R < E$. Let $K\in \Nat$ such that
$$
  2^{K-1} \ \ \leqslant \ \ E-R \ \ < \ \ 2^K
$$
For each state $s$ in $\cN$, let  $m_s$ be the least common multiple of
the denominators of the probability values $P_{\cN}(s,\alpha,t)$ for some
$\alpha \in \Act_{\cN}(s)$ and $\alpha$-successor $t$ of $s$.
Let
$$
   m \ \ = \ \ \prod_{\stackrel{s\in S_{\cN}}{s\not= \final}} m_s
$$
Then, $m$ is a natural number and the number of digits of $m$ in a binary
(or decimal) encoding is bounded by the size of $\cN$.
Moreover, as $\cN$ is acyclic, for each maximal path $\fpath$ from $\sinit$ to
$\final$, the probability $\probability(\fpath)$ can be written in the form
$\ell/m$ for some natural number $\ell$. This yields that for any deterministic
scheduler $\sched$ for $\cN$ we have:
\begin{equation}
  \label{prob-k-m}
  \Pr^{\sched}_{\cN,\sinit}(\Diamond^{\geqslant R} \final)
  \ \in \ \left\{ \ \frac{k}{m} \ : k \in \Nat \ \right\}
  \tag{*}
\end{equation}
The value $\lambda$ is defined by:
$$
  \lambda  \ \ \ = \ \ \ 1 - \frac{1}{2 K m}
$$
By the Bernoulli inequality:
$$
  \lambda^K \ \ \ = \ \ \ \Bigl( 1 - \frac{1}{2 K m} \Bigr)^K
  \ \ \ \geqslant \ \ \
  1 - \frac{K}{2 K m} \ \ \ = \ \ \
  1 - \frac{1}{2 m}
$$

\tudparagraph{2ex}{{\it Reward parameters $X_0,X_1,\ldots,X_K$.}}
The reward values $X_0, X_1,\ldots,X_K$ are defined inductively by:
$$
  X_i \ \ \ = \ \ \
  (1{-}\lambda )\cdot ( X - 2^i+1 ) \ \ + \ \ \lambda X_{i-1}
  \qquad
  \text{for $i=1,\ldots,K$}
$$
where the choice of $X_0=X$ will be explained below.
For the following statements up to and including Claim 2,
it suffices to deal with any value $X$ such that
$X \geqslant 2^K$ and $X \geqslant 2Km = 1/(1{-}\lambda)$.
Note that the constraint $X \geqslant 2^K$ yields
$$
  X_0 \ = \ X \ > \ X_1 \ > \ X_2 \ > \ \ldots \ > \ X_K \ > \ 0.
$$
Moreover:
\begin{eqnarray*}
   X_i & \ = \ &
   \lambda^i X \ \ + \ \
   (1-\lambda) \cdot \sum_{j=0}^i \lambda^j (X-2^{i-j}+1)
   \\
   \\[0ex]
   & = &
   \lambda^i X \ \ + \ \
   (1-\lambda) \cdot \sum_{j=0}^i \lambda^j (X+1)
   \ \ - \ \
   (1-\lambda) 2^i \cdot \sum_{j=0}^i \left(\frac{\lambda}{2}\right)^j
   \\
   \\[0ex]
   & = &
   \lambda^i X \ \ + \ \
   (1-\lambda) (X+1) \cdot \frac{1-\lambda^{i+1}}{1-\lambda}
   \ \ - \ \
   (1-\lambda) 2^i \cdot
    \frac{1-\left(\frac{\lambda}{2}\right)^{i+1}}
         {1-\frac{\lambda}{2}}
   \\
   \\[0ex]
   & = &
   \lambda^i X \ \ + \ \ \bigl(1-\lambda^{i+1}\bigr)(X+1)
   \ \ - \ \
   \frac{1-\lambda}{2-\lambda}\cdot \bigl( 2^{i+1}-\lambda^{i+1}\bigr)
   \\
   \\[0ex]
   & = &
   X \ \ + \ \ \lambda^i(1-\lambda) X \ \ - \ \ \lambda^{i+1}
   \ \ - \ \
   \frac{1-\lambda}{2-\lambda}\cdot \bigl( 2^{i+1}-\lambda^{i+1}\bigr)
   \\
   \\[0ex]
   & = &
   X \ \ + \ \ \lambda^i(1-\lambda) X
   \ \ - \ \
   \frac{1-\lambda}{2-\lambda}\cdot  2^{i+1}
   \ \ - \ \ \frac{\lambda^{i+1}}{2-\lambda}
\end{eqnarray*}
Recall that we require $X \geqslant 2Km$.
This implies:
$$
  X \ \ \ > \ \ \
  \frac{1-\frac{1}{2Km}}{1+\frac{1}{2Km}} \cdot 2Km
  \ \ \ = \ \ \
  \frac{\lambda}{2{-}\lambda} \cdot \frac{1}{1{-}\lambda}
$$
Hence, $(1{-}\lambda) X \ > \ \lambda/(2{-}\lambda)$. Therefore:
$$
 \lambda^i(1-\lambda) X \ \ > \ \ \frac{\lambda^{i+1}}{2-\lambda}
$$
This yields
$$
   X_i \ \ \ > \ \ \ X \ \ - \ \ \frac{1-\lambda}{2-\lambda}\cdot  2^{i+1}
$$
and therefore $X_i + 2^i > X$ as we have:
\begin{eqnarray*}
  X_i + 2^i
  & \ \ > \ \ &
  X \ \ - \ \ \frac{1-\lambda}{2-\lambda}\cdot  2^{i+1} \ \ + \ \ 2^i
  \\
  \\[0ex]
  & = &
  X \ \ + \ \
    2^i \cdot \Bigl( \ 1 \ - \ 2\cdot \frac{1 -\lambda}{2-\lambda} \ \Bigr)
  \\
  \\[0ex]
  & = &
  X \ \ + \ \ 2^i \cdot \frac{2-\lambda - 2 + 2\lambda}{2-\lambda}
  \\
  \\[0ex]
  & = &
  X \ \ + \ \ \frac{2^{i+1} \lambda}{2-\lambda}
  \ \ \ > \ \ \ X
\end{eqnarray*}
Furthermore, the inductive definition of the values $X_i$ yields:
$$
  \Delta_i \ \ \ \eqdef \ \ \
  \frac{\lambda^{K-i} X_i - \lambda^{K-i+1} X_{i-1}}
       {\lambda^{K-i}-\lambda^{K-i+1}}
  \ \ \ = \ \ \
  \frac{X_i - \lambda X_{i-1}}{1-\lambda}
  \ \ \ = \ \ \
  X-2^i+1
$$

\tudparagraph{2ex}{{\it Definition of the parameters of the initialization.}}
We define
$$
   T \ \ = \ \ R+X-\frac{1}{2}
$$
and choose a rational number $p \in ]0,1[$ such that
$$
   pT \ \ \ \geqslant \ \ \ T-\frac{1}{4} \ \ \ = \ \ \ R+X-\frac{3}{4}
$$
and
$$
   \frac{1}{p} \cdot \Bigl( \ p T  \ + \ (1{-}p) (R + 2^K + X) \ \Bigr)
   \ \ \  \leqslant \ \ \ T+\frac{1}{4} \ \ \ = \ \ \ R+X-\frac{1}{4}
$$
The latter constraint is equivalent to
$$
  \frac{1-p}{p} \cdot (R + 2^K + X) \ \ \ \leqslant \ \ \ \frac{1}{4}
$$
which again is equivalent to
$$
 \frac{1}{p} \cdot (R + 2^K + X) \ \ \ \leqslant \ \ \
 \frac{1}{4} \ + \ R+2^K+X
$$
For example, we can deal with
$$
  p \ \ \ = \ \ \
  \max \
  \left\{ \ \ \
     \frac{R+2^K+X}{\frac{1}{4} + R+2^K+X}, \ \ \
     \frac{R+X-\frac{3}{4}}{R+X-\frac{1}{2}} \ \ \
  \right\}
$$
Note that then indeed $0< p < 1$.
Moreover, the choice of $p$ ensures that:
$$
  T-\frac{1}{4} \ \ \ \leqslant \ \ \
  \CExp{\sched} \ \ \ \leqslant \ \ \
  T+\frac{1}{4}
$$
for each scheduler $\sched$ for $\cM$
with $\Pr^{\sched}_{\cM,\sinit}(\Diamond \goal)>0$.
In particular:
\begin{eqnarray*}
  \CExp{\max} & \ \ \leqslant \ \ & T + \frac{1}{4} \ \ \ = \ \ \
  R + X - \frac{1}{4} \ \ \ < \ \ \ R+X
  \\
  \\[0ex]
  \CExp{\min} & \ \ \geqslant \ \ & T - \frac{1}{4} \ \ \ = \ \ \
  R+X-\frac{3}{4} \ \ \ > \ \ \ R+X-1
\end{eqnarray*}

\tudparagraph{2ex}{{\it Optimal decisions in the final state.}}
As before, if $\sched$ is a scheduler for $\cM$ then we write
$\CExp{\sched}$ for
$\CExpState{\sched}{\cM,\sinit}(\accdiaplus \goal |\Diamond \goal)$.
We use here the fact that
$\rew_{\cM}(\fpath) \leqslant E \leqslant R+2^K$ for each
path from $\sinit$ to $\final$ in $\cM$ and $X_i \leqslant X$ for
$i=0,1,\ldots,K$.
By the choice of the reward values $X_0,\ldots,X_K$, we obtain that
action $\accept_i$ is optimal for each finite path $\fpath$ from
$\sinit$ to $\final$ with $R_i \leqslant \rew_{\cM}(\fpath) < R_{i+1}$.
To see this, we rely on Lemma \ref{lemma:rho-theta-x-y}
and the observation:
$$
  \underbrace{\rew_{\cM}(\fpath)}_{\geqslant R_i} \ + \
  \Delta_i
  \ \ \ \geqslant \ \ \
  R_i + X-2^i +1
  \ \ \ = \ \ \
  R+X \ \ \ > \ \ \ \CExp{\max}
$$
Thus, if $\rew_{\cM}(\fpath) \geqslant R_i$ then action $\accept_i$
yields a better (larger) conditional expectation than $\accept_{i-1}$.
Likewise, for the paths $\fpath$ from $\sinit$ to $\final$
with $\rew(\fpath) < R_i$, action $\accept_{i-1}$ is better than
$\accept_i$ as we have:
$$
  \underbrace{\rew_{\cM}(\fpath)}_{\leqslant R_i{-}1} \ + \
  \Delta_i
  \ \ \ \leqslant \ \ \
  R_i -1 + X-2^i+1
  \ \ \ = \ \ \
  R+X-1 \ \ \ < \ \ \ \CExp{\min}
$$
Action $\reject$ is the optimal one for exactly the paths $\fpath$
from $\sinit$ to $\goal$ with $\rew(\fpath) < R$ as we have:
$$
   \underbrace{\rew_{\cM}(\fpath)}_{\leqslant R{-}1} \ + \
   \frac{\lambda^{K-i} X_i - 0}{\lambda^{K-i}-0}
   \ \ \ \leqslant \ \ \
   R-1 + X_i \ \ \ \leqslant  \ \ \ R+X-1
   \ \ \ < \ \ \ \CExp{\min}
$$
Let $\SchedOpt$ denote the class of scheduler
$\sched$ for $\cM$ such that for each finite paths $\fpath$ from
$\sinit$ to $\final$ we have:
$$
  \sched(\fpath) \ \ \ = \ \ \
  \left\{
    \begin{array}{lcl}
      \accept_i & : & \text{if $R_i \leqslant \rew_{\cM}(\fpath) < R_{i+1}$}
      \\[0.8ex]
      \reject & : &  \text{if $\rew_{\cM}(\fpath) < R$}
    \end{array}
  \right.
$$
The above shows that for each scheduler $\usched$ for $\cM$ there is
a scheduler $\sched \in \SchedOpt$ with
$\CExp{\sched}\geqslant \CExp{\usched}$.

\tudparagraph{1ex}{{\it Functions $f$ and $g$.}}
We define functions $f, g : [0,1] \to \Real$ as follows:
\begin{eqnarray*}
  f(q) & \ \ = \ \ &
  \frac{ pT \ + \ (1{-}p) q \lambda^K (R+X)}{p \ + \ (1{-}p)q \lambda^K}
  \\
  \\[0ex]
  g(q) & \ = \ &
  \frac{ pT \ + \ (1{-}p) q (R+2^K+X_K)}{p \ + \ (1{-}p)q}
\end{eqnarray*}
Both $f$ and $g$ are strictly monotonous as their first derivatives
$f'$ and $g'$ are positive. Note that
$$
  \text{if} \ \ h(q) = \frac{a + bq}{c + dq}
  \ \ \text{then} \ \
  h'(q) \ = \
  \frac{bc - ad}{(c+dq)^2}
$$
In the case of $f$, we deal with
$a= pT$, $b = (1{-}p)\lambda^K (R{+}X)$, $c = p$ and
$d = (1{-}p)\lambda^K$. Then:
$$
 \begin{array}{lcl}
  bc - ad & \ \ = \ \ &
   (1{-}p)\lambda^K (R{+}X)p \ - \ pT(1{-}p) \lambda^K
  \\
  \\[-1ex]
  & \ \ = \ \ &
   p(1{-}p)\lambda^K (R{+}X {-} T)
  \\
  \\[-1ex]
  & \ \ = \ \ &
   p(1{-}p)\lambda^K \cdot \frac{1}{2} \ \ > \ \ 0
 \end{array}
$$
In the case of $g$ we deal with
$a= pT$, $b = (1{-}p) (R{+}2^K{+}X_K)$, $c = p$ and
$d = 1{-}p$. Then:
$$
 \begin{array}{lcl}
  bc - ad & \ \ = \ \ &
  (1{-}p) (R{+}2^K{+}X_K)p \ - \ pT(1{-}p)
  \\
  \\[-1ex]
  & \ \ = \ \ &
  p(1{-}p)(R{+}2^K{+}X_K {-} T) \ \ > \ \ 0
\end{array}
$$
Note that $2^K + X_K > X$ and $T = R{+}X{-}\frac{1}{2} < R{+}X$, which yields
$R{+}2^K{+}X_K {-} T > 0$.

\tudparagraph{1ex}{{\it Claim 1:}}
For each scheduler $\sched \in \SchedOpt$ for $\cM$ with
$q \, =\, \Pr^{\sched}_{\cN,s_0}(\Diamond \final)$ we have:
$$
  f(q) \ \ \ \leqslant \ \ \ \CExp{\sched} \ \ \ \leqslant \ \ \ g(q)
$$

\tudparagraph{1ex}{{\it Proof of Claim 1:}}
As before, we write $\CExp{\sched}$ for
$\CExpState{\sched}{\cM,\sinit}(\accdiaplus \goal |\Diamond \goal)$.
Let
$$
  q_r \ \ = \ \ \Pr^{\sched}_{\cM,\sinit}(\Diamond^{=r} \final)
$$
and
$$
  A \ \ = \ \
  \sum_{i=0}^K \ \sum_{r=R_i}^{R_{i+1}-1}
    q_r \cdot \lambda^{K-i} \cdot (r + X_i)
$$
Then:
$$
  \CExp{\sched}
  \ \ \ = \ \ \
  \frac{ pT + (1{-}p)A}
       { p + (1{-}p)(q_0 \lambda^K + q_1\lambda^{K-1} + \ldots + q_K)}
$$
We first show that $f(q) \leqslant \CExp{\sched}$.
Thanks to Lemma \ref{lemma:rho-theta-x-y}, it suffices to show that
for each $i\in \{1,\ldots,K{-}1\}$ and $r \in \Nat$ with
$R_i \leqslant r < R_{i+1}$ we have:
$$
  \frac{\lambda^{K-i} (r+X_i) \ - \ \lambda^{K} (R+X)}
       {\lambda^{K-i} - \lambda^K}
  \ \ \ \geqslant \ \ \ \CExp{\max}
$$
To prove this, we show by induction on $i$:
$$
  \frac{X_i - \lambda^i X}{1-\lambda^i}
  \ \ \ \geqslant \ \ \ X-2^i+1
$$
The case $i=1$ is obvious as
$$
  \frac{X_1 - \lambda X}{1-\lambda}
  \ \ = \ \
  \frac{X_1 - \lambda X_0}{1-\lambda}
  \ \ = \ \ \Delta_1 \ \ = \ \ X-1
$$
Induction step:
\begin{eqnarray*}
  \frac{X_i - \lambda^i X}{1-\lambda^i}
  & \ \ = \ \ &
  \frac{X_i - \lambda X_{i-1}}{1-\lambda^i}
  \ \ + \ \
  \frac{\lambda X_{i-1} - \lambda^i X}{1-\lambda^i}
  \\
  \\[0ex]
  & = &
  \underbrace{\frac{X_i - \lambda X_{i-1}}{1-\lambda}}_{= X{-}2^i+1}
  \cdot \frac{1-\lambda}{1-\lambda^i}
  \ \ \ + \ \ \
  \underbrace{\frac{X_{i-1} - \lambda^{i-1} X}{1-\lambda^{i-1}}}_
       {\geqslant X{-}2^{i-1}+1}
  \cdot \lambda \cdot
  \frac{1-\lambda^{i-1}}{1-\lambda^i}
  \\
  \\[0ex]
  & \geqslant &
  (X-2^i+1) \cdot \frac{1-\lambda}{1-\lambda^i}
  \ \ \ + \ \ \
  (X-2^{i-1}+1) \cdot
   \lambda \cdot \frac{1-\lambda^{i-1}}{1-\lambda^i}
  \\
  \\[0ex]
  & = &
  (X-2^{i-1}+1) \cdot
  \frac{1 - \lambda + \lambda(1-\lambda^{i-1})}{1-\lambda^i}
  \ \ - \ \
  2^{i-1} \cdot \underbrace{\frac{1-\lambda}{1-\lambda^i}}_{\leqslant 1}
  \\
  \\[-2ex]
  & \geqslant &
  (X-2^{i-1}+1) - 2^{i-1} \ \ \ = \ \ \
   X-2^i+1
\end{eqnarray*}
As $R_i \leqslant r < R_{i+1}$, we have $r \geqslant R+2^i-1$.
We obtain:
\begin{eqnarray*}
  \frac{\lambda^{K-i} (r+X_i) \ - \ \lambda^{K} (R+X)}
       {\lambda^{K-i} - \lambda^K}
  & \ \ = \ \ &
  \frac{(r+X_i) \ - \ \lambda^i (R+X)}
       {1 - \lambda^i}
  \\
  \\[0ex]
  & \geqslant &
  \frac{(R+2^i-1 + X_i) \ - \ \lambda^i (R+X)}
       {1 - \lambda^i}
  \\
  \\[0ex]
  & \geqslant &
  R \ \ + \ \ \frac{2^i-1}{1-\lambda^i} \ \ + \ \
  \frac{X_i \ - \ \lambda^i X}
       {1 - \lambda^i}
  \\
  \\[0ex]
  & \geqslant &
  R + 2^i -1 + X - 2^i + 1
  \\[2ex]
  & = & R+X \ \ \ > \ \ \ \CExp{\max}
\end{eqnarray*}
This yields $f(q) \leqslant \CExp{\sched}$.
The proof of the statement $g(q) \geqslant \CExp{\sched}$ is analogous.
We first observe that for $i < K$:
$$
  \frac{X_K - \lambda^i X_{K-i}}{1-\lambda^i}
  \ \ \ \geqslant \ \ \ X - 2^{K-i+1} +1
$$
Thus, if $R_{K-i} \leqslant r < R_{K-i+1} = R-1 +2^{K-i+1}$ then
\begin{eqnarray*}
  & &
  \frac{(R+2^K + X_K) - \lambda^{i} (r+X_{K-i})}
       {1 - \lambda^{i}}
  \\
  \\[0ex]
  & \ > \ &
  \frac{(R+2^K+X_k) \ - \ \lambda^{i} (R-1+2^{K-i+1}+X_{K-i})}
       {1 - \lambda^{i}}
  \\
  \\[0ex]
  & = &
  \frac{(R + 2^{K-i+1}) \ - \ \lambda^i (R  +  2^{K-i+1})}{1-\lambda^i}
  \ \ + \ \
  \underbrace{\frac{X_K - \lambda^{i} X_{K-i}}{1-\lambda^{i}}}_{
     \geqslant X-2^{K-i+1}+1}
  \ \ + \ \
  \underbrace{\frac{2^K - 2^{K-i+1} + \lambda^i}{1-\lambda^{i}}}_{\geqslant 0}
  \\[-1ex]
  & \geqslant &
  R +  2^{K-i+1} %
  \ +  \
  X - 2^{K-i+1} + 1
  \\
  \\[0ex]
  & \ > \ & R+X \ \ > \ \ \CExp{\max}
\end{eqnarray*}
By Lemma \ref{lemma:rho-theta-x-y},
we obtain $g(q) \geqslant \CExp{\max}$.

\tudparagraph{2ex}{{\it Definition of the threshold value.}}
The threshold $\threshold$ for conditional expectations
in $\cM$ is defined as follows:
$$
  \threshold \ \ \ = \ \ \
  f\Bigl(\frac{1}{2}\Bigr) \ \ \ = \ \ \
  \frac{ pT \ + \ (1{-}p) \frac{1}{2} \lambda^K (R+X)}
       { p \ + \ (1{-}p)\frac{1}{2} \lambda^K}
$$
We now prove the first part of the soundness of the reduction.
The precise value of $X$ is still irrelevant, except that we require
$X \geqslant 2^K$ and $X \geqslant 2Km$.

\tudparagraph{2ex}{{\it Claim 2:}}
 \ \ \
 $\Pr^{\max}_{\cN,s_0}(\Diamond^{\geqslant R}\final) \ \geqslant \ \frac{1}{2}$
 \quad implies \quad
 $\CExpState{\max}{\cM,\sinit}(\accdiaplus \goal |\Diamond \goal)
  \ \geqslant \ \threshold$

\tudparagraph{2ex}{{\it Proof of Claim 2.}}
Suppose
$\Pr^{\max}_{\cN,s_0}(\Diamond^{\geqslant R}\final) \ \geqslant \ \frac{1}{2}$.
We pick a deterministic scheduler $\tsched$ for $\cN$ with
$$
 q \ \ \eqdef \ \
 \Pr^{\tsched}_{\cN,s_0}(\Diamond^{\geqslant R}\final) \
 \ \geqslant \ \ \frac{1}{2}
$$
Let $\sched$ be the unique scheduler for $\cM$ in $\SchedOpt$ that extends
$\tsched$ by decisions for the paths ending in $\final$.
More precisely, $\sched(\sinit \, \tau \, \fpath) = \tsched(\fpath)$
for each finite path $\fpath$ in $\cN$ from $s_0$ to some state
$s\in S_{\cN}\setminus \{\final\}$.
For the $\sched$-paths $\fpath$ from
$\sinit$ to $\final$ we have $\sched(\fpath) = \accept_i$ if
$R_i \leqslant \rew_{\cM}(\fpath) < R_{i+1}$ and
$\sched(\fpath)=\reject$ if $\rew_{\cM}(\fpath) < R$.
By Claim 1 and the monotonicity of $f$ we get:
$$
  \CExp{\sched} \ \ \ \geqslant \ \ \ f(q)
  \ \ \ \geqslant \ \ \ f\Bigl(\frac{1}{2}\Bigr) \ \ \ = \ \ \ \threshold
$$
Hence,  $\CExp{\max} \geqslant  \threshold$.

\tudparagraph{2ex}{{\it Definition of the reward value $X$.}}
While the arguments presented so far hold for any value $X$,
an adequate choice of $X$ is crucial for
the reverse implication. We define:
$$
  X \ \ = \ \ m \cdot 2^K
$$
Note that $X$ meets the constraints
$X \geqslant 2^K$ and $X \geqslant 2Km = 1/(1{-}\lambda)$
that have been required before.
We then have $X_0=X > X_1 > \ldots > X_K >0$.

\tudparagraph{2ex}{{\it Claim 3:}}
 \ \ \
 $\CExpState{\max}{\cM,\sinit}(\accdiaplus \goal |\Diamond \goal)
  \ \geqslant \ \threshold$
 \quad implies \quad
  $\Pr^{\max}_{\cN,s_0}(\Diamond^{\geqslant R}\final)
    \ \geqslant \ \frac{1}{2}$

\tudparagraph{1ex}{{\it Proof of Claim 3.}}
We now suppose that $\CExp{\max} \geqslant \threshold$.
Let $\sched$ be a scheduler with
$\CExp{\sched} \geqslant \threshold$.
We may suppose w.l.o.g. that $\sched \in \SchedOpt$.
The goal is to show that
$$
  q \ \ \eqdef \ \
  \Pr^{\sched}_{\cN,s_0}(\Diamond \final)
  \ \ \geqslant \ \ \frac{1}{2}
$$
We suppose by contradiction that $q < \frac{1}{2}$.
By \eqref{prob-k-m} we obtain:
$$
  q \ \ \leqslant \ \ \frac{1}{2} - \frac{1}{m}
$$
Recall that by Claim 1, we have $\CExp{\sched} \leqslant g(q)$.
Using the monotonicity of $f$ and $g$ it suffices to show that
\begin{center}
  $f\bigl(\frac{1}{2} \bigr) \ \ \geqslant \ \
   g\bigl(\frac{1}{2}-\frac{1}{m}\bigr)$
\end{center}
Again we can rely on Lemma \ref{lemma:rho-theta-x-y}.
By the choice of $\lambda$ we have:
$$
  \frac{1}{2}\lambda^K \ \ \ \geqslant \ \ \
  \frac{1}{2} \Bigl( 1 - \frac{1}{2m} \Bigr)
  \ \ \ = \ \ \
  \frac{1}{2} - \frac{1}{4m}
  \ \ \ > \ \ \
  \frac{1}{2} - \frac{1}{m}
$$
Hence (by Lemma \ref{lemma:rho-theta-x-y}), the task is show that:
$$
  \frac{1}{2}\lambda^K (R+X) \ \ - \ \
  \Bigl( \frac{1}{2}-\frac{1}{m} \Bigr) \cdot (R+2^K+X_K)
  \ \ \  > \ \ \
  \Bigl( \frac{1}{2}\lambda^K - \frac{1}{2}-\frac{1}{m} \Bigr)
  \cdot \threshold
$$
Obviously, this is equivalent to the following statement:
$$
  \frac{1}{2}\lambda^K (R+X-\threshold)
  \ \ \ > \ \ \
  \Bigl( \frac{1}{2}-\frac{1}{m} \Bigr) \cdot (R+2^K +X_K -\threshold)
$$
As $X_K \leqslant X$ and
$\frac{1}{2} \lambda^K \geqslant \frac{1}{2} -\frac{1}{4m}$ (see above),
it suffices to show that
$$
  \Bigl(\frac{1}{2} - \frac{1}{4m} \Bigr) \cdot X
  \ \ \ > \ \ \
  \Bigl( \frac{1}{2}-\frac{1}{m} \Bigr) \cdot (2^K +X)
$$
which is equivalent to the statement
$$
  \Bigl(\frac{1}{m} - \frac{1}{4m} \Bigr) \cdot X
  \ \ \ > \ \ \
  \Bigl( \frac{1}{2}-\frac{1}{m} \Bigr) \cdot 2^K
$$
Indeed we have:
$$
   \Bigl(\frac{1}{m} - \frac{1}{4m} \Bigr) \cdot X
   \ \ \ = \ \ \
   \frac{1}{m} \cdot \frac{3}{4} \cdot m \cdot 2^{K}
   \ \ \ > \ \ \ 2^{K-1} \ \ \ > \ \ \
   \Bigl( \frac{1}{2}-\frac{1}{m} \Bigr) \cdot 2^K
$$
This completes the proof of Claim 3.

\tudparagraph{2ex}{{\it Size of the generated MDP.}}
The size of the graph of $\cM$ is linear in the size of the graph structure
of $\cN$.
It remains to check that the length of the binary encoding of
all parameters $p,\lambda,T,X_0,\ldots,X_K$ are polynomially bounded in
the size of $\cN$.
This is indeed the case as the logarithmic lengths of $m$ and $E$
are polynomially bounded in the size of $\cN$ and
$K = \lfloor \log (E{-}R) \rfloor +1$. (Recall that we suppose $E> R$.)
Note that the number of digits required to represent $m$ is bounded by
$\ell \cdot |S_{\cN}| \cdot |\Act_{\cN}|$
where $\ell$ is the logarithmic length
of the largest reward value in $\cN$.

\tudparagraph{2ex}{{\it Transforming $\cM$ into an MDP with integer rewards.}}
In the presented construction, the
reward values $T,X_1,\ldots,X_K$ are non-negative rational numbers.
Section \ref{section:rational-rewards} presents a polynomial
transformation of MDPs with rational rewards into MDPs with integer rewards.
In the setting here, we can rely on an alternative approach. This makes
use of the fact that the non-integer rational rewards only appear
for the state-action pairs $(t,\tau)$ and $(\final,\accept_i)$ that lead to
a trap state.
(Recall that all reward values in the given MDP
$\cN$ are natural numbers.)
In particular, there are no nondeterministic choices in $\cM$
after firing transitions with non-integer rewards.

Instead of moving with probability 1 from $t$ to $\goal$ while earning reward
$T$ we can redefine the transition probabilities and
reward value for $(t,\tau)$ by
$$
   P_{\cM}(t,\tau,\goal) \ = \ \frac{1}{T}, \qquad
   P_{\cM}(t,\tau,t) \ = \ \frac{T{-}1}{T}, \qquad
   \rew_{\cM}(t,\tau) \ = \ 1
$$
Then, the expected total reward for the path fragments from $t$
to state $\goal$ in $\cM$ equals $T$.
Note that with $T = \frac{k}{\ell}$ we have:
\begin{eqnarray*}
  \sum_{i=0}^{\infty}
     \left( \frac{k{-}\ell}{\ell} \right)^i \cdot \frac{\ell}{k} \cdot (i{+}1)
  & \ \ = \ \ &
  \frac{\ell}{k} \cdot \frac{1}{\bigl( 1 - \frac{k{-}\ell}{k}\bigr)^2}
  \ \ \ = \ \ \
  \frac{\ell}{k}\cdot \frac{1}{\bigl( \frac{\ell}{k} \bigr)^2}
  \ \ \ = \ \ \
  \frac{k}{\ell}
\end{eqnarray*}
The treatment of $(\final,\accept_i)$ is analogous.
We can introduce a fresh state $t_i$ such that $\cM$ moves from $\final$
to $t_i$ with probability 1 and reward 0. The behavior in state $t_i$ is
purely probabilistic (say the enabled action is $\tau$) where the
transition probabilities and the reward value are given by:
$$
   P_{\cM}(t_i,\tau,\goal) \ = \ \frac{\lambda^{K-i}}{X_i}, \qquad
   P_{\cM}(t_i,\tau,\fail) \ = \ \frac{1{-}\lambda^{K-i}}{X_i}, \qquad
   P_{\cM}(t_i,\tau,t_i) \ = \ \frac{X_i{-}1}{X_i}
$$
and $\rew_{\cM}(t,\tau)=1$.

\tudparagraph{2ex}{{\it
   Strict lower bound for the maximal conditional expectations.}}
We now turn to the PSPACE-hardness of the question
\begin{center}
  \begin{tabular}{l}
    ``does
      $\CExpState{\max}{\cM,\sinit}(\accdiaplus \goal | \Diamond \goal)
         \, > \, \threshold$ hold?''
   \end{tabular}
\end{center}
We can rely on the strict monotonicity of functions $f$ and $g$
to adapt the statements and proofs of Claims 2 and 3
for the strict bound $f(\frac{1}{2}-\frac{1}{2m})$ rather than the
non-strict bound $f(\frac{1}{2})$.
With \eqref{prob-k-m} we obtain:
\begin{eqnarray*}
  & \ \ \ \ &
  \Pr^{\max}_{\cN,s_0}(\Diamond^{\geqslant R} \final)
  \ \ \geqslant \ \ \frac{1}{2}
  \\
  \\[0ex]
  \text{iff} & &
  \Pr^{\max}_{\cN,s_0}(\Diamond^{\geqslant R} \final)
  \ \ > \ \ \frac{1}{2} - \frac{1}{2m}
  \\
  \\[0ex]
  \text{iff} & &
  \CExp{\max} \ \ > \ \ f\Bigl( \frac{1}{2}-\frac{1}{2m} \Bigr)
\end{eqnarray*}
This yields the PSPACE-hardness of the threshold problem
for strict bounds.
\Ende
\end{proof}

\begin{lemma}[Threshold problem in PSPACE for acyclic MDPs]
  \label{lemma:CExp-in-PSPACE}
  All four variants of the threshold problem for maximal conditional
  expectations in acyclic MDPs
  are solvable by polynomially space-bounded algorithms.
\end{lemma}

\begin{proof}
Let $\cM$ be an acyclic MDP. We sketch a polynomially space-bounded algorithm
that decides whether $\CExp{\max} \geqslant \threshold$ for some given positive
rational number $\threshold$.
The treatment of the threshold problem with a strict lower bound
``does $\CExp{\max} > \threshold$ hold?'' is analogous and omitted here.
By duality of non-strict (resp.~strict) upper bounds and strict
(resp.~non-strict) lower bounds and the closedness of PSPACE under complements
we obtain that also the problems
``does $\CExp{\max} \leqslant \threshold$ hold?'' and
``does $\CExp{\max} < \threshold$ hold?'' belong to PSPACE.

To design a (deterministic) polynomially space-bounded
algorithm for the threshold problem
``does $\CExp{\max} \geqslant \threshold$ hold?''
we reuse the same concepts as in the
algorithm for the threshold problem presented in
Section~\ref{sec:threshold} and Appendix~\ref{appendix:threshold},
but now in a recursive procedure that enumerates
only the relevant state-reward pairs $(s,r)$.

Let $\REC(s,r)$ denote a recursive procedure that takes as input a state
$s$ and a reward value $r\in \{0,1,\ldots \saturation\}$.
It returns a triple $(\alpha,y,\theta) = (\action(s,r),y_{s,r},\theta_{s,r})$
where $\alpha$ is the decision of a reward-based scheduler $\sched$ for
the state-reward pair $(s,r)$ and $y=y_{s,r}$ and $\theta =\theta_{s,r}$ 
are the corresponding
probability and expectation values, i.e.,
\begin{center}
   $y = \Pr^{\residual{\sched}{(s,r)}}_{s}(\Diamond \goal)$ \ \ and \ \
   $\theta = \Exp{\residual{\sched}{(s,r)}}{s}$.
\end{center}
The initial call is $\REC(\sinit,0)$. If the returned probability value
$y_{\sinit,0}$ is positive then the algorithm terminates with the answer
``yes'' or ``no'', depending on whether
$\theta_{\sinit,0}/y_{\sinit,0} \geqslant \threshold$.
If $y_{\sinit,0}=0$ then the algorithms terminates with the answer ``no''.

The terminal cases are calls $\REC(s,r)$ where
$s\in \{\goal,\fail\}$.
For the trap states $\goal$ and $\fail$, the returned triple consists
of a dummy action name, probability value 1 for $\goal$ and 0 for $\fail$
and expectation 0.
Recall that by \eqref{assumption:A1} there are no other trap states.%
\footnote{Recursive calls $\REC(s,\saturation)$
   where $\saturation$ is  a precomputed saturation point
   could also be treated as terminal cases that return
   the triple $(\maxsched(s),p^{\max}_s,\Exp{\maxsched}{s})$ where
   $\maxsched$ is as in Lemma \ref{lemma:Sched-max-exp}
   and $p_s^{\max} = \Pr^{\max}_s(\Diamond \goal) =
     \Pr^{\maxsched}_s(\Diamond \goal)$.
   However, the precomputation of a saturation point $\saturation$ is
   irrelevant in the acyclic case and can be omitted.}

Suppose now $s \in S \setminus \{\goal,\final\}$.
The call $\REC(s,r)$ inspects all actions
$\alpha \in \Act(s)$ and recursively calls the procedure
$\REC(t,R)$ for all $\alpha$-successors $t$ of $s$
where $R = \min \{r{+}\rew(s,\alpha),\saturation\}$.
The obtained triples $(\action(t,R),y_{t,R},\theta_{t,R})$ are used
to compute the values
\begin{eqnarray*}
     y_{s,r,\alpha} & \ = \ &
         \sum_{t\in S} P(s,\alpha,t)\cdot y_{t,R}
     \\
     \\[0ex]
     \theta_{s,r,\alpha} & = &
         \rew(s,\alpha) \cdot y_{s,r,\alpha} \ \ + \ \
         \sum_{t\in S} P(s,\alpha,t)\cdot \theta_{t,R}
\end{eqnarray*}
$\REC(s,r)$ then picks one action $\alpha \in \Act(s)$ where
 $$
   \Delta_{s,r,\alpha} \ \ \ \eqdef \ \ \
   \theta_{s,r,\alpha} - (\threshold {-} r)\cdot y_{s,r,\alpha}
 $$
 is maximal.
 If there are two or more candidate actions, it selects an
 action $\alpha$ where $y_{s,r,\alpha}$ is maximal under all
 actions in $\beta \in \Act(s)$ where $\Delta_{s,r,\beta}$ is
 maximal.
Finally, $\REC(s,r)$ returns
$(\action(s,r),y_{s,r},\theta_{s,r})=
 (\alpha,y_{s,r,\alpha},\theta_{s,r,\alpha})$.

The argument for the soundness is fairly
the same as for the algorithm presented in
Section \ref{algo:threshold-cexp}
(see
Lemma \ref{lemma:soundness-threshold-algorithm} and
Remark \ref{threshold-problem-MDP-without-zero-reward-cycles}).
It remains to show that the presented recursive approach for acyclic MDPs
is polynomially space bounded.
The recursion depth is bounded by the length of a longest path in
$\cM$ (where ``length'' refers to the number of transitions rather than the
reward values), and therefore bounded by $|S|$.
The space requirements per recursive call are polynomial in the size of $\cM$.
Thus, the overall space complexity is polynomially bounded.
\Ende
\end{proof}

\begin{remark}[PSPACE-completeness for acyclic MDPs with rational rewards]
\label{remark:PSPACE-completenss-acyclic-rational}
{\rm
The recursive algorithm presented in the proof of 
Lemma \ref{lemma:CExp-in-PSPACE} also works for MDPs where the reward
values $\rew(s,\alpha)$ are (positive or negative) rational numbers.
The asymptotic space requirements are the same.
This yields that all variants of the threshold problem for 
both maximal and minimal
conditional expectations ``does $\CExp{\max}\bowtie \threshold$ hold?''
and ``does $\CExp{\min}\bowtie \threshold$ hold?'' in MDPs with
rational rewards are PSPACE-complete.
\footnote{
The minimal conditional expectations $\CExp{\min}$ are defined
analogously to $\CExp{\max}$, i.e.,
$\CExp{\min} = \inf_{\sched} \ \CExp{\sched}$ 
where $\sched$ ranges over all
schedulers for $\cM$ satisfying the scheduler requirement
\eqref{assumption:SR}.
}
As before, $\bowtie$ is one of the comparison operators 
$<$, $\leqslant$, $>$ or $\geqslant$.
Note that PSPACE-completeness for minimal expectations is a consequence of the
PSPACE-completeness of the threshold problems for maximal expectations
in acyclic MDPs with rational rewards as we can 
multiply all reward values $\rew(s,\alpha)$ with value ${-}1$.
Later in Corollary \ref{corollary:Cmin-PSPACE-completeness for acyclic MDPs} 
we will show that PSPACE-hardness
of the threshold problem for minimal expectations even holds for
acyclic MDPs with non-negative rewards.
\Ende
 }
\end{remark}

\section{Rational and negative rewards}
\label{appendix:rational}

The algorithm presented for the computation of conditional expectations
crucially relies on the assumption that all rewards are
non-negative integer values.
However, MDPs with rational rewards can be easily transformed into
MDPs with integer rewards of the same size. This transformation together
with the presented algorithm for MDPs with non-negative integer rewards
yields a method for computing conditional expectations in MDPs with
non-negative rational rewards.
See Section \ref{section:rational-rewards}.
The treatment of MDPs with positive and negative rewards appears to be harder.
For the case of acyclic MDPs we provide a polynomial transformation to MDPs
with non-negative rewards. The blow-up of the transformed MDP is at most
quadratic.
See Section \ref{section:negative-rewards}.

\subsection{From rational rewards to integer rewards}
\label{section:rational-rewards}
\label{appendix:rational-rewards}

Given an MDP $\cM$ with reward function
$\rew_{\cM} : S \times \Act \to \Rational$,
we take the least common multiple $M$ of the denominators of the values
$\rew_{\cM}(s,\alpha)$ with $s\in S$ and $\alpha \in \Act(s)$.
Let now $\cM'$ be the MDP that results from $\cM$ when the reward function
is replaced with
$$
   \rew_{\cM'}(s,\alpha) \ \ = \ \ M \cdot \rew_{\cM}(s,\alpha)
$$
Then, $\rew_{\cM'}(s,\alpha)\in \Integer$ for all state-action pairs
$(s,\alpha)$ and
$\rew_{\cM'}(\fpath) =M \cdot \rew_{\cM}(\fpath)$ for each finite path
$\fpath$ in $\cM$ resp.~$\cM'$. Hence:
$$
  M \cdot
  \CExpState{\sched}{\cM,\sinit}%
   \ \ \ = \ \ \
  \CExpState{\sched}{\cM',\sinit}%
$$
for each scheduler $\sched$ of $\cM$ resp.~$\cM'$.
In particular:
$$
  \CExpState{\max}{\cM,\sinit}%
   \ \ \ = \ \ \
  \frac{1}{M}
  \cdot
  \CExpState{\max}{\cM',\sinit}%
$$
and the analogous statement for minimal conditional expectations.
Obviously, if all rewards in $\cM$ are non-negative, then so are the rewards
in $\cM'$.
Thus, maximal conditional expectations in MDPs 
with non-negative rational rewards
are computable in exponential time and the corresponding
threshold problem is PSPACE-hard.

\subsection{From integer rewards to non-negative integer rewards
   for acyclic MDPs}
\label{section:negative-rewards}
\label{appendix:negative-rewards}

We now explain how to transform a given acyclic MDP $\cM$ with
a reward function $\rew : S \times \Act \to \Integer$
(that might have negative values)
into an acyclic MDP $\cM'$ with non-negative rational rewards of size
$\cO\bigl( \ \mathit{size}(\cM)^2 \ \bigr)$.
The transformation proceeds in two phases. The first phase transforms
$\cM$ into a layered acyclic MDP $\cM_1$ where all transitions from
a state at layer $i$ move to a state at layer $i{+}1$.
Thus, all maximal paths in $\cM_1$ have the same length.
The MDP $\cM_1$ might still contain negative rewards.
The size of $\cM_1$ is polynomially (quadratically)
bounded in the size of $\cM_1$ and
$\cM$ and $\cM_1$ are equivalent with respect to the random variable
that assigns the accumulated reward to the paths from $\sinit$ to $\goal$ or
$\fail$.
The second phase then replaces $\cM_1$ with an MDP $\cM_2$ that has the same
graph structure as $\cM_1$ (i.e., is also layered) and has non-negative
rewards.

\tudparagraph{1.5ex}{{\it Phase 1: Construction of a layered MDP.}}
Let $s_0,s_1,\ldots,s_N,s_{N+1}$ be a topological sorting of the
states in $\cM$, i.e., if there is a transition from $s_i$ to $s_j$ then
$i<j$ where $s_N=\goal$, $s_{N+1} =\fail$ (and $s_0=\sinit$).
We extend $\cM$ to an MDP $\cM_1$ with $N{+}1$ layers
such that all transitions for any state of layer $i$ lead to state of
layer $i{+}1$. Formally, the state space of $\cM_1$ is
$S_{\cM_1} = L_0 \cup L_1 \cup \ldots \cup L_N$
where $L_i$ stands for the states at layer $i$.
We have $L_0=\{\sinit\}$, $L_N = \{\goal,\fail\}$ and for $1 \leqslant i < N$:
$$
  L_i \ \ = \ \
  \{s_i\} \ \cup \
  \bigl\{\ \extstate{i}{k}{\alpha}{j} \ : \
           0 \leqslant i < j \leqslant N{+}1
  \bigr\}
$$
(The states $t_{N,N+1}$ are not needed and could be dropped.)
Intuitively, state $\extstate{i}{k}{\alpha}{j}$ stands for an
intermediate pseudo-state at layer $i$ that is visited when firing the
transition $s_k \overto{\alpha} s_j$ in $\cM$. More precisely, the transition
$s_k \overto{\alpha} s_j$ in $\cM$ with $k{+}1<j$ is mimicked by the
path
$$
  s_k \ \ \alpha \ \ \extstate{k{+}1}{k}{\alpha}{j}
  \ \ \tau \ \ \extstate{k{+}2}{k}{\alpha}{j} \ \ \tau \ \ \ldots
  \ \ \tau \ \ \extstate{j{-}1}{k}{\alpha}{j} \ \ \tau \ \ s_j
$$
in $\cM_1$. The behaviour in the states $\extstate{i}{k}{\alpha}{j}$
is deterministic and the unique successor of
$\extstate{i}{k}{\alpha}{j}$ is
$\extstate{i{+}1}{k}{\alpha}{j}$ if $i < j{-}1$
and $s_j$ if $i=j{-}1$.
Formally, the action set of $\cM_1$ is $\Act \cup \{\tau\}$ and
\begin{center}
  $\Act_{\cM_1}(s_i) = \Act_{\cM}(s_i)$
  \quad and \quad
  $\Act_{\cM_1}(\extstate{i}{k}{\alpha}{j}) = \{\tau\}$ \ if $k < i < j$
\end{center}
For the initial state $s_0=\sinit$ we have
$\Act_{\cM_1}(s_0)=\Act_{\cM}(s_0)$,
while $\Act_{\cM_1}(\goal) = \Act_{\cM_1}(\fail)=\varnothing$.
The reward function of $\cM_1$ is defined by:
\begin{center}
   $\rew_{\cM_1}(s_i,\alpha)=\rew(s_i,\alpha)$ \quad and \quad
   $\rew_{\cM_1}(\extstate{i}{k}{\alpha}{j},\tau) = 0$ if $k < i< j$.
\end{center}
The transition probabilities are defined as follows.
Let $i,k\in \{0,1,\ldots,N{-}1\}$ and $j\in \{1,\ldots,N\}$.
\begin{center}
 \begin{tabular}{l}
  \begin{tabular}{l}
    $P_{\cM_1}(s_i,\alpha,s_{i+1}) \ = \ P_{\cM}(s_i,\alpha,s_{i+1})$
  \end{tabular}
  \\[1.5ex]

  \begin{tabular}{ll}
      $P_{\cM_1}(s_i,\alpha,\extstate{i{+}1}{i}{\alpha}{j}) \  =  \
       P_{\cM}(s_0,\alpha,s_j)$
      & if $j > i{+}1$
    \\[1.5ex]

    $P_{\cM_1}(\extstate{i}{k}{\alpha}{j}, \tau,
            \extstate{i{+}1}{k}{\alpha}{j})
     \ = \ 1$
    & if $j > i{+}1$
    \\[1ex]

    $P_{\cM_1}(\extstate{j{-}1}{k}{\alpha}{j}, \tau, s_j) \ = \ 1$
  \end{tabular}
 \end{tabular}
\end{center}
The transition probabilities of the incoming transitions of state
$s_{N+1}=\fail$
are defined accordingly (recall that layer $N$ consists of the two states
$s_N=\goal$ and $s_{N+1}=\fail$):
\begin{center}
 \begin{tabular}{l@{\hspace*{0.2cm}}l}

    $P_{\cM_1}(s_{N-1},\alpha,\fail) \ = \ P_{\cM}(s_{N-1},\alpha,\fail)$
    \\[1.5ex]

    $P_{\cM_1}(s_{i},\alpha,\extstate{i{+}1}{k}{\alpha}{N{+}1})
    \ = \ P_{\cM}(s_i,\alpha,\fail)$
    & if $i < N{-}1$
    \\[1.5ex]

    $P_{\cM_1}(\extstate{i}{k}{\alpha}{N{+}1},\tau,
               \extstate{i{+}1}{k}{\alpha}{N{+}1})
    \ = \ 1$
    & if $i < N{-}1$
    \\[1.5ex]

    $P_{\cM_1}(\extstate{N{-}1}{k}{\alpha}{N{+}1},\tau,\fail) \ = \ 1$
  \end{tabular}
\end{center}
and $P_{\cM_1}(\cdot)=0$ in all remaining cases.
This construction ensures that
$s \in L_i$ and $P_{\cM_1}(s,\alpha,t) > 0$ implies $t\in L_{i+1}$.
Thus, $\cM_1$ is indeed layered.
In particular, $\cM_1$ is acyclic and all maximal paths
from $s_0=\sinit$ to $\goal$ or $\fail$ have the length $N$.

Clearly, the size of $\cM_1$ is quadratic in the size of $\cM$.
There is a one-to-one correspondence between the schedulers in
$\cM_1$ and $\cM$ and
$$
  \Pr^{\sched}_{\cM_1,\sinit}(\varphi) \ \ = \ \
  \Pr^{\sched}_{\cM,\sinit}(\varphi)
$$
for each scheduler $\sched$ and stutter-invariant measurable
property $\varphi$, where we suppose a labeling function for $\cM_1$
such that the labels of the fresh states $\extstate{i}{k}{\alpha}{j}$
agree with the label of state $s_k$.
In particular, 
$\CExpState{\sched}{\cM_1,\sinit}
 = 
 \CExpState{\sched}{\cM,\sinit}
$
for each scheduler $\sched$ for $\cM$ resp.~$\cM_1$.
Therefore:
\begin{center}
  $\CExpState{\max}{\cM_1,\sinit}
   = \
   \CExpState{\max}{\cM,\sinit}$
  \ \ \ and \ \ \
  $\CExpState{\min}{\cM_1,\sinit}
  = \ 
 \CExpState{\min}{\cM,\sinit}$
\end{center}

\tudparagraph{1.5ex}{{\it Phase 2: from integer rewards in
   layered MDPs to non-negative integer rewards.}}
Let
$$
   Y \ \ = \ \
   - \min \
   \bigl\{ \ \rew_{\cM_1}(s,\alpha) \ : \
             s\in S_{\cM_1}, \ \alpha \in \Act_{\cM_1}(s) \
   \bigr\}
$$
where we suppose that $\cM_1$ contains indeed negative reward values.
Then:
$$
  \rew_{\cM_1}(s,\alpha) + Y  \ \ \geqslant \ \ 0
$$
for all states $s$ in $\cM_1$ and $\alpha \in \Act_{\cM_1}(s)$.
We define $\cM_2$ as the MDP that results from $\cM_1$ by replacing the
reward function with
$$
  \rew_{\cM_2}(s,\alpha) \ \ = \ \ \rew_{\cM_1}(s,\alpha) + Y
$$
Clearly, $\cM_1$ and $\cM_2$ have the same paths and
the same schedulers.
As  all maximal paths in $\cM_1$ have the same length $N$,
we have
$$
  \rew_{\cM_2}(\fpath) \ \ = \ \
  \rew_{\cM_1}(s,\alpha) \ + \ N \cdot Y
$$
for each path from $\sinit$ to $\goal$.
This yields
$\CExpState{\sched}{\cM_2,\sinit}
 = 
 \CExpState{\sched}{\cM_1,\sinit}$
for each scheduler $\sched$ for $\cM_1$ and $\cM_2$.
Hence:
$$
 \begin{array}{lcl}
    \CExpState{\max}{\cM_2,\sinit}
    & \ = \ \ &
    \CExpState{\max}{\cM_1,\sinit}
      \ + \ N \cdot Y
  \\[2ex]
    \CExpState{\min}{\cM_2,\sinit}
    & \ = \ \ &
    \CExpState{\min}{\cM_1,\sinit}
      \ + \ N \cdot Y
 \end{array}
$$

\begin{corollary}
  \label{corollary:Cmin-PSPACE-completeness for acyclic MDPs}
  All four variants of the
  threshold problem for minimal conditional expectations 
  (``does $\CExpState{\min}{\cM,\sinit} \bowtie \ \threshold$ hold ?''
  where $\bowtie \ \in \{>,\geqslant,<,\leqslant\}$)
  in acyclic MDPs with rational rewards are PSPACE-complete.
  PSPACE-hardness even holds for 
  acyclic MDPs with non-negative 
  integer rewards.
\end{corollary}

\begin{proof}
Membership to PSPACE is a consequence of
Remark \ref{remark:PSPACE-completenss-acyclic-rational}.
For the PSPACE-hardness, we provide a reduction from the threshold
problems for maximal conditional expectations in acyclic MDPs
with non-negative integer rewards
(see Theorem \ref{theorem:PSPACE-complete-CExp-acyclic}).
Given an acyclic MDP $\cM$, let $\cM^-$ denote the MDP arising from $\cM$
by multiplying all rewards with ${-}1$. Obviously:
$$
 \begin{array}{lcl}
   \CExpState{\max}{\cM,\sinit}
   & \ \ = \ \ \ &
   - \, \CExpState{\min}{\cM^{-},\sinit}
 \end{array}
$$
We now apply the above approach to transform $\cM^-$ into a layered
MDP $\cM^-_2$ with non-negative integer rewards such that
$$
 \begin{array}{lcl}
   \CExpState{\min}{\cM^{-},\sinit}
   & \ \ = \ \ \ &
   \CExpState{\min}{\cM^{-}_2,\sinit}
   - \ C 
 \end{array}
$$
where $C$ is the constant $N \cdot Y$ with $N$ and $Y$ being as above.
Hence:
$$
  \CExpState{\max}{\cM,\sinit} > \threshold
  \quad \text{iff} \quad
  \CExpState{\min}{\cM^{-},\sinit} < -\threshold
  \quad \text{iff} \quad
  \CExpState{\min}{\cM^{-}_2,\sinit} < \ C-\threshold
$$
and analogous statements for other comparison operators.
\Ende
\end{proof}

\section{Implementation and experimental results}

\label{appendix:implementation}

We have extended the popular model checker PRISM~\cite{prism40,PrismWebSite} by a
prototypical implementation of the algorithms presented in this paper
to facilitate initial experiments.%
\footnote{%
We would like to thank Steffen M\"arcker for his work
on the infrastructure and algorithms
presented in~\cite{BKKM14} in PRISM.
For the implementation and additional information on the
experiments (performed on
a computer with two Intel Xeon L5630 4-core CPUs at 2.13GHz, no Turbo,
192GB RAM,
running Linux)
see
\mbox{\url{https://wwwtcs.inf.tu-dresden.de/ALGI/PUB/TACAS17/}}
}
Our implementation supports checking for finiteness, computing an
upper bound and saturation point and solving the
threshold problems as well as the computation of $\CExp{\max}$. It is
currently limited to the case $F=G$ and does not yet support MDPs with
zero-reward cycles in the threshold/scheduler computation
phases. The latter allows to avoid having to solve a linear program
for each level in the threshold algorithm. The implementation uses
both the symbolic, MTBDD-based engine as well as the explicit engine
of PRISM: in the first phase of checking finiteness of $\CExp{\max}$,
in the computation of an upper bound and for computing $\cM$ the
symbolic engine is used.
This phase relies heavily on multiple model transformations,
e.g., a newly introduced symbolic implementation
for collapsing end components.
For the second phase, i.e., the threshold algorithm and the scheduler
improvement algorithm, we convert the symbolically represented MDP
and related information
(rewards, $\cM$ and the probabilities and expected rewards for
reaching $\goal$ in $\cM$) into an explicit representation which
facilitates an easy manipulation during those algorithms. In future
work, we are however interested in an ``explicit'' implementation of
the first phase and a ``symbolic'' implementation of the second phase,
as that might be beneficial in certain cases in practice as well.
We have also implemented the specialized algorithm for
the threshold problem for acyclic MDPs
(Appendix~\ref{appendix:PSPACE}),
using a computed table to avoid
identical recursions.
We use a double-precision floating point representation for the
numerical values and the approximative computations provided by PRISM.
It would be desirable to use exact representations and computations, which remains future work.%
\footnote{%
  \cite{HadMon14} raised issues
  with the termination criterion commonly
  used in value iteration computations.
  In separate work, we are implementing their proposed improvement,
  which can then be applied to our implementation for better precision.}

\begin{table}[t]
\caption{Statistics for selected experiments: number of states of the
  original MDP $\cM$ and of the ``cleaned-up'' MDP $\hat{\cM}$ used in
  the second phase; the value $R$ used in the upper
  bound computation, the saturation point $\saturation$, the lower
  ($\CExp{\maxsched}$) and upper ($\CExp{\ub})$ bound on
  $\CExp{\max}$ and the computed value $\CExp{\max}$;
  computation time $t$ (in seconds), of which
  $t_1$ for the first computation phase (finiteness,
  computation of bounds and $\maxsched$) and $t_2$ for the scheduler
  improvement algorithm; number of calls to the threshold algorithm.}
\label{tab:results}
\centering
\begin{tabular}{l|r|r|r|r|r|r|r|r|r|r|r}
 & \tableCL{$\cM$}
 & \tableCL{$\hat{\cM}$}
 & \tableCL{R}
 & \tableCL{$\saturation$}
 & \tableCL{$\CExp{\maxsched}$}
 & \tableCL{$\CExp{\ub}$}
 & \tableCL{$\CExp{\max}$}
 & \tableCL{$t$}
 & \tableCL{$t_1$}
 & \tableCL{$t_2$}
 & \tableC{\,calls}\\\hline
\multicolumn{11}{l}{\consensus{} case study}\\\hline
N=2,K=2\, & 
 \numStates{272} & 
 \numStates{189} & 
 187 &
 272 &
 \numValue{56.00} &  %
 \numValue{87.97} &  %
 \numValue{75.10} &  %
 \numSec{1.24} &  %
 \numSec{0.84} &  %
 \numSec{0.22} &  %
 7 \\\hline
N=2,K=8 & 
 \numStates{1040} & 
 \numStates{765} & 
 763 &
 3763 &
 \numValue{799.57} &  %
 \,\numValue{1491.18} &  %
 \numValue{867.30} &  %
 \numSec{60.59} &  %
 \numSec{53.33} &  %
 \numSec{6.52} &  %
 8 \\\hline
N=3,K=3 & 
 \numStates{3968} & 
 \,\numStates{2292} & 
 \,2290 &
 \,1279 &
 \,\numValue{278.95} &  %
 \,\numValue{364.18} &  %
 \,\numValue{363.46} &  %
 \,\numSec{308.02} &  %
 \,\numSec{301.90} &  %
 \,\numSec{3.69} &  %
 3 \\\hline
N=3,K=4 & 
 \numStates{5216} & 
 \numStates{3036} & 
 3034 &
 2097 &
 \numValue{479.84} &  %
 \numValue{594.97} &  %
 \numValue{588.56} &  %
 \numSec{710.82} &  %
 \numSec{699.26} &  %
 \numSec{8.43} &  %
 3 \\\hline
\multicolumn{11}{l}{\wlan{} case study}\\\hline
b=2,k=2 & 
 \numStates{28598} & 
 \numStates{821} & 
 290 &
 68 &
 \numValue{32.00} &  %
 \numValue{40.00} &  %
 \numValue{40.00} &  %
 \numSec{91.05} &  %
 \numSec{7.50} &  %
 \numSec{0.19} &  %
 4 \\\hline
b=2,k=3 & 
 \,\numStates{35197} & 
 \numStates{2435} & 
 844 &
 152 &
 \numValue{67.34} &  %
 \numValue{92.00} &  %
 \numValue{92.00} &  %
 \numSec{96.21} &  %
 \numSec{16.95} &  %
 \numSec{0.81} &  %
 4 \\\hline
\end{tabular}
\vspace{-2ex}
\end{table}

\tudparagraph{1ex}{\it Experiments.}
We have carried out experiments with adapted
case studies from the PRISM benchmark suite~\cite{KwiatkowskaNP12},
with Table~\ref{tab:results} showing statistics and results for
selected instances.
We performed experiments with the
\consensus{} (parameters: number of processors $N$ and factor $K$
influencing the range of the random walk)
and \wlan{} (parameters: backoff counter maximum $b$ and number of
collisions $k$).
For \consensus{}, we consider the
maximal conditional expectation for the number of steps
with state set $F=G$ being ``finished
and all coins are 1''. For \wlan{}, we consider the
maximal conditional expectation for the accumulated time
with state set $F=G$ being ``$k$ collisions have occurred''.
Originally, each time step had a reward of $50$, which we
rescaled to $1$. In general, scaling the rewards by dividing by the
greatest common divisor of all the reward values and rescaling of the
result is beneficial for performance in our setting.

As can be seen, the time for the first phase tends to dominate. This
is mostly due to the upper bound computation (both building the
symbolic reward counter product and the unconditional
reward computation), dealing with an MDP of size
$\mathcal{O}(\Size(\hat{\cM}) \cdot R)$.
The scheduler-improvement algorithm 
(where the saturation point $\saturation$
provides the scheduler memory
requirements per state of $\hat{\cM}$)
has to call the threshold algorithm only a few times, as $\saturation$
 is of reasonable size and the optimal scheduler is discovered
after optimizing just a few top levels.
For the \wlan{} case study, the total computation times are
artificially inflated due to an inefficient conversion between the
symbolic and explicit MDP, which we will fix in the next version of
the implementation. Overall, these first experiments are encouraging,
indicating that improvements in the computation of an upper bound
would be particularly worthwhile.
As it can be the case that $\CExp{\ub} = \CExp{\max}$, as for the
\wlan{} results, an additional heuristic would be to prepend a first
threshold check for threshold $\threshold = \CExp{\ub}$
to catch these cases without having to start the full
scheduler-improvement algorithm.

\end{document}